\documentclass[11pt,reqno,twoside]{book}

\usepackage{fixltx2e} 

\usepackage{cmap} 

\usepackage[T1]{fontenc}
\usepackage[utf8]{inputenc}
\usepackage{graphicx}
\usepackage{placeins}

\usepackage{verbatim}

\usepackage[utf8]{inputenc}
\usepackage{imakeidx}
\makeindex[columns=2, title=Alphabetical Index, 
           options= -s example_style.ist]
           \usepackage[nottoc]{tocbibind}

\usepackage{setspace}

\let\counterwithin\relax  
\usepackage{lmodern} 
\usepackage[scale=0.88]{tgheros} 

\usepackage{bm} 

\usepackage{bbold}

\usepackage{amsmath,amsbsy,amsgen,amscd,amsthm,amsfonts,amssymb,thmtools} 

\usepackage[centering,top=1.5in,bottom=1.2in,left=1.4in,right=1.4in]{geometry}

\usepackage{titling}
\usepackage{musicography}
\usepackage[sf,bf,compact]{titlesec}

\usepackage{booktabs,longtable,tabu} 
\usepackage[font=small,margin=12pt,labelfont={sf,bf},labelsep={space}]{caption}

\usepackage{enumitem}
\setitemize{itemsep=0pt} 
\setenumerate{itemsep=0pt}
\setlist[1]{labelindent=\parindent,
 }

\newcommand{\su}{\mathsf{u}}

\usepackage[usenames,dvipsnames]{xcolor}

\definecolor{dark-gray}{gray}{0.3}
\definecolor{dkgray}{rgb}{.4,.4,.4}
\definecolor{dkblue}{rgb}{0,0,.5}
\definecolor{medblue}{rgb}{0,0,.75}
\definecolor{rust}{rgb}{0.5,0.1,0.1}

\usepackage{url}
\usepackage[colorlinks=true]{hyperref}
\hypersetup{linkcolor=dkblue}    
\hypersetup{citecolor=rust}      
\hypersetup{urlcolor=rust}     

\usepackage[final]{microtype} 

\newtheoremstyle{myThm} 
    {\topsep}                    
    {\topsep}                    
    {\itshape}                   
    {}                           
    {\sffamily\bfseries}                   
    {.}                          
    {.5em}                       
    {}  

\newtheoremstyle{myRem} 
    {\topsep}                    
    {\topsep}                    
    {}                   
    {}                           
    {\sffamily}                   
    {.}                          
    {.5em}                       
    {}  

\newtheoremstyle{myDef} 
    {\topsep}                    
    {\topsep}                    
    {}                   
    {}                           
    {\sffamily\bfseries}                   
    {.}                          
    {.5em}                       
    {}  

\theoremstyle{myThm}
\newtheorem{theorem}{Theorem}[chapter]
\newtheorem{lemma}[theorem]{Lemma}
\newtheorem{proposition}[theorem]{Proposition}
\newtheorem{corollary}[theorem]{Corollary}

\newtheorem{assumption}[theorem]{Assumption}
\newtheorem{algorithm}[theorem]{Algorithm}

\theoremstyle{myRem}
\newtheorem{remarkx}[theorem]{Remark}
 \newenvironment{remark}
  {\pushQED{\qed}\remarkx}
  {\popQED\endremarkx}

\theoremstyle{myDef}
\newtheorem{definitionx}[theorem]{Definition}
\newenvironment{definition}
  {\pushQED{\qed}\definitionx}
  {\popQED\enddefinitionx}

 \newtheorem{examplex}[theorem]{Example}
 \newenvironment{example}
  {\pushQED{\qed}\examplex}
  {\popQED\endexamplex}

   \newtheorem{observationx}[theorem]{Observation}
 \newenvironment{observation}
  {\pushQED{\qed}\observationx}
  {\popQED\endexamplex}

\usepackage{fancyhdr}
\usepackage{nopageno} 
\let\originalleft\left
\let\originalright\right
\renewcommand{\left}{\mathopen{}\mathclose\bgroup\originalleft}
\renewcommand{\right}{\aftergroup\egroup\originalright}

\usepackage{mathtools}
\mathtoolsset{centercolon}  

\renewcommand{\phi}{\varphi}
\newcommand{\eps}{\varepsilon}

\newcommand{\zerovct}{\vct{0}} 
\newcommand{\Id}{{I}} 

\providecommand{\mathbbm}{\mathbb} 

\newcommand{\R}{\mathbbm{R}}

\newcommand{\N}{\mathbbm{N}}
\newcommand{\Z}{\mathbbm{Z}}
\newcommand{\la}{\langle}
\newcommand{\ra}{\rangle}

\definecolor{mygreen}{rgb}{0.1,0.75,0.2}

\newcommand{\as}{\color{mygreen}}
\newcommand{\nc}{\normalcolor}

\newcommand{\dz}{d_z}

\newcommand{\pr}{\rho}
\newcommand{\post}{\pi}
\newcommand{\noise}{\nu}
\newcommand{\loss}{\mathsf{L}}
\newcommand{\losss}{\ell}
\newcommand{\like}{\mathsf{l}}
\newcommand{\reg}{\mathsf{R}}
\newcommand{\precision}{\Omega}
\newcommand{\V}{V}
\newcommand{\Y}{Y}
\newcommand{\varg}{v}

\newcommand{\pMC}{\post^N_{\mbox {\tiny{\rm MC}}}}
\newcommand{\prMC}{\pr^N_{\mbox {\tiny{\rm MC}}}}
\newcommand{\ppCN}{\post^N_{\mbox {\tiny{\rm pCN}}}}
\newcommand{\pIS}{\post^N_{\mbox {\tiny{\rm IS}}}}
\newcommand{\pMCMC}{\post^N_{\mbox {\tiny{\rm MCMC}}}}

\newcommand{\mpr}{\hat{m}}
\newcommand{\mpost}{m}
\newcommand{\Cpr}{\hat{C}}
\newcommand{\Cpost}{C}

\newcommand{\Prob}{\operatorname{\mathbbm{P}}}
\newcommand{\Expect}{\operatorname{\mathbb{E}}}
\newcommand{\Sam}{N}
\newcommand{\sam}{n}
\newcommand{\pdfs}{{\mathscr{P}}}

\newcommand{\Du}{d}
\newcommand{\du}{i}
\newcommand{\Dy}{k}

\newcommand{\Ru}{\mathbbm{R}^\Du}
\newcommand{\Ry}{\mathbbm{R}^\Dy}

\newcommand{\vct}[1]{{#1}}
\newcommand{\mtx}[1]{{#1}}

\newcommand{\umap}{u_{\mbox {\tiny{\rm MAP}}}}
\newcommand{\upm}{u_{\mbox {\tiny{\rm PM}}}}

\newcommand{\pmh}{p_{\mbox {\tiny{\rm MH}}}}
\newcommand{\pimh}{\pi_{\mbox {\tiny{\rm MH}}}}

\newcommand{\mA}{A^{-1}}

\newcommand{\cN}{\mathcal{N}}

\newcommand{\dkl}{d_{\mbox {\tiny{\rm KL}}}}

\newcommand{\dtv}{d_{\mbox {\tiny{\rm TV}}}}
\newcommand{\dchi}{d_{\mbox {\tiny{$ \chi^2$}}}}
\newcommand{\dhell}{d_{\mbox {\tiny{\rm H}}}}
\newcommand{\Nc}{\mathcal{N}}

\usepackage[]{algorithm}

\usepackage{graphicx, subfig}
\usepackage{algorithmic}
\usepackage{authblk}
\usepackage[square,numbers]{natbib}
\usepackage{chngcntr}
\usepackage{mathrsfs}
\counterwithin{table}{chapter}
\counterwithin{algorithm}{chapter}

\title{{\huge\emph{Inverse Problems and Data Assimilation}}\\
} 

\author{\vspace{.55in}  Daniel Sanz-Alonso $^{\musFlat{}},$ Andrew Stuart $^{\musSharp{}}$ and Armeen Taeb $^{\musNatural{}}$\\
	\vspace{.35in} 
	$^{\musFlat{}}$ Department of Statistics, University of Chicago \vspace{.05in} 
	\\ $^{\musSharp{}}$ Department of Computing and Mathematical Sciences, Caltech \vspace{.05in} 
	\\  $^{\musNatural{}}$ Department of Statistics, University of Washington}

\vspace{.25in} 

\date{}

\makeatletter\@addtoreset{section}{part}\makeatother%
\numberwithin{equation}{chapter}

\renewcommand{\triangleq}{:=}
\newcommand{\upperRomannumeral}[1]{\uppercase\expandafter{\romannumeral#1}}

\renewcommand{\hat}{\widehat}
\renewcommand{\bar}{\overline}

\newcommand{\An}{{\mathcal{A}}}
\newcommand{\Pred}{{\mathcal{P}}}

\newcommand{\J}{{\mathsf{J}}}
\newcommand{\I}{{\mathsf{I}}}

\newcommand{\Jtp}{\J_{\mbox {\tiny{\rm TP}}}}
\newcommand{\Jdm}{\J_{\mbox {\tiny{\rm DM}}}}

\newcommand{\Jl}{\J_\ell^{\mbox {\tiny{lin}}}}
\newcommand{\Juc}{\J_\ell^{\mbox {\mbox{\tiny{UC} }  }}}
\newcommand{\Jtpl}{\J_{\mbox{\tiny{\rm TP,}}\ell}^{\ell} }

\newcommand{\Juctp} {\J_{\mbox {\tiny{\rm TP}},\ell}^{\mbox {\mbox{\tiny{UC} }  }}}
\newcommand{\Jucdm} {\J_{\mbox {\tiny{\rm DM}},\ell}^{\mbox {\mbox{\tiny{UC} }  }}}
\newcommand{\Jucdmn} {\J_{\mbox {\tiny{\rm DM}},\ell}^{(n),\mbox {\mbox{\tiny{UC} }  }}}
\newcommand{\Jtpn}{\J_{\mbox {\tiny{\rm TP}},\ell}^{(n)}}
\newcommand{\Juctpn} {\J_{\mbox {\tiny{\rm TP}},\ell}^{(n),\mbox {\mbox{\tiny{UC} }  }}}

\newcommand{\rtp}{r_{\mbox {\tiny{\rm TP}}}}
\newcommand{\rdm}{r_{\mbox {\tiny{\rm DM}}}}

\newcommand{\uin}{u_\ell^{(n)}}

\newcommand{\vin}{v_\ell^{(n)}}

\begin{document}
\maketitle 

\newpage
\section*{\Large{\sffamily{Introduction}}}


\subsection*{Aim and Overview of the Notes}

The aim of these notes is to provide a clear and concise 
 mathematical introduction
to the subjects of \index{inverse problem}Inverse Problems and \index{data assimilation}Data Assimilation, and their
inter-relations, together with bibliographic pointers to literature
in this area that goes into greater depth.  The target audiences
are advanced undergraduates and beginning graduate students in the
mathematical sciences, together
with researchers in the sciences and engineering who are interested
in the systematic underpinnings of methodologies widely used in their
disciplines.

In its most basic form, \index{inverse problem}inverse problem theory is the study of how to estimate model parameters from data. Often the data provide indirect information about these parameters, corrupted by noise. The theory of \index{inverse problem}inverse problems, 
however, is much richer than just parameter estimation. For example, 
the underlying theory can be used to
determine the effects of noisy data on the accuracy of the solution; it can 
be used to determine what kind of \index{observation}observations are needed to accurately 
determine a parameter; and it can be used to study the uncertainty in a
parameter estimate and, relatedly, is useful, for example,
in the design of strategies for control or \index{optimization}optimization under
uncertainty, and for risk analysis. The theory
thus has applications in many fields of science and engineering.

To apply the ideas in these notes, the starting point is a mathematical model 
mapping the unknown parameters to the \index{observation}observations: termed the \index{forward!problem}``forward'' or 
``direct'' problem, and often a subject of research in its own right. A good \index{forward!model}forward model will not only identify how the data is dependent on parameters, but 
also what sources of noise or model uncertainty are present in the postulated
relationship between unknown parameters and data. For example, if  the desired
\index{forward!problem}forward problem cannot be solved analytically, then the \index{forward!model}forward model may be approximated by a numnerical simulation; 
in this case, discretization may be considered as a source of error. Once a relationship between model parameters, sources of error, and data is clearly defined, the \index{inverse problem}inverse problem of estimating parameters from data can be addressed. The theory of
\index{inverse problem}inverse problems can be separated into two cases: (1) the ideal case where data is not corrupted by noise and is derived from a known perfect model; and (2) the practical case where data is incomplete and imprecise. The first case is useful for classifying \index{inverse problem}inverse problems and determining if a given set of \index{observation}observations
can, in principle, allow to fully reconstruct the model parameters; this provides insight into
conditions needed for existence, uniqueness, and stability of a solution to the \index{inverse problem}inverse problem. 
The second case is useful for the formulation of practical algorithms
to learn about parameters, and uncertainties in their estimates, and will
be the focus of these notes.

A model for which a solution exists, is unique, and changes continuously with input (stability) is termed  \index{well-posed}``well-posed''.  
Conversely, a model lacking any of these properties is termed \index{ill-posed}``ill-posed''. \index{ill-posed}Ill-posedness is present in many \index{inverse problem}inverse problems,
and mitigating it is an extensive part of the subject.
Out of the different approaches to formulating an 
\index{inverse problem}inverse problem, 
our notes emphasize the \index{Bayesian}Bayesian framework.  Nonetheless,
practical algorithms in this area include a variety of related optimization 
approaches, and these are  also   discussed in detail.

The goal of the \index{Bayesian}Bayesian framework is to find a probability measure that assigns a probability to each possible solution for a parameter $u$, given the 
data $y$. \index{Bayes formula}Bayes formula states that
\begin{equation*}
\Prob(u \vert y) = \frac{1}{\Prob(y)} \Prob(y \vert u)\Prob(u).
\end{equation*}
 This formula 
enables calculation of the \index{posterior}posterior probability on 
$u\vert y$, $\Prob(u \vert y) $, in terms of the  
product of the data \index{likelihood}likelihood $\Prob(y \vert u)$ and the 
\index{prior}prior information on the parameter encoded in $\Prob(u)$. The \index{likelihood}likelihood
describes the probability of the observed data $y$ if the input parameter
were set to be $u$; it is determined by the \index{forward!model}forward model, and the
structure of the noise. The normalization constant $\Prob(y)$ ensures that
$\Prob(u \vert y)$  is a probability measure. There are four
primary benefits to this framework:  (1)  it provides a clear theoretical 
setting in which the \index{forward!model}forward model choice, the description of how noise 
enters the data and the \index{forward!model}forward model,
and \emph{a \index{prior}priori} information on the unknown parameter are all explicit; 
 (2)  it provides information about the entire solution space for possible
input parameter choices;  (3)  it naturally leads to quantification of 
uncertainty and risk in parameter estimates;  (4)   it is generalizable to a wide class of \index{inverse problem}inverse problems, in finite and infinite dimension, and comes
with a \index{well-posed}well-posedness  theory 
mitigating the ill-posedness of a naive deterministic approach.

The first part of the notes is dedicated to studying the \index{Bayesian}Bayesian framework for \index{inverse problem}inverse problems. Techniques such as \index{importance sampling}importance sampling and \index{Markov chain!Monte Carlo}\index{MCMC}Markov 
Chain Monte Carlo (MCMC) methods are introduced; these methods have the
desirable property that in the limit of an infinite number of samples
they reproduce the full \index{posterior}posterior distribution.  Since it is often 
computationally intensive to implement these methods, especially in 
high-dimensional problems, techniques to approximate
the \index{posterior}posterior by a \index{Dirac}Dirac or a \index{Gaussian}Gaussian distribution are  also  discussed,
along with related \index{optimization}optimization algorithms to determine the best
approximation.

The second part of the notes covers \index{data assimilation}data assimilation. 
This refers to a particular class of \index{inverse problem}inverse problems in which
the unknown parameter is the initial condition of a \index{dynamical system}dynamical system or,
in the case of \index{dynamics!stochastic}stochastic dynamics, the entire sequence of subsequent 
\index{state}states of the system, and the data comprises partial and noisy \index{observation}observations 
of the (possibly stochastic) \index{dynamical system}dynamical system. A primary use of \index{data assimilation}data assimilation is in forecasting, where the purpose is to provide better future estimates than can be obtained using either the data or the model alone. 
All the methods from the first part of the course
may be applied directly, but there are other new methods which
exploit the Markovian structure to update the \index{state}state of the system
sequentially, rather than to learn about the initial condition. (But,
of course, knowledge of the initial condition may be used to inform
the \index{state}state of the system at later times.)

\begin{table}
	\begin{center}
	\bgroup
\def\arraystretch{1.5}
		\begin{tabular}{ | c | c |c |}
			\hline
			Topic & Inverse Problems & Data Assimilation  \\ \hline
		 Bayesian Formulation& Chapter 1 &  Chapter 7  \\ \hline
			 Linear Setting & Chapter 2  & Chapter 8    \\ \hline
			 Optimization Perspective & Chapter 3  & Chapter 9    \\ \hline
			 			 Gaussian Approximation & Chapter 4  & Chapter 10    \\ \hline
			 			 			 Sampling & Chapters 5 and 6 & Chapters 11 and 12    \\ \hline	 			 			 			 			 Kalman Inversion  & \multicolumn{2}{c|}{Chapter 13}    \\ \hline
		\end{tabular}
\egroup
		\caption*{ {\sffamily{ \bf Table 1}} Structure of the notes: the organization of the material emphasizes the unity between the subjects of inverse problems and data assimilation.  }		
		\label{Updateswithgradients}
	\end{center}
\end{table}

 The third and final part of the notes describes methods for generic
\index{inverse problem}inverse problems that build on \index{data assimilation} data 
assimilation ideas, thus bringing together the material in the first two parts. The structure of the notes, as well as the presentation, emphasize the inter-relations between inverse problems and data assimilation. As summarized in Table 1,  each chapter in the first part (inverse problems) has its counterpart in the second part (data assimilation).  

\subsection*{Use of the Notes for Teaching and Independent Learning}
These notes were first developed out of Caltech course 
ACM 159 (now ACM/IDS 154) in Fall 2017, and substantially modified for 
the University of Chicago course STAT 31550 in Winter 2019; now
the notes form the basis of courses taught regularly in both institutions. To cater to students with diverse backgrounds and interests,  the instructors complement the material covered in class with hands-on assignments. The first two parts of the notes include several exercises that the instructors have used for this purpose. Additionally, when teaching these classes, we have found it pedagogically beneficial to ask students to complete an independent project, implementing the methods studied in class to solve an applied problem of their choice. This applied problem often arises from the students' own research;
the bibliographic references included at the end of each chapter also form a resource 
to help students to choose and formulate their own projects. Finally, the notes are intended to be concise and self-contained, and thus to be useful not only as a  classroom  teaching resource, 
but also for independent self-guided learning.

\subsection*{Notation}
Throughout the notes we use $\mathbb{N}$ to denote the positive integers $\{1,2,3, \cdots \},$
and $\mathbb{Z}^+$ to denote the non-negative integers $\mathbb{N} \cup \{0\}=\{0,1,2,3, \cdots \}.$
 The symbol $I_\Du$ denotes the identity matrix on $\Ru$, and
$Id$ denotes the identity mapping.  
We use $|\cdot|$ to denote the Euclidean norm corresponding to the 
inner-product $\langle a ,  b \rangle = a^\top b;$  we also use the
notation $|\cdot|$ to denote the induced norm on matrices.  

A  symmetric matrix $A$
is \index{positive definite}positive definite (resp. positive semi-definite) if $\langle u, Au \rangle$ is positive (resp. non-negative)
for all $u \ne 0$. This will sometimes  be denoted by  $A>0$ (resp. $A \ge 0).$
For $A>0,$ we denote by  $| \cdot |_A $ the weighted norm defined by $|v|_A^2  = v^\top A^{-1} v $.
 The corresponding weighted Euclidean inner-product is given by
$\langle \cdot \;,  \; \cdot \rangle_{A}:=\langle \cdot \;, A^{-1}\cdot \rangle.$
 We use $\otimes$ to denote the outer product between
two vectors: $(a \otimes b)c=\langle b,c \rangle a.$ 
We let $B(u,\delta)$ denote the open ball of radius $\delta$ at $u$, in the
Euclidean norm.
 We also use ${\rm det}$ and ${\rm Tr}$ to denote the determinant 
and trace functions on matrices. 

Throughout, we denote by $\Prob (\cdot), \Prob(\cdot \; | \; \cdot )$ the probability density function (pdf) of
a random variable and its conditional pdf, respectively.
We write 
$$ \rho(f)  =\Expect^\rho[f]=\int_{\Ru}f(u)\rho(u)du$$
to denote expectation of $f: \Ru \mapsto  \mathbb{R}$
with respect to pdf
$\rho$ on $\Ru.$ 
The distribution of the random variables in these notes 
 will often have density with respect to \index{Lebesgue}Lebesgue
measure, but occasional use of \index{Dirac}Dirac masses will be required; we will
use the notational convention that \index{Dirac}Dirac mass at point $v$
has ``density'' $\delta(\cdot-v)$, also denoted by $\delta_v(\cdot).$ When a random
variable $u$ has pdf $\rho$
we will write $u \sim \rho.$ We use $\Rightarrow$ to denote \index{weak convergence}weak
convergence of probability measures,  that is, $\rho_n \Rightarrow \rho$ if $\rho_n(f) \to \rho(f)$ for all bounded and continuous $f: \Ru \mapsto  \R$.

\subsection*{Acknowledgments}
These notes were created in \LaTeX \, by the students in ACM 159, 
based on lectures presented by the instructor Andrew Stuart, and 
on input from the course TA Armeen Taeb. The authors are very grateful to
these students, without whom the notes would not exist.
The individuals responsible for 
typesetting the notes, listed in alphabetic order, are: Blancquart, Paul; Cai, Karena; Chen, Jiajie; Cheng, Richard; Cheng, Rui; Feldstein, Jonathan; Huang, De; Id{\'i}ni, Benjamin; Kovachki, Nikola; Lee, Marcus; Levy, Gabriel; Li, Liuchi; Muir, Jack; Ren, Cindy; Seylabi, Elnaz; Sch{\"a}fer, Florian; Singhal, Vipul; Stephenson, Oliver; Song, Yichuan;  Su, Yu; Teke, Oguzhan; Williams, Ethan; Wray, Parker; Zhan, Eric; Zhang, Shumao; Xiao, Fangzhou. Furthermore, the following students added content to the notes, beyond the materials presented by the
instructors: Parker Wray --  created an early draft of the Overview; Jiajie Chen --  found an  alternative proof of early presentations of under-determined inverse problems and smoothing in Gaussian data assimilation; Fangzhou Xiao --  providing numerical
illustrations of  \index{prior}prior, \index{likelihood}likelihood and \index{posterior}posterior; Elnaz Seylabi and Fangzhou Xiao -- catching  many  typographical errors in an  early  draft of these notes; Cindy Ren -- numerical simulations to enhance understanding of \index{importance sampling}importance sampling;
Cindy Ren and De Huang --  improving the constants in initial presentations of the approximation error  of \index{importance sampling}importance sampling; Richard Cheng and Florian Sch{\"a}fer -- illustrations to enhance understanding of the coupling argument used to study convergence of \index{MCMC}MCMC algorithms by presenting the finite \index{state-space!finite}state-space case; and Ethan Williams and Jack Muir -- numerical simulations and illustrations of \index{Kalman filter!ensemble}ensemble Kalman filter and \index{Kalman filter!extended}extended Kalman filter that appeared in an early version of these notes.
The authors are also grateful to Tapio Helin (LUT University) 
who used the notes in his own course and
provided very helpful feedback on an early draft. Finally, the authors
are thankful to Yuming Chen, Andrew Dennehy, Ruoxi Jiang, Phillip Lo, and Walter Zhang (University of Chicago)
and Eitan Levin (Caltech) for their generous feedback; they are also grateful to
Hwanwoo Kim (University of Chicago) for making
substantial improvements to the figures initially provided by the individuals
listed above.

The work of Daniel Sanz-Alonso has been funded by DOE, NGIA, and NSF (USA), and by FBBVA
(Spain). The work of Andrew Stuart has been funded by AFOSR, ARL, DoD, NIH, NSF, and ONR (USA), by EPSRC (UK), and by ERC (EU). The work of Armeen Taeb has been funded by the Resnick Fellowship (USA) and by the ETH Foundations of Data Science (Switzerland).  All of this funded research has helped to shape the 
presentation of the material in these notes and is gratefully acknowledged.   \\[0.1in]


\newpage \tableofcontents

\part{Inverse Problems}

 \chapter{\Large{\sffamily{Bayesian Inverse Problems and Well-Posedness}}} \label{ch1}
In this chapter we introduce the \index{Bayesian}Bayesian approach to inverse problems  in which  the unknown parameter and the observed data are viewed as random variables. In this probabilistic formulation, the solution of the inverse problem is the \index{posterior}posterior distribution  on the
parameter given the data.  
 We will show that the \index{Bayesian}Bayesian formulation leads to a form of \index{well-posed}well-posedness: small perturbations of the \index{forward!model}forward model or the observed data translate into small perturbations of the \index{posterior}posterior distribution. \index{well-posed}Well-posedness requires a
notion of \index{distance}distance between probability measures. We introduce the \index{distance!total variation}total 
variation  and \index{distance!Hellinger}Hellinger distances, giving characterizations of them,
and bounds relating them, that will be used throughout these notes.
We prove \index{well-posed}well-posedness in the \index{distance!Hellinger}Hellinger distance.

The chapter is organized as follows. Section \ref{sec:11} introduces the formulation of \index{Bayesian!inverse problem}Bayesian inverse problems.  In Section \ref{sec:12} we derive a formula for the \index{posterior}posterior pdf and explain how several estimators for the unknown parameter can be obtained using the \index{posterior}posterior. Section \ref{sec:13} describes the \index{well-posed}well-posedness of the \index{Bayesian}Bayesian formulation together with the necessary background on distances between probability measures. The chapter closes with bibliographical remarks in Section \ref{sec:14}.

\section{Formulation of \index{Bayesian!inverse problem}Bayesian Inverse Problems}\label{sec:11}
We consider the following setting. We let $G:  \Ru \to \Ry$ define
the \index{forward!model}forward model and aim to recover an unknown parameter $u \in \Ru$ from data $y \in \Ry$ given by
\begin{equation}
 y = G(u) + \eta, 
\label{eq:jc0}
\end{equation}
where $\eta \in \Ry$ represents \index{observation!noise}observation noise.
We view $(u, y ) \in \Ru \times \Ry$ as a random variable, whose distribution is specified by means of the following assumption on the distribution of $(u,\eta) \in \Ru \times \Ry$ and the relationship between $u,$ $y$ and $\eta$ postulated in equation \eqref{eq:jc0}.
\begin{assumption}\label{a:jc1}
The distribution of the random variable $(u,\eta) \in \Ru \times \Ry$ is defined by:
 \begin{itemize}
 \item $u \sim \pr(u), u \in \Ru $.
 \item $\eta \sim \noise(\eta), \eta \in \Ry$.
 \item $ u$ and $\eta$ are independent, written $u \perp \eta.$
 \end{itemize}
\end{assumption}
Here $\pr$ and $\noise$ describe the pdfs of the random variables $u$ and $\eta,$ respectively. 
Then $\pr(u)$ is called the \index{prior}\emph{prior} pdf and, for each fixed $u\in \Ru,$  $y |  u \sim \noise\bigl(y - G(u)\bigr)$ determines the \index{likelihood}\emph{likelihood} function.  In this
probabilistic perspective, the solution to the inverse problem is the conditional distribution of $u$ given $y$, which is called the \index{posterior}\emph{posterior} distribution, and will be denoted by $u | y\sim \post^y(u).$  The \index{posterior}posterior pdf determines, for any candidate
 parameter value in $\Ru$, how probable that parameter is, based on \index{prior}prior assumptions
and the link between parameter and data, all expressed probabilistically. In particular,
the \index{posterior}posterior contains information about the level of uncertainty in the parameter 
recovery: for instance, large \index{posterior}posterior covariance typically
indicates that the data contains insufficient information 
to accurately recover the input parameter.

 \section{Formula for \index{posterior}Posterior pdf: \index{Bayes theorem}Bayes Theorem}\label{sec:12}
\index{Bayes theorem}Bayes theorem is a bridge connecting the \index{prior}prior, the \index{likelihood}likelihood and the \index{posterior}posterior.
 \begin{theorem}[Bayes Theorem]
\label{t:bayes}
 Let Assumption \ref{a:jc1} hold, and assume that
\[
 Z = Z(y) :=\int_{\Ru} \noise\bigl( y - G(u)\bigr)\pr(u) du  >0.\]
 Then $u | y\sim \post^y(u),$ where
\begin{equation}\label{eq:bayesformula}
 \post^y(u) = \frac{1}{Z} \noise\bigl(y - G(u)\bigr)\pr(u).
\end{equation}
 \end{theorem}
 \begin{proof}
 Denote by $\Prob (\cdot)$ the pdf of a random variable and by $\Prob(\cdot | \cdot )$ its conditional pdf. 
We have
 \[
 \begin{aligned}
 \Prob(u , y) &=  \Prob (u | y) \Prob(y),  \textrm{ if } \Prob(y) > 0 ,\\
 \Prob(u, y) &=  \Prob (y | u ) \Prob(u),  \textrm{ if } \Prob(u) > 0 .\\
\end{aligned}
 \]
 Note that the marginal pdf on $y$ is given by
 \begin{align*}
 \Prob(y) &= \int_{\Ru} \Prob( u, y) du  \\
 & = \int_{\Ru} \Prob(y|u) \Prob(u) du = Z >0.
 \end{align*}Then
 \begin{equation}
\label{eq:jc1}
 \Prob(u | y) = \frac{1}{\Prob(y) } \Prob(y | u) \Prob(u)   =  \frac{1}{\Prob(y)}\noise\bigl(y - G(u )\bigr) \pr(u)
 \end{equation}
 for both $\Prob(u)=\rho(u) > 0$ and $\Prob(u)=\rho(u) = 0$. 
\end{proof}

We will often denote the \index{likelihood}likelihood function by 
$\like(u) := \noise \bigl(y - G(u)\bigr)$. We then write 
$$ \post^y(u) = \frac{1}{Z} \like(u) \pr(u),$$
omitting the data $y$ in the \index{likelihood}likelihood function; when no confusion arises
we will also simply write $\post(u)$ for  the \index{posterior}posterior pdf, rather than
$\post^y(u).$

\begin{remark}
The proof of Theorem \ref{t:bayes} shows that in order to apply \index{Bayes formula}Bayes formula \eqref{eq:bayesformula} one needs to guarantee that the normalizing constant $\Prob(y) = Z$ 
is positive; in other words, the marginal density of the observed data $y$ needs 
to be positive.  This is simply the natural assumption that the observed data could 
indeed have been observed, given the probabilistic conditions in Assumption \ref{a:jc1}.
From now on it will be assumed without further notice that $\Prob(y) = Z >0.$ 
Finally, we remark that throughout these notes we will denote normalizing constants generically by $Z,$ and depending on the context the normalizing constant  may sometimes 
be interpreted as the marginal density of an underlying data set. 
\end{remark}

The \index{posterior}posterior distribution $\post^y(u)$ contains all the knowledge on the parameter $u$ available in the \index{prior}prior and the data. In applications it is often useful, however, to summarize the \index{posterior}posterior distribution through a few numerical values. Summarizing the \index{posterior}posterior is particularly important if the parameter is high-dimensional, since then visualizing the \index{posterior}posterior or detecting regions of high \index{posterior}posterior probability is nontrivial. Two natural numerical summaries are the \index{posterior!mean estimator}posterior mean and the \index{posterior}posterior mode. 
\begin{definition}\label{def:map}
The {\em \index{posterior!mean estimator}posterior mean estimator} of $u$ given data $y$ is the mean of the \index{posterior}posterior distribution:
$$\upm = \int_{\Ru} u \post^y(u) \, du.$$ 
	The {\em maximum a posteriori (\index{MAP estimator}MAP) estimator} of $u$ given data $y$ is the mode of the  
\index{posterior}posterior distribution $\post^{\vct{y}}(\vct{u})$,  defined as
	$$\umap =\arg\max_{\vct{u}\in\Ru}\post^{\vct{y}}(\vct{u}).$$
\end{definition}
\noindent This maximum may not be uniquely defined, 
in which case we talk about \emph{a}, 
rather than \emph{the}, \index{MAP estimator}MAP estimator.

The importance of the \index{MAP estimator}MAP and the 
\index{posterior!mean estimator}posterior mean already suggest 
the need to compute maxima (for the MAP estimator) and 
integrals (for the posterior mean) in
order to extract actionable information from the
\index{Bayesian}Bayesian formulation of inverse problems 
and \index{data assimilation}data assimilation.  For this reason, \index{optimization}optimization (to compute
maxima) and \index{sampling}sampling (to compute integrals) will play an important role in these notes.  
In practice it is often useful to quantify the uncertainty in the parameter reconstruction, and numerical summaries such as the \index{posterior!mean estimator}posterior mean and the \index{MAP estimator}MAP estimators can be complemented by \index{credible intervals}credible intervals, that is, parameter regions of prescribed \index{posterior}posterior 
probability. 
In order to make tractable the computation of estimators and credible intervals, the \index{posterior}posterior can be approximated by a simple 
distribution, such as a \index{Gaussian}Gaussian or a \index{Gaussian!mixture} Gaussian
mixture; \index{optimization}optimization can be used to determine such
approximations. In a similar spirit, sampling may be viewed as approximating
the \index{posterior}posterior by a combination of \index{Dirac}Dirac masses to enable computation of integrals.
An \index{optimization}optimization perspective for inverse problems and \index{data assimilation}data assimilation will be studied in Chapters \ref{chap:optimization} and \ref{ch:optfilteringsmoothing}, respectively, and \index{Gaussian!approximation}Gaussian approximations will be discussed in Chapters \ref{ch4} and \ref{lecture10}, respectively; \index{Dirac}Dirac approximations constructed via \index{sampling}sampling will be studied in Chapters \ref{lecture5} and \ref{ch:6} (inverse problems) and in Chapters \ref{ch11} and \ref{ch12} (\index{data assimilation}data assimilation).  

We next consider two simple examples of a direct application of \index{Bayes theorem}Bayes theorem. 

\begin{example}[MAP and Posterior Mean Estimators]
Let $\Du=\Dy=1$, $\eta \sim \noise = \Nc(0, \gamma^2),$ and let
\[
\pr(u) = \begin{cases}
\frac{1}{2} ,  &  u\in (-1,1), \\
0 ,  &  u\in (-1,1 )^c. \\
\end{cases}
\]
Suppose that the \index{observation}observation is generated by $y = u + \eta$. Using \index{Bayes theorem}Bayes Theorem \ref{t:bayes}, we derive the \index{posterior}posterior pdf
\[
\post^y(u) =
 \begin{cases}
 \frac{1}{2Z} \exp(  -\frac{1}{2\gamma^2}|y- u |^2   ),  &  u \in (-1,1) , \\
 0 ,  & u \in (-1,1)^c , \\
 \end{cases}
\]
where $Z$ is a normalizing constant ensuring that $\int_{\R} \post^y(u) du = 1$. 
 Now we find 
the \index{MAP estimator}MAP estimator.  
From the explicit formula for $\post^y$, we have
\[
\umap  = \arg\max_{u \in \R} \post^y(u) = \begin{cases}
y &  \textrm{if } y \in (-1, 1),  \\
-1&  \textrm{if } y  \leq - 1 , \\
1 & \textrm{if } y\geq 1 . \\
\end{cases}
\]
In this example, the \index{prior}prior on $u$ is supported
on $(-1,1)$ and the \index{posterior}posterior on $u|y$ is supported on $(-1,1)$.
If the data lies in $(-1,1)$ then the \index{MAP estimator}MAP estimator is the data itself; otherwise it is the extremal point of
the \index{prior}prior support which matches the sign of the data. The \index{posterior!mean estimator}posterior mean is 
$$\upm = \frac{1}{2Z} \int_{-1}^1 u \exp\Bigl( -\frac{1}{2\gamma^2} |y - u|^2 \Bigr) \, du,$$
which may be  approximated using, for instance, the \index{sampling}sampling methods described in Chapters \ref{lecture5} and \ref{ch:6}. 
\end{example}

The following example illustrates once again the application of \index{Bayes theorem}Bayes theorem, and shows that the \index{posterior}posterior may concentrate near a low-dimensional manifold in the input parameter space $\R^d$. In such a case it is important to understand the geometry of the support of the \index{posterior}posterior density, which cannot be captured by point estimation or \index{Gaussian!approximation}Gaussian approximations.

\begin{example}[Concentration of Posterior on Manifold]
Let $\Du =2, \Dy=1, \pr \in C(\R^2,  \R), $ 
 and suppose that there is $\pr_{\max} >0$ such that, for all  $u \in \R^2,$ $0 < \pr(u) \leq \pr_{\max} <  \infty$. 
  Suppose that the \index{observation}observation is generated by
 \begin{align*}
 y &= G(u) + \eta,\\
G(u) & = u_1^2 + u_2^2,\\
\eta & \sim \noise = \Nc(0, \gamma^2), \quad 0< \gamma \ll 1,
 \end{align*}
and assume that $y > 0$. Using \index{Bayes theorem}Bayes theorem we obtain
the \index{posterior}posterior pdf
\[
\post^y(u) = \frac{1}{Z} \exp\left(  - \frac{1}{2 \gamma^2}| u_1^2 + u_2^2 - y|^2   \right) \pr(u).
\]
We now show that the \index{posterior}posterior concentrates near the manifold
defined by the circumference  $\{ u \in \R^2: u_1^2 + u_2^2=y\}.$
Denote $A^{\pm}:=\{ u \in \R^2 :  | u_1^2 + u_2^2 - y |^2 \leq \gamma^{2 \pm \delta}   \}$, for some fixed $\delta \in (0,2)$. The set $A^-$ is defined so that it captures most of the
\index{posterior}posterior probability, and $A^+$ so that it captures little of the \index{posterior}posterior probability. They are defined this way because the \index{observation!noise}observational  noise has 
variance $\gamma^2$; considering a neighbourhood of the circumference
which  scales as  $\gamma$ raised to a power slightly smaller than $2$ 
captures most of the \index{posterior}posterior probability; considering a 
neighbourhood of the circumference in which the exponent is slightly 
larger than this captures little of the \index{posterior}posterior probability.
 Define $B$ to be the closed ball of radius $2\sqrt{y}$ centered at the origin.
Let $u^+ \in A^+ \subset B, u^- \in (A^-)^c$ and let $\pr_{\min} = \inf_{u\in B} \pr(u)$.  Since $\pr(u)$ is positive and continuous and $B$ is compact, $\pr_{\min} > 0$. Taking the \index{small noise limit}small noise limit yields 
 \[
 \frac{\post^y(u^+)}{ \post^y(u^-) } \geq \exp\left( -\frac{1}{2} \gamma^{\delta} + \frac{1}{2} \gamma^{-\delta} \right)  \frac{\pr_{\min}}{\pr_{\max}}\to \infty, \ \textrm{as } \gamma \to 0^+.
 \]
Therefore, noting that $y>0$, the \index{posterior}posterior $\post^y$ concentrates, as $\gamma \to 0^+,$ on the circumference with radius $\sqrt{y}$.
  \begin{figure}[h]
   \centering
   \includegraphics[width =0.45\textwidth, height = 0.4\textwidth]{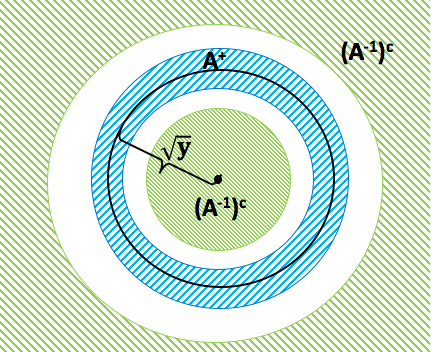}
   \caption{The \index{posterior}posterior measure concentrates on a circumference with radius $\sqrt{y}$. Here, the blue shadow area is $A^+$ and the green shadow area is $(A^-)^c$. }
 \end{figure}
\end{example}


\section{\index{well-posed}Well-Posedness of \index{Bayesian!inverse problem}Bayesian Inverse Problems}\label{sec:13}

In this section we show that the \index{Bayesian}Bayesian formulation of \index{inverse problem}inverse problems
leads to a form of \index{well-posed}well-posedness. 
More precisely, we study the sensitivity of the \index{posterior}posterior pdf to perturbations of the \index{forward!model}forward model $G.$  In many \index{inverse problem}inverse problems the ideal \index{forward!model}forward model $G$ is not accessible but can be approximated by some computable $G_\delta;$ consequently $\post^y$ is replaced by $\post_\delta^y$. An example that is often found in applications, to which the theory contained herein may be generalized,
is when  $G$ is an operator acting on an infinite-dimensional space
which is approximated, for the purposes of computation, by some finite-dimensional operator $G_\delta$. We seek to prove that, under certain assumptions, the small difference between $G$ and $G_\delta$ (\index{forward!error}forward error) leads to a 
similarly small difference between $\post^y$ and $\post^y_\delta$
(inverse error):
      
\vspace{0.1in}
\noindent{\bf Meta Theorem: \index{well-posed}Well-Posedness}\\
       \[|G-G_\delta|=O(\delta)\ \Longrightarrow \ d(\post^y,\post_\delta^y)=O(\delta),\]
{\em for small enough $\delta>0$ and some metric $d(\cdot,\cdot)$ on probability densities.}
\vspace{0.1in}

This result will be formalized in Theorem \ref{thm1} below, which  shows that the $O(\delta)$-convergence of $\post_\delta^y$ with respect to some \index{distance}distance $d(\cdot,\cdot)$ can be guaranteed under certain assumptions on the \index{likelihood}likelihood. We will conclude the chapter by showing an example where these assumptions hold true. In order to discuss these issues we will need to
introduce metrics on probability densities.

\subsection{Metrics on Probability Densities}
Here we introduce the \index{distance!total variation}total variation and the \index{distance!Hellinger}Hellinger distance, both of which have been used to show \index{well-posed}well-posedness results. In this chapter we will use the \index{distance!Hellinger}Hellinger distance to establish \index{well-posed}well-posedness of \index{Bayesian!inverse problem}Bayesian \index{inverse problem}inverse problems, and in Chapter \ref{lecture7} we employ the \index{distance!total variation}total variation distance to establish \index{well-posed}well-posedness of \index{Bayesian}Bayesian formulations of \index{filtering}filtering and \index{smoothing}smoothing in \index{data assimilation}data assimilation.

\begin{definition}\label{def:hellinger}
The \index{distance!total variation} \emph{total variation distance} between two pdfs $\post$ and $\post'$ is defined by 
\[\dtv (\post,\post')\coloneqq\frac{1}{2}\int|\post(u)-\post'(u)|du=\frac{1}{2}\|\post-\post'\|_{L^1}.\]
The \index{distance!Hellinger}\emph{Hellinger distance} between two pdfs $\post$ and $\post'$ is defined by
\[\dhell(\post,\post')\coloneqq\Big(\frac{1}{2}\int|\sqrt{\post(u)}-\sqrt{\post'(u)}|^2du\Big)^{1/2}=\frac{1}{\sqrt{2}}\|\sqrt{\post}-\sqrt{\post'}\|_{L^2}.\]
\end{definition}

In the rest of this subsection we will establish bounds between the \index{distance!Hellinger}Hellinger and \index{distance!total variation}total variation distance, and show how both distances can be used to bound the difference of expected values computed with two different densities; these results will be used in subsequent chapters.  Before doing so, the next lemma motivates our choice of normalization constant $1/2$ for \index{distance!total variation}total variation distance and $1/\sqrt{2}$ for \index{distance!Hellinger}Hellinger distance: they are chosen so that the maximum possible distance between two densities is one.  
The proof also shows that $\post$ and $\post'$ have \index{distance!total variation}total variation and \index{distance!Hellinger}Hellinger distance equal to one if and only if they have disjoint supports, that is, if $\int \post(u)\post'(u) du=0.$

      \begin{lemma} \label{l:dh1}
For any pdfs $\post$ and $\post'$,
       \[0\leq \dtv(\post,\post')\leq 1,\quad 0\leq \dhell(\post,\post')\leq 1.\]
      \end{lemma}
      \begin{proof} The lower bounds follow immediately from the 
definitions, so we only need to prove the upper bounds. For \index{distance!total variation}total variation distance
      \[\dtv(\post,\post')=\frac{1}{2}\int|\post(u)-\post'(u)|du\leq \frac{1}{2}\int\post(u)du+\frac{1}{2}\int\post'(u)du=1, \]
      and for \index{distance!Hellinger}Hellinger distance
      \begin{align*}
      \dhell(\post,\post')=&\ \biggl(\frac{1}{2}\int \Bigl|\sqrt{\post(u)}-\sqrt{\post'(u)}\Bigr|^2du\biggr)^{1/2}\\
      =&\ \bigg(\frac{1}{2}\int\Big(\post(u)+\post'(u)-2\sqrt{\post(u)\post'(u)}\,\, \Big)du \bigg)^{1/2}\\
      \leq&\ \Big(\frac{1}{2}\int\big(\post(u)+\post'(u)\big)du\Big)^{1/2}\\
      =&\ 1.\end{align*}
      \end{proof}      
The following result gives bounds between \index{distance!total variation}total variation and \index{distance!Hellinger}Hellinger distance.       

      \begin{lemma} \label{l:dh2} For any pdfs $\post$ and $\post'$,
       \[\frac{1}{\sqrt{2}}\dtv(\post,\post')\leq \dhell(\post,\post')\leq \sqrt{\dtv(\post,\post')}.\]
      \end{lemma}
      \begin{proof} Using the \index{Cauchy-Schwarz inequality}Cauchy–Schwarz inequality 
      \begin{align*}
      \dtv(\post,\post')=&\ \frac{1}{2}\int\Big|\sqrt{\post(u)}-\sqrt{\post'(u)}\Big|\Big|\sqrt{\post(u)}+\sqrt{\post'(u)}\Big|du\\
      \leq&\ \biggl(\frac{1}{2}\int\Big|\sqrt{\post(u)}-\sqrt{\post'(u)}\Big|^2du\biggr)^{1/2}\biggl(\frac{1}{2}\int\Big|\sqrt{\post(u)}+\sqrt{\post'(u)}\Big|^2du\biggr)^{1/2}\\
      \leq&\ \dhell(\post,\post')\biggl(\frac{1}{2}\int\big(2\post(u)+2\post'(u)\big)du\biggr)^{1/2}\\
      =&\ \sqrt{2}\dhell(\post,\post').
      \end{align*}
      Notice that $|\sqrt{\post(u)}-\sqrt{\post'(u)}|\leq |\sqrt{\post(u)}+\sqrt{\post'(u)}|$ since $\sqrt{\post(u)},\sqrt{\post'(u)}\geq0$. Thus we have
      \begin{align*}
      \dhell(\post,\post')=&\ \biggl(\frac{1}{2}\int\Big|\sqrt{\post(u)}-\sqrt{\post'(u)}\Big|^2du\biggr)^{1/2}\\
      \leq&\ \biggl(\frac{1}{2}\int\Big|\sqrt{\post(u)}-\sqrt{\post'(u)}\Big|\Big|\sqrt{\post(u)}+\sqrt{\post'(u)}\Big|du\biggr)^{1/2}\\
      \leq&\ \Big(\frac{1}{2}\int\big|\post(u)-\post'(u)\big|du\Big)^{1/2}\\
      =&\ \sqrt{\dtv(\post,\post')}.
      \end{align*}
      \end{proof}

The following two lemmas show that if two densities are close in \index{distance!total variation}total variation or in \index{distance!Hellinger}Hellinger distance, expectations computed with respect to both densities are also close. In addition, the following lemma also provides a useful characterization of the \index{distance!total variation}total variation distance that will be used repeatedly throughout these notes. 

      \begin{lemma} \label{lemmatv} Let $f$ be a function such that $
|f|_\infty:=\sup_{u\in \Ru}|f(u)|<\infty.$ It holds that
      \[\big|\Expect^\post[f]-\Expect^{\post'}[f]\big|\leq 2|f|_\infty \dtv(\post,\post').\]
Moreover, the following \index{variational!characterization of total variation}variational characterization of the \index{distance!total variation}total variation distance holds:
\begin{equation}
\dtv(\post, \post') = \frac12 \sup_{|f|_\infty \le 1} \big|\Expect^\post[f]-\Expect^{\post'}[f]\big|.
\end{equation} 
      \end{lemma}
      \begin{proof}
      For the first part of the lemma, note that 
     	\begin{align*}
     	\Big|\Expect^\post[f]-\Expect^{\post'}[f]\Big|=&\ \Big|\int_{\Ru}f(u)\bigl(\post(u)-\post'(u)\bigr)du\Big|\\
     	\leq&\ 2|f|_\infty\cdot\frac{1}{2}\int_{\Ru}|\post(u)-\post'(u)|du\\
     	=&\ 2|f|_\infty\dtv(\post,\post').
     	\end{align*}
     	This in particular shows that, for any $f$ with $|f|_\infty =1,$ 
     	 $$\dtv(\post, \post') \ge \frac12 \big|\Expect^\post[f]-\Expect^{\post'}[f]\big|.$$
     Our goal now is to show a choice of $f$ with $|f|_\infty =1$ that achieves equality.	Define $f(u):= \text{sign} \Bigl( \post(u)- \post'(u) \Bigr)$, so that $f(u) \Bigl( \post(u)- \post'(u) \Bigr) = |\post(u) - \post'(u)|.$ Then it holds that $|f|_\infty =1,$ and 
     	\begin{align*}
     	\dtv(\post, \post') &= \frac12 \int_{\Ru} |\post(u) - \post'(u)| \, du \\
     	&= \frac12 \int_{\Ru} f(u) \Bigl( \post(u) - \post'(u) \Bigr) \, du \\
     	& = \frac12 \Big|\Expect^\post[f]-\Expect^{\post'}[f]\Big|.
     	\end{align*}
     	This completes the proof of the \index{variational!characterization of total variation}variational characterization.
      \end{proof}
 
      \begin{lemma} \index{distance!Hellinger} \label{Lemma: hellinger} Let $f$ be a function such that 
$f_2:=\bigl(\Expect^\post[|f|^2]+\Expect^{\post'}[|f|^2]\bigr)^{\frac12}< \infty$. It holds that
      \[\big|\Expect^\post[f]-\Expect^{\post'}[f]\big|\leq 2 f_2 \dhell(\post,\post').\]
      \end{lemma}
      \begin{proof}
      Using the \index{Cauchy-Schwarz inequality}Cauchy–Schwarz inequality 
     	\begin{align*}
     	\Big|\Expect^\post[f]-\Expect^{\post'}[f]\Big|=&\ \biggl|\int_{\Ru}f(u)\Bigl(\sqrt{\post(u)}-\sqrt{\post'(u)}\Bigr)\Bigl(\sqrt{\post(u)}+\sqrt{\post'(u)} \, \Bigr) \,du\biggr|\\
     	\leq&\ \biggl(\frac{1}{2}\int \Big|\sqrt{\post(u)}-\sqrt{\post'(u)}\Big|^2du\biggr)^{1/2} \biggl(2\int|f(u)|^2\Big|\sqrt{\post(u)}+\sqrt{\post'(u)}\Big|^2 \,du\biggr)^{1/2}\\
     	\le &\ \dhell(\post,\post')\Big(4\int|f(u)|^2\bigl(\post(u)+\post'(u)\bigr)du\Big)^{1/2}\\
     	= &\ 2f_2 \,\dhell(\post,\post').
     	\end{align*}
      \end{proof}
      
\begin{remark}
Note  that the result for \index{distance!Hellinger}Hellinger only assumes that $f$ is square integrable with respect to $\post$ and $\post'$. In contrast, the result for \index{distance!total variation}total variation distance assumes that $f$ is bounded, which is a stronger condition. 
Lemma \ref{l:dh2} also demonstrates that smallness in 
the Hellinger metric is a more stringent condition than smallness in
\index{distance!total variation}total variation. 
Our aim in the following section is to show  \index{well-posed}well-posedness in some metric on probability densities. 
The preceding observations suggest that establishing such a result 
in the \index{distance!Hellinger}Hellinger metric makes a stronger statement than doing so in \index{distance!total variation}total
variation.
\end{remark}

\subsection{Approximation Theorem}
We denote by 
\[\like(u)=\noise\bigl(y-G(u)\bigr)\quad \text{and}\quad \like_\delta(u)=\noise\bigl(y-G_\delta(u)\bigr)\]
the \index{likelihood}likelihoods associated with $G(u)$ and $G_\delta(u),$ so that
\[\post^y(u)=\frac{1}{Z}\like(u)\pr(u)\quad \text{and}\quad \post_\delta^y(u)=\frac{1}{Z_\delta}\like_\delta(u)\pr(u),\]
where $Z, Z_\delta>0$ are the corresponding normalizing constants. Before we proceed to our main result, we make some assumptions.
\begin{assumption}\label{a:dh1}
There exist $\delta^+>0$ and $K_1,K_2<\infty$ such that, for all $\delta\in (0,\delta^+),$ 
	\begin{itemize}
	\item[(i)] $|\sqrt{\like(u)} - \sqrt{\like_\delta(u)}|\leq \phi(u)\delta$, for some $\phi(u)$ 
such that $\Expect^{\pr} [\phi^2(u)] \le K_1^2$;
	\item[(ii)] $\sup_{u\in\Ru}(|\sqrt{\like(u)}|+|\sqrt{\like_\delta(u)}|) \le K_2.$
	\end{itemize}
\end{assumption}

\begin{remark}
Assumption \ref{a:dh1} only involves conditions on the \index{likelihood}likelihood $\like$ and the approximate \index{likelihood}likelihood $\like_\delta.$ Our presentation in this chapter emphasizes the
situation in which this approximation is necessitated in order to approximate the \index{forward!model}forward 
model $G.$ However, another important scenario which is covered by the theory is approximation due 
to perturbations of the data $y.$ As an example, we will establish in Chapter \ref{lecture7}  a \index{well-posed}well-posedness result that guarantees stability of \index{Bayesian}Bayesian \index{smoothing}smoothing under perturbations of the data. 
More generally, the theoretical framework introduced here is very flexible, and it may be employed to
study the stability of many \index{Bayesian}Bayesian formulations of \index{inverse problem}inverse problems and \index{data assimilation}data assimilation under a wide range of perturbations. 
\end{remark}

Now we state the main result of this section:
\begin{theorem}[\index{well-posed}Well-Posedness of \index{posterior}Posterior] Under Assumption \ref{a:dh1} we have
      \[\dhell(\post^y,\post_\delta^y)\leq c \delta,\quad \delta\in(0,\Delta),\]
      for some $\Delta>0$ and some $c\in(0,+\infty)$ independent of $\delta$.
      \label{thm1}\\
\end{theorem}

Notice that this theorem together with Lemma \ref{Lemma: hellinger} guarantee that expectations computed with respect to $\post^y$ and $\post_\delta^y$ are order $\delta$ apart. 
To prove Theorem \ref{thm1}, we first show a lemma which characterizes
the normalization factor $Z_\delta$ in the small $\delta$ limit.
   
     \begin{lemma}\label{lem1} Under Assumption \ref{a:dh1} there exist $\Delta>0$, $c_1,c_2\in(0,+\infty)$ such that 
      \[|Z-Z_\delta|\leq c_1\delta\quad \text{and}\quad Z,Z_\delta>c_2,\quad \text{for}\ \delta\in(0,\Delta).\]
      \end{lemma}
      \begin{proof} 
Since $Z=\int \like(u)\pr(u)du$ and
$Z_\delta=\int \like_\delta(u)\pr(u)du$ 
we have
      \begin{align*}
      |Z-Z_\delta|=& \biggl| \int \bigl(\like(u)-\like_\delta(u)\bigr)\pr(u)du    \biggr|   \\
      \leq&\ \Big(\int \Bigl|\sqrt{\like(u)}- \sqrt{\like_\delta(u)}\Bigr|^2\pr(u)du\Big)^{1/2}\Big(\int \Bigl|\sqrt{\like(u)}+\sqrt{\like_\delta(u)}\Bigr|^2\pr(u)du\Big)^{1/2}\\
      \leq&\ \Big(\int \delta^2 \phi(u)^2\pr(u)du\Big)^{1/2}\Big(\int K_2^2\pr(u)du\Big)^{1/2}\\
      \leq&\ K_1 K_2\delta,\quad \delta\in(0,\delta^+). 
      \end{align*}
      Therefore, for $\delta\leq \Delta\coloneqq\min\{\frac{Z}{2K_1 K_2},\delta^+\}$, we have
      \[Z_\delta\geq Z-|Z-Z_\delta|\geq \frac{1}{2}Z.\]
      The lemma follows by taking $c_1=K_1 K_2$ and $c_2=\frac{1}{2}Z$.\\
      \end{proof}   
 
      \begin{proof}[Proof of Theorem \ref{thm1}]
   We break the total error into two contributions, one reflecting
the difference between $Z$ and $Z_\delta$, and the other the difference between $\like$ and $\like_\delta$:
      \begin{align*}
      \dhell(\post^y,\post_\delta^y)=&\ \frac{1}{\sqrt{2}}\Bigl\|\sqrt{\post^y}-\sqrt{\post_\delta^y}\Bigr\|_{L^2}\\
      =&\  \frac{1}{\sqrt{2}}\Big\|\sqrt{\frac{\like\pr}{Z}}-\sqrt{\frac{\like\pr}{Z_\delta}}+\sqrt{\frac{\like\pr}{Z_\delta}}-\sqrt{\frac{\like_\delta\pr}{Z_\delta}}\Big\|_{L^2}\\
      \leq&\ \frac{1}{\sqrt{2}}\Big\|\sqrt{\frac{\like\pr}{Z}}-\sqrt{\frac{\like\pr}{Z_\delta}}\Big\|_{L^2}+\frac{1}{\sqrt{2}}\Big\|\sqrt\frac{\like\pr}{Z_\delta}-\sqrt{\frac{\like_\delta\pr}{Z_\delta}}\Big\|_{L^2}.
      \end{align*}
      Using Lemma \ref{lem1} we have, for $\delta\in(0,\Delta)$, 
      \begin{align*}
      \Big\|\sqrt{\frac{\like\pr}{Z}}-\sqrt{\frac{\like\pr}{Z_\delta}}\Big\|_{L^2}=&\ \Big|\frac{1}{\sqrt{Z}}-\frac{1}{\sqrt{Z_\delta}}\Big|\Big(\int \like(u)\pr(u)du\Big)^{1/2}\\
      =&\ \frac{|Z-Z_\delta|}{(\sqrt{Z}+\sqrt{Z_\delta})\sqrt{Z_\delta}}\\
      \leq&\ \frac{c_1}{2c_2}\delta,
      \end{align*}
      and 
      \[\Big\|\sqrt{\frac{\like\pr}{Z_\delta}} - \sqrt{\frac{\like_\delta\pr}{Z_\delta}}\Big\|_{L^2}=\frac{1}{\sqrt{Z_\delta}}\Big(\int\Bigl|\sqrt{\like(u)}- \sqrt{\like_\delta(u)}\Bigr|^2\pr(u) du\Big)^{1/2} \leq \sqrt{\frac{K_1^2}{c_2}}\delta.\]
      Therefore
      \[\dhell(\post^y,\post_\delta^y)\leq \frac{1}{\sqrt{2}}\frac{c_1}{2c_2}\delta+\frac{1}{\sqrt{2}}\sqrt{\frac{K_1^2}{c_2}}\delta=c \delta,\]
      with $c=\frac{1}{\sqrt{2}}\frac{c_1}{2c_2}+\frac{K_1}{\sqrt{2c_2}}$,
which is independent of $\delta$.
      \end{proof}

\subsection{Example: \index{well-posed}Well-Posedness for Parameter Estimation in an ODE} 
Many \index{inverse problem}inverse problems arise from differential equations with unknown input parameters. Here we consider a simple but typical example where $G(u)$ comes from the solution of an \index{ordinary differential equation}ordinary differential equation (\index{ODE}ODE), which needs to be solved numerically. Let $x(t)$ be the solution to the initial value problem 
\begin{equation}
\label{eq:dh1}
\frac{d x}{dt}=F(x;u),\quad x(0)=0,
\end{equation} 
where $F:\Ry\times \Ru\rightarrow \Ry$ is a function such that $F(x;u)$ and the partial Jacobian $D_xF(x;u)$ are uniformly bounded with respect to $(x,u)$, i.e.
\[|F(x;u)|,|D_xF(x;u)|< F_{{\rm max}}  ,\quad \text{ for all } (x,u)\in \Ry\times \Ru,\]
for some constant $F_{{\rm max}}$, and thus $F(x;u)$ is \index{Lipschitz}Lipschitz in $x$ in that,  for all $u \in \Ru$, 
\[|F(x;u)-F(x';u)|\leq F_{{\rm max}} |x-x'|,\quad \text{ for all } x,x'\in \Ry.\]
Note that $u \in \Ru$ defines parametric dependence of the
vector field defining the differential equation.

Now consider the \index{inverse problem}inverse problem setting 
$$y=G(u)+\eta,$$ 
where 
$$G(u):=x(1)=x(t)|_{t=1},$$
and $\eta\sim \Nc(0,\gamma^2 I_\Dy)$. We assume that the exact mapping $G(u)$ is replaced by some numerical approximation $G_\delta(u)$. In particular, $G_\delta(u)$ is given by using the \index{forward Euler method}forward Euler method to solve the \index{ODE}ODE \eqref{eq:dh1}. Define $X_0=0$ and 
\[X_{\ell+1}=X_\ell+\delta F(X_\ell;u),\quad \ell\geq0,\]
where $\delta=\frac{1}{L}$ for some large integer $L$. Finally define $G_\delta(u):=X_L$. \\

In what follows, we will prove that $G_\delta(u)$ is uniformly bounded and 
close to $G(u)$ when $\delta$ is small, and that $G$ and $G_\delta$ both
satisfy the same global bound. Then we will use these results to show 
that Assumption \ref{a:dh1} is satisfied. Therefore, we can apply Theorem~\ref{thm1} to this example to establish that the approximate 
\index{posterior}posterior $\post_\delta^y$, defined by approximate
\index{forward!model}forward model $G_\delta$, is close to the true 
\index{posterior}posterior $\post^y$ 
with exact \index{forward!model}forward model $G$.\\

In showing that Assumption \ref{a:dh1} is satisfied, we use
Lemmas \ref{l:maruyama} and \ref{lem:euler} below. 
Recall that $\eta\sim \Nc(0,\gamma^2I_\Dy)$, and thus 
\[\sqrt{\like(u)}=\sqrt{\noise\bigl(y-G(u)\bigr)}=\frac{1}{(2\pi)^{\Dy/4}\gamma^{\Dy/2}}\exp\Big(-\frac{1}{4\gamma^2}|y-G(u)|^2\Big),\]
\[\sqrt{\like_\delta(u)}=\sqrt{\noise\bigl(y-G_\delta(u)\bigr)}=\frac{1}{(2\pi)^{\Dy/4}\gamma^{\Dy/2}}\exp\Big(-\frac{1}{4\gamma^2}|y-G_\delta(u)|^2\Big).\]
\begin{itemize}
\item For Assumption \ref{a:dh1}(i) notice that the function $e^{-w}$  is \index{Lipschitz}Lipschitz for $w>0$, with \index{Lipschitz}Lipschitz constant $1$. Therefore we have
\begin{align*}
\Big|\sqrt{\like(u)}-\sqrt{\like_\delta(u)}\Big|\leq&\ \frac{1}{(2\pi)^{\Dy/4}\gamma^{\Dy/2}}\cdot \frac{1}{4\gamma^2}\cdot \big||y-G(u)|^2-|y-G_\delta(u)|^2\big|\\
= &\ \frac{1}{(2\pi)^{\Dy/4}\gamma^{\Dy/2}}\cdot \frac{1}{4\gamma^2} \cdot |2y-G(u)-G_\delta(u)||G(u)-G_\delta(u)|\\
\leq &\ \frac{1}{(2\pi)^{\Dy/4}\gamma^{\Dy/2}}\cdot \frac{1}{4\gamma^2} \cdot (2|y|+2F_{{\rm max}}) c\delta\\
=&\ \tilde{c}\delta.
\end{align*}
That is to say, Assumption \ref{a:dh1}(i) is satisfied with $\phi(u)=\tilde{c}$ and $\int_{\Ru} \phi^2(u)\pr(u)du=\tilde{c}^2<\infty$. 
\item Assumption \ref{a:dh1}(ii) is satisfied, since
\[\sqrt{\like(u)}=\frac{1}{(2\pi)^{\Dy/4}\gamma^{\Dy/2}}\exp\Big(-\frac{1}{4\gamma^2}|y-G(u)|^2\Big)\leq \frac{1}{(2\pi)^{\Dy/4}\gamma^{\Dy/2}},\]
\[\sqrt{\like_\delta(u)}=\frac{1}{(2\pi)^{\Dy/4}\gamma^{\Dy/2}}\exp\Big(-\frac{1}{4\gamma^2}|y-G_\delta(u)|^2\Big)\leq \frac{1}{(2\pi)^{\Dy/4}\gamma^{\Dy/2}}.\]
\end{itemize}

The preceding verification of Assumption \ref{a:dh1}
used the following two lemmas, and the first of these
uses the Gronwall inequality which follows them. 
Define $t_\ell=\ell\delta$, $x_\ell=x(t_\ell)$. The following lemma gives an estimate on the error generated from using the \index{forward Euler method}forward Euler method. 

\begin{lemma} \label{l:maruyama} Let $E_\ell:=x_\ell-X_\ell$. Then there is $c<\infty$ independent of $\delta$ such that
\[|E_\ell|\leq c\delta,\quad 0\leq \ell\leq L. \]
In particular,
\[|G(u)-G_\delta(u)|=|E_L|\leq c\delta.\]
      \end{lemma}
      \begin{proof}
For simplicity of exposition, we consider the case $\Dy=1$; the case
$\Dy>1$ is almost identical, simply requiring the integral form for
the remainder term in the Taylor expansion.  Using Taylor expansion
in the case $\Dy=1,$ there is $\xi_\ell\in[t_\ell,t_{\ell+1}]$ such that 
     	\begin{align*}
     	x_{\ell+1}=&\ x_\ell+\delta \frac{dx}{dt}(t_\ell)+\frac{\delta^2}{2}\frac{d^2x}{dt^2}(\xi_\ell)\\
     	=&\ x_\ell+\delta F(x_\ell;u)+\frac{\delta^2}{2}D_xF\bigl(x(\xi_\ell);u\bigr)F\bigl(x(\xi_\ell);u\bigr).
     	\end{align*}
     	Thus we have
     	\begin{align*}
     	|E_{\ell+1}|=&\ |x_{\ell+1}-X_{\ell+1}|\\
     	=&\ \Big|x_\ell-X_\ell+\delta\Bigl(F(x_\ell;u)-F(X_\ell;u)\Bigr)+\frac{\delta^2}{2}D_xF\bigl(x(\xi_\ell);u\bigr)F\bigl(x(\xi_\ell);u\bigr)\Big|\\
     	\leq&\ |x_\ell-X_\ell|+\delta \bigl|F(x_\ell;u)-F(X_\ell;u)\bigr|+\frac{\delta^2}{2} \bigl|D_xF\bigl(x(\xi_\ell);u\bigr) \bigr| \bigl|F\bigl(x(\xi_\ell);u\bigr)\bigr|\\
     	\leq&\ |E_\ell|+\delta F_{{\rm max}}|E_\ell|+ \frac{\delta^2}{2}F_{{\rm max}}^2.
     	\end{align*}
     	Noticing that $|E_0|=0$, the \index{discrete Gronwall inequality}discrete Gronwall inequality (Theorem
\ref{t1.6}) gives
     	\begin{align*}
     	|E_\ell|\leq&\ (1+\delta F_{{\rm max}})^\ell|E_0|+\frac{(1+\delta F_{{\rm max}})^\ell-1}{\delta F_{{\rm max}}}\cdot\frac{\delta^2}{2}F_{{\rm max}}^2\\
     	\leq&\ \bigg( \Big(1+\frac{F_{{\rm max}}}{L}\Big)^L-1\bigg)\cdot\frac{F_{{\rm max}}\delta}{2}\\
     	\leq&\ \frac{(e^{F_{{\rm max}}}-1)F_{{\rm max}}}{2} \delta.  
     	\end{align*}
     	The lemma follows by taking $c=\frac{(e^{F_{{\rm max}}}-1)F_{{\rm max}}}{2}$.
      \end{proof}

\begin{lemma}\label{lem:euler} For any $u \in \Ru,$  
\[|G(u)|,|G_\delta(u)| \le F_{{\rm max}}.\]
      \end{lemma}
      \begin{proof} For $G(u)$ we use that $F(x;u)$ is uniformly bounded, so that
     	\[|G(u)|=|x(1)|=\Big|\int_0^1F(x(t);u)dt\Big|\leq \int_0^1\bigl|F(x(t);u)\bigr|dt\leq F_{{\rm max}}.\]
     	As for $G_\delta(u)$, we first notice that 
     	$$|X_{\ell+1}| = |X_\ell +\delta F(X_\ell;u)| \le |X_\ell| + \delta |F(X_\ell;u)| \le |X_\ell| + \delta F_{{\rm max}},$$
     	and by induction 
     	$$|X_{\ell}| \le |X_0| +  \ell \delta F_{{\rm max}}  = \ell \delta F_{{\rm max}}.$$
     	In particular, 
     	$$|G_\delta(u) | = |X_L| \le L \delta F_{{\rm max}} = F_{{\rm max}}.$$
      \end{proof}

The following \index{discrete Gronwall inequality}discrete Gronwall inequality is used
several times in these notes, and is stated and proved here
for completeness.

\begin{theorem}[Discrete Gronwall Inequality] \label{t1.6}
\index{discrete Gronwall inequality}
Let a positive sequence $\{Z_{\ell}\}_{\ell=0}^L$ satisfy
$$Z_{\ell+1} \le CZ_{\ell}+D, \qquad \forall \ell=0, \dots, L-1$$
for some constants $C,D$ with $C>0$.
Then
$$Z_{\ell} \le \frac{D}{1-C}(1-C^{\ell})+Z_{0}C^{\ell} \qquad \forall \ell=0, \dots, L, \quad
C \ne 1$$
and
$$Z_{\ell} \le \ell D+Z_0 \qquad \forall \ell=0, \dots, L, \quad C=1.$$
\end{theorem}

\begin{proof}
The proof is by induction. We start with the case $C \ne 1.$ 
The result holds for $\ell=0$.
Assume it is true for $\ell<L$. Then, using the defining inequality,
$$Z_{\ell+1} \le \frac{CD}{1-C}(1-C^{\ell})+Z_{0}C^{\ell+1}+D.$$
Rearranging yields
$$Z_{\ell+1} \le \frac{D}{1-C}(1-C^{\ell+1})+Z_{0}C^{\ell+1}$$
and the result follows by induction.

When $C=0$ we again note that the result holds for $\ell=0$.
Assume it is true for $\ell<L$. Then, using the defining inequality with
$C=1$,
$$Z_{\ell+1} \le \ell D+Z_0+D=(\ell+1)D+Z_0$$
and the result follows by induction.
\end{proof}

\section{Discussion and Bibliography}\label{sec:14}
The book by Kaipio and Somersalo \cite{kaipio2006statistical} provides an introduction to the \index{Bayesian}Bayesian approach to \index{inverse problem}inverse problems, especially in the context of differential equations,  and the book \cite{calvetti2007introduction} gives an introduction to \index{Bayesian}Bayesian scientific computing. An overview of the 
subject of \index{Bayesian!inverse problem}Bayesian \index{inverse problem}inverse problems in differential equations, with
a perspective informed by the geophysical sciences,
is given in the book by Tarantola  \cite{tarantola2005inverse} (see, especially, Chapter 5). For non-statistical
approaches to \index{inverse problem}inverse problems, we refer to the books \cite{tikhonov1977solutions,engl1996regularization,vogel2002computational} and the lecture notes \cite{bal2012introduction,miller2003fundamentals}.

The subject of \index{Bayesian!inverse problem}Bayesian \index{inverse problem}inverse problems may be developed
beyond the specific setting of equation \eqref{eq:jc0} 
to study problems of the form 
\begin{equation*}
 y = G(u,\eta).
\end{equation*}
Our emphasis on additive noise $\eta$, often assumed to be \index{Gaussian}Gaussian,
simplifies some algorithms and enables us to be explicit about
some formulae, but is not fundamental in any way. We refer to \cite{dunlop2019multiplicative} for \index{well-posed}well-posedness theory and a study of \index{MAP estimator}MAP estimation with multiplicative noise.  In addition, the setting of equation \eqref{eq:jc0} presupposes that the  \index{forward!model}forward model $G$ is given to us, but in some cases the forward model itself may need to be learned from data.

In the paper \cite{stuart2010inverse} the \index{Bayesian}Bayesian approach to \index{regularization}regularization is reviewed, 
developing a function space viewpoint on the subject; a similar development
of this approach is described in \cite{lasanen2012non,lasanen2012nonb}.
A \index{well-posed}well-posedness theory 
and some algorithmic approaches which are used when adopting the \index{Bayesian!inverse problem}Bayesian 
approach to \index{inverse problem}inverse problems are introduced.
The function space viewpoint on the subject is developed in more detail in the
 lecture notes \cite{dashti2013bayesian}. 
An  early application of this function space methodology to a large-scale
applied \index{inverse problem}inverse problem, taken from
the geophysical sciences, may be found in \cite{martin2012stochastic}. 
The paper \cite{lieberman2010parameter} demonstrates the potential
for the use of dimension reduction techniques from control theory
within statistical \index{inverse problem}inverse problems.

We refer to \cite{gibbs2002choosing}  for further study on the subject of
metrics, and other distance-like functions, on probability measures. 
The first published paper to discuss stability and \index{well-posed}well-posedness
of the \index{Bayesian!inverse problem}Bayesian \index{inverse problem}inverse problem is \cite{marzouk2009stochastic}, in which
the \index{divergence!Kullback-Leibler}Kullback-Liebler divergence 
(see Chapter \ref{ch4}) is employed.
Related results on stability and \index{well-posed}well-posedness, 
but using other distances and divergences, 
may be found in \cite{latz2020well}.
The articles \cite{stuart2010inverse,dashti2013bayesian,trillos2017bayesian} 
study \index{well-posed}well-posedness of \index{Bayesian!inverse problem}Bayesian \index{inverse problem}inverse problems in the Hellinger metric,
with respect to perturbations in the data; papers \cite{cotter2010approximation,harlim2019kernel}  consider stability of the \index{posterior}posterior distribution with respect to numerical approximation of
partial differential equations appearing in the \index{forward!model}forward model.
The papers \cite{hosseini2017well,hosseini2017wellb} discuss generalizations of the \index{well-posed}well-posedness theory to various classes of specific \index{non-Gaussian}non-Gaussian \index{prior}priors. 
On the other hand, \cite{iglesias2014well} contains an interesting set of examples
where the Meta Theorem stated in this chapter fails
in the sense that, whilst \index{well-posed}well-posedness holds, the \index{posterior}posterior is
H\"older with exponent less than one, rather than \index{Lipschitz}Lipschitz, with respect
to perturbations.

The \index{Bayesian!inverse problem}Bayesian approach to \index{inverse problem}inverse problems builds on, and benefits from, the vast literature on \index{Bayesian!statistics}Bayesian statistics. The paper \cite{fienberg2006did} provides a historical overview of the development and popularization of \index{Bayesian!statistics}Bayesian statistics, starting with the introduction of \index{Bayes formula}Bayes formula in 1763  \cite{bayes1763lii} and emphasizing the leading role of Savage \cite{savage1972foundations} in axiomatizing and popularizing the subjective view of probability pioneered by de Finetti \cite{de2017theory}. We refer to \cite{gelman2013bayesian} for a recent and comprehensive textbook on \index{Bayesian}Bayesian methodology. 
See \cite{nickl} for an overview of Bayesian inversion and, in particular, statistical
consistency results in this context.

A topic of debate in \index{Bayesian!statistics}Bayesian statistics, and specifically in
the \index{Bayesian!inverse problem}Bayesian approach to \index{inverse problem}inverse problems, is how to construct \index{prior}prior
probability measures from available \index{prior}prior information, which is typically not
described probabilistically.
 The papers \cite{owhadi2015brittleness,owhadi2015brittlenessb} demonstrate that this is an important question: different
\index{prior}priors, both consistent with available \index{prior}prior information, can lead to
wildly different \index{Bayesian!inference}Bayesian inference when computing \index{posterior}posterior expectations:
what the authors term \index{Bayesian!brittleness}\emph{Bayesian brittleness.}
Arguably, this issue may be dealt with through application of the scientific
method: 
  a given \index{prior}prior and \index{likelihood}likelihood 
are postulated, and 
\index{posterior}posterior predictions are made;
data acquired after making posterior predictions 
may then be used to evaluate the \index{Bayesian}Bayesian probabilistic 
model employed, and in particular the \index{prior}prior and
\index{likelihood}likelihood and, if necessary, modify it. 
The body of work on \index{Bayesian!brittleness}Bayesian brittleness
builds on related analysis in the context of \index{forward!uncertainty quantification}forward 
uncertainty quantification \cite{owhadi2013optimal}, a topic concerned with propagating uncertainty on parameters
through a model into predictions. The subject of uncertainty
quantification, both the forward and inverse varieties,
is overviewed in \cite{sullivan2015introduction,
smith2013uncertainty}.


 \chapter{\Large{\sffamily{The Linear-Gaussian Setting }}}
\label{sec: posteriordistribution}

Recall the \index{inverse problem}inverse problem of estimating an unknown parameter $\vct{u} \in \Ru$ from data $y\in \Ry$ under the model assumption 
\begin{equation}\label{eq:yurelation}
y = G(u) + \eta.
\end{equation} 
In this chapter we study the \index{linear-Gaussian setting}linear-Gaussian setting, where the \index{forward!model}forward model $G(\cdot)$ is linear and both the \index{prior}prior on $u$ and the distribution of the \index{observation!noise}observation noise $\eta$ are \index{Gaussian}Gaussian. This setting is highly amenable to analysis and arises frequently in applications. Moreover, as we will see throughout these notes, many methods employed in nonlinear or \index{non-Gaussian}non-Gaussian settings build on ideas from the \index{linear-Gaussian setting}linear-Gaussian case by performing linearization or invoking \index{Gaussian!approximation}Gaussian approximations. 
After establishing a formula for the \index{posterior}posterior pdf in Section \ref{sec:21}, we investigate in Section \ref{sec:22} the effect that the choice of \index{prior}prior has on our solution  by quantifying the spread of the \index{posterior}posterior distribution in the \index{small noise limit}small noise (approaching zero) limit. This investigation provides intuitive understanding concerning the impact of the \index{prior}prior for overdetermined, determined, and underdetermined regimes, corresponding to $\Du <\Dy,\Du =\Dy,$ and $\Du>\Dy,$ respectively. Extensions of the theory and references to the literature are discussed in Section \ref{sec:23}. 

The following will be assumed throughout this chapter.
\begin{assumption}
The relationship between unknown $u\in \Ru$, data $y\in \Ry,$ 
and noise $\eta\in \Ry$ defined by equation \eqref{eq:yurelation} holds. Moreover,
\label{a:sz1}
\begin{itemize}
	\item Linearity of the \index{forward!model}forward model: $G(\vct{u})=\mtx{A}\vct{u}$, for some $\mtx{A}\in\R^{\Dy \times \Du}$.
    \item \index{Gaussian}Gaussian \index{prior}prior: $\vct{u}\sim \pr(\vct{u})=\Nc(\mpr,\Cpr)$, where $\Cpr$ is \index{positive definite}positive definite.
    \item\index{Gaussian}Gaussian noise: $\vct{\eta}\sim \noise(\vct{\eta})=\Nc(\zerovct,\mtx{\Gamma})$, where $\mtx{\Gamma}$ is \index{positive definite}positive definite.
    \item $u$ and $\eta$ are independent: $ u\perp \eta.$ 
\end{itemize}
\end{assumption}

\section{Derivation of the \index{posterior}Posterior Distribution}\label{sec:21}


Under Assumption  \ref{a:sz1} the \index{likelihood}likelihood on $\vct{y}$ given $\vct{u}$ is \index{Gaussian}Gaussian, 
\begin{align}
	\vct{y}|\vct{u}\sim \Nc(\mtx{A}\vct{u},\mtx{\Gamma}).
\end{align}
Therefore, using Bayes formula \eqref{eq:bayesformula} we see that the \index{posterior}posterior $\post^y(u)$ is given by 
\begin{align*}
	\post^{\vct{y}}(\vct{u})&=\frac{1}{Z}\noise(\vct{y}-\mtx{A}\vct{u})\pr(\vct{u})\\
	&=\frac{1}{Z}\exp\left(-\frac{1}{2}|\vct{y}-\mtx{A}\vct{u}|_{\mtx{\Gamma}}^2\right)\exp\left(-\frac{1}{2}|\vct{u}-\mpr |_{\Cpr}^2\right)\\
	&=\frac{1}{Z}\exp\left(-\frac{1}{2}|\vct{y}-\mtx{A}\vct{u}|_{\mtx{\Gamma}}^2-\frac{1}{2}|\vct{u}-\mpr|_{\Cpr}^2\right)\\
	&=\frac{1}{Z}\exp\bigl(-\J(\vct{u})\bigr),
\end{align*}
with
\begin{equation}
\label{eq:ju}
\J(\vct{u})=\frac{1}{2}|\vct{y}-\mtx{A}\vct{u}|_{\mtx{\Gamma}}^2+\frac{1}{2}|\vct{u} - \mpr|_{\Cpr}^2.
\end{equation}
 Note that here
\begin{equation}
\label{eq:likeg}
\log \like(u)=-\frac{1}{2}|\vct{y}-\mtx{A}\vct{u}|_{\mtx{\Gamma}}^2.
\end{equation} 
Since the \index{posterior}posterior pdf can be written as the exponential of a quadratic in $u$ it follows that the \index{posterior}posterior is \index{Gaussian}Gaussian. Its mean and covariance are given in the following result.

\begin{theorem}[\index{posterior}Posterior is \index{Gaussian}Gaussian]
\label{t:sz1}
   	Under Assumption \ref{a:sz1} 
the \index{posterior}posterior distribution is  \index{Gaussian}Gaussian, 
   		\begin{align}
			\vct{u}|\vct{y} \sim \post^{\vct{y}}(\vct{u}) =\Nc(\mpost,\Cpost).\label{eq: posterior}
		\end{align}
The \index{posterior!mean estimator}posterior mean $\mpost$ and covariance $\Cpost$ are given by the
following formulae:
		\begin{align}
			\mpost&=(\mtx{A}^{\top}\mtx{\Gamma}^{-1}\mtx{A}+\Cpr^{-1})^{-1}(\mtx{A}^{\top}\mtx{\Gamma}^{-1}\vct{y} + \Cpr^{-1}\mpr),\label{eq: postmean} \\
			\Cpost&=(\mtx{A}^{\top}\mtx{\Gamma}^{-1}\mtx{A}+\Cpr^{-1})^{-1}.
			\label{eq: postvar}
		\end{align}

	\label{thm: posteriorgaussian}
\end{theorem}
\begin{proof}
	Since $\post^{\vct{y}}(\vct{u})=\frac{1}{Z}\exp\bigl(-\J(\vct{u})\bigr)$
with $\J(u)$ given by \eqref{eq:ju},
a quadratic function of $\vct{u}$, it follows that
the \index{posterior}posterior  is \index{Gaussian}Gaussian.
Denoting the mean and variance of $\post^{\vct{y}}(\vct{u})$ by
$\mpost$ and $\Cpost$, we can write $\J(\vct{u})$ in the following form
	\begin{align}
		\J(\vct{u})&=\frac{1}{2}|\vct{u}-\mpost|_{\Cpost}^2+ q ,
		\label{eq: ju2}
	\end{align}
where the term $q$ does not depend on $\vct{u}$.
Now matching the coefficients of the quadratic and linear terms 
in equations \eqref{eq:ju} and \eqref{eq: ju2}, we get
	\begin{align*}
		\Cpost^{-1}&=\mtx{A}^{\top}\mtx{\Gamma}^{-1}\mtx{A}+\Cpr^{-1}, \\
		\Cpost^{-1}\mpost&=\mtx{A}^{\top}\mtx{\Gamma}^{-1}\vct{y} + \Cpr^{-1}\mpr.
	\end{align*}
Therefore equations \eqref{eq: postmean} and \eqref{eq: postvar} follow. 
\end{proof}

 We saw in the previous chapter that the \index{posterior!mean estimator}posterior mean estimator and the \index{MAP estimator}MAP estimator are typically different.
However, equation \eqref{eq: ju2} shows that in the
current linear-Gaussian setting the \index{posterior!mean estimator}posterior 
mean $\mpost$ minimizes $\J(u)$ given in \eqref{eq:ju}. 
Thus, the MAP estimator and the posterior mean coincide. 

\begin{corollary}[Characterization of  Bayes Estimators]
The \index{posterior!mean estimator}posterior mean and \index{MAP estimator}MAP estimators under Assumptions \ref{a:sz1} agree, and are given by $\umap = \upm = \mpost$ defined in equation \eqref{eq: postmean}.
\end{corollary}

Furthermore, the formula \eqref{eq:ju} demonstrates \nc that the \index{posterior!mean estimator}posterior mean is found as a compromise
between maximizing the \index{likelihood}likelihood (by making the \index{loss}\emph{loss} term 
$\frac{1}{2}|\vct{y}-\mtx{A}\vct{u}|_{\mtx{\Gamma}}^2$ small) and minimizing 
deviations from the \index{prior}prior mean (by making the \index{regularization}\emph{regularization} term 
$\frac{1}{2}|\vct{u} -\mpr|_{\Cpr}^2$ small). The relative importance given to both \index{objective}objectives is determined by the relative size of the \index{prior}prior covariance $\Cpr$ and the noise covariance $\Gamma.$
An important feature of the \index{linear-Gaussian setting}linear-Gaussian setting is that the \index{posterior}posterior covariance $\Cpost$ does not depend on the data $y;$ this is not true in general. 

 We conclude this subsection with an example.

\begin{example}	\label{exp: map}
	Let $\mtx{\Gamma}=\gamma^2\Id$, $\Cpr=\sigma^2\Id,$ $\mpr =0,$ and 
set $\lambda=\frac{\gamma^2}{\sigma^2}$. Then
	$$\J_\lambda(\vct{u})\triangleq\gamma^2\J(\vct{u})=\frac{1}{2}|\vct{y}-\mtx{A}\vct{u}|^2+\frac{\lambda}{2}|\vct{u}|^2.$$
Since $\mpost$ minimizes $\J_\lambda(\cdot)$ it follows that
	\begin{align}
		(\mtx{A}^{\top}\mtx{A}+\lambda\Id)\mpost=\mtx{A}^{\top}\vct{y}.
		\label{eq: expmean}
	\end{align}
\end{example}

Example \ref{exp: map} provides a link between \index{Bayesian!inversion}Bayesian inversion and
\index{optimization}optimization approaches to inversion: $\J_\lambda(\vct{u})$ can be seen as 
the \index{objective}objective function in a linear regression model with a regularizer 
$\frac{\lambda}{2}|\vct{u}|^2$, as used in ridge regression. Equation \eqref{eq: expmean} for $\mpost$ is exactly the normal equation 
with regularizer in the least-squares problem.
In fact, in the general linear-Gaussian setting of Assumption \ref{a:sz1}, equation \eqref{eq: postmean} can also be viewed as a generalized normal equation. This perspective helps us 
understand the structure of \index{Bayesian}Bayesian \index{regularization}regularization by linking it to the deep understanding of \index{optimization}optimization approaches to \index{inverse problem}inverse problems. A more extensive account of the \index{optimization}optimization perspective and its interplay with \index{Bayesian}Bayesian formulations will be given in the following chapter. 

 \section{\index{small noise limit}Small Noise Limit of the \index{posterior}Posterior Distribution}\label{sec:22}
 In this section we study the small \index{observation!noise}observation noise limit of the \index{posterior}posterior in the \index{linear-Gaussian setting}linear-Gaussian setting. While most of the ideas and results
can be extended beyond this setting, explicit calculations that are possible
in the \index{linear-Gaussian setting}linear-Gaussian setting provide helpful intuition. Throughout this section we assume the following.
 
 \begin{assumption}
\label{a:sz2}
In addition to Assumption \ref{a:sz1} (the \index{linear-Gaussian setting}linear-Gaussian setting), we
assume that $\vct{\eta}:=\gamma \vct{\eta}_0,$ where $\vct{\eta}_0 \sim \Nc(\zerovct,\mtx{\Gamma}_0)$; thus $\Gamma=\gamma^2\Gamma_0.$
\end{assumption}

Note that substituting $\Gamma = \gamma^2 \Gamma_0$  into \eqref{eq: postmean} and \eqref{eq: postvar} we obtain that
 \begin{align}
\mpost =& (A^{\top} \Gamma_0^{-1} A + \gamma^2 \Cpr^{-1})^{-1} (A^{\top}\Gamma_0^{-1}y +\gamma^2 \Cpr^{-1} \mpr), \label{eq:meansmallnoise}\\
\Cpost =&  \gamma^2 (A^{\top} \Gamma_0^{-1} A + \gamma^2 \Cpr^{-1})^{-1}. \label{eq:covsmallnoise}
 \end{align}

In the next three subsections we study the behavior of the \index{posterior!mean estimator}posterior mean $m$ and covariance $C$
as $\gamma \to 0^+$ ---the \index{small noise limit}small noise limit. We remark that $m,$ $C,$ and the \index{posterior}posterior $\pi^y$ depend on the noise level $\gamma$, but we will not make explicit said dependence in our notation. 
We separately consider the overdetermined, determined, and underdetermined regimes.  We recall that $\Rightarrow$ denotes \index{weak convergence}weak convergence of probability measures. We will use repeatedly that \index{weak convergence}weak convergence of \index{Gaussian}Gaussian distributions is equivalent to the convergence of their means and covariances. In particular, the weak limit of a sequence of \index{Gaussian}Gaussians with means converging to $m^+$ and covariance matrices converging to the zero matrix is a \index{Dirac}Dirac mass $\delta_{m^+}.$

 \subsection{Overdetermined Case}
 We start with the overdetermined case $\Du<\Dy$. 
\begin{theorem}[\index{small noise limit}Small Noise Limit of \index{posterior}Posterior Distribution -- Overdetermined] \label{lec2:thm1}
Suppose that Assumption \ref{a:sz2} holds, that ${\rm Null}(A)= {0}$ and that $\Du<\Dy.$
Then, in the limit $\gamma \to 0^+,$  
$$\post^y \Rightarrow \delta_{m^+},$$
where $m^+$ is the solution of
the least-squares problem 
\begin{equation}\label{lec1:thm1_1}
m^+  =\arg\min_{u \in \Ru} |\Gamma_0^{-1/2} (y- Au)   |^2.
\end{equation}
\end{theorem}

\begin{proof}
 
Since Null$(A)= {0}$ and $\Gamma_0$ is invertible
we deduce that there is $\alpha>0$ such that, for all $u \in \Ru,$ 
 \[
  \langle u, A^{\top} \Gamma^{-1}_0 Au \rangle = | \Gamma_0^{-1/2} Au |^2 \geq \alpha |u|^2.
 \]
 Thus $A^{\top} \Gamma_0^{-1} A$ is \index{positive definite}positive definite (and hence invertible). It follows that as $\gamma \to 0^+$, the \index{posterior}posterior covariance converges to the zero matrix, $\Cpost \to 0,$ and the \index{posterior!mean estimator}posterior mean satisfies the limit
 \[
 \mpost \to  m^{*} =   (A^{\top}\Gamma_0^{-1}A)^{-1} A^{\top}\Gamma_0^{-1} y.
 \]
This proves the \index{weak convergence}weak convergence of $\post^y$ to $\delta_{m^{*}}$. It remains to characterize $m^{*}$.  Since Null$(A)= {0}$, the minimizers of the scaled loss\footnote{Note that this is a rescaling by
$\gamma^{2}$ of the negative log-likelihood
from equation \eqref{eq:likeg}.}\index{loss}
 \[
 \loss(u) :=\frac{1}{2} | \Gamma_0^{-1/2}(y - Au)|^2
  \]
  are unique and satisfy the normal equations $A^{\top}\Gamma_0^{-1}Au = A^{\top} \Gamma_0^{-1}y$. Hence $m^{*}$ solves the desired least-squares problem and coincides with $m^+$ given in \eqref{lec1:thm1_1}.
 \end{proof}

We have shown that in the overdetermined case where $A^{\top}\Gamma_0^{-1} A$ is invertible, the \index{observation!noise}small observational noise limit leads to a \index{posterior}posterior which is a \index{Dirac}Dirac, centered at the solution of the least-squares problem \eqref{lec1:thm1_1}.   Therefore, in this limit the \index{prior}prior plays no role in the \index{Bayesian!inference}Bayesian inference.

\begin{theorem}[\index{posterior!consistency}Posterior Consistency -- Overdetermined]\label{th:postconsistency}
Suppose that the assumptions of Theorem \ref{lec2:thm1} hold and that the data satisfies
\begin{equation} \label{eq:datapostcons}
y=Au^\dagger+\gamma \eta_0^\dagger, \quad\quad {\rm for \,\, fixed} \quad u^\dagger \in \Ru, \eta_0^\dagger \in \Ry.
\end{equation}
Then, for any sequence  $M(\gamma)\rightarrow\infty$ as $\gamma\to 0^+$, 
	\begin{align}
		\Prob ^{\post^y} \bigl(|\vct{u}-\vct{u}^\dagger|^2>M(\gamma)\gamma^2\bigr)\rightarrow0,
		\label{eq: maincovergence}
	\end{align}
	where $\Prob^{\post^y}$ denotes probability under the \index{posterior}posterior distribution. 
	\end{theorem}
	
\begin{remark}
	For any $\varepsilon>0$, set $M(\gamma)=\frac{\varepsilon^2}{\gamma^2}$ in Theorem \ref{thm: maintheorem} to obtain 
	$$\Prob^{\post^y} \bigl(|\vct{u}-\vct{u}^\dagger|>\varepsilon \bigr)\rightarrow0.$$
	This shows that the \index{posterior}posterior probability concentrates around the truth in the \index{small noise limit}small noise limit. 
\end{remark}

\begin{proof}[Proof of Theorem \ref{th:postconsistency}]
Throughout this proof we let $c$ be a constant independent of $\gamma$ that may change from line to line, and we denote by $\Expect$ expectation with respect to the \index{posterior}posterior distribution, which is \index{Gaussian}Gaussian with mean $\mpost$ and covariance $\Cpost$ given by equations \eqref{eq:meansmallnoise} and \eqref{eq:covsmallnoise}. Denote
$$m^{*} =   (A^{\top}\Gamma_0^{-1}A)^{-1} A^{\top}\Gamma_0^{-1} y$$
as in the proof of the previous theorem. 
We have that
\begin{equation}
\Expect \bigl[ |u-u^\dagger|^2  \bigr] \le c\Bigl( \Expect  \bigl[  |u-\mpost|^2  \bigr]  + |\mpost - m^{*}|^2 + |m^{*}-u^\dagger|^2  \Bigr).
\end{equation}
We now bound each of the three terms in the right-hand side.

For the first one,
\begin{align*}
\Expect  \bigl[ |u-\mpost|^2  \bigr]  &= \Expect\bigl[(u-\mpost)^{\top} (u-\mpost)  \bigr]  = \Expect\bigl[ {\rm Tr}[(u-\mpost)\otimes(u-\mpost)]   \bigr] \\
 &={\rm Tr} \Expect\bigl[ (u-\mpost)\otimes(u-\mpost)   \bigr]  \\
&=  {\rm Tr}(\Cpost) \le \gamma^2 {\rm Tr}  \Bigl[ (A^{\top} \Gamma_0^{-1} A)^{-1} \Bigr].
\end{align*}
For the second term, note that 
\begin{align*}
(A^{\top} \Gamma_0^{-1} A) m^{*} &=  A^{\top}\Gamma_0^{-1}y, \\
(A^{\top} \Gamma_0^{-1} A + \gamma^2 \Cpr^{-1}) \mpost &= A^{\top}\Gamma_0^{-1}y +\gamma^2 \Cpr^{-1} \mpr.
\end{align*}
Therefore
\[ \mpost-m^{*} = \gamma^2 (A^{\top}\Gamma_0^{-1} A)^{-1} (\Cpr^{-1} \mpr   -\Cpr^{-1}\mpost).\] 
Since $\mpost$ converges it is bounded, and so there is $c>0$ such that
\[ |\mpost-m^{*}|^2 \le c \gamma^4. \]
Finally, for the third term we write
\begin{align*}
m^{*} &=   (A^{\top}\Gamma_0^{-1}A)^{-1} A^{\top}\Gamma_0^{-1} Au^\dagger  + \gamma (A^{\top}\Gamma_0^{-1}A)^{-1} A^{\top}\Gamma_0^{-1}  \eta_0^{\dagger} \\
 &= u^\dagger  + \gamma (A^{\top}\Gamma_0^{-1}A)^{-1} A^{\top}\Gamma_0^{-1}  \eta_0^{\dagger},
\end{align*}
which gives
\[ |m^{*} - u^\dagger|^2 \le c \gamma^2. \]
Using \index{Markov inequality}Markov inequality and the three bounds above,
 $$\Prob^{\post^y} \bigl(|\vct{u}-\vct{u}^\dagger|^2>M(\gamma)\gamma^2 \bigr)\leq\frac{\Expect [|\vct{u}-\vct{u}^\dagger|^2]}{M(\gamma)\gamma^2}\leq\frac{c}{M(\gamma)}\rightarrow0,~{\rm as}~\gamma\rightarrow 0^+.$$
\end{proof}

\subsection{Determined Case}
As a byproduct of the proof of Theorem \ref{lec2:thm1}, we can determine the limiting behavior of $\post^y$ in the boundary case $\Du = \Dy$.
\begin{theorem}[\index{small noise limit}Small Noise Limit of \index{posterior}Posterior Distribution -- Determined]\label{lec1:thm1.5}
Suppose that Assumption \ref{a:sz2} holds, ${\rm Null}(A)= 0,$ and $\Du=\Dy$. Then, in the \index{small noise limit}small noise limit $\gamma \to 0^+$, $$\post^y \Rightarrow \delta_{A^{-1 } y}.$$
\end{theorem}
\begin{proof}
In the proof of Theorem \ref{lec2:thm1}, the assumption $\Du < \Dy$ is used only in that $A$ is not a square matrix and thus $A, A^{\top}$ are not invertible. Denote by $(\mpost, \Cpost)$ the mean and variance of the \index{posterior}posterior $u|y$. Using the same argument, we have $\Cpost \to 0$ and
\[
\mpost \to m^{*} =  (A^{\top}\Gamma_0^{-1}A)^{-1} A^{\top}\Gamma_0^{-1} y.
\]
Using that $A, A^\top$ are square invertible matrices we obtain
\[
m^{*} = (A^{-1} \Gamma_0(A^{\top})^{-1} ) A^{\top} \Gamma_0^{-1 } y = A^{-1 }y.
\]
Therefore, $\post^y(u) \Rightarrow \delta_{m^{*}} = \delta_{A^{-1}y}$.
\end{proof}

Note that here, as in the overdetermined case, the \index{prior}prior plays no role in the \index{small noise limit}small noise limit. Moreover, it can be shown as above that \index{posterior!consistency}posterior consistency holds. The proof is very similar to that in the overdetermined case, 
and therefore omitted. 

\begin{theorem}[\index{posterior!consistency}Posterior Consistency -- Determined]
	Suppose that the assumptions of Theorem \ref{lec1:thm1.5} hold, and that the data satisfies
\begin{equation} \label{eq:datapostcons}
y=Au^\dagger+\gamma \eta_0^\dagger, \quad\quad {\rm for \,\, fixed} \quad u^\dagger, \eta_0^\dagger \in \Ru.
\end{equation}
	Then, for any sequence $M(\gamma)\rightarrow\infty$ as $\gamma \to 0^+$, 
	\begin{align}
		\Prob ^{\post^y} \bigl(|\vct{u}-\vct{u}^\dagger|^2>M(\gamma)\gamma^2 \bigr)\rightarrow0.
		\label{eq: maincovergence}
	\end{align}
	\label{thm: maintheorem}
\end{theorem}

 \subsection{Underdetermined Case}
Finally we consider the underdetermined case $\Du > \Dy$. We assume that $A \in \R^{\Dy \times \Du}$ with ${\rm Rank}(A) = \Dy$ and write
 \begin{equation}\label{lec1:thm2_0}
 A = (A_0\  0) Q^{\top} =(A_0 \ 0 )(Q_1\  Q_2)^{\top}= A_0Q_1^{\top},
 \end{equation}
 with $A_0 \in \R^{\Dy \times \Dy} $ an invertible matrix, $Q = (Q_1 \ Q_2) \in \R^{\Du \times \Du}$ an orthogonal matrix so that $Q^{\top} Q= I$, $Q_1 \in \R^{\Du \times \Dy}, Q_2 \in \R^{\Du \times (\Du-\Dy)}$. We have the following result: 
  \begin{theorem}[\index{small noise limit}Small Noise Limit of \index{posterior}Posterior Distribution -- Underdetermined]\label{lec1:thm2}
Suppose that Assumption \ref{a:sz2} holds, that ${\rm Rank}(A) = \Dy,$ and $\Du>\Dy$. In the \index{small noise limit}small noise limit $\gamma \to 0^+$, 
$$\post^y \Rightarrow \Nc(m^+, C^+),$$ where
\[
\begin{aligned}
m^+ &=\Cpr Q_1 (Q_1^{\top} \Cpr Q_1)^{-1}  A_0^{-1}y + Q_2(Q_2^{\top}\Cpr^{-1}Q_2)^{-1}Q_2^{\top} \Cpr^{-1}\mpr, \\
C^+ &=Q_2(Q_2^{\top}\Cpr^{-1}Q_2)^{-1} Q_2^{\top}. \\
\end{aligned}
\]
  \end{theorem}

Since ${\rm Rank}(C^+) = {\rm Rank}(Q_2) =\Du-\Dy<\Du$ this theorem
demonstrates that, in the \index{small noise limit}small \index{observation!noise}observational noise limit, the \index{posterior}posterior
has no uncertainty in a subspace of dimension $\Dy$, but
retains uncertainty in a subspace of dimension $\Du-\Dy$. As a consequence, there is no \index{posterior!consistency}posterior consistency in the underdetermined case.

\begin{example}[Small Noise Limit -- Underdetermined]  To help understand the result in Theorem \ref{lec1:thm2}, 
we consider a simple explicit example. Assume that $A = (A_0 \ 0) \in \R^{\Dy \times \Du}, \Gamma = \gamma^2 \Gamma_0 =\gamma^2  I_\Dy, \Cpr =I_\Du$, $\mpr = 0$.  Let $u =(u_1,u_2)^\top  
\sim \Nc(0, I_\Du)$, with $u_1 \in \Ry, u_2 \in \R^{\Du - \Dy}$. \nc
The data then satisfies
\[
y = Au + \eta = A_0 u_1 +\eta, \  \eta \sim \Nc(0, \gamma^2 I_\Dy).
\]
The \index{posterior}posterior $u| y$ is $\post^y(u) = \frac{1}{Z_{\gamma}} \exp(-  \J_{\gamma}(u) )$, where 
\begin{align}
 \J_{\gamma}(u) &= \frac{ 1 }{2\gamma^2} | y - A_0 u_1 |^2 + \frac{1}{2} | u |^2  \notag \\
 &=\left(  \frac{ 1 }{2\gamma^2} | y - A_0 u_1 |^2 + \frac{1}{2} | u_1 |^2\right)
 +\frac{1}{2} |u_2|^2.  \label{lec1:thm2_1}
\end{align}
It is clear that 
\[
\post^{y}(u_1) \Rightarrow \delta_{A_0^{-1} y} (u_1).
\]
Once $u_1$ is fixed as $A_0^{-1}y$, the first term in \eqref{lec1:thm2_1} is a constant $\frac{1}{2} | A_0^{-1}y |^2$. Since $u_1$ and $u_2$ are independent we can derive, formally, the limiting \index{posterior}posterior as follows
\[
 \post^y(u) \Rightarrow \delta_{A_0^{-1} y} (u_1) \otimes \frac{1}{Z} \exp \Bigl(- \frac{1}{2} |u_2|^2\Bigr) = \delta_{A_0^{-1} y} (u_1) \otimes \Nc(0, I_{\Du-\Dy}),
\]
where $Z = \int_{\R^{\Du-\Dy}} \exp(-\frac{1}{2}|u_2|^2 ) du_2$. In fact, this is exactly the limiting \index{posterior}posterior measure given in Theorem \ref{lec1:thm2}.
\end{example}

To prove Theorem \ref{lec1:thm2}, we use the following decomposition of the identity $I_\Du.$
\begin{lemma}
Let $\Cpr \in \R^{\Du \times \Du}$ be invertible and $Q = [ Q_1 \ Q_2 ]$ be an orthogonal matrix with $Q_1 \in \R^{\Du \times \Dy}, Q_2 \in \R^{\Du \times (\Du-\Dy)}$. We have the following decomposition of $I_\Du$
\begin{equation}\label{lec1:lem}
I_\Du =  \Cpr Q_1 (Q_1^{\top} \Cpr Q_1)^{-1}   Q_1^{\top} +Q_2(Q_2^{\top}\Cpr^{-1}Q_2)^{-1} Q_2^{\top}\Cpr^{-1}.
\end{equation}
\end{lemma}
\begin{proof}
Denote by $R$ the right-hand side of \eqref{lec1:lem}. Since $Q$ is orthogonal, we have $Q_1^{\top} Q_2 = 0, Q_2^{\top}Q_1 = 0$ and thus
\[
Q_1^{\top}(R- I) = 0, \quad Q^{\top}_2\Cpr^{-1}  (R-I) = 0 .
\]
If $B :=(Q_1 \ \Cpr^{-1} Q_2)$ is full rank, the above identities imply that $B^{\top}(R-I) = 0$ and thus $R= I$. Note that
\[
Q^{\top} B =  \left[
\begin{array}{c}
Q_1^{\top} \\
Q_2^{\top} \\
\end{array}
\right] 
[ Q_1 \ \Cpr^{-1} Q_2 ]   =
 \left[
\begin{array}{cc}
I_\Dy  & Q_1^{\top}\Cpr^{-1} Q_2 \\
0 &  Q_2^{\top}\Cpr^{-1} Q_2 \\
\end{array}
\right].
\]
Since the last matrix is invertible, $B$ is invertible and the proof is complete.
\end{proof}

\begin{proof}[Proof of Theorem \ref{lec1:thm2}]
Using \eqref{lec1:lem} we can decompose $u$ as follows
\[
\begin{aligned}
u &=\underbrace{  \Cpr Q_1 (Q_1^{\top} \Cpr Q_1)^{-1} }_{S} \underbrace{Q_1^{\top} u}_{u_1} + \underbrace{Q_2(Q_2^{\top}\Cpr^{-1}Q_2)^{-1}}_{T} \underbrace{Q_2^{\top}\Cpr^{-1}u}_{u_2}  = Su_1 + Tu_2 . \\
\end{aligned}
\]
Here $u_1$ and $u_2$ are \index{Gaussian}Gaussian with $u_2 \sim \Nc(Q_2^{\top}\Cpr^{-1} \mpr , Q_2^{\top}\Cpr^{-1} Q_2  )$. The identity
\[
{\rm Cov}(u_1,u_2) = Q_1^{\top} {\rm Cov}(u,u) \Cpr^{-1} Q_2 = Q_1^{\top} Q_2 = 0
\]
shows that $u_1$ and $u_2$ are independent, written $u_1 \perp u_2$. From \eqref{lec1:thm2_0}, we have
\begin{equation}\label{lec1:thm2_2}
y = Au + \eta  = A_0 Q_1^{\top} u+ \eta= A_0 u_1 + \eta.
\end{equation}
Since $u\perp \eta$ and $u_1 \perp u_2$, we have that $u_2 \perp y, u_1$. We apply conditional probability to yield
\[
\post^y(u_1, u_2)  :=\Prob(u_1, u_2 | y) =\Prob(u_2) \Prob(u_1 | y)  .
\]
Equation \eqref{lec1:thm2_2} and Theorem \ref{lec1:thm1.5} shows that $\Prob(u_1 | y ) \Rightarrow \delta_{A_0^{-1}y} (u_1)$ as the noise vanishes, that is, as  $\gamma \to 0^+$.  
Note that $u_2 \perp u_1$ and $u_2 \perp y$. The limiting \index{posterior}posterior measure  $(u_1,u_2) | y$ is
\begin{equation}\label{lec1:thm2_3}
\post^y(u_1, u_2) \Rightarrow \Prob(u_2) \otimes \delta_{A_0^{-1}y}(u_1)
\end{equation}
as $\gamma \to 0^+$. Recall $u = Su_1 + Tu_2$ and $u_2 \sim \Nc(Q_2^{\top}\Cpr^{-1} \mpr , Q_2^{\top}\Cpr^{-1} Q_2  )$. The mean and variance of the limiting \index{posterior}posterior measure $u|y$ is
\[
\begin{aligned}
m^+ &= \Expect[ Su_1 + Tu_2 | y ] =  SA_0^{-1}y +T\Expect[ u_2 ] =SA_0^{-1}y + TQ_2^{\top}\Cpr^{-1}\mpr, \\
C^+ & = {\rm Cov }(Su_1 + Tu_2 | y) = {\rm  Cov }( Tu_2) = T   Q_2^{\top}\Cpr^{-1} Q_2 T^{\top}  =Q_2(Q_2^{\top}\Cpr^{-1}Q_2)^{-1} Q_2^{\top}. \\
\end{aligned}
\]
We have thus completed the proof.
\end{proof}

Equation \eqref{lec1:thm2_3} shows that in the limit of zero  \index{observation!noise}observational noise, the uncertainty is only in the variable $u_2$. Since ${\rm Span}(T) = {\rm Span}(Q_2)$ and $u = SA_0^{-1} y + Tu_2$, the uncertainty we observed is in ${\rm Span}(Q_2)$. The \index{prior}prior plays a role in the \index{posterior}posterior measure, 
in the limit of zero \index{observation!noise}observational noise, but only in the
variables $u_2.$

\section{Discussion and Bibliography}\label{sec:23}
The linear setting plays, for several reasons, a central role in the study of
\index{inverse problem}inverse problems. First, linear \index{inverse problem}inverse problems are ubiquitous in applications, and are challenging to solve when the matrix defining the linear \index{forward!model}forward model is ill-conditioned, or when the system is severely underdetermined. 
Second, in the \index{linear-Gaussian setting}linear-Gaussian setting explicit solutions are available; these explicit solutions can be used to give insight into the solution of nonlinear \index{inverse problem}inverse problems. Underlying the derivation of these formulae is the fact, shown in this chapter, that a \index{Gaussian}Gaussian \index{likelihood}likelihood function supplemented with a \index{Gaussian}Gaussian \index{prior}prior leads to a \index{posterior}posterior that is again \index{Gaussian}Gaussian. In statistical terms, this constitutes an example of a 
\emph{conjugate \index{prior}prior} \cite{gelman2013bayesian}, namely a choice of \index{prior}prior for a given \index{likelihood}likelihood such that the \index{posterior}posterior belongs to the same family as the \index{prior}prior. A third reason for the central importance of linear \index{inverse problem}inverse problems is that they arise naturally in sequential \index{data assimilation}data assimilation, as we will see in the second part of these notes. 
The paper \cite{franklin1970well},
which concerns the \index{linear-Gaussian setting}linear-Gaussian setting, was arguably the first to formulate
\index{Bayesian!inversion}Bayesian inversion in function space, for the specific problem of determining
the initialization of the heat equation from the solution at later times.
The paper \cite{lehtinen1989linear} studied the \index{linear-Gaussian setting}linear-Gaussian setting
more generally. A computational framework for discretization of linear-Gaussian Bayesian inverse problems in function space was introduced in \cite{bui2013computational}.

In this chapter we have studied several \index{small noise limit}small noise limits, and established a basic form of \index{posterior!consistency}posterior consistency.
 Intuitively, small \index{observation!noise}observation noise would seem 
desirable in the reconstruction of the unknown parameter; however,
and perhaps counterintuitively,  it often makes the computational  solution to the \index{inverse problem}inverse problem more challenging. A concrete manifestation of this phenomenon is analyzed in the context of \index{importance sampling}importance sampling in \cite{agapiou2017importance}. For a treatment of \index{posterior!consistency}posterior consistency in infinite dimensions we refer to \cite{knapik2011bayesian,agapiou2013posterior,nickl2017bernstein}, and for
 the consistency problem in the classical statistical setting
to the books \cite{gine2015mathematical,van1998asymptotic}. 
In certain large data regimes, the \index{Bernstein-von Mises theorem}Bernstein-von Mises theorem \cite{doob1949application} guarantees that the \index{Bayesian}Bayesian \index{posterior}posterior solution is approximately \index{Gaussian!approximation}Gaussian  \cite{nickl2017bernstein,nickl2019bernstein,giordano2020consistency} 
and that the \index{prior}prior distribution plays a negligible role in the \index{posterior}posterior, thus providing theoretical support to the \index{Bayesian}Bayesian approach. We emphasize, however, that in the underdetermined \index{inverse problem}inverse problem setting one cannot expect the conclusions to hold, as demonstrated in this chapter. Furthermore, recent work
\cite{nickl2021some} demonstrates specific phenomena, including potential 
obstacles to \index{posterior!consistency}consistency theorems, that may result
in the setting of infinite-dimensional \index{Bayesian!inversion}Bayesian inversion. For non-statistical \index{optimization}optimization-based
approaches to \index{inverse problem}inverse problems, and consistency in particular, see \cite{engl1996regularization} and the references therein.

 \chapter{\Large{\sffamily{Optimization Perspective }}}\label{chap:optimization}

In this chapter we explore the properties of \index{Bayesian!inversion}Bayesian inversion
from the perspective of an \index{optimization}optimization problem which corresponds
to maximizing the \index{posterior}posterior probability: that is, to finding a 
\index{MAP estimator}maximum a posteriori (MAP) estimator, or mode of the \index{posterior}posterior distribution. We demonstrate the properties of the point estimator resulting
from this \index{optimization}optimization problem, showing its positive and negative
attributes, the latter motivating our work in the following three chapters. 
We also introduce, and study, basic gradient-based \index{optimization}optimization algorithms.

The chapter is organized as follows. We first introduce the problem setting in Section \ref{sec:31}. Two theoretical results are presented in Section \ref{sec:32}. The first shows that the \index{MAP estimator}MAP estimator is attained under appropriate assumptions, while the second provides an interpretation of \index{MAP estimator}MAP estimation in terms of maximizing the probability of infinitesimally small balls. Section \ref{sec:33} contains several examples that illustrate some possible limitations of \index{MAP estimator}MAP estimation. \index{gradient descent}Gradient descent and \index{gradient descent!stochastic}stochastic gradient descent algorithms are described in Section \ref{sec:34}. Both of these algorithms are important examples of gradient-based \index{optimization}optimization algorithms, which we interpret as arising from 
time-discretization of an underlying  differential equation. 
The chapter closes in Section \ref{sec:35} with bibliographical remarks. 

\section{The Setting}\label{sec:31}

Once again we work in the \index{inverse problem}inverse problem setting of finding $u\in \Ru$
from $y\in \Ry$ given by
$$ y = G(u) + \eta$$
 with noise $\eta\sim\noise$ and \index{prior}prior $u\sim\pr,$ as in Assumption \ref{a:jc1}. The \index{posterior}posterior pdf $\post^y(u)$ on $u|y$ is given by \index{Bayes theorem}Theorem \ref{t:bayes} and has the form
\[
\post^y (\vct u) = \frac{1}{Z} \noise (\vct y- \vct G(\vct u)) \pr (\vct u).
\]
Generalizing the definition from the previous chapter,
concerning only the Gaussian setting, we define a \index{loss}\emph{loss function}
\[
\loss (\vct u) = - \log \noise \bigl(\vct y-\vct G(\vct u)\bigr),
\]
and a \emph{regularizer}
\[
\reg(\vct u) = - \log \pr(\vct u).
\]
 Note that the loss is equal to the negative \index{likelihood} log-likelihood: 
$\loss(u)=-\log \like(u).$ 
When added together, these two functions of $u$ comprise
an \index{objective}\emph{objective function} of the form
\[
\J(\vct u) = \loss(\vct u) + \reg(\vct u).
\]
Furthermore
\[
\post^y (\vct u) = \frac{1}{Z} \noise \bigl(\vct y-\vct G(\vct u)\bigr) \pr (\vct u) \propto e^{-\J(\vct u)}.
\]
We see that minimizing the \index{objective}objective function $\J(\cdot)$
is equivalent to maximizing the \index{posterior}posterior pdf $\post^y(\cdot)$. 
Therefore, recalling Definition \ref{def:map}, the \index{MAP estimator}MAP estimator can be rewritten in terms of $\J$ as follows:
\begin{align*}
\umap &= \arg\max_{\vct{u}\in\Ru}\post^{\vct{y}}(\vct{u}) \\
&= \arg \min_{u \in \Ru} \J(u).
\end{align*} 

We will provide conditions under which the \index{MAP estimator}MAP estimator is attained in Theorem \ref{thm:achievable-obj}, and we will give an interpretation of \index{MAP estimator}MAP estimators in terms of maximizing the probability of infinitesimal balls in Theorem
\ref{thm:ball-density}. This interpretation can be used to generalize the definition of \index{MAP estimator}MAP estimators to measures that do not possess a \index{Lebesgue}Lebesgue density.

\begin{example}[MAP Estimator -- Linear-Gaussian Setting]
Consider the \index{linear-Gaussian setting}linear-Gaussian setting of Assumption \ref{a:sz1}.
Then, since the \index{posterior}posterior is \index{Gaussian}Gaussian, its mode agrees with its mean, which is given by $\mpost$ as defined in Theorem \ref{t:sz1}.
\end{example}

\begin{example}[Loss Function -- Gaussian Observational Noise]
	\label{ex:loss-l2}
	If $\eta = \Nc(\zerovct ,\mtx \Gamma)$, then
	\(\noise\bigl(\vct y-\vct G(\vct u)\bigr) \propto \exp(- \frac{1}{2} |\vct y-\vct G(\vct u)|_{\mtx \Gamma} ^2)\). So the \index{loss}loss in this case is \(\loss(\vct u) = \frac{1}{2} |\vct y-\vct G(\vct u)|_{\mtx \Gamma}^2\), a \(\mtx \Gamma\)-weighted \(\losss_2\) loss.
\end{example}

\begin{example}[\(\losss_2\) Regularizer -- Gaussian Prior]
	\label{ex:reg-l2}
	If we have \index{prior}prior \(\pr(\vct u) = \Nc(\zerovct ,\mtx \Cpr)\), then 
ignoring $u$-independent normalization factors, which appear as constant
shifts in $\J(\cdot)$, we may take the regularizer as \(\reg(\vct u) = \frac{1}{2} |\vct u|_{\mtx \Cpr}^2\). In particular, if \(\mtx \Cpr = \lambda^{-1} \Id \), then \(\reg(\vct u) = \frac{\lambda}{2} |\vct u|^2\), an \(\losss_2\) regularizer.
\end{example}

If we combine Example \ref{ex:loss-l2} and Example \ref{ex:reg-l2}, we obtain a canonical \index{objective}objective function 
$$\J(\vct u) = \frac{1}{2} |\vct y-\vct G(\vct u)|_{\mtx \Gamma}^2 + \frac{\lambda}{2} |\vct  u|^2.$$ 
To connect with future discussions, here \(\lambda\) corresponds to \index{prior}prior precision, and may be learned from data:  an example of a 
\index{Bayesian!hierarchical}\emph{hierarchical} 
formulation of Bayesian inversion. 

\begin{example}[\(\losss_1\) Regularizer -- Laplace Prior]
	\label{ex:reg-l1}
	As an alternative to the \(\losss_2\) regularizer, 
consider \(\vct u=(u_1,\ldots,u_\Du)\) with \(u_\du\) having \index{prior}prior
distribution \index{i.i.d.} i.i.d. \index{Laplace distribution}Laplace. Then \(\pr(\vct u) \propto \exp(-\lambda \sum_{\du=1}^\Du |u_\du|) = \exp(-\lambda |\vct u|_1)\). In this case
\(\reg( \vct u ) = \lambda |\vct u|_1\), an \(\losss_1\) regularizer. 
If we combine this \index{prior}prior with the weighted \(\losss_2\) \index{loss}loss above,
then we obtain the \index{objective}objective function $$\J(\vct u) = \frac{1}{2} |\vct y - \vct G ( \vct u ) |_{\mtx \Gamma}^2 + \lambda |\vct u |_1.$$ 
Even though this \index{objective}objective function promotes sparse solutions, samples
from the underlying \index{posterior}posterior distribution are typically not sparse. 
\end{example} 

\section{Theory}\label{sec:32}

For any \index{optimization}optimization problem for an \index{objective}objective function with a finite
infimum, it is of interest to determine whether the infimum is attained. 
We have the following result which shows that, under suitable conditions on $\J$, the infimum of $\J$
is attained and hence that the formulation of the \index{MAP estimator}MAP estimator through maximization
of $\post^y$ (equivalently minimization of $\J$) is well-defined.

		\begin{theorem}[Attainable \index{MAP estimator}MAP Estimator]
			\label{thm:achievable-obj}
			Assume that $\J$ is non-negative, continuous and that $\J(u) \to \infty$ as $|u| \to \infty. $
Then $\J$ attains its infimum. Therefore, 
the \index{MAP estimator}MAP estimator of $u$ based on the \index{posterior}posterior $\post^y(u) \propto \exp\bigl(-\J(u)\bigr)$ is attained.
		\end{theorem}
		\begin{proof}
%
			By the assumed growth and non-negativity of $\J$, there is $R$ such that $\inf_{u\in \Ru} \J(u) = \inf_{u\in B(0,R)}   \J(u),$ 
where  (recall)  $B(0,R)$ denotes the closed ball of radius $R$ around the origin. Since $\J$ is assumed to be continuous, its infimum over $B(0,R)$ is attained and the proof is complete. 
		\end{proof}

\begin{remark}
Suppose that: 
\begin{enumerate}
				\item \(\vct G \in C(\Ru, \Ry)\), i.e. \(\vct G\) is a continuous function; 
				\item the \index{objective}objective function \(\J(\vct u)\) has \(\losss_2\) \index{loss}loss as defined in Example \ref{ex:loss-l2} and \(\losss_p\) regularizer \(\reg(\vct u) = \frac{\lambda}{p} |\vct u|_p^p\), \(p \in [1,\infty)\).
			\end{enumerate}
Then the assumptions on $\J$ in Theorem \ref{thm:achievable-obj} are satisfied.  This shows that if \(\vct G\) is continuous, the infimum of $\J$ defined with \(\losss_2\) \index{loss}loss and \(\losss_p\) regularizer is attained at the \index{MAP estimator}MAP estimator of the  corresponding \index{Bayesian}Bayesian problem with \index{posterior}posterior pdf proportional to $\exp\bigl(-\J(u)\bigr)$. 
\end{remark}

\begin{remark}
Notice that the assumption that $\J(u) \to \infty$ is not restrictive: this condition needs to hold in order to be able to normalize $\post^y (u)\propto \exp\bigl(-\J(u)\bigr)$ into a pdf, which is implicitly assumed in the second part of the theorem statement. 
\end{remark}

Intuitively, the \index{MAP estimator}MAP estimator maximizes \index{posterior}posterior probability. We make this
precise in the following theorem, which links the \index{objective}objective function $\J(\cdot)$
to small ball probabilities. 

		\begin{theorem}[\index{objective}Objective Function and \index{posterior}Posterior Probability]
			\label{thm:ball-density}
			 Under the same assumptions as in Theorem \ref{thm:achievable-obj}, let
			\[
			\alpha(\vct u,\delta) := \int_{\vct v \in B(\vct u, \delta)} \post^y(\vct v) d\vct v = \Prob^{\post^y}\bigl(B(\vct u,\delta)\bigr)
			\] 
			be the \index{posterior}posterior probability of a ball with radius \(\delta\) 
centered at \(u\). Then, for all \(\vct u, \vct u' \in \Ru\), we have
			\[
			\lim_{\delta \rightarrow 0} \frac{\alpha(\vct u, \delta)}{\alpha(\vct u', \delta)} = e^{\J(\vct u') - \J(\vct u)}.
			\]
		\end{theorem}
		\begin{proof}
Let $u, u' \in \Ru$ and let $\epsilon>0.$ By continuity of $\J$ we have that, for all $\delta$ sufficiently small,
\begin{align*}
e^{-\J(u) - \epsilon} \le e^{-\J(v)} &\le e^{-\J(u) + \epsilon} \quad \text{for all} \,\, v \in B(u, \delta), \\
e^{-\J(u') - \epsilon} \le e^{-\J(v)} &\le e^{-\J(u') + \epsilon} \quad \text{for all}  \,\, v \in B(u', \delta).
\end{align*}
Therefore, for all $\delta$ sufficiently small,
\begin{align*}
&B_\delta e^{-\J(\vct u) - \epsilon}
			\leq 
			\int_{v \in B(u,\delta)} e^{-\J(\vct v)} d\vct v \leq 
			 B_\delta e^{-\J(\vct u) + \epsilon}, \\
&B_\delta e^{-\J(\vct u') - \epsilon}
			\leq 
			\int_{v \in B(u',\delta)} e^{-\J(\vct v)} d\vct v \leq 
			 B_\delta e^{-\J(\vct u') + \epsilon}, 
\end{align*}
		where \(B_\delta \) is the \index{Lebesgue}Lebesgue measure of a ball with radius \(\delta\). Taking the ratio of $\alpha$'s and using the above bounds we obtain that, for all $\delta$ sufficiently small,		
		\[
			e^{\J(\vct u') - \J(\vct u) -2\epsilon} \leq \frac{\alpha(\vct u, \delta) }{\alpha(\vct u', \delta)} \leq  e^{\J(\vct u') - \J ( \vct u) + 2\epsilon}. 
			\]
			Since $\epsilon>0$ is arbitrary, the desired result follows.

		\end{proof}

\begin{remark}
This theorem shows that maximizing the probability of an
infinitesimally small ball is the same as minimizing the \index{objective}objective
function $\J(\cdot).$ This is intuitive in finite dimensions, but
the proof above generalizes beyond measures
which possess a \index{Lebesgue}Lebesgue density, and may be used in infinite
dimensions. 
\end{remark}

\section{Examples}\label{sec:33}

By means of examples, we now probe whether the \index{MAP estimator}MAP
estimator captures useful information about the \index{posterior}posterior distribution.

\begin{example}[Summarizing Single-Peaked Posterior]
	If the \index{posterior}posterior is single-peaked, such as a \index{Gaussian}Gaussian or a \index{Laplace distribution}Laplace distribution, as shown in Figure \ref{fig:gaussian-laplace}, the \index{MAP estimator}MAP estimator, i.e. minimizer of the \index{objective}objective function, reasonably summarizes the most likely value of the unknown parameter.
\end{example}

\begin{figure}
  \centering
  \includegraphics[width=0.7\columnwidth]{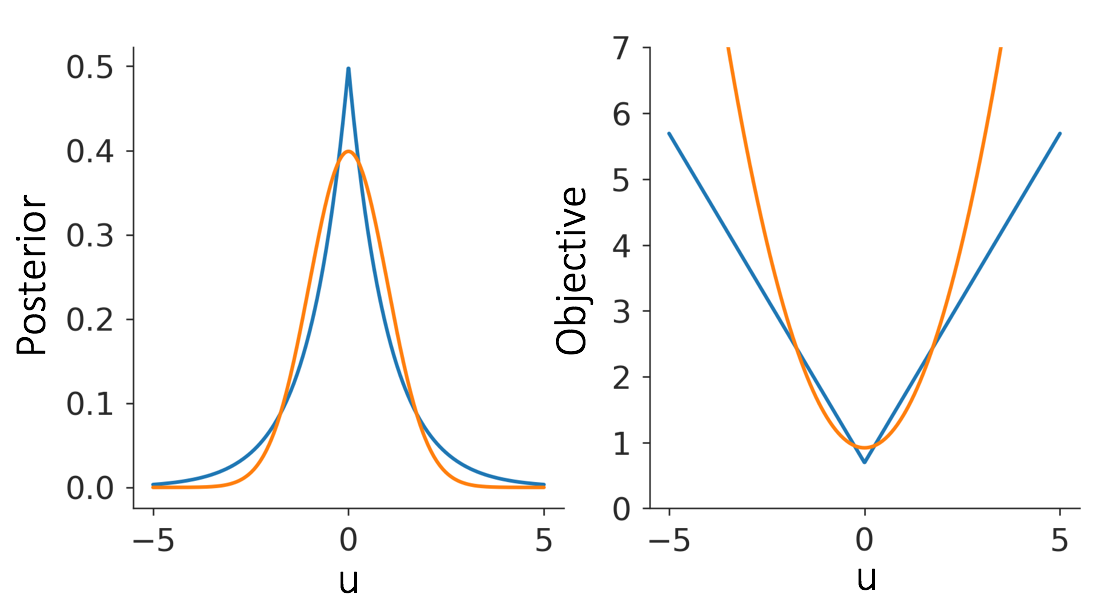}
  \caption{\index{posterior}Posterior (left) and \index{objective}objective function (right) for \(\Nc(0,1)\) \index{posterior}posterior (orange) and \index{Laplace distribution} \(\mathrm{Laplace}(0,1)\) \index{posterior}posterior (blue). }
  \label{fig:gaussian-laplace}
\end{figure}

We next consider several examples where a point estimator ---or a \(\delta\)-radius ball with small $\delta$--- fails to adequately summarize the \index{posterior}posterior distribution.
\begin{example}[Summarizing Multiple-Peaked Posterior]
	If the \index{posterior}posterior is rather unevenly distributed, such as a slab-and-spike distribution, as shown in Figure~\ref{fig:slab-n-spike}, then it is
less clear that the \index{MAP estimator}MAP estimator usefully summarizes the \index{posterior}posterior. For example, for the case in Figure \ref{fig:slab-n-spike} we may want the solution output of our \index{Bayesian}Bayesian problem to be a weighted average of two \index{Gaussian}Gaussian distributions, or two point estimators each with a separate 
mean located at one of the two minima of the \index{objective}objective functions, 
and weight describing the probability mass associated with each of
those two points.
\end{example}

\begin{figure}
	\centering
	\includegraphics[width=0.9\columnwidth]{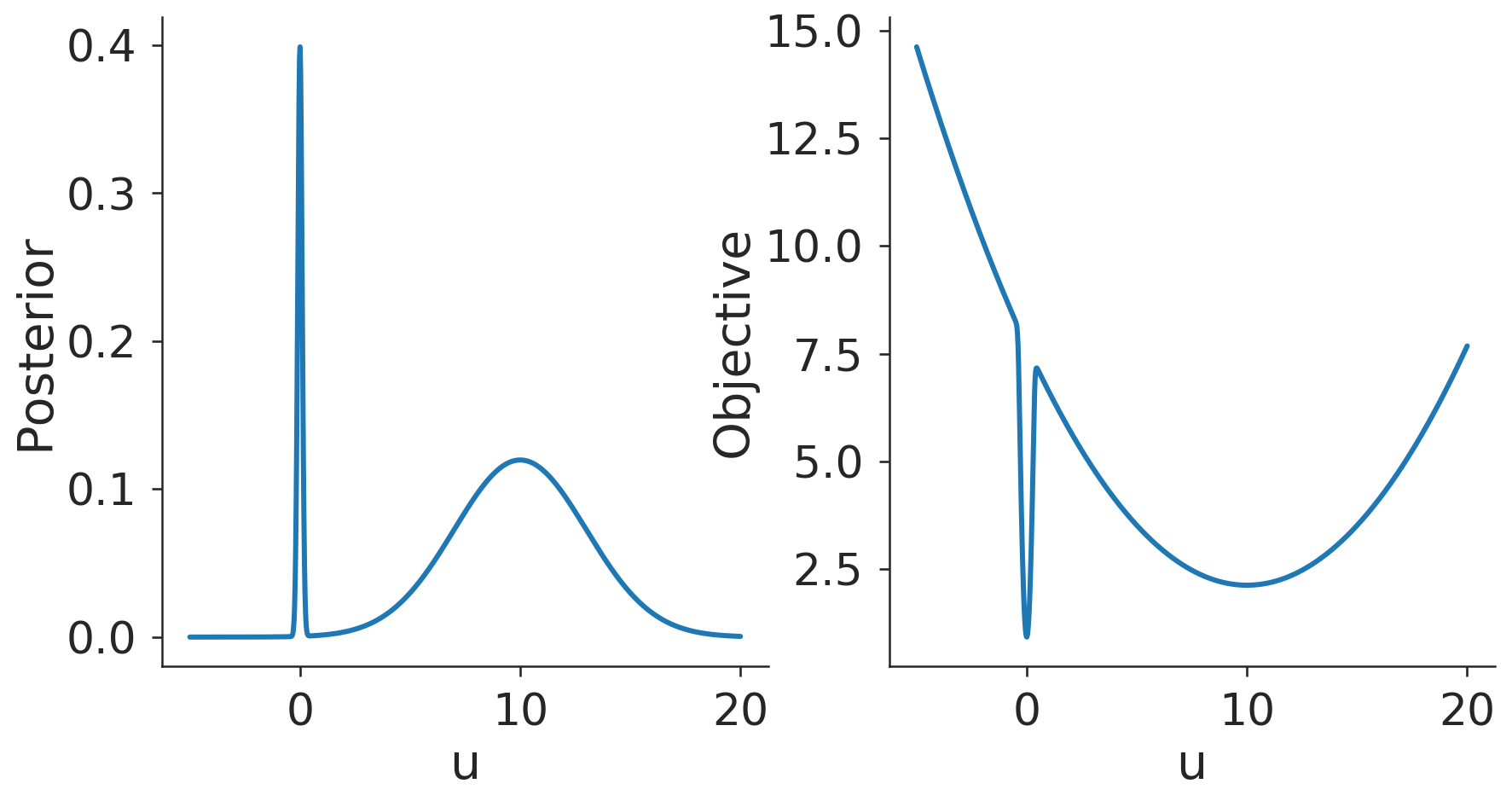}
	\caption{\index{posterior}Posterior (left) and \index{objective}objective function (right) for a \index{posterior}posterior that is a sum of two \index{Gaussian}Gaussian distributions, \(\Nc(0,0.1^2)\) with probability \(0.1\) and \(\Nc(10,3^2)\) with probability \(0.9\).}
	\label{fig:slab-n-spike}
\end{figure}

\begin{example}[Summarizing Rough Posteriors]
	In addition to a multiple-peak \index{posterior}posterior, there are cases where the \index{objective}objective function and the associated \index{posterior}posterior pdf are simply very rough. In these cases, the small-scale roughness should be ignored, while the large-scale variation should be captured.  For example, the \index{objective}objective function in Figure \ref{fig:rough}  is very rough and has a unique minimizer at a point far from \(0\). However, it also has a larger-scale pattern: it tends to be smaller around \(0\), while larger away from \(0\). The \index{MAP estimator}MAP estimator cannot capture this large scale pattern, as it is found by minimizing the \index{objective}objective function. It is arguably the
case that \(u=0\) is a better point estimate. An alternative way to interpret this phenomenon is that there is a natural ``temperature'' to this problem, in the sense that variations lower than this temperature could be viewed as random noise that do not capture meaningful information.
\end{example}
\begin{figure}
	\centering
	\includegraphics[width=0.9\columnwidth]{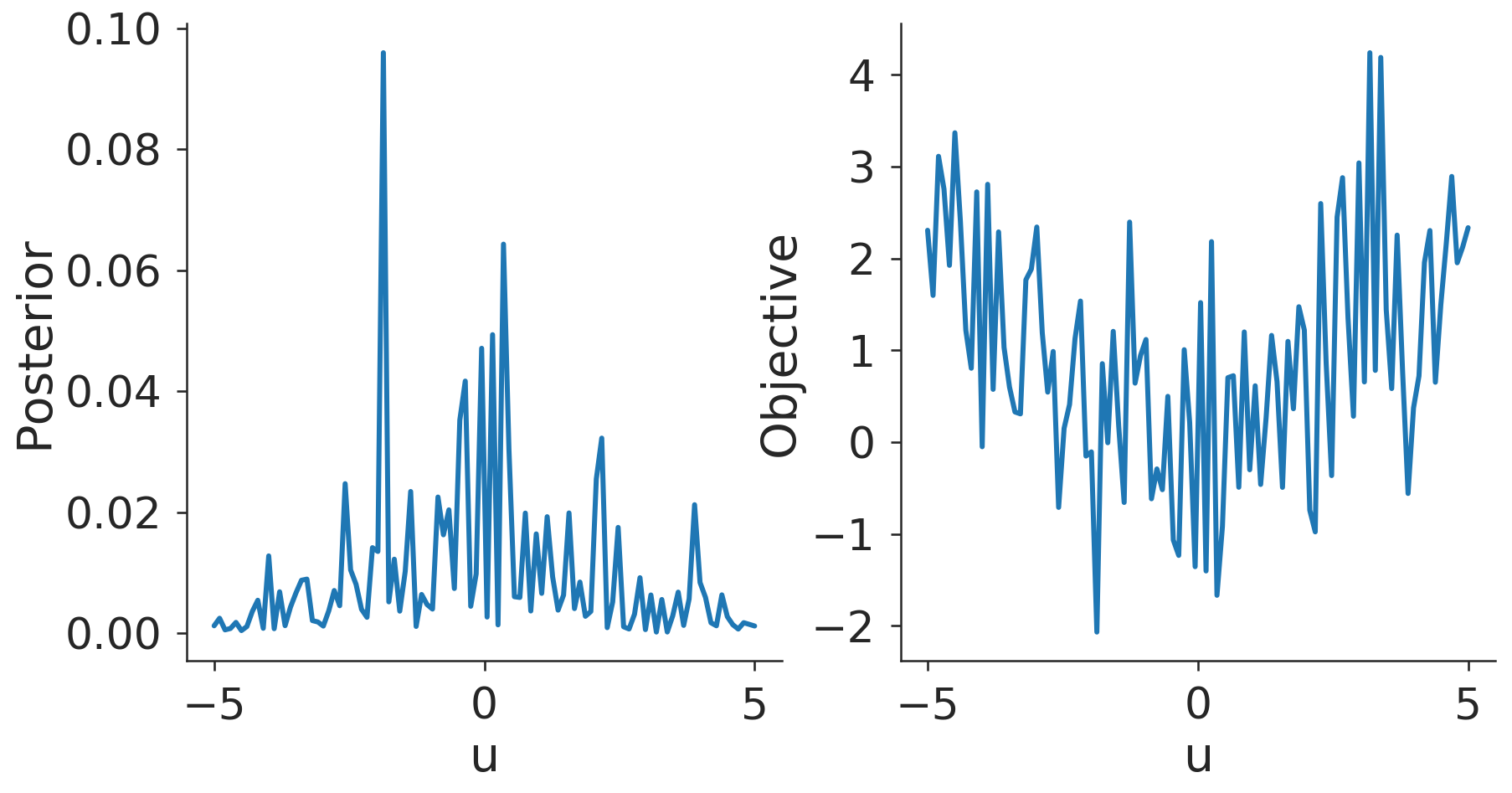}
	\caption{\index{posterior}Posterior (left) and \index{objective}objective function (right) from an \index{objective}objective function that is very rough in the small scale, but contains a regular pattern on the larger scale. This specific example is generated by white noise summed with a quadratic function for the \index{objective}objective function, and the \index{posterior}posterior is computed from the \index{objective}objective function.}
	\label{fig:rough}
\end{figure}

The preceding examples suggest that multi-peak distributions, or multi-minimum 
\index{objective}objective functions, can cause problems for \index{MAP estimator}MAP estimation.
Next we illustrate that if the dimension \(\Du\) of the parameter \(\vct u \in \Ru\) is high, then a single point estimator, even if a \index{MAP estimator}MAP estimator, is typically not a good summary of the \index{posterior}posterior.
\begin{figure}[ht!]
	\centering
	\includegraphics[width=0.7\columnwidth]{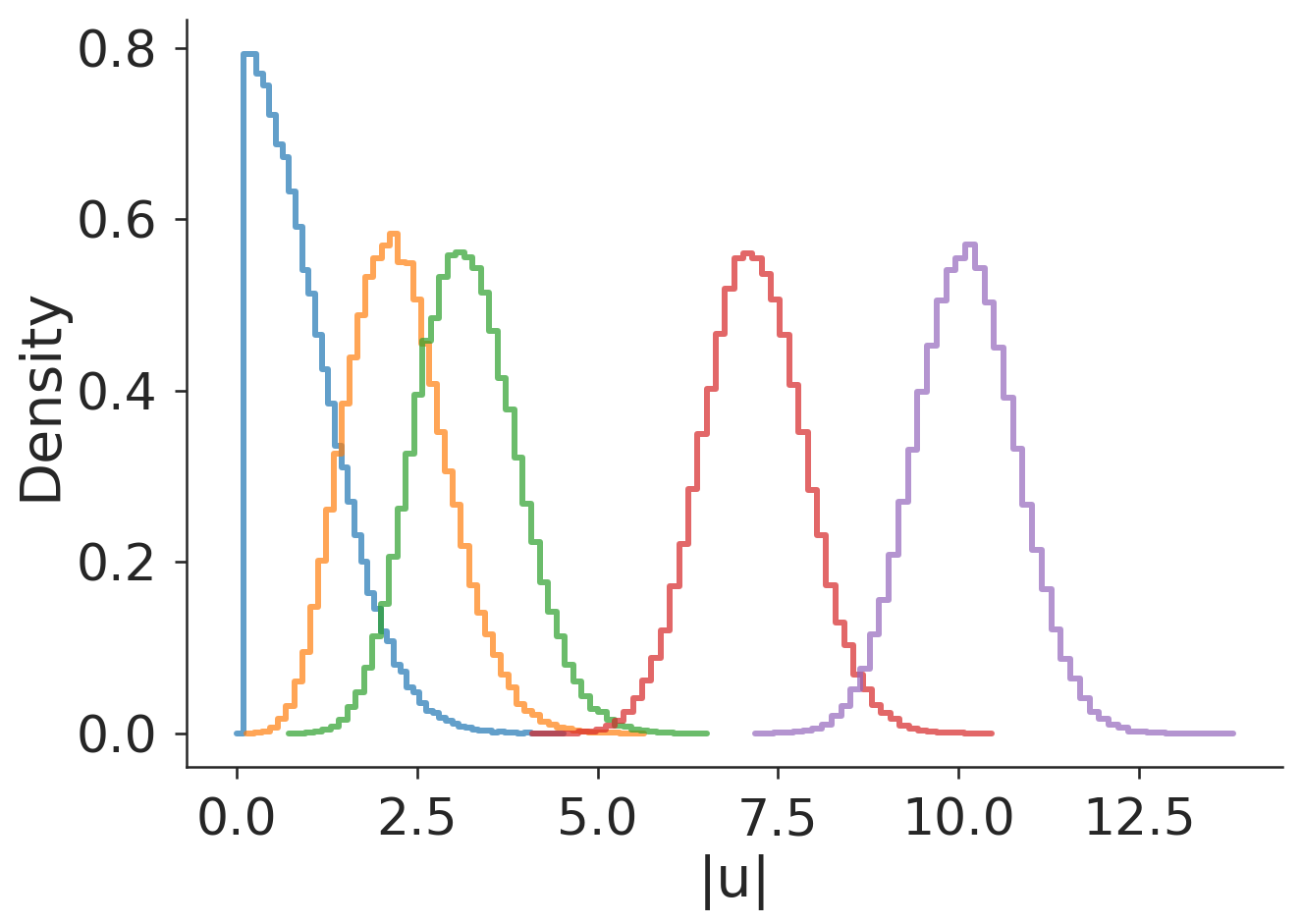}
	\caption{Empirical density of \(\ell_2\) norm of \(\Nc(\zerovct,\Id)\) random vectors for various dimension: \(\Du=1\) (blue), \(\Du=5\) (orange), \(\Du=10\) (green), \(\Du=50\) (red), and \(\Du=100\) (purple). The empirical density is obtained from \(10000\) samples for each distribution.}
	\label{fig:highdim}
\end{figure}
\begin{example}[Summarizing High-Dimensional Posterior]
	We consider what is the ``typical size'' of a vector \(\vct u\) drawn 
from the standard \index{Gaussian}Gaussian distribution \(\Nc(\zerovct, \Id)\), as the dimension
increases. In Figure \ref{fig:highdim} we display the empirical density of the norm of such random vectors. We can see that at low dimensions, such as when \(\Du=1\), obtaining a value close to the mode \(\vct u=0\) is highly likely. In higher dimensions, however,
the probability for a vector from this distribution to have a small $\ell_2-$norm becomes increasingly small as $\Du$ grows. 
 For example, let us consider the probability for the norm to be less than \(5\). Then \(\Prob (|\vct u| <5)\) is \(0.99999943\) when \(\Du=1\), \(0.99986\) when \(\Du=5\), \(0.99465\) when \(\Du=10\), \(0.001192\) when \(\Du=50\), and \(1.135 \times 10^{-15}\) when \(\Du=100\). So we see that, as the dimension increases, with probability close to \(1\) a sample from the \index{posterior}posterior would 
have a norm far from \(0\). Indeed, for \(\Du=1000\), the \(5\)th and \(95\)th percentiles are respectively \(30.3464\) and \(32.7823\). This means when \(d=1000\), we most likely will find a vector with size around \(31\), not \(0\). 
Another way to see this is that, since the components $u_\du$ of $u$
are \index{i.i.d.}i.i.d. standard unit \index{Gaussian}Gaussians we have that, by the strong law of large
numbers,
$$\frac{1}{\Du}\sum_{\du=1}^\Du u_\du^2 \to 1$$
as $\Du \to \infty$ almost surely. Thus, with high probability, the  
$\ell_2-$norm is of size $\sqrt{\Du}.$
This example suggests that in high dimension, a point estimator may not capture enough information about the density.
\end{example}

The preceding examples demonstrate that \index{MAP estimator}MAP estimators should be treated
with caution, as they may not capture the desired \index{posterior}posterior information in
many cases. This motivates the study of alternative ways ---beyond \index{MAP estimator}MAP estimators--- to capture information from the \index{posterior}posterior distribution. One such approach is to fit one or several \index{Gaussian}Gaussian distributions to the \index{posterior}posterior by minimizing an appropriate \index{distance}distance-like measure between distributions. This is the topic of the next chapter. 
However, in the remainder of this chapter we discuss gradient-based methods for minimization.  These may be
useful for MAP estimation, and also for fitting Gaussian approximations.

\section{Gradient-Based Optimization Algorithms}\label{sec:34}
\label{ssec:GA}

In this section we discuss algorithms for the minimization of
\(\J: \R^{\Du} \mapsto \R.\) Algorithms for the \index{optimization}optimization of functions
of this type are numerous, and vary considerably in type. In order 
to focus our discussion, we devote our attention entirely to 
gradient-based algorithms. These are organized around a single
important principle, and are also of interest due
to their use in parameter estimation arising in machine learning
(a form of \index{inverse problem}inverse problem).

\subsection{Gradient Flow}

Our starting point is the differential equation
\begin{equation}
\label{eq:gf}
\frac{du}{dt}=-D\J(u), \quad u(0)=u_0.
\end{equation}
A straightforward calculation shows that
\begin{equation}
\label{eq:gf}
\frac{d}{dt}\bigl(\J(u)\bigr)=\Bigl\langle D\J(u), \frac{du}{dt} 
\Bigr\rangle=-|D\J(u)|^2.
\end{equation}
This calculation is at the core of gradient-based \index{optimization}optimization
algorithms. Since the time-derivative of $u(t)$ gives the tangent to the
trajectory, it demonstrates that evolving in the direction of the
negative gradient of $\J(u)$ will cause  
$\J\bigl(u(t)\bigr)$ to be non-increasing as a function of time; indeed
$\J\bigl(u(t)\bigr)$ will actually decrease until $u$ is at a critical
point of $\J(\cdot)$: a point at which the gradient is zero, including
local minima, local maxima and saddle points.

For any $K\in \R^{\Du \times \Du},$ that we will assume \index{positive definite}positive definite in what follows, we may also consider the \emph{preconditioned} gradient flow
\begin{equation}
\label{eq:gf2}
\frac{du}{dt}=-KD\J(u), \quad u(0)=u_0.
\end{equation} 

\subsection{\index{gradient descent}Gradient Descent}

In order to turn the gradient flow \eqref{eq:gf2} into an \index{optimization}optimization algorithm,
we discretize it by the Euler method with variable time-step $\alpha_\ell>0.$

\FloatBarrier
\begin{algorithm}
\caption{\label{algMH} \index{gradient descent}Gradient Descent Algorithm}
\begin{algorithmic}[1]
\vspace{0.1in}
\STATE {\bf Input}: \index{objective}Objective function $\J:\R^\Du \to \R,$ \index{positive definite}positive definite matrix $K,$ initialization $u_0\in \R^\Du,$ number of steps $L,$  rule for choosing the step-sizes  $\{\alpha_\ell\}_{\ell = 0}^{L-1}.$  \\
\vspace{.04in}
\STATE For $\ell = 0,1,\dots,L -1$ do:
$$u_{\ell+1}=u_\ell-\alpha_\ell K D\J(u_\ell).$$
\STATE{\bf Output}: Deterministic iterates $u_0, u_1, \ldots, u_L.$
\label{alg_1}
\end{algorithmic}
\end{algorithm}
\FloatBarrier

It is natural to ask how $\alpha_\ell$ should be chosen. In order to get 
insight into this issue, we study in detail the case where $K=I$ and $\J(u)$ is quadratic. The latter condition ensures that  the iteration for $u_\ell$ is linear in the case of fixed $\alpha_\ell$;
it is however nonlinear when $\alpha_\ell$ is adapted, as it is here, on
the basis of $u_\ell.$

Let $A \in \R^{\Du \times \Du}$ be \index{positive definite}positive definite, let $b \in \R^\Du,$
and define
\begin{equation}
\J(u)=\frac12|b-Au|_{A}^2.
\end{equation}
This strictly convex function has minimum $u^\star$ which is the solution
of the linear system
\begin{equation}
\label{eq:linmin}
Au^\star=b.
\end{equation}
The gradient flow \eqref{eq:gf} gives the
linear differential equation
$$\frac{du}{dt}=b-Au$$
and has unique globally attracting fixed point at $u^\star.$

The resulting discrete time-step algorithm is
$$u_{\ell+1}=u_\ell+\alpha_\ell(b-Au_\ell).$$
The first question we ask is how $\alpha_\ell$ should be chosen to maximize 
the decrease in $\J(\cdot)$ in one step of the algorithm. 
We address this in the next lemma and then, using this optimal
time-step, we study the convergence properties of the algorithm.
With this goal in mind, it is helpful to define the {\em residual function}
$r: \R^d \to \R^d$ by $r(u)=b-Au.$ Given the sequence $\{u_\ell\}$
we may then define the {\em residual vector} $r_\ell=r(u_\ell).$
Then $\J(u)=\frac12|r(u)|_A^2$, $\J(u_\ell)=\frac12|r_\ell|^2_{A}$ and
$$u_{\ell+1}=u_\ell+\alpha_\ell r_\ell.$$

\begin{lemma}
Choosing
\begin{equation*}
\alpha_\ell = \frac{|r_\ell|^2}{|r_\ell|_{{\mA}}^2}
\end{equation*}
leads to the maximal decrease in $\J(\cdot)$ and to the algorithm
\begin{equation}
\label{eq:gdopt}
u_{\ell+1}=u_\ell+\frac{|r_\ell|^2}{|r_\ell|_{\mA}^2} r_\ell.
\end{equation}
\end{lemma}

\begin{proof}
We have
\begin{align*}
Au_{\ell+1}&=Au_\ell+\alpha_\ell Ar_\ell,\\
b&=b,
\end{align*}
so that subtracting gives
$$r_{\ell+1}=r_\ell-\alpha_\ell Ar_\ell.$$
From this it follows that
\begin{equation}\label{eq:auxeq}
\J(u_{\ell+1})=\J(u_\ell) - \alpha_\ell |r_\ell|^2 + \frac12 \alpha_\ell^2
|r_\ell|_{{\mA}}^2.
\end{equation}
The right-hand side is quadratic in $\alpha_\ell$ and minimized at the
prescribed choice of $\alpha_\ell.$
\end{proof}

\begin{theorem}[Conditioning of $A$ and Decrease of $\J$]
Let $A$ have maximal and minimal eigenvalues $\lambda_{\max} \ge 
\lambda_{\min}>0,$ respectively. Then
\begin{equation}\label{eq:auxineq}
\J(u_{\ell+1}) \le \Bigl(1-\frac{\lambda_{\min}}{\lambda_{\max}}\Bigr)\J(u_\ell).
\end{equation}
\end{theorem}

\begin{proof}
Substituting the optimal choice of $\alpha_\ell$ into equation \eqref{eq:auxeq} gives
\begin{align*}
\J(u_{\ell+1})&=\J(u_\ell)-\frac12 \frac{|r_\ell|^4}{|r_\ell|_{{\mA}}^2}\\
&=\J(u_\ell)- \frac{ |r_\ell|^4}{|r_\ell|_{{\mA}}^2 |r_\ell|_{A}^2}\J(u_\ell).
\end{align*}
Applying the result of Lemma \ref{l:ratios} below gives the desired result. 
\end{proof}

\begin{remark}
Inequality \eqref{eq:auxineq} suggests slow convergence of the algorithm for matrices
$A$ which have a large condition number, i.e. for which
$\lambda_{\max} \gg \lambda_{\min}.$
In principle this can be ameliorated by preconditioning the algorithm
by choosing $K=A^{-1}$ so that the preconditioned steepest
descent iteration becomes
$$u_{\ell+1}=u_\ell+\alpha_\ell(A^{-1}b-u_\ell).$$
The optimal choice of $\alpha_\ell$ for this iteration becomes $\alpha_\ell = 1,$ which gives $u_{\ell +1 } = A^{-1} b.$ Thus, the algorithm converges in one step, regardless of the initial condition. However, implementing the algorithm with $K=A^{-1}$ would require computation of $A^{-1}b;$
the goal of the descent algorithm is, of course, to avoid computation
of $A^{-1}$ in the first place. This discussion illustrates nonetheless the potential practical advantage of preconditioning using a \index{positive definite}positive definite matrix  $K\approx A^{-1}$ whose action on vectors
can nonetheless be computed much more cheaply than that of $A^{-1}$ itself. 
\end{remark}

\begin{lemma}
For any $u \in \R^{\Du},$
\label{l:ratios}
\begin{equation*}
\frac{|u|^4}{|u|_{{\mA}}^2 |u|_{A}^2}
  \geq
  \frac{\lambda_{\min}}{\lambda_{\max}}.
\end{equation*}

\end{lemma}

\begin{proof}
Since $A$ is assumed to be \index{positive definite}positive definite, the eigenvalue problem for $A$ has solutions with the form
\begin{align*}
      A\phi_{i} &= \lambda_{i}\phi_{i},\quad i=1,\dots, \Du, \\
      \la \phi_{i}, \phi_{j}\ra &= \delta_{ij}, \quad i,j = 1, \ldots, \Du,
\end{align*}
where we may assume the ordering
$$0<\lambda_{\min}:= \lambda_{1} \leq \cdots \le
\lambda_\Du =:\lambda_{\max}.$$
Expanding $u \in \R^d$ in this eigenbasis, we have
\begin{equation*}
u = \sum_{\du=1}^{\Du}u_{\du}\phi_{\du}
\end{equation*}
with $u_\du=\la u, \phi_\du \ra.$
Now, noting that
$$|u|^2= \sum_{\du=1}^\Du  u_\du^2, \quad |u|^2_{{\mA}}=\sum_{\du=1}^\Du \lambda_\du u_\du^2,
\quad |u|^2_{A}=\sum_{\du=1}^\Du \frac{u_\du^2}{\lambda_\du},$$
we get
\begin{align*}
|u|^2 &= \sum_{\du=1}^{\Du}u_{\du}^{2} \geq
\lambda_{\min}|u|_{A}^2,\\
|u|_{{\mA}}^2 &= \sum_{\du=1}^{\Du}\lambda_{\du}u_{\du}^{2} \leq
\lambda_{\max}|u|^2.
\end{align*}
The desired result follows.
\end{proof}

\subsection{\index{gradient descent!stochastic}Stochastic Gradient Descent}

%
%
%


Here we consider optimizing a stochastically defined \index{objective}objective function. 
This concerns the setting where
\begin{equation}\label{eq:JforSGD}
\J(u)=\int_{B} F(u,z)  \zeta  (z) \, dz,
\end{equation}
$B \subseteq \R^{ \dz },$ and $ \zeta  $ is the pdf of a random variable
$z \in B$. The goal is \index{optimization}optimization of $\J(u).$

\index{gradient descent!stochastic}Stochastic gradient descent is designed to numerically solve this \index{optimization}optimization problem in cases where explicit evaluation
of $\J(u)$, and its gradient $D\J(u)$, is not possible because
doing so involves an integration over $B$. It is assumed, however,
that $D_u F(u,z)$ can be evaluated  
for any fixed $z \in B \subseteq\R^{ \dz }.$ 
The proposed algorithm is then the following:
\FloatBarrier
\begin{algorithm}
\caption{\label{algSGD} \index{gradient descent!stochastic}Stochastic Gradient Descent Algorithm}
\begin{algorithmic}[1]
\vspace{0.1in}
\STATE {\bf Input}: \index{objective}Objective function $\J$ defined implicitly by \eqref{eq:JforSGD}, \index{positive definite}positive definite matrix $K,$ initialization $u^{(0)}\in \R^\Du,$ number of steps $L,$  rule for choosing the step-sizes  $\{\alpha_\ell\}_{\ell = 0}^{L-1}.$  \\
\vspace{.04in}
\STATE For $\ell = 0,1,\dots,L -1$ do: \\
$u^{(\ell+1)}=u^{(\ell)}-\alpha_\ell K D_u F(u^{(\ell)},z^{(\ell)})\quad   \text{with} \,\,  z^{(\ell)}\sim   \zeta  \,\, \index{i.i.d.}\text{i.i.d.}$
\STATE{\bf Output}: Random iterates $u^{(0)}, u^{(1)}, \ldots, u^{(L)}.$
\label{alg_1}
\end{algorithmic}
\end{algorithm}
\FloatBarrier
The output of the algorithm defines
an (in general) inhomogeneous \index{Markov chain}Markov chain; it will be homogeneous
if $\alpha_\ell$ is constant in $\ell$. \index{Markov chain}Markov chains are discussed in more detail in Chapter \ref{ch:6}. In what follows we will show the convergence of the algorithm in a simple setting, amenable to a concrete analysis. We will also motivate the importance of the algorithm in a machine learning context.

Our convergence analysis will rely on the following assumption.
\begin{assumption}\label{assumptionobjective}
The \index{objective} objective function $\J$ in \eqref{eq:JforSGD} satisfies:
\begin{itemize}
\item[(i)] There exists $c_1$ such that, for all $u \in \R^d,$ $\sup_{z \in B} | D_u F(u,z) |^2 \le c_1 .$
\item[(ii)] There exists $c_2>0$ such that, for all $u, v \in \R^d,$ 
\begin{equation}\label{eq:strongconvex}
\J(v) \ge \J(u) + \langle D\J(u), v - u \rangle + \frac{c_2}{2} |u - v|^2. 
\end{equation}
\end{itemize}
\end{assumption}

Note that item (i) in Assumption \ref{assumptionobjective} implies a \index{Lipschitz}Lipschitz condition on $F$ over its second argument, while the second item assumes strong convexity of $\J.$ In particular, this second condition implies that, if $\J$ is sufficiently smooth, its Hessian satisfies $D^2 \J \ge c_2 I, $ that is, for all $u \in \R^d$ the matrix $D^2 \J(u) - c_2 I$ is \index{positive definite}positive definite. 

\begin{theorem}[Convergence of \index{gradient descent!stochastic}Stochastic Gradient Descent]
\label{t:add1}
Suppose that Assumption \ref{assumptionobjective} holds. Suppose further that the step-sizes  are positive with $\alpha_\ell \to 0 $ and $\sum_{\ell = 0}^{\infty}  \alpha_\ell= \infty.$ Then the \index{objective!function}objective function $\J$  has a unique minimizer $u^\star$ and  the output of Algorithm \ref{algSGD} satisfies $\Expect \bigl[ |u^{(\ell)} - u^\star|^2 \bigr] \to 0$ as $\ell \to \infty.$
\end{theorem}
\begin{proof}
The existence and uniqueness of the minimizer $u^\star$ of $\J$ follows by the strong convexity in Assumption \ref{assumptionobjective}  item (ii).
Denote $e_\ell = \Expect \bigl[ |u^{(\ell)} - u^\star|^2 \bigr].$ Then, from the definition of the stochastic gradient descent updates, we have that
\begin{align}\label{eq:auxrecur}
\begin{split}
e_{\ell+1} &= \Expect \Bigl[ | u^{(\ell)} - u^\star  - \alpha_\ell D_u F(u^{(\ell)} , z^{(\ell)}) |^2   \Bigr] \\
& = e_\ell + \alpha_\ell^2 \Expect\Bigl[ |D_u F(u^{(\ell)},z^{(\ell)}) |^2  \Bigr]  
- 2 \alpha_\ell \Expect \Bigl[ \bigl\langle u^{(\ell)} - u^\star, D_u F(u^{(\ell)}, z^{(\ell)} ) \bigr\rangle \Bigr] .
\end{split}
\end{align}
By the law of total expectation and the definition of $\J$ in \eqref{eq:JforSGD}, we can rewrite the last expectation in the right-hand side as
\begin{align}\label{eq:auxder}
\begin{split}
\Expect \biggl[  \Expect\Bigl[ \bigl\langle u^{(\ell)} - u^\star, D_u F(u^{(\ell)}, z^{(\ell)} ) \bigr\rangle  \mid  u^{(1)}, \ldots, u^{(\ell)},   z^{(1)}, \ldots, z^{(\ell-1)} \Bigr] \biggr] 
= \Expect\Bigl[ \langle u^{(\ell)} - u^\star, D \J(u^{(\ell)}) \rangle \Bigr]. 
\end{split}
\end{align}
Therefore, using Assumption  \ref{assumptionobjective} items (i) and (ii) to bound the second and third terms in the right-hand side of   \eqref{eq:auxrecur}, we deduce that
\begin{align*}
e_{\ell+1} &\le (1 - \alpha_\ell c_2) e_\ell + \alpha_\ell^2 c_1.
\end{align*}
It follows that, for any $\epsilon >0,$ 
$$ e_{\ell+1} - \epsilon \le (1 - \alpha_\ell c_2) (e_\ell - \epsilon) + \alpha_\ell ( \alpha_\ell c_1 - \epsilon c_2). $$
Note that, for all sufficiently large $\ell$,
$\alpha_\ell ( \alpha_\ell c_1 - \epsilon c_2) <0$. Thus we obtain that, 
for all sufficiently large $\ell,$ 
$$ e_{\ell+1} - \epsilon \leq (1 - \alpha_\ell c_2) (e_\ell - \epsilon).$$
Iterating this inequality gives that, for  some $\ell$ sufficiently 
large and all $m \in \N$, 
$$e_{\ell+m} - \epsilon \le \prod_{j = 0}^{m-1} (1 - \alpha_{\ell + j} c_2) ( e_\ell - \epsilon).$$
Recall that for $x \in (0,1)$ we have that $\log(1-x) \le -x$ (a proof can be found in Chapter \ref{ch4}, Lemma \ref{l:kllemma}).
 Now notice that, as $m \to \infty,$
\begin{equation}
0 \le \prod_{j = 0}^{m-1} (1 - \alpha_{\ell + j}  c_2) \le \exp \biggl(  \,  \sum_{j= 0}^{m-1} -\alpha_{\ell + j}  c_2   \biggr)\to 0
\end{equation} 
since by assumption $\sum_{\ell = 0}^{\infty} \alpha_\ell= \infty.$
Thus $e_{\ell+m} \le \epsilon$ for all $m$ large enough, and the desired result follows since $\epsilon$ is arbitrary. 
\end{proof}


\begin{example}[Stochastic Gradient Descent in Machine Learning]
\label{ex:add1}
Although the original motivation for the algorithm was settings in which
$D\J(u)$ is not explicitly calculable, the methodology has gained importance
in machine learning \index{optimization}optimization tasks where the motivation is different.
Consider an \index{objective}objective function defined by
$$\J(u)=\frac12 \sum_{i=1}^{ \dz } |y^i-G^i(u)|^2,$$
where each $y^i \in \R^\Dy$ represents data arising from a \index{forward!model}forward model $G^i$.
We may write this \index{objective}objective in the form of equation \eqref{eq:JforSGD} as follows. Define
$$F(u,z)=  \frac{ \dz }{2} \sum_{i=1}^{ \dz } z_i |y^i-G^i(u)|^2$$
with $z:=(z_1,\ldots, z_{ \dz }) \in \R^{ \dz }.$
Define $e^i:=(0,\ldots,1,\ldots,0)^\top$, the $i$-th unit
vector and let
$$ \zeta  (z)=\frac{1}{ \dz }\sum_{i=1}^{ \dz } \delta(z - e^i).$$
Then
$$\J(u)=\int_{B} F(u,z)  \zeta (z) \, dz$$
with $B$ any bounded set containing all the unit vectors 
$\{e^i\}_{i=1}^{ \dz }.$ This is because if $z \sim  \zeta ,$ then $\Expect[z] =  {\dz}^{-1}  (1, \ldots, 1)^\top \in \R^{ \dz }.$

In this setting the \index{gradient descent!stochastic}stochastic gradient descent algorithm becomes
\begin{align} 
\label{eq:algx}
\begin{split}
i(\ell) &\sim \mathfrak{u}(\{1, \ldots,  \dz \}) \, \,\index{i.i.d.}\text{i.i.d.}  \\
u^{(\ell+1)} &=u^{(\ell)} + \alpha_\ell  \dz   DG^{i(\ell)}(u^{(\ell)})^\top\Bigl(y^{i(\ell)}-G^{i(\ell)}(u^{(\ell)})\Bigr),
\end{split}
\end{align}
where the notation signifies that $i(\ell)$ is chosen uniformly at random from the index set
$\{1,\ldots,  \dz \}.$ 
In the context of machine learning this algorithm has several potential
advantages over standard \index{gradient descent}gradient descent: i) if $ \dz $ is massive (large data sets)
then it is not necessary to hold the entirety of $D\J(u)$ in memory
at any one time; ii) if the data is received in a streaming fashion then
the algorithm can be implemented in a non-random fashion where the indices
$i(\ell)$ are traversed systematically as the components of the data are
received; (iii)  it is observed empirically that the randomness
induced by sampling terms from the summand defining $D\J(u)$ promotes
improved \index{optimization}optimization  for nonconvex $\J(u)$,
in comparison with standard \index{gradient descent}gradient descent,
because the randomness allows escape from local minima and allows for
more rapid traversing of saddle-point neighbourhoods.
\end{example}

We now consider the setting of Example \ref{ex:add1} in which
$G^i(u)=A^i u$ for some \index{positive definite}positive definite matrix $A^i \in \R^{\Du \times \Du}$, $y^i \in \R^{\Du}$
and we modify the definition of $\J$ so that each term employs
a different norm:
$$\J(u)=\frac12 \sum_{i=1}^{ \dz } |y^i -A^i u|_{A^i}^2.$$
We define
$$\bar{y}=\frac{1}{ \dz }\sum_{i=1}^{ \dz } y^i, \quad \bar{A}=\frac{1}{ \dz }\sum_{i=1}^{ \dz } A^i.$$
A straightforward calculation reveals that $\bar{A}$ is \index{positive definite}positive definite
and $\J(u)$ has a unique minimizer 
$u^\star$ solving the equation $\bar{A}u^\star=\bar{y}.$

In this setting the analog of the algorithm from \eqref{eq:algx} becomes
\begin{equation} 
\label{eq:algy}
u^{(\ell+1)}  =u^{(\ell)}  + \alpha_{\ell}  \Bigl(y^{i(\ell)}-A^{i(\ell)}u^{(\ell)} \Bigr),
\end{equation}
where $i(\ell)$ is chosen uniformly at random from $\{1,\dots,  \dz \}$
\index{i.i.d.}i.i.d. at every step, and independently from $u^{(\ell)}.$
 This gives an (in general inhomogeneous)
\index{Markov chain}Markov chain.
Theorem \ref{t:add1} concerning stochastic gradient descent 
made the assumption that the time-step $\alpha_\ell$ decreases to zero
with increasing $\ell.$ Here we choose a fixed time-step $\alpha$
leading to a homogeneous \index{Markov chain}Markov chain;
we prove a positive result about the convergence of the algorithm
in an average sense.

\begin{theorem}[Convergence of Stochastic Gradient Descent -- Constant Step-Size]
\ Let $\alpha_\ell=\alpha>0$ and assume 
that, in \eqref{eq:algy},  $\lim_{\ell \to \infty} \Expect\bigl[ u^{(\ell)} \bigr]$ exists. Then the limit is
given by $u^\star.$
\end{theorem}

\begin{proof}
Take expectation in \eqref{eq:algy} conditional on knowing $u^{(\ell)}$ to obtain
\begin{align*}
\Expect \bigl[  u^{(\ell+1)}   |u^{(\ell)} \bigr]&   =u^{(\ell)} + \alpha \Bigl(\Expect[y^{i(\ell)} ]-\Expect \bigl[ A^{i(\ell)} u^{(\ell)} \bigr] \Bigr)\\
&=u^{(\ell)} + \alpha \Bigl(\bar{y}-\bar{A}u^{(\ell)}\Bigr).
\end{align*}
Taking expectation over $u^{(\ell)}$ gives
$$\Expect \bigl[ u^{(\ell+1)} \bigr] =\Expect \bigl[u^{(\ell)}\bigr] + \alpha \Bigl(\bar{y}-\bar{A}\Expect \bigl[ u^{(\ell)}\bigr]\Bigr).$$
Taking the limit $\ell \to \infty$ and assuming $\lim_{\ell \to \infty} \Expect \bigl[ u^{(\ell)} \bigr]$ 
exists and is given by $u^{\dagger}$ yields 
$$u^{\dagger}=u^{\dagger} + \alpha \bigl(\bar{y}-\bar{A}u^{\dagger}\bigr).$$
Hence $\bar{A}u^{\dagger}=\bar{y}$ and by the invertibility of $\bar{A}$
it follows that $u^{\dagger}=u^{\star}.$
\end{proof}

\section{Discussion and Bibliography}\label{sec:35}
Standard textbooks on \index{optimization}optimization include \cite{nocedal2006numerical,dennis1996numerical,boyd2004convex}. 
The \index{optimization}optimization perspective on inversion predates the development of
the \index{Bayesian}Bayesian approach as a computational tool, because it is typically
far cheaper to implement. The subject of classical \index{regularization}regularization 
techniques for inversion is discussed in \cite{engl1996regularization}.
The concept of \index{MAP estimator}MAP estimators, which links probability to
\index{optimization}optimization, is discussed in the books 
\cite{kaipio2006statistical,tarantola2005inverse} 
in the finite-dimensional setting.
The paper \cite{dashti2013map} studies this connection precisely: it
defines the \index{MAP estimator}MAP estimator for infinite-dimensional \index{Bayesian!inverse problem}Bayesian \index{inverse problem}inverse problems, and the corresponding \index{variational}variational formulation, in the setting of \index{Gaussian}Gaussian
\index{prior}priors and Gaussian noise. 
The paper \cite{helin2015maximum} studies related ideas, but in the 
\index{non-Gaussian}non-Gaussian setting, and \cite{agapiou2017sparsity} generalizes 
the \index{variational}variational formulation of \index{MAP estimator}MAP estimators to \index{non-Gaussian}non-Gaussian \index{prior}priors that 
are sparsity promoting. Recent work sets \index{MAP estimator}MAP estimators
for PDE-based \index{inverse problem}inverse problems within the existing framework of
statistical estimation theory \cite{nickl2018convergence}, 
and also within the framework of $\Gamma$-convergence \cite{ayanbayev2021convergence}. 
The paper \cite{worthen2014towards} shows an  example  of \index{optimization}optimization
based inversion  in  a large-scale geophysical application.

A discussion of gradient-based descent in both continuous and
discrete time may be found in \cite{stuart1998dynamical}.
Stochastic analogues of \eqref{eq:gf} may be used to sample
the probability distribution $\exp\bigl(-\beta \J(u)\bigr)$
and an introduction to this subject may be found in
\cite{pavliotis2014stochastic}. The idea of using \emph{stochastic
approximation} for solving nonlinear equations defined via
an expectation was introduced in the paper \cite{robbins1951stochastic}.
The specific analysis in the case of such equations defined
as a gradient, and in particular the statement and proof
of a  result closely related to 
Theorem \ref{t:add1}, may be found in \cite{kiefer1952stochastic}. 
The link to machine learning, described in Example \ref{ex:add1}, is 
overviewed in \cite{goodfellow2016deep}. The paper \cite{bottou2018optimization} provides an accessible introduction to \index{optimization}optimization methods for large-scale machine learning.



 \chapter{\Large{\sffamily{Gaussian Approximation}}} 
\label{ch4}

Recall the \index{inverse problem}inverse problem of finding $u$ from $y$ given by \eqref{eq:jc0}, 
and the \index{Bayesian}Bayesian formulation which follows from Assumption \ref{a:jc1}. 
In the previous chapter we explored the idea of obtaining a point estimator using an \index{optimization}optimization perspective arising
from maximizing the \index{posterior}posterior pdf. We related this idea to finding the center of a ball of radius $\delta$ with maximal probability in the limit $\delta \to 0^+.$
Whilst the idea is intuitively appealing, and reduces the complexity of \index{Bayesian!inference}Bayesian inference from determination of a pdf to
determination of a single point, the approach has a number of limitations, 
in particular for noisy, multi-peaked or high-dimensional 
\index{posterior}posterior distributions; the examples in the previous chapter 
illustrated these limitations. 

In this chapter we again adopt an \index{optimization}optimization approach to the problem of \index{Bayesian!inference}Bayesian inference, but instead seek a \index{Gaussian}Gaussian distribution $p = \Nc(\vct{\mu}, \vct{\Sigma})$ that minimizes some \index{distance}distance-like measure from the \index{posterior}posterior 
$\post^y(\vct{u})$. However, rather than using a metric to define the
\index{distance}distance, we use the Kullback-Leibler \index{divergence!Kullback-Leibler}divergence introduced in Section \ref{sec:41}. Since this  \index{divergence}divergence 
is not symmetric, we obtain to two distinct minimization problems described, in turn, in Sections \ref{sec:existence} and \ref{sec:43}. Both  approaches are compared in Section \ref{sec:44}. In Section \ref{sec:posterior} we show how \index{Bayes theorem}Bayes theorem itself can be formulated through a closely related minimization principle. The chapter closes in Section \ref{sec:46} with extensions and bibliographical remarks.

\section{The Kullback-Leibler Divergence}\label{sec:41}
\begin{definition}\label{def:KLdivergence}
Let $\post,  \post' >0$ be two pdfs on $\Ru$.\footnote{The definition extends to situations where the support of $\post'$ is not
the whole of $\Ru$, provided $\post$ is absolutely continuous with respect to
$\post'$.}
The \index{divergence!Kullback-Leibler}{\em Kullback-Leibler divergence},  also known as {\em relative entropy}, 
of $\post$ with respect to $\post'$ is defined by 
\begin{equation*}
\begin{split}
\dkl(\post \| \post') &:=  \int_{\Ru} \log{\left(\frac{\post(u)}{\post'(u)}\right) \post(u)du  }\\
 &= \Expect{^\post\left[ \log{\left(\frac{\post}{\post'}\right) }\right]} \\
 &= \Expect{^{\post'}\left[ \log{\left(\frac{\post}{\post'}\right)\frac{\post}{\post'} }\right]}.
\end{split}
\end{equation*}
\end{definition}


\index{divergence!Kullback-Leibler}Kullback-Leibler  is a \index{divergence}divergence in that $\dkl(\post \|\post') \ge 0,$ with equality if and only if $\post = \post'.$
From the definition it is clear
that $\dkl(\post \|\post')=0$ if $\post = \post';$ that it is otherwise
strictly positive is proved in Lemma \ref{l:kllemma} below, as a consequence
of the analogous property for the \index{distance!Hellinger}Hellinger or \index{distance!total variation}total variation distances. 
However, unlike \index{distance!Hellinger}Hellinger and \index{distance!total variation}total variation, it does not define a metric. In particular, 
the  \index{divergence!Kullback-Leibler}Kullback-Leibler  divergence is not symmetric: in general,
$$ \dkl(\post \| \post') \neq \dkl(\post' \| \post),$$
a fact that will be important in this chapter. 
Nevertheless, it is useful for at least four reasons: (1) it provides an upper bound for many \index{distance}distances,  as illustrated in Lemma \ref{l:kllemma} below; (2) its logarithmic structure allows explicit computations
that are difficult using actual \index{distance}distances; (3) it satisfies many convenient analytical properties such as being convex in both arguments and lower-semicontinuous in the topology of \index{weak convergence}weak convergence; and (4) it has an information theoretic and physical interpretation.

\begin{example} \label{ex:klg}
Consider two Gaussian densities $p_1$ and $p_2$ on $\R^d$
with means $\mu_1, \mu_2$ and positive definite covariance matrices $\Sigma_1, \Sigma_2$.
Then
$$ \dkl(p_1\|p_2)=\frac{1}{2}\Bigl(\log \frac{{\rm det} \Sigma_2}
{{\rm det} \Sigma_1}-d 
+|\mu_1-\mu_2|_{\Sigma_2}^2+{\rm Tr}(\Sigma_2^{-1}\Sigma_1)\Bigr).$$ 
\end{example}

 The following lemma establishes upper-bounds on \index{distance!total variation}total variation and \index{distance!Hellinger}Hellinger distances in terms of the \index{divergence!Kullback-Leibler}Kullback-Leibler divergence. Note that as a corollary we obtain a proof of the fact that $\dkl(\post \|\post') > 0$ if $\post \neq \post'.$ 
\begin{lemma} \label{l:kllemma}
The \index{divergence!Kullback-Leibler}Kullback-Leibler  divergence provides
the following upper bounds for \index{distance!Hellinger}Hellinger and \index{distance!total variation}total variation distance:
$$ \dhell(\post, \post')^2 \leq \frac{1}{2}\dkl(\post \| \post'),
\quad \dtv(\post, \post')^2 \leq \dkl(\post \| \post').$$
\end{lemma}

\begin{proof}

The second inequality follows from the first one by Lemma \ref{l:dh2}; thus
we prove only the first inequality. 
Consider the function $\phi:\R^+ \mapsto \R$ 
defined by
$$\phi(x) = x - 1 - \log{x}.$$
Note that
\begin{equation*}
\begin{split}
\phi'(x) &= 1 - \frac{1}{x}, \\
\phi''(x) &= \frac{1}{x^2}, \\
\phi(\infty) & = \phi(0) =\infty.
\end{split}
\end{equation*}
Thus, the function is convex on its domain.
As the minimum of $\phi$ is attained at $x = 1$, and as $\phi(1) = 0$,
we deduce that $\phi(x) \geq 0$  for all $x \in (0, \infty).$
Hence, 
\begin{equation*}
\begin{split}
x-1 &\geq \log{x} \hspace{1cm} \text{ for all } x \ge 0, \\
\sqrt{x} - 1&\geq \frac{1}{2}\log{x} \hspace{1cm} \text{for all } x  \ge 0.
\end{split}
\end{equation*}
We can use this last inequality to bound the \index{distance!Hellinger}Hellinger distance:
\begin{equation*}
\begin{split}
\dhell(\post,\post')^2 &= \frac{1}{2}\int \left( 1 - \sqrt{\frac{\post'}{\post}} \right)^2 \post du \\
&=  \frac{1}{2}\int \left( 1 + \frac{\post'}{\post} -  2\sqrt{\frac{\post'}{\post}} \right) \post du \\
&= \int \left( 1 - \sqrt{\frac{\post'}{\post}} \right) \post du \leq -\frac{1}{2}\int \log\Bigl({\frac{\post'}{\post}}\Bigr)\post du = \frac{1}{2}\dkl(\post \| \post').
\end{split}
\end{equation*}
\end{proof}

\section{Best \index{Gaussian}Gaussian Fit by Minimizing $\dkl(p\|\post)$}
\label{sec:existence}

In this section we prove the existence of a best \index{Gaussian!approximation}Gaussian approximation $p = \Nc(\mu, \Sigma)$ to a given pdf $\pi$ in the sense that $\dkl ( p \| \pi)$ is minimized. As part of our analysis, we will show that Gaussian pdfs $p$ that minimize $\dkl ( p \| \pi)$ can be found by solving a stochastic optimization algorithm to determine optimal mean and covariance. Therefore, the stochastic gradient descent algorithm studied in Chapter \ref{chap:optimization} provides a natural method to find a best Gaussian fit. 
While the existence of a minimizer and the applicability of stochastic gradient descent apply more broadly, we focus our discussion on the case where $\pi = \pi^y$ is a \index{posterior}posterior distribution satisfying the following assumption:

\begin{assumption}\label{asp:1}
The \index{posterior}posterior distribution  $\pi(u) = \frac{1}{Z} \exp \bigl(-\loss(u)\bigr) \rho(u)$ satisfies: 
\begin{itemize}
\item The \index{loss}loss function $\loss(u)$ 
is non-negative and bounded above. 
\item The \index{prior}prior is a centered isotropic \index{Gaussian}Gaussian: $\pr(u) =  \Nc(0, \lambda^{-1}\Id).$
\end{itemize}
\end{assumption}
Let $\mathcal{A}$ be the set of \index{Gaussian}Gaussian distributions on $\Ru$ with \index{positive definite}positive definite covariance,
$$ \mathcal{A} = \{\Nc(\vct{\mu},\mtx{\Sigma}): \vct{\mu} \in \Ru, \mtx{\Sigma} \in \R^{\Du\times \Du} \index{positive definite}\text{ positive definite}\}.$$

We have the following theorem, which establishes the existence of a  best \index{Gaussian!approximation}Gaussian approximation. We remark, however, that minimizers need not be unique. Note that $\mathcal{A}$ is an open set
since the set of \index{positive definite}positive definite matrices is open.
It is thus implicit in the theorem that the infimum
is indeed attained with \index{positive definite}positive definite covariance. 

\begin{theorem}[Best \index{Gaussian!approximation}Gaussian Approximation]
\label{t:bi1}
Under Assumption \ref{asp:1}, there exists at least one probability distribution $p \in \mathcal{A}$ at which the infimum 
$$\inf_{p \in \mathcal{A}} \dkl (p \| \post) $$
is attained.
\end{theorem}

\begin{proof}
The  \index{divergence!Kullback-Leibler}Kullback-Leibler divergence can be computed explicitly as
\begin{equation*}
\begin{split}
\dkl (p \| \post) = &\Expect^p  \bigl[ {\log{p}}\bigr] - \Expect^p\bigl[ \log \post\bigr] \\
 = &\Expect^p
\left[
-\frac{1}{2}\vert \vct{u} - \vct{\mu} \vert^2_{\vct{\Sigma}} 
- \frac{1}{2}\log \left( (2\pi)^\Du \text{det} \vct{\Sigma} \right) \right. + \left. \loss (\vct{u}) + \frac{\lambda}{2} \vert \vct{u} \vert^2 +\log Z 
\right].
\end{split}
\end{equation*}
Note that $Z$ is the normalization constant for $\post$ and is independent
of $p,$ and hence of $\vct{\mu}$ and $\vct{\Sigma}.$
We can represent a random variable $\vct{u} \sim p$ by writing
 $\vct{u} = \vct{\mu} + \vct{\Sigma}^{1/2}\xi$, where $\xi \sim \Nc(\vct{0}, \vct{\Id})$, 
and hence
$$ \vert \vct{u} \vert^2 = \vert \vct{\mu} \vert^2 + \vert \Sigma^{1/2}\xi|^2 + 2\langle \vct{\mu}, \vct{\Sigma}^{1/2}\xi \rangle.$$
Using this we obtain
\begin{equation*}
\begin{split}
\dkl (p \| \post) &=  - \frac{\Du}{2} - \frac{\Du}{2}\log(2\pi) - \frac{1}{2}\log \text{det}\vct{\Sigma}   +\Expect^p \bigl[ \loss (\vct{u}) \bigr]+  \frac{\lambda}{2} \vert \vct{\mu} \vert^2+ \frac{\lambda}{2} \text{tr}(\vct{\Sigma})   + \log Z.\\
\end{split}
\end{equation*}
 Define 
\begin{equation*}
\begin{split}
\J (\mu,\Sigma) &= \frac{\lambda}{2} \vert \vct{\mu} \vert^2 + \frac{\lambda}{2} \text{tr}(\vct{\Sigma}) - \frac{1}{2}\log \text{det}\vct{\Sigma}+\Expect^p \bigl[ \loss (\vct{u}) \bigr],\\
\J_0 (\mu,\Sigma) &= \frac{\lambda}{2} \vert \vct{\mu} \vert^2 + \frac{\lambda}{2} \text{tr}(\vct{\Sigma}) - \frac{1}{2}\log \text{det}\vct{\Sigma}.
\end{split}
\end{equation*}
Note that since $\loss$ is assumed to be bounded above, $ \J  \to \infty$ if and only
if $ \J_0  \to \infty.$ Furthermore, writing positive definite $\Sigma=QDQ^\top$
where $Q$ is orthogonal and $D$ is diagonal with non-negative entries
$\{\sigma_i\}_{i=1}^d$ we find that
\begin{equation*}
\begin{split}
 \J (\mu,\Sigma) &= \frac{\lambda}{2} \vert \vct{\mu} \vert^2 + 
\frac12 \sum_{i=1}^d \Bigl(\lambda \sigma_i-\log(\sigma_i)\Bigr)
+\Expect^p \bigl[ \loss (\vct{u}) \bigr],\\
 \J_0 (\mu,\Sigma) &= \frac{\lambda}{2} \vert \vct{\mu} \vert^2 + 
\frac12 \sum_{i=1}^d \Bigl(\lambda \sigma_i-\log(\sigma_i)\Bigr)
\end{split}
\end{equation*}
For any $\Sigma,$ $ \J_0  (\mu,\Sigma) \to \infty$   
(and hence $ \J  (\mu,\Sigma) \to \infty$) as $|\mu|\to \infty.$
Furthermore, for any $\mu$ and any $i$, $ \J_0 (\mu,\Sigma) \to \infty$ 
(and hence $ \J  (\mu,\Sigma) \to \infty$)
as $\sigma_i\to 0^+$ or $\sigma_i \to \infty.$
Now define, for $\Sigma=QDQ^\top$ as above, 
$$\tilde{\mathcal{A}} := \{(\vct{\mu},\mtx{\Sigma}): \vct{\mu} \in \Ru, \mtx{Q} \in \R^{\Du\times \Du}: Q^\top Q=I , \,\,\,|\mu|\le M,\,\,\, r\le \sigma_i \le R 
\,\,\forall i\}.$$
Note that $ \J  (0,\Id) <\infty.$
Thus there are $M,r,R>0$ such that the infimum of $ \J (\mu,\Sigma)$ over $\mu\in \R^\Du$ and \index{positive definite}positive definite $\Sigma$  
is equal to the infimum of $ \J (\mu,\Sigma)$ over the
closed and bounded set $\tilde{\mathcal{A}}.$
Since $ \J $ is continuous in $\tilde{\mathcal{A}}$ 
it achieves its infimum, and the proof is complete.
\end{proof}

\begin{remark}[Minimizing $\dkl(p\|\post)$ with \index{gradient descent!stochastic}Stochastic Gradient Descent]
The proof of Theorem \ref{t:bi1} shows that a best \index{Gaussian!approximation}Gaussian approximation $p \in \mathcal{A}$ to $\pi$ can be found by minimizing the \index{objective}objective 
\begin{align*}
 \J  (\mu,\Sigma) &=  \Expect^{\xi \sim \Nc(0,I) }  \Bigl[  \frac{\lambda}{2} \vert \vct{\mu} \vert^2 + \frac{\lambda}{2} \text{tr}(\vct{\Sigma}) - \frac{1}{2}\log \text{det}\vct{\Sigma}+\loss ( \mu + \Sigma^{1/2} \xi) \Bigr] \\
 &= \int_{\R^d}  F  (u, z)  \zeta(z)\, dz,  
\end{align*}
where $u = (\mu, \Sigma),$ $ \zeta  = \Nc(0,I),$ and 
$$ F  (u,z) =  \frac{\lambda}{2} \vert \vct{\mu} \vert^2 + \frac{\lambda}{2} \text{tr}(\vct{\Sigma}) - \frac{1}{2}\log \text{det}\vct{\Sigma}+\loss ( \mu + \Sigma^{1/2} z).$$
Thus, this optimization problem can be solved using the stochastic gradient descent algorithm described in Chapter \ref{chap:optimization}.
\end{remark}

\section{Best \index{Gaussian}Gaussian Fit by Minimizing $\dkl(\post\|p)$}
\label{sec:43}
In this section we show that the best \index{Gaussian!approximation}Gaussian approximation in Kullback-Leibler with respect to its second argument is unique and given by moment matching.
\begin{theorem}[Best \index{Gaussian!approximation}Gaussian Approximation by Moment Matching]\label{th:gaussianapprox1}
Assume that $\bar{\mu} := \Expect^{\post}[u]$ is finite and that
$\bar{\Sigma} := \Expect^{\post}\bigl[(u-\bar{\mu})\otimes(u-\bar{\mu})\bigr]$ is
\index{positive definite}positive definite. Then the 
infimum $$\inf_{p \in \mathcal{A}} \dkl (\post\|p)$$
is attained at the element in $\mathcal{A}$ with mean $\bar{\mu}$ and
covariance $\bar{\Sigma}.$
\end{theorem}

\begin{proof}
By definition 
\begin{align}
\dkl(\post \|p) = -\Expect^{\post}{[\log{p}]} + \Expect^{\post}{[\log{\post}]}.
\label{Dkl}
\end{align}
Since the second term does not involve $p$, we study minimization of
\begin{align*}
-\Expect^{\post}{[\log{p}]} &= -\Expect^{\post}{\left[ \log\left(\frac{1}{\sqrt{(2\pi)^\Du\text{det}\Sigma}}\text{exp}\left(-\frac{1}{2}\big|u-\mu\big|_{\Sigma}^2\right)\right)\right]}\\
 &= \frac{1}{2}\Expect^{\post}{\left[\big|u-\mu\big|_{\Sigma}^2\right]} +\frac{1}{2}\log{\text{det}\Sigma}+\frac{\Du}{2}\log{2\pi}.
\end{align*}
Let $\precision = \Sigma^{-1}$. Then our task is equivalent to minimizing the
following function of $\mu$ and $\precision$:
\begin{align*}
 \J  (\mu,\precision) = \frac{1}{2}\Expect^{\post}{\left[ \langle u-\mu, \precision(u-\mu)\rangle\right]-\frac{1}{2}\log{\text{det}\precision}}.
\end{align*}
First we find the critical points of $ \J $ by taking its first order partial derivative with respect to $\mu$ and $\precision$ and setting both to zero:
\begin{align*}
\partial_\mu  \J  &= -\Expect^{\post}\left[\precision(u-\mu)\right] = 0;\\
\partial _\precision  \J  &= \frac{1}{2}\partial_\precision \left( \Expect^{\post}\left[(u-\mu)\otimes(u-\mu):\precision\right]\right)-\frac{1}{2\text{det}\precision}\partial_\precision \text{det}\precision\\
& = \frac{1}{2}\Expect^{\post}\left[(u-\mu)\otimes(u-\mu)\right] - \frac{1}{2}\precision^{-1} = 0; \quad\quad 
\end{align*}
here we have used the relation $\partial_\precision\text{det}\precision = \text{det}{\precision}\cdot \precision^{-1}.$ Solving the above two equations gives us the critical point, expressed in terms
of mean and covariance, 
$$(\bar{\mu},\bar{\Sigma}) = (\Expect^{\post}[u], \Expect^{\post}[(u-\bar{\mu})\otimes(u-\bar{\mu})]).$$ 
The fact that the critical point $(\bar{\mu},\bar{\Sigma}^{-1})$ is a minimizer of $ \J $ follows because $ \J $ is convex. Indeed, note that $ \J $ is the  sum of two convex functions: a \index{positive definite}positive definite quadratic form and a negative log-determinant.  
\end{proof}

\begin{remark}[Minimizing $\dkl(\post \|p)$ with \index{Monte Carlo}Monte Carlo]
Theorem \ref{th:gaussianapprox1} shows that the Gaussian $p \in \mathcal{A}$ closest to $\post$ in the sense of minimizing $\dkl(\post \|p)$ is the Gaussian with the same mean and covariance $(\bar{\mu},\bar{\Sigma})$ as $\post.$ Both mean and covariance can be computed using Monte Carlo methods, a family of algorithms designed to computed expected values with respect to a given target distribution using samples. Monte Carlo algorithms will be studied in Chapter \ref{lecture5}.
\end{remark}

\section{Comparison Between $\dkl(\post \|p)$ and $\dkl(p \|\post)$}\label{sec:44}
It is instructive to compare the two different minimization
problems, both leading to a ``best \index{Gaussian}Gaussian'', that we described
in the preceding two sections.
We write the two relevant \index{divergence}divergences as follows and then explain the
nomenclature:
\begin{align*}
&\dkl(p\|\post) = \Expect^{p}{\left[\log\left(\frac{p}{\post}\right)\right]}= \Expect^{p}{[\log{p}]} - \Expect^{p}{[\log{{\post}}]}, \quad \text{``Mode-seeking''}\\
&\dkl(\post\|p) = \Expect^{\post}{\left[\log\left(\frac{\post}{p}\right)\right]}= \Expect^{\post}{[\log{\post}]} - \Expect^{\post}{[\log{{p}}]}. \quad \text{``Mean-seeking''}
\end{align*}

\begin{figure}[!tbp]
\hspace*{\fill}%
  \centering
  \subfloat[Minimizing $\dkl(p\|\post)$]{\includegraphics[width=0.4\textwidth]{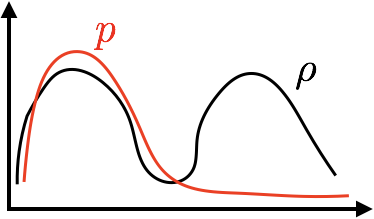}}
  \hfill
  \subfloat[Minimizing $\dkl(\post\|p)$]{\includegraphics[width=0.4\textwidth]{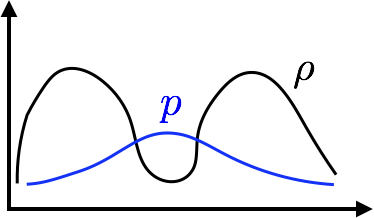}}
  \hspace*{\fill}
  \caption{(a) Minimizing $\dkl(p\|\post)$ can lead to serious information \index{loss}loss while (b) minimizing $\dkl(\post\|p)$ ensures a comprehensive consideration of all components of $\post$.}
\label{f:mvm}
\end{figure}

Note that when minimizing $\dkl(p\|\post)$ we want $\log{\frac{p}{\post}}$ to be small  in regions of high probability under $p,$ which can happen when $p\simeq \post$ or when $p$ is much smaller than $\post.$ This illustrates
the fact that minimizing $\dkl(p\|\post)$ may miss out components of $\post$. 
For example, in Figure \ref{f:mvm}(a)  $\post$ is a bimodal 
distribution but minimizing $\dkl(p\|\post)$ over \index{Gaussian}Gaussians $p$ can only give a single mode 
approximation which is achieved by matching one of the modes;
we may think of this as  ``mode-seeking''. In contrast, when minimizing 
$\dkl(\post\|p)$ over \index{Gaussian}Gaussians $p$ we want $\log{\frac{\post}{p}}$ to be small where $p$ appears as the denominator. This implies that wherever $\post$ has some mass we must let $p$ also have some mass there in order to keep $\frac{\post}{p}$ as close as possible to one. Therefore, the minimization is carried out by allocating the mass of $p$ in a way such that on average the \index{divergence}divergence between $p$ and $\post$ attains its minimum, as shown in Figure \ref{f:mvm}(b); hence 
the label ``mean-seeking.''  Different applications will favor different
choices between the mean and mode seeking approaches to \index{Gaussian!approximation}Gaussian approximation. 
\label{sec:difference}

\section{Variational\index{variational!formulation of Bayes theorem} Formulation of \index{Bayes theorem}Bayes Theorem}
\label{sec:posterior}
This chapter has been concerned with finding the best \index{Gaussian!approximation}Gaussian approximation
to a measure with respect to  \index{divergence!Kullback-Leibler}Kullback-Leibler  divergences. \index{Bayes theorem}Bayes Theorem \ref{t:bayes}
itself can be formulated through a closely related minimization
principle.  Consider a \index{posterior}posterior $\post^y(u)$ in the following form:
\begin{align*}
\post^y(u) = \frac{1}{Z}\text{exp}\bigl(-\loss(u)\bigr)\pr(u),
\end{align*}
where $\pr(u)$ is the \index{prior}prior, $\loss(u)$ is the negative \index{likelihood}log-likelihood,
and $Z$ the normalization constant. We assume here for exposition that all pdfs are positive. Dropping the superscript $y$ from $\post^y$ for
notational simplicity, we express $\dkl(p\|\post)$ in terms of the \index{prior}prior as follows:
\begin{align*}
\dkl(p\|\post) &= \int_{\Ru} \text{log}\left( \frac{p}{\post}\right)p \, du\\
&= \int_{\Ru} \text{log}\left( \frac{p}{\pr}\frac{\pr}{\post}\right)p \, du\\
&=  \int_{\Ru} \text{log}\left( \frac{p}{\pr}\text{exp}\bigl(\loss(u)\bigr)Z\right)p \, du\\
&= \dkl(p\|\pr) + \Expect^{p}[\loss(u)] +\log{Z}.
\end{align*}
If we define 
$${\mathcal J}(p) =  \dkl(p\|\pr) + \Expect^{p}[\loss(u)],$$ 
then we have the following:
\begin{theorem}[\index{Bayes theorem}Bayes Theorem as an Optimization Principle]
The \index{posterior}posterior distribution $\post$ is given by the following minimization principle:
$$\post = \text{\rm argmin}_{p\in \pdfs} {\mathcal J}(p),$$ 
where $\pdfs$ contains all pdfs on $\Ru$.
\end{theorem}
\begin{proof}
Note that
$$\dkl(p\|\post)={\mathcal J}(p)+\log{Z}.$$
Since $Z$ is the normalization constant for $\post$ and is independent of $p$, 
the minimizer of $\dkl(p\|\post)$ over $p\in\pdfs$ will also be the minimizer of ${\mathcal J}(p).$ Since the unique global minimizer of $\dkl(p\|\post)$ is attained at $p=\post,$ the result follows. 
\end{proof}

The posterior distribution $\pi$
is the minimizer of ${\mathcal J}(p)$ over all
pdfs. However, we can approximate $\pi$
by minimizing ${\mathcal J}(p)$ over a subset of all pdfs.
 The following example of this connects to earlier parts
of the chapter; further  discussion on other computational methods and theoretical insights that stem from viewing Bayes theorem as an optimization problem
may be found in the conclusion  Section \ref{sec:46}.

\begin{example}[Optimization over Gaussians]
If we approximate $\pi$ by minimizing ${\mathcal J}(p)$ over Gaussians
then we obtain the 
methodology studied in Section \ref{sec:existence}. 
\end{example}


\section{Discussion and Bibliography}\label{sec:46}
The definition of the \index{divergence!Kullback-Leibler}Kullback-Leibler divergence, and upper-bounds
in terms of probability metrics, can be found in \cite{gibbs2002choosing}.
For a basic introduction to \index{variational!Bayesian method}variational Bayesian methods,
including the moment-matching version of \index{Gaussian!approximation}Gaussian
approximation, see \cite{bishop}. The idea of approximating a \index{target distribution}target distribution $\post$ by  minimizing the \index{divergence!Kullback-Leibler}Kullback-Leibler divergence within a family of admissible distributions is popular in probabilistic machine learning. \index{variational!Bayesian method}Variational Bayesian methods \cite{jordan1999introduction,wainwright2008graphical} minimize $\dkl( p \| \post)$; in contrast, expectation 
propagation methods \cite{minka2013expectation}, which seek a factorized
approximate distribution, proceed by minimizing $\dkl( \post \| p).$ 
We refer to \cite{wainwright2008graphical,blei2017variational} for accessible introductions to \index{variational!Bayesian method}variational Bayesian methods and further pointers to the literature. 

In this chapter we have focused on \index{Gaussian!approximation}Gaussian approximations, but other families of admissible distributions can be considered. The family of admissible distributions should in practice be large enough to allow for accurate approximation of the \index{target distribution}target distribution, while also allowing for efficient \index{optimization}optimization. \index{Gaussian!approximation}Gaussian approximations are useful in \index{Bayesian!inverse problem}Bayesian \index{inverse problem}inverse problems and are invoked by many \index{data assimilation}data assimilation algorithms, as we shall see in Chapter \ref{lecture10}. In probabilistic machine learning it is common to invoke mean-field rather than \index{Gaussian!approximation}Gaussian approximations, and a variety of efficient \index{optimization}optimization algorithms are available in this context \cite{bishop}. Recent works that employ variational inference techniques for the solution of \index{inverse problem}inverse problems include \cite{agrawal2021variational,law2021sparse}. 

The problem of finding a \index{Gaussian!approximation}Gaussian approximation of a general 
finite-dimensional probability distribution is studied in \cite{lu2016gaussian}, and
infinite-dimensional formulations are considered in \cite{pinski2015kullback} and the
companion paper \cite{pinski2015algorithms}. 
\index{Gaussian!approximation}Gaussian approximation of \index{small noise limit}small noise diffusions are studied in \cite{sanz2016gaussian}.
The approximation in Theorem \ref{t:bi1}
consists of a single \index{Gaussian!approximation}Gaussian distribution. 
If the \index{posterior}posterior has more than one mode, a single \index{Gaussian}Gaussian may not be appropriate. 
For an approximation composed of \index{Gaussian}Gaussian mixtures, the reader is referred to \cite{lu2016gaussian}. The paper \cite{trillos2019variational} highlights how minimization of \index{divergence!Kullback-Leibler}Kullback-Leibler divergence arises naturally in the \index{optimization}optimization of local entropy and heat regularized costs in deep learning.

The formulation of \index{Bayes theorem}Bayes theorem as an \index{optimization}optimization principle is
 well known; see the book \cite{mackay2003information} and the
paper \cite{bissiri2016general} for clear expositions of this subject. There are at least three advantages of viewing \index{Bayes theorem}Bayes theorem as an optimization problem. First, the \index{variational!formulation of Bayes theorem}variational formulation provides a natural way to approximate the \index{posterior}posterior by restricting the minimization problem to distributions satisfying some computationally desirable property. For instance, \index{variational!Bayesian method}variational Bayesian methods often restrict the minimization to densities with a 
factorizable
structure implied by independence with respect to the components of the
unknown $u$; similarly, in Section \ref{sec:existence} we have studied restriction to the class of \index{Gaussian}Gaussian distributions. Second, \index{variational!formulation of Bayes theorem} variational formulations can be used to show convergence of \index{posterior}posterior distributions indexed by some parameters using techniques from calculus of variations. For instance, the papers \cite{trillos2017consistency} and \cite{garcia2018continuum} exploit the variational formulation of \index{Bayes theorem}Bayes theorem to establish convergence of \index{Bayesian}Bayesian procedures. Third, \index{variational!formulation of Bayes theorem} variational formulations provide natural paths, defined by a gradient flow, towards the \index{posterior}posterior. Understanding these flows and their rates of convergence is helpful in the design and choice of \index{sampling}sampling algorithms \cite{trillos2018bayesian}. 



 For more information about the properties of the \index{exponential family}exponential family we refer to \cite{nielsen2009statistical}, and for background on matrix calculations that were used in this chapter we refer to  \cite{petersen2008matrix}.

\chapter{\Large{\sffamily{Monte Carlo Sampling and Importance Sampling }}}\label{lecture5}
In this chapter we introduce \index{Monte Carlo}Monte Carlo sampling and \index{importance sampling}importance sampling. 
These are two general techniques for estimating expectations with respect to
a given pdf $\post.$ \index{Monte Carlo}Monte Carlo generates 
 independent 
samples from $\post$ and combines them with equal weights, whilst 
\index{importance sampling}importance sampling uses independent samples,
weighted appropriately, from a different distribution. 
 In quantifying the error in \index{Monte Carlo}Monte Carlo and \index{importance sampling}importance sampling,
we will use a \index{distance!between random probability measures}distance on random probability measures that reduces
to \index{distance!total variation}total variation in the case of deterministic probability
measures; and we will introduce the \index{divergence!$\chi^2$}$\chi^2$ divergence.

In \index{Bayesian!inverse problem}Bayesian \index{inverse problem}inverse problems, we are typically unable  to directly generate
samples from the \index{posterior}posterior distribution $\post^y$ itself,
so that Monte Carlo sampling is not viable; however, importance
sampling may be used. For example, it is often possible to generate
samples from the \index{prior}prior; importance sampling can then be used
to reweight samples from the \index{prior}prior
distribution, to approximate \index{posterior}posterior expectations.

Recall that for any pdf $p$ and function $\varphi:\Ru\longrightarrow\R,$ we denote 
\begin{equation}
\label{eq:measa}
p(\varphi)=\Expect^{p}\left[\varphi(u)\right]=\int_{\Ru} \varphi(u)p(u) \, du.
\end{equation}
Thus we view the pdf $p$ as a linear 
functional on the space of real-valued functions on $\Ru$. 
Our task in this chapter is to evaluate $\post(f)$ for \index{target distribution}target distribution $\post$ on $\Ru$
and for a given test function $f:\Ru\longrightarrow\R$. Thus, we are interested in computing 
\begin{equation}
\label{eq:measb}
\post(f)=\int_{\Ru} f(u)\post(u) \, du.
\end{equation}
\index{Monte Carlo}Monte Carlo sampling approximates this integral using samples
from the \index{target distribution}target $\post.$ 

To describe \index{importance sampling}importance sampling, we note that for any pdf  $\pr$ such that the support of $\post$ is contained in the support of $\pr,$  equation \eqref{eq:measb}
can be rewritten as 
\begin{equation}
\label{eq:measc}
\post(f)=\int_{\Ru} f(u)\Bigl(\frac{\post(u)}{\pr(u)}\Bigr)\pr(u) \, du=\rho(fw),
\end{equation}
where
$$w(u)=\frac{\post(u)}{\pr(u)}.$$
We assume that the ratio $\post(u)/\rho(u)$ is only known up to a normalization constant and write
\begin{equation}
\label{eq:MC1}
w(u)=\frac{\post(u)}{\pr(u)} = \frac{1}{Z}\like(u), 
\end{equation}
where the unknown normalizing constant is defined by $Z = \rho(\like).$
Noting that $w(u)=Z^{-1}\like(u)$, we obtain from
\eqref{eq:measc}
\begin{equation}
\label{eq:measd}
\post(f)=\frac{\rho(f\like)}{\rho(\like)}.
\end{equation}
\index{importance sampling}Importance sampling methods are based on approximating the two integrals 
on the right-hand side of this identity with \index{Monte Carlo}Monte Carlo,  using samples from $\pr.$
Note that it is not necessary to know $Z$ to implement this method.

A particular application of
\index{importance sampling}importance sampling in the context of \index{Bayes theorem}Bayes theorem is the setting
where $\pr$ is the \index{prior}prior, $\post$ the \index{posterior}posterior
and $g$ the \index{likelihood}likelihood. However, the \index{importance sampling}importance sampling method is 
not restricted to this splitting of the \index{posterior}posterior into a product
of \index{likelihood}likelihood and \index{prior}prior; and
indeed, depending on the specific test function $f$ of interest, the importance
sampling method  may be far from optimal if applied with this choice of $\pr$.

To summarize, \index{Monte Carlo}Monte Carlo approximates $\post(f)$ using \eqref{eq:measb} and 
samples from $\post$; \index{importance sampling}importance sampling approximates $\post(f)=\pr(fw)$
using \eqref{eq:measd} and samples from $\pr$. Underlying the approximations
of integrals are approximations of measures. For this reason, it is convenient in
this chapter to generalize the concept of pdf to include 
\index{Dirac}Dirac mass distributions. A \index{Dirac}Dirac mass at $v$ will be viewed as
having pdf $\delta(\cdot-v)$ where $\delta(\cdot)$ integrates to one
and takes the value zero everywhere except at the origin. This \index{Dirac}Dirac
mass is also sometimes written as $\delta_{v}(\cdot).$ 

This chapter is organized as follows. We first introduce and analyze \index{Monte Carlo}Monte Carlo sampling in Section \ref{sec:MCdefn&thm}. \index{importance sampling} Importance sampling is then studied in Section \ref{sec:ISdefn&thm}. We close in Section \ref{sec:53} with pointers to the extant literature on this subject. 

\section{\index{Monte Carlo}Monte Carlo Sampling}
\label{sec:MCdefn&thm}
\index{Monte Carlo}Monte Carlo sampling applies when it is possible to generate
\index{i.i.d.}i.i.d. samples  $u^{(\sam)} \sim \post,$ $1\le \sam \le \Sam.$ The method approximates the \index{target distribution}target distribution $\post$ by a sum of \index{Dirac}Dirac masses located at the samples $u^{(\sam)},$ each given equal weight $1/N.$ This leads to the \index{Monte Carlo}Monte Carlo estimator $\pMC$ of $\post$ given by 
\begin{equation}
\label{eq:indeed}
\pMC := \frac{1}{\Sam} \sum_{\sam=1}^{\Sam}\delta(u-u^{(\sam)}).
\end{equation}
We summarize this simple procedure in the following algorithm:
\FloatBarrier
\begin{algorithm}
\caption{\label{algMH} \index{Monte Carlo}Monte Carlo Sampling Algorithm}
\begin{algorithmic}[1]
\vspace{0.1in}
\STATE {\bf Input}: \index{target distribution}Target distribution $\post$, number of samples $\Sam.$  \\
\vspace{.04in}
\STATE Sample $u^{(\sam)} \sim \post$ \index{i.i.d.}i.i.d. \quad $\sam \in \{1,\dots, \Sam\}.$
\STATE{\bf Output}: \index{target distribution}Target approximation $\post \approx \pMC: = \frac{1}{\Sam} \sum_{\sam=1}^\Sam \delta(u-u^{(\sam)}).$
\label{alg_1}
\end{algorithmic}
\end{algorithm}
\FloatBarrier

This algorithm leads to the following estimator of $\post(f):$
\begin{equation*}
\pMC (f) = \frac{1}{\Sam}\sum_{\sam=1}^{\Sam}f(u^{(\sam)}),\,\, u^{(\sam)} \sim \post \quad \index{i.i.d.}{\rm  i.i.d.}
\end{equation*}

We are interested in determining whether the estimator $\pMC (f)$ of $\post(f)$ is accurate regardless of the specific test function $f$. For this reason, we seek to understand whether the \index{Monte Carlo}Monte Carlo estimator $\pMC$ is a good approximation to $\post$ in a suitable metric. This perspective will also
be useful in analyzing \index{importance sampling}importance sampling in this chapter,
and when analyzing sequential methods  
 for \index{data assimilation}data assimilation in  
Chapters \ref{ch11} and \ref{ch12}.
Note that $\pMC$ is a \emph{random} probability measure due to sampling, and so in order to formalize this question we need a \index{distance!between random probability measures}distance between random probability measures. To this end, for random probability measures $\post$ and $\post'$, we define
\begin{equation}
\label{eq:rdist}
	d(\post, \post') = \sup_{|f|_\infty \leq 1} \Bigl( \Expect\Bigl[\bigl(\post(f) - \post'(f)\bigr)^2\Bigr] \Bigr)^{1/2},
\end{equation}
where the expectation is taken over the random variable, in our case 
the randomness from sampling $\post$.  It is possible to show that \index{distance!between random probability measures}$d(\cdot,\cdot)$ indeed defines a distance between random probability measures. Furthermore, when \(\post,\post'\) are deterministic,
then we have \(d(\post,\post') = 2\dtv(\post,\post')\). Using
this \index{distance!between random probability measures}distance between random probability measures, we have the following result.

\newcommand*\diff{\mathop{}\!\mathrm{d}}
      \begin{theorem}[\index{Monte Carlo}Monte Carlo Error]
\label{t:MC}
       For $f:\Ru\longrightarrow\mathbb{R}$ denote $|f|_\infty := \sup_{u\in \Ru} |f(u)|.$ We have
       \begin{equation*}
        \begin{split}
        \sup_{|f|_\infty\leq1} \left|\Expect{\left[\pMC (f)-\post(f)\right]}\right|  &= 0,\\
        d(\pMC, \post)^2 &\leq \frac{1}{\Sam}. 
        \end{split}
        \end{equation*}
      \end{theorem}
      
      \begin{proof}
To prove the first result, namely that the estimator is unbiased, we use linearity of the expected value and that $u^{(\sam)} \sim \post$:
\begin{align*}
\Expect \left[\pMC (f) \right]  &= \Expect \Bigl[   \frac{1}{\Sam} \sum_{\sam=1}^{\Sam} f\bigl(  u^{(n)} \bigr) \Bigr] \\
&= \frac{1}{\Sam} N \post(f) = \post(f) = \Expect \bigl[ \post(f) \bigr]. 
\end{align*}
Therefore the supremum over $|f|_\infty \le 1$ is a supremum over a quantity that is
zero, for any $f$, and the result follows.

For the second result, note that since $\pMC (f)$ is unbiased, its variance agrees with its mean squared error. Now using that the $u^{(\sam)} \sim \post$ are independent we deduce that
\begin{align*}
\text{Var} \left[\pMC (f) \right]  &= \text{Var} \Bigl[   \frac{1}{\Sam} \sum_{\sam=1}^{\Sam} f \bigl( u^{(\sam)}\bigr) \Bigr] \\
&= \frac{1}{\Sam^2} \Sam \text{Var}_\post[f] =  \frac{1}{\Sam} \text{Var}_{\post}[f] . 
\end{align*}
For $|f|_\infty \le 1,$ we have
        \begin{equation*}
        \text{Var}_{\post}[f] = \post(f^2) -\post(f)^2 \leq \post(f^2) \le 1, 
        \end{equation*}
and therefore
        \begin{equation*}
        \sup_{|f|_\infty\leq1} \left|\Expect{\left[\left(\pMC (f)-\post(f)\right)^2\right]}\right| = \sup_{|f|_\infty\leq1} \left| \frac{1}{\Sam}\text{Var}_{\post}[f]\right| \leq \frac{1}{\Sam}.
        \end{equation*}
      \end{proof}

The theorem shows that the \index{Monte Carlo}Monte Carlo estimator  $\pMC$ is
an unbiased approximation for the \index{posterior}posterior $\post$ and that, by choosing
$\Sam$ large enough,
expectation of any bounded function $f$
can in principle be approximated by \index{Monte Carlo}Monte Carlo
sampling to arbitrary accuracy. Furthermore, although the
convergence is slow with respect to $\Sam$ --the mean squared error decays like $N^{-1}$ so the typical error only decays like $N^{-1/2}$ -- 
there is no dependence on the
dimension of the problem or on the properties of $f$, other than its
supremum. Moreover, the proof of Theorem \ref{t:MC} shows that, in fact, the Monte Carlo error in the approximation of $\post(f)$ is 
determined by the variance of $f$ under $\pi.$

      \begin{example}[Approximation of an Integral]
\label{ex:q}
       Let $f:\mathbb{R}\longrightarrow\mathbb{R}$ be a sigmoid function defined on $\mathbb{R}$ and shown in Figure \ref{F:MC1}(a) below as the blue solid curve.
For the \index{target distribution}target distribution $\post$ we take a mixture of two \index{Gaussian}Gaussians found
by choosing from $\Nc(-5,1)$ with probability $1/10$ and from
$\Nc(5,1)$ with probability $9/10.$ We wish to approximate the expected value,
under $\post$, of $f(u) \times \mathbb{I}_{[a,b]}(u)$ where
$$\mathbb{I}_{[a,b]}(u) = \begin{cases} 1 & \text{ if } u \in [a,b], \\ 0 & \text{ otherwise}. \end{cases}$$ We use \index{Monte Carlo}Monte Carlo sampling to generate $\Sam$ random samples $u^{(1)},\ldots,u^{(\Sam)}$ and compute the error between the actual integral and the \index{Monte Carlo}Monte Carlo estimator. The integral and estimator are in the form: 
        \begin{equation*}
        \begin{split}
        &\post(f) = \int_{a}^{b} f(u)\post(u) \diff u,\\
        &\pMC (f) = \frac{1}{\Sam} \sum_{\sam=1}^{\Sam}f(u^{(\sam)})\mathbb{I}_{[a,b]}(u^{(\sam)}).
        \end{split}
        \end{equation*}
The results of a set of numerical experiments with $a=-5, b=5$ and varying $\Sam$ are shown
in Figure \ref{F:MC1}(b).
A randomly chosen subset of the samples used when $\Sam=100$ is displayed in Figure \ref{F:MC1}(a); only samples in $[-5,5]$ are shown, since other samples do not contribute
to the estimator in this case.
  \end{example}
    
       \begin{figure}[thbp]
        \centering
        \begin{tabular}{@{}c@{}}
            \includegraphics[width=0.5\textwidth]{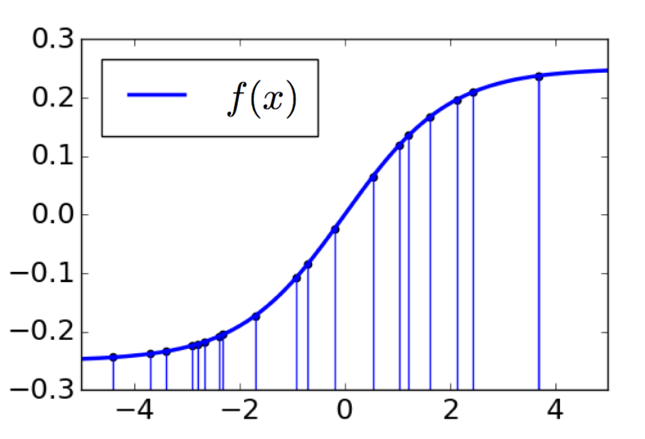}
        \end{tabular}
        \hspace{0.0in}
        \begin{tabular}{@{}c@{}}
            \includegraphics[width=0.55\textwidth]{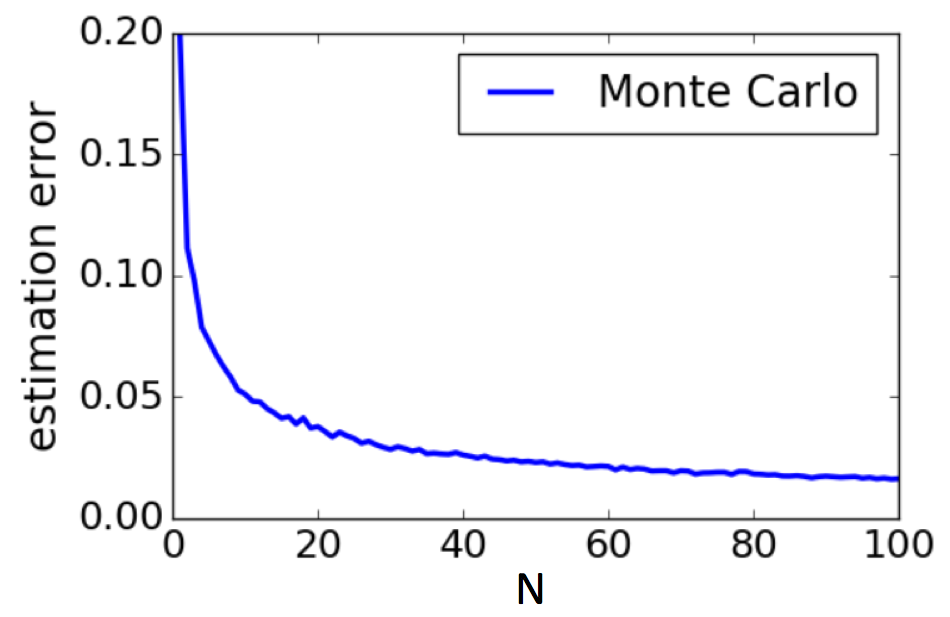}
        \end{tabular}
        \caption{Large sample size $N$ reduces the estimation error by the \index{Monte Carlo}Monte Carlo method.}
        \label{F:MC1}
       \end{figure}

\section{Importance Sampling}
\label{sec:ISdefn&thm}
\index{Monte Carlo}Monte Carlo sampling can only be used when it is possible to sample from the desired 
\index{target distribution} target distribution $\post$. When it is not possible to sample from $\post$, 
we can draw samples from another \index{proposal!distribution}{\em proposal} distribution $\pr$ instead.
Consider $\post$ as in equation \eqref{eq:MC1}. Given a test function $f: \R^d \to \R$ we can rewrite its expectation with respect to $\post$ in terms of expected values with respect to $\pr$ as in equation  \eqref{eq:measd}.
Approximating the numerator and the denominator using \index{Monte Carlo}Monte Carlo with
samples from $\pr$ gives
\begin{align*}
\post(f) & \approx  \sum_{\sam=1}^{\Sam }w^{(\sam)} f\left(u^{(\sam)}\right),\,\, u^{(\sam)} \sim \pr \,\, \index{i.i.d.}{\rm i.i.d.}  \\
&= \pIS (f), 
\end{align*}
where 
\begin{equation*}
w^{(\sam)} := \frac{\like\left(u^{(\sam)}\right)}{\sum_{m=1}^{\Sam}\like\left(u^{(m)}\right)}, \quad \quad \pIS := \sum_{\sam=1}^{\Sam }w^{(\sam)}\delta(u-u^{(\sam)}).
\end{equation*}     
Thus, given $\Sam$ samples $u^{(1)}, \ldots, u^{(\Sam)}$ generated \index{i.i.d.}i.i.d. according to $\pr,$ we can estimate $\post$ with the \emph{particle approximation measure} $\pIS.$

\FloatBarrier
\begin{algorithm}
\caption{\label{algIS} \index{importance sampling} Importance Sampling Algorithm}
\begin{algorithmic}[1]
\vspace{0.1in}
\STATE {\bf Input}: \index{target distribution}Target distribution $\post(u) = \frac{1}{Z} \like(u) \pr(u)$, \index{proposal!distribution}proposal distribution $\pr,$ number of samples $\Sam.$  \\
\vspace{.04in}
\STATE  Sample $u^{(\sam)} \sim \pr$\index{i.i.d.} i.i.d.  \quad $\sam \in \{1,\dots, \Sam\}.$ 
\STATE  Compute $$w^{(\sam)} := \frac{\like\left(u^{(\sam)}\right)}{\sum_{m=1}^{\Sam}\like\left(u^{(m)}\right)}.$$
\STATE{\bf Output}: \index{target distribution}Target approximation $\post \approx \pIS := \sum_{\sam=1}^{\Sam }w^{(\sam)} \delta(u-u^{(\sam)}).$
\label{alg_1}
\end{algorithmic}
\end{algorithm}
\FloatBarrier

We emphasize that implementation of this algorithm does not assume knowledge of 
the normalizing constant $Z,$ but only that $g$ can be evaluated and that $\pr$ can 
be sampled from. In particular, note that the  algorithm is invariant under $g \mapsto \lambda g$ for any scalar $\lambda.$ 
  Algorithm \ref{algIS}  leads to the following estimator of $\post(f):$
\begin{equation*}
\pIS(f)=  \sum_{\sam=1}^{\Sam}  w^{(\sam)}  f(u^{(\sam)}),\,\, u^{(\sam)} \sim \pr \quad \index{i.i.d.}{\rm i.i.d.}
\end{equation*}

  \begin{example}[Change of Measurement]
\label{ex:q2}
We consider a similar set-up as in Example \ref{ex:q}, integrating
a sigmoid function, shown in blue in Figure \ref{Fig. 3}, with respect to
a pdf $\post$ which is bimodal, shown in red in
Figure \ref{Fig. 3}; we again restrict the support of the desired integral.
We estimate the integral using \index{importance sampling}importance sampling based on $\Sam$ random 
samples $u^{(1)},\ldots,u^{(\Sam)}$ from the measure 
$\pr=\Nc(\mu,\,\sigma^{2})$, shown in green in Figure \ref{Fig. 3}. 
The estimator of the integral is given by 
        \begin{align*}
        \begin{split}
        \pIS(f) &= \sum_{\sam=1}^{\Sam}w^{(\sam)} f(u^{(\sam)}){\mathbbm{1}}_{[a,b]}(u^{(\sam)}),\\
        w^{(\sam)}&=\frac{\like(u^{(\sam)})}{\sum_{m=1}^{\Sam}\like(u^{(m)})}.
        \end{split}
        \end{align*}
Here $g$ is a function proportional to the ratio of the densities of
$\post$ and $\pr.$        
If $\post(u^{(\sam)})>\pr(u^{(\sam)})$, the samples should have been denser, 
so we raise the weight on $f(u^{(\sam)})$ in proportion to
$\frac{\post(u^{(\sam)})}{\pr(u^{(\sam)})}> 1$. If $\post(u^{(\sam)})<\pr(u^{(\sam)})$, the samples should have been less dense, so we lower the weight on $f(u^{(\sam)})$ in proportion to $\frac{\post(u^{(\sam)})}{\pr(u^{(\sam)})}< 1$.
      \end{example}

        \begin{figure}[thbp]
        \centering
        \begin{tabular}{@{}c@{}}
            \includegraphics[width=0.5\textwidth]{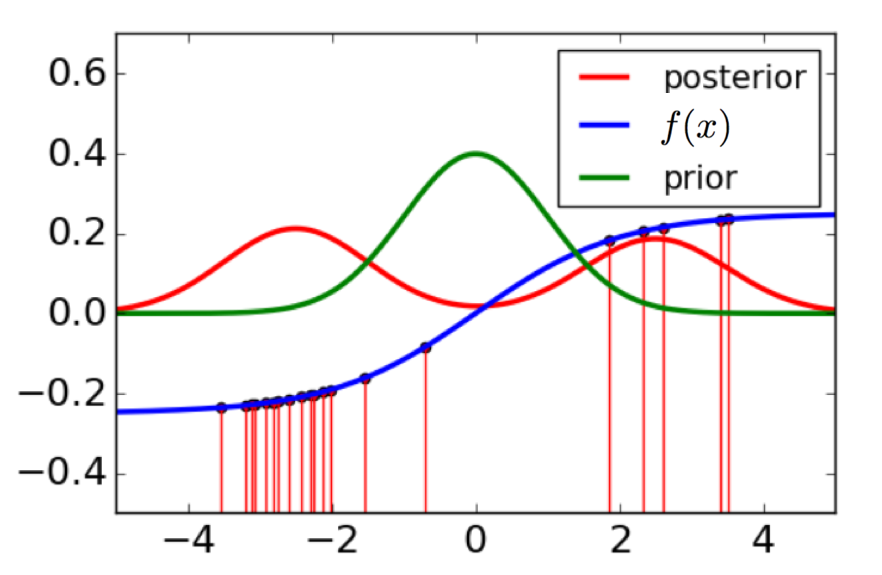}
        \end{tabular}
        \hspace{0.0in}
        \begin{tabular}{@{}c@{}}
            \includegraphics[width=0.5\textwidth]{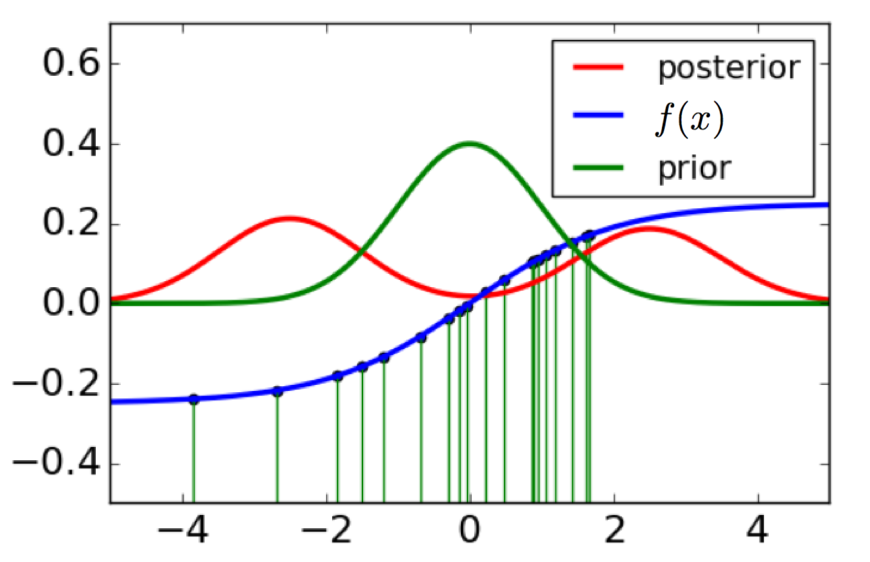}
        \end{tabular}
        \caption{\index{importance sampling}importance sampling is a change of measure via the importance weights. The red curve shows a bimodal distribution $\post$ and the green curve shows a \index{Gaussian}Gaussian distribution $\pr$. The blue curve is the function to be integrated, on its support $[-5,5].$ The upper figure shows samples 
from the \index{posterior}posterior $\post$ itself; these would be used for \index{Monte Carlo}Monte Carlo
sampling; the lower curve shows samples from the \index{prior}prior $\pr$, as used
for \index{importance sampling}importance sampling. The importance weights capture and compensate for the difference of sampling from these two distributions.}
          \label{Fig. 3}
		\end{figure}

We now introduce the \index{divergence!$\chi^2$}$\chi^2$ divergence between probability
distributions, and discuss some of its properties, before going on to
use it to quantify the accuracy of \index{importance sampling}importance sampling.

\begin{definition}
Let $\post, \post'>0$ be two pdfs on $\Ru.$\footnote{The definition extends to situations where the support of $\post'$ is not
the whole of $\Ru$, provided $\post$ is absolutely continuous with respect to
$\post'$.} The \index{divergence!$\chi^2$} \emph{$\chi^2$ divergence} of $\post$ with respect to $\post'$ is 
\begin{equation}\label{eq:chi2}
\dchi(\post \| \post') := \int_{\Ru} \Bigl( \frac{\post(u)}{\post'(u)} -1  \Bigr)^2 \, \post'(u) \, du.
\end{equation}
\end{definition}
The \index{divergence!$\chi^2$}$\chi^2$  divergence is not a \index{distance}distance as it is, in general, not symmetric; it
is, however, distance-like and captures the closeness of the two distributions;
this is analogous to the \index{divergence!Kullback-Leibler}Kullback-Leibler divergence defined in the
preceding chapter. 
The next lemma shows that the \index{divergence!$\chi^2$}$\chi^2$ divergence may be used to upper bound the \index{divergence!Kullback-Leibler}Kullback-Leibler divergence and therefore, by Lemma \ref{l:kllemma}, also the \index{distance!total variation}total variation and \index{distance!Hellinger}Hellinger distances. 

\begin{lemma}\label{l:chisq}
The \index{divergence!$\chi^2$}$\chi^2$ divergence provides the following upper bounds for the \index{divergence!Kullback-Leibler}Kullback-Leibler divergence:
$$\dkl(\post \| \post') \le \log\Bigl(\dchi(\post\| \post') + 1 \Bigr), \quad \quad \dkl( \post \| \post') \le \dchi(\post \| \post').$$
\end{lemma}
\begin{proof}
The second inequality is a direct consequence of the first one, noting that, for $x \ge 0,$ $\log(x + 1 ) \le x.$
To prove the first inequality note that by \index{Jensen inequality}Jensen inequality
\begin{align*}
\dkl(\post \| \post')  &= \int_{\R^\Du} \log \biggl( \frac{\post(u)}{\post'(u)} \biggr) \post(u) \, du \\
& \le \log \biggl(  \int_{\R^\Du} \frac{\post(u)}{\post'(u)}   \frac{\post(u)}{\post'(u)} \post'(u) \, du \biggr)  \\
& = \log\Bigl(\dchi(\post\| \post') + 1 \Bigr),
\end{align*}
where for the last equality we used that 
\begin{align*}
\dchi(\post\| \post') &= \int_{\Ru} \Bigl( \frac{\post(u)}{\post'(u)} -1  \Bigr)^2 \, \post'(u) \, du \\
& = \int_{\Ru} \Bigl( \frac{\post(u)}{\post'(u)} \Bigr)^2 \, \post'(u) \, du  - 2  \int_{\Ru} \Bigl( \frac{\post(u)}{\post'(u)} \Bigr) \, \post'(u) \, du  + \int_{\Ru}  \post'(u) \, du \\
&= \int_{\Ru} \Bigl( \frac{\post(u)}{\post'(u)} \Bigr)^2 \, \post'(u) \, du  - 1.
\end{align*}
\end{proof}

The next result shows that,  similarly as for \index{Monte Carlo}Monte Carlo sampling, the mean squared error of $\pIS(f)$ as an estimator of $\post(f)$ is order $N^{-1}.$ However, there are two main differences: the estimator is now biased, and the constant in the mean squared error depends on the \index{divergence!$\chi^2$}$\chi^2$ divergence between the \index{target distribution}target and the \index{proposal!distribution}proposal.

    \begin{theorem}[\index{importance sampling}Importance Sampling Error]
      \label{IS}
       We have
       \begin{equation*}
        \begin{split}
        \sup_{|f|_\infty\leq1} \left|\Expect{\left[\pIS(f)-\post(f)\right]}\right|  \leq 2 \frac{1 + \dchi(\post\|\pr)}{\Sam},\\
     d(\pIS, \post)^2 \leq 4 \frac{1 + \dchi(\post\|\pr)}{\Sam}.
        \end{split}
       \end{equation*}
      \end{theorem}
      
      \begin{proof}
The proof of the first item (bias) uses the second item (variance). 
Nonetheless, we start with the proof for the bias, 
because bias and variance are often thought of, conceptually, in that order.
        Given
        \begin{equation*}
        \post(u) = \frac{1}{Z}\like(u)\pr(u) = \frac{1}{\pr(\like)}\like(u)\pr(u),
        \end{equation*}
        the proof of Lemma \ref{l:chisq} shows that 
$$\dchi(\post\|\pr) = \frac{\pr(\like^2)}{\pr(\like)^2} -1.$$ 
To ease the notation we introduce
$$\zeta: = \frac{\pr(\like^2)}{\pr(\like)^2}.$$
        We rewrite 
        \begin{equation*}
        \begin{split}
        \post(f) &= \frac{\pr(\like f)}{\pr(\like)} \simeq \frac{\prMC(\like f)}{\prMC(\like)} =\pIS(f).\\
        \end{split}
        \end{equation*}
       Then we have
       \begin{equation}\label{eq:splitting}
        \begin{split}
        \pIS(f) -\post(f) &= \pIS(f) -  \frac{\pr(\like f)}{\pr(\like)}\\
        &=  \frac{\pIS(f)\left(\pr(\like)-\prMC(\like)\right)}{\pr(\like)} - \frac{\left(\pr(\like f)-\prMC( \like f)\right)}{\pr(\like)}.
        \end{split}
        \end{equation}
The expectation of the second term is zero and hence 
       \begin{equation*}
       \begin{split}
\left| \Expect{\left[\pIS(f) -\post(f)\right]}\right|
&= \frac{1}{\pr(\like)}\left|\Expect{\left[\pIS(f)\left(\pr(\like)-\prMC(\like)\right)\right]}\right|\\
        &\leq \frac{1}{\pr(\like)}\left|\Expect{\left[\left(\pIS(f)-\post(f)\right)\left(\pr(\like)-\prMC(\like)\right)\right]}\right|,\\
        \end{split}
        \end{equation*}
       since $\Expect{\left[\pr(\like)-\prMC(\like)\right]}=0$.
       Using the \index{Cauchy-Schwarz inequality}Cauchy-Schwarz inequality, the second result
from this theorem (whose proof follows) and
Theorem \ref{t:MC} we have, for all $|f|_\infty \leq 1$,
       
       \begin{equation*}
       \begin{split}
\left| \Expect{\left[\pIS(f) -\post(f)\right]}\right|
&\leq \frac{1}{\pr(\like)} \left(\Expect{\left[\left(\pIS(f)-\post(f)\right)^2\right]}\right)^{1/2}\left(\Expect{\left[\left(\pr(\like)-\prMC(\like)\right)^2\right]}\right)^{1/2}\\
        &\leq \frac{1}{\pr(\like)}\left(\frac{4\zeta}{\Sam}\right)^{1/2}\left(\frac{\pr(\like^2)}{\Sam}\right)^{1/2} = \frac{2\zeta}{\Sam}.
       \end{split}
       \end{equation*}

We now prove the second result. We use the splitting of
$\pIS(f) -\post(f)$ into the sum of two terms
as derived in equation \eqref{eq:splitting}. Using Theorem \ref{t:MC}, the basic inequality $(a-b)^2 \le 2(a^2 + b^2)$ and that for all $|f|_\infty \leq 1, |\pIS(f)| \leq 1$, we have, for all $|f|_\infty \leq 1$,
        \begin{equation*}
        \begin{split}
        &\left|\Expect{\left[\left(\pIS(f) -\post(f)\right)^2\right]}\right|\\
        &\leq \frac{2}{\pr(\like)^2} \left( \Expect{\left[\left(\pIS(f)\right)^2\left(\pr(\like)-\prMC(\like)\right)^2\right]}+ \Expect{\left[\left(\pr(\like f)-\prMC(\like f)\right)^2\right]} \right)\\
        &\leq \frac{2}{\pr(\like)^2} \left( \Expect{\left[\left(\pr(\like)-\prMC(\like)\right)^2\right]}+ \Expect{\left[\left(\pr(\like f)-\prMC(\like f)\right)^2\right]} \right)\\
        &= \frac{2}{\pr(\like)^2\Sam}\bigl(\text{Var}_{\pr}\left[\like\right] +\text{Var}_{\pr}\left[ \like f\right] \bigr)\\
        &\leq \frac{2}{\pr(\like)^2\Sam} \left( \pr(\like^2)+\pr(\like^2 f^2)\right)\\
        &\leq  \frac{4\pr(\like^2)}{\pr(\like)^2\Sam} = \frac{4\zeta}{\Sam}.
        \end{split}
        \end{equation*}
        Therefore, 
        \begin{align*}
         d(\pIS, \post)^2 =  \sup_{|f|_\infty\leq1} \left|\Expect{\left[\left(\pIS(f)-\post(f)\right)^2\right]}\right| 
        \leq \frac{4\zeta}{\Sam}.
\end{align*}

      \end{proof}

\begin{remark}
In Theorem \ref{IS}  we measure the quality of $\pIS$ as an approximation 
of the \index{target distribution}target $\post$ by considering the worst-case bias and mean squared error over 
the class of bounded test functions $\{f: \Ru \to \R: |f|_\infty =1\}$. We show that 
worst-case error upper-bounds can be obtained in terms of the \index{divergence!$\chi^2$}$\chi^2$ divergence between the \index{target distribution}target and the \index{proposal!distribution}proposal,
quantifying the intuitive fact that, over a broad \emph{class} of test functions,
the performance of \index{importance sampling}importance sampling depends on the closeness between \index{target distribution}target
and \index{proposal!distribution}proposal. Note, however, that for a \emph{specific} function $f$ careful choice of
$\pr$ in the \index{importance sampling}importance sampling methodology may lead to considerable improvement
over \index{Monte Carlo}Monte Carlo sampling. 

Unlike \index{Monte Carlo}Monte Carlo, the \index{importance sampling}importance sampling estimator $\pIS(f)$ is biased for $\post(f)$. The theorem shows, however, that the bias decays at a rate that is twice that of the standard deviation, and so for large $\Sam$ the mean squared error is dominated by the variance. As for \index{Monte Carlo}Monte Carlo, the rate of convergence of the variance 
is governed by the inverse of $\Sam$, and the dimension $\Du$ does not directly appear in the upper-bound. However,
for \index{importance sampling}importance sampling to be accurate (with a limited number of samples $\Sam$)
it is important that \index{target distribution}target and \index{proposal!distribution}proposal are close in \index{divergence!$\chi^2$}$\chi^2$ divergence, 
a condition that will not be typically satisfied in high dimensions. 
\end{remark}

      \begin{example}[Explicit Bound for a \index{linear-Gaussian setting}Linear-Gaussian \index{inverse problem}Inverse Problem]\label{ExampleIS}
     Let $a \in \R$ be given, and consider the one-dimensional \index{inverse problem}inverse problem 
     $$y = au + \eta, \quad \quad \eta \sim \Nc(0, \gamma^2),$$
     supplemented with a \index{Gaussian}Gaussian \index{prior}prior $u \sim \rho(u) = \Nc(0,{\widehat{c}}\, ^2).$ Defining 
     $$\like(u):= \exp\Bigl( -\frac{a^2}{2\gamma^2}u^2 + \frac{ay}{\gamma^2} u  \Bigr)$$
    we can write the \index{posterior}posterior distribution $\post(u)$ in the form \eqref{eq:MC1}, namely
    $$\post(u) = \frac{1}{Z} \like(u) \pr(u).$$
     Setting $\delta^2 := a^2 {\widehat{c}}\, ^2/\gamma^2$ a direct calculation shows that 
     \begin{align*}
     \pr(\like) &= \frac{1}{\sqrt{\delta^2 + 1}} \exp\Bigl( \frac12 \frac{\delta^2 y^2}{a^2 {\widehat{c}}\, ^2 + \gamma^2}  \Bigr), \\
     \pr(\like^2) & = \frac{1}{ \sqrt{2\delta^2 + 1}} \exp\Bigl(  \frac{2\delta^2 y^2}{2a^2 {\widehat{c}}\, ^2 + \gamma^2}  \Bigr), 
     \end{align*}
     and so, noting that $\frac{y}{ \sqrt{a^2 \widehat{c}\, ^2 + \gamma^2}} \sim \Nc(0,1)$ under our model, we obtain that
     \begin{align*}
    \zeta &=\frac{\pr(\like^2)}{\pr(\like)^2} \\
     & = \frac{\delta^2 + 1}{ \sqrt{2\delta^2 + 1} }\exp \Bigl( \frac{\delta^2}{2\delta^2 + 1} z^2  \Bigr), \quad \quad z \sim \Nc(0,1).
     \end{align*}  
     Theorem \ref{IS} then guarantees that 
     $$d(\pIS, \post)^2 \le \frac{4\zeta}{N}.$$
     It is illustrative to note that $\zeta$ ---and hence the \index{divergence!$\chi^2$}$\chi^2$-divergence between the \index{posterior}posterior and the \index{prior}prior--- is an increasing function of $\delta^2 =  a^2 {\widehat{c}}\, ^2/\gamma^2$. This is intuitive, since (i) larger $a$ and ${\widehat{c}}\, ^2$ make the \index{prior}prior less informative; and (ii) smaller $\gamma$ makes the data more informative. In either of these two limiting regimes, we expect the \index{posterior}posterior to become further away from the \index{prior}prior. 
      \end{example}

\section{Discussion and Bibliography}\label{sec:53}
A classic reference on the \index{Monte Carlo}Monte Carlo method is \cite{hammersley1964percolation}. Recent textbooks covering both methodological and theoretical aspects of \index{Monte Carlo}Monte Carlo methods include \cite{liu2008monte,robert2013monte}. 
In practice, a wide range of probabilities, integrals and summations can be approximated by the \index{Monte Carlo}Monte Carlo method. An advantage of Monte Carlo methods is that the
convergence rate is independent of the dimension of the vector space supporting
the random variable; indeed, the $N^{-1/2}$ rate can be obtained for
infinite-dimensional problems, in principle.  
A caveat of \index{Monte Carlo}Monte Carlo methods is that they converge slowly.
A faster convergence rate can be attained using quasi-random, low discrepancy sequences rather than random samples from the \index{target distribution}target. These quasi-random points can be suitably chosen in order to provide greater uniformity than random or pseudo-random sequences. The convergence theory, practical limitations, and scalability to high dimension of the resulting \emph{quasi-Monte Carlo} methods are overviewed in \cite{caflisch1998monte,dick2013high,sloan1998quasi}.  The subject of multi-level Monte Carlo (MLMC) has made the use of \index{Monte Carlo}Monte Carlo methods practical in new areas of application;
see \cite{giles2015multilevel} for an overview. The methodology applies
when approximating expectations over infinite-dimensional spaces,
and distributes the computational budget over different levels of
approximation, with the goal of optimizing the cost per unit 
error, noting that the latter balances sampling and approximation based
 sources of error.  

\index{importance sampling}Importance sampling is reviewed in \cite{tokdar2010importance}. The methodology was first developed as an approach to reduce the variance of \index{Monte Carlo}Monte Carlo integration \cite{kahn1953methods,kahn1955use}.  The chapter notes \cite{anderson2014monte} give a comparison of \index{Monte Carlo}Monte Carlo and \index{importance sampling}importance sampling with examples. Early investigations of  \index{importance sampling}importance sampling focused on the following question: given a test function $f$, how should one choose the \index{proposal!distribution}proposal $\pr$ so that the estimator $\pIS(f)$ of $\post(f)$ has a small mean squared error? This question has led to a plethora of algorithms for simulation of rare events, which is still a very active area of research.
The presentation in this chapter closely follows the papers \cite{agapiou2017importance,sanz2020bayesian}, which study \index{importance sampling}importance sampling from the perspective of \index{filtering}filtering and \index{sequential importance resampling}sequential importance resampling. In this context, it is important to guarantee the accuracy of  the \index{importance sampling}importance sampling estimator $\pIS(f) \approx \post(f)$ for a variety of test functions. This can be achieved by ensuring that $\pIS$ is close to $\post,$ as shown in Theorem \ref{IS}.
 In order for \index{importance sampling}importance sampling to be accurate for a wide family of test functions,  \index{target distribution}target and \index{proposal!distribution}proposal need to be sufficiently close, since otherwise the \emph{effective sample size} will be low \cite{agapiou2017importance,sanz2020bayesian,martino2017effective}. 
 Necessary sample size results for \index{importance sampling}importance sampling in terms of several \index{divergence}divergences between \index{target distribution}target and \index{proposal!distribution}proposal are established in \cite{sanz2018importance,chatterjee2018sample}. The papers \cite{bugallo2017adaptive,kawai2017adaptive} consider advanced \index{importance sampling}importance sampling via adaptive algorithms.  Some recent adaptive methods are based on the idea of finding the \index{proposal!distribution}proposal distribution within some parametric family that is closest to the \index{target distribution}target distribution in an appropriate sense \cite{ryu2014adaptive,akyildiz2021convergence,deniz2022global}.


 \chapter{\Large{\sffamily{Markov Chain Monte Carlo}}}
\label{ch:6}
In this chapter we study \index{Markov chain!Monte Carlo}Markov chain Monte Carlo \index{MCMC}(MCMC),
a methodology that delivers approximate samples from a given 
\index{target distribution}{\em target} distribution $\post.$  The methodology applies to
settings in which $\post$ is the posterior distribution in
\eqref{eq:bayesformula}, but it is also widely used in numerous applications beyond Bayesian inference. 
  As with \index{Monte Carlo}Monte Carlo and \index{importance sampling}importance sampling, \index{MCMC}MCMC may be viewed as approximating the \index{target distribution}target distribution by a sum of \index{Dirac}Dirac masses, thus 
allowing the approximation of expectations with respect to the \index{target distribution}target. Implementation of \index{Monte Carlo}Monte Carlo presupposes that independent samples from the \index{target distribution}target can be obtained. \index{importance sampling}Importance sampling and \index{MCMC}MCMC bypass this restrictive assumption: \index{importance sampling}importance sampling by appropriately weighting independent samples from a \index{proposal!distribution}proposal distribution, and \index{MCMC}MCMC by drawing correlated samples from a \index{Markov kernel}Markov kernel that has the \index{target distribution}target as \index{invariant distribution}invariant distribution. 

The concepts of \index{Markov kernel}Markov kernel and \index{invariant distribution}invariant distribution will hence play a central role in this chapter, and we start in Section \ref{sec:61} with a recap of the elements of this
theory needed in the remainder of the chapter.
Then in Section \ref{sec:62}  we provide a general discussion of \index{Markov chain!sampling}Markov chain sampling,  which assumes the
existence of an ergodic \index{Markov chain}Markov chain, with a kernel from which samples may
 be drawn iteratively, with \index{invariant distribution}invariant distribution equal to the \index{target distribution}target
$\post$. Following that, in Section \ref{sec:63} we discuss \index{Metropolis-Hastings}Metropolis-Hastings
sampling which assumes the existence of a \index{Markov kernel}Markov kernel from which samples may 
readily be drawn, and then uses a correction mechanism to obtain
a new \index{Markov chain}Markov chain with \index{invariant distribution}invariant distribution equal to the \index{target distribution}target $\post$.
The relationship between \index{Metropolis-Hastings}Metropolis-Hastings sampling and \index{Markov chain}Markov chain
sampling is analogous to the relationship between \index{importance sampling}importance sampling
and \index{Monte Carlo}Monte Carlo sampling. 
After introducing the general \index{Metropolis-Hastings}Metropolis-Hastings methodology, and showing its invariance with respect to the \index{target distribution}target distribution in Section \ref{sec:64}, 
we will specify to the case where $\post$ is a \index{posterior}posterior distribution 
given via \index{Bayes theorem}Bayes theorem from the product of the \index{likelihood}likelihood function and 
the \index{prior}prior distribution. In this context, we will analyze in Section \ref{sec:65} the convergence 
of the \index{pCN}pCN algorithm, which uses 
the \index{prior}prior and the \index{likelihood}likelihood separately as part of its design, and is prototypical
of many useful \index{Metropolis-Hastings}Metropolis-Hastings methods,
especially for high-dimensional sampling problems.   
The chapter closes in Section \ref{sec:66} 
with extensions and bibliographical remarks. 

\section{\index{Markov chain}Markov Chains in $\R^d$}\label{sec:61}
We recall that $p:\R^{\Du} \times \R^{\Du} \to \R$ is called a 
\emph{\index{Markov kernel}Markov kernel} if,
\begin{enumerate}[label=(\roman*)]
\item $p(u,v) \ge 0$ for all $(u,v) \in \R^{\Du} \times \R^{\Du};$ and
\item $\int_{\R^\Du} p(u,v) \; dv = 1$ for all $u\in \R^{\Du}.$
\end{enumerate}
Thus if $p:\R^{\Du} \times \R^{\Du} \to \R$ is a \index{Markov kernel}Markov kernel, then $p(u,\cdot): \R^{\Du} \to \R^+$ is a pdf on $\R^d.$ We also recall that $\post$ 
is an \index{invariant distribution}\emph{invariant distribution} of the \index{Markov kernel}Markov kernel $p(u,v)$ if, for any $v\in \Ru,$ 
\begin{equation}
\int_{\Ru} \post(u) p(u,v)\, du = \post(v). 
\end{equation}
A \emph{sample path} of the \index{Markov chain}Markov chain generated by kernel $p$ is defined as follows:
given initial distribution $\pi_0$, generate $\{u^{(n)}\}_{n \in \mathbb{Z}^+}$
inductively:
\begin{align*}
u^{(0)} &\sim \pi_0,\\
u^{(n+1)} &\sim p(u^{(n)},\cdot), \quad n \in \mathbb{Z}^+.
\end{align*}
Note that $\{u^{(n)}\}_{n \in \mathbb{Z}^+}$  is a random sequence and hence, for each 
$n  \in \mathbb{Z}^+$, there is a marginal distribution on $u^{(n)}$, denoted $\pi_n$;
in later discussions correlations between $u^{(n)}$ and $u^{(m)}$ for $n \ne m$ 
will also be relevant. The following result is fundamental.

\begin{lemma}
Let $\post$ be an \index{invariant distribution}invariant distribution of the \index{Markov kernel}Markov kernel $p$. Let $\{u^{(n)}\}_{n \in \mathbb{Z}^+}$ be a sample path generated with kernel $p$ and initial distribution $\post_0 = \post.$
Then it follows that $u^{(n)} \sim \pi$ for all 
$n  \in \mathbb{Z}^+$.
\end{lemma}

\begin{proof}
By induction it suffices to show that if $\pi_n=\pi$ then $\pi_{n+1}=\pi.$ Let $A$
denote an arbitrary subset in $\mathbb{R}^d$. We first note that
$$\mathbb{P}(u^{(n+1)} \in A|u^{(n)})=\int_{A} p(u^{(n)},v) \, dv.$$
Thus, using $\pi_n=\pi$, exchanging the order of integration and using
the invariance of $\pi$ with respect to kernel $p$, we find that
\begin{align*}
\pi_{n+1}(A)&=\mathbb{P}(u^{(n+1)} \in A)\\
&=\mathbb{E}^{u^{(n)} \sim \pi_n}  \Bigl[ \mathbb{P} \bigl(u^{(n+1)} \in A|u^{(n)} \bigr) \Bigr]\\
&=\int_{\mathbb{R}^d}\pi_n(u)\Bigl(\int_{A} p(u,v) \, dv\Bigr) \, du\\
&=\int_{\mathbb{R}^d}\pi(u)\Bigl(\int_{A} p(u,v)\, dv\Bigr) \, du\\
&=\int_{A} \Bigl(\int_{\mathbb{R}^d}\pi(u) p(u,v) \, du\Bigr) \, dv\\
&=\int_{A} \pi(v) \, dv\\
&=\pi(A).
\end{align*}
Since $A$ is arbitrary the proof is complete.
\end{proof}

In the following it will be useful to compute expectations with respect to
the distribution on sample paths $\su=\{u^{(n)}\}_{n \in \mathbb{Z}^+}$
implied by the \index{Markov kernel}Markov kernel $p$ and initial distribution $\pi_0$. 
To this end we let $\mathbb{E}$ denote
expectation with respect to the distribution on sample paths and define,
for real-valued functions $f$ and $g$ on the sample paths, 
\begin{align*}
\text{Var}\bigl(f(\su)\bigr)&:=\mathbb{E} \Bigl[\bigl(f(\su)-\mathbb{E}(f)\bigr)^2\Bigr],\\ 
\text{Cov}\bigl(f(\su),g(\su)\bigr)&:=\mathbb{E} \Bigl[\bigl(f(\su)-\mathbb{E}(f)\bigr)\bigl(g(\su)-\mathbb{E}(g)\bigr)\Bigr].
\end{align*}
We will be particularly interested in the case in which the initial distribution 
of the \index{Markov chain}Markov chain is $\pi;$ the preceding lemma shows that each element $u^{(n)}$ of 
the sample path $\su$  is then distributed according to $\pi.$
We then write $\mathbb{E}^{u^{(0)} \sim \pi}.$ If, abusing notation,
 $f$ and $g$ depend only on a single element $u^{(n)}$  of $\su$,  then we have
\begin{align*}
\text{Var}\bigl(f(u^{(n)})\bigr)&:=\text{Var}_{\pi}\bigl(f(u)\bigr):=\post \Bigl(\bigl(f(u)-\post(f)\bigr)^2\Bigr),\\
\text{Cov}\bigl(f(u^{(n)}),g(u^{(m)})\bigr)&:=\mathbb{E}^{u^{(0)} \sim \pi} \Bigl[\bigl(f(u^{(n)})-\post(f)\bigr)\bigl(g(u^{(m)})-\post(g)\bigr) \Bigr].
\end{align*}

\section{\index{Markov chain!sampling}Markov Chain Sampling}\label{sec:62}
The idea of \index{MCMC}MCMC is simple to state: given a \index{target distribution}target distribution $\post,$ 
find a \index{Markov kernel}Markov kernel that can be sampled from and has $\post$ as its 
\index{invariant distribution}invariant distribution. 
Samples  $\{ u^{(\sam)} \}_{\sam=1}^{N}$  drawn iteratively from the 
kernel may be used to approximate \index{posterior}posterior expectations. 
The samples are given uniform weights $1/N$ but, in contrast to standard \index{Monte Carlo}Monte Carlo, they are not independent and they are not drawn exactly from the \index{target distribution}target 
$\post.$ However, if the chain is guaranteed to satisfy {\em sample path \index{ergodicity}ergodicity},
then the resulting estimator for $\post(f)$ is asymptotically unbiased and satisfies a central limit theorem for suitable test functions $f.$ We display the algorithm, define
the estimator and then state a theorem summarizing convergence.

\FloatBarrier
\begin{algorithm}
\caption{\index{Markov chain!sampling}Markov Chain Sampling Algorithm \label{alg_1}}
\begin{algorithmic}[1]
\vspace{0.1in}
\STATE {\bf Input}: \index{target distribution}Target distribution $\post$, initial distribution $\post_0,$ \index{Markov kernel}Markov kernel  $q(u,v)$ with \index{invariant distribution}invariant distribution $\post$, number of samples $\Sam.$  \\
\vspace{.04in}
\STATE {\bf Initial Draw}: Draw initial sample $u^{(0)} \sim \post_0.$ 
 \\
\vspace{.04in}
\STATE{{{\bf Subsequent Samples}}}: For $\sam = 0,1,\dots,\Sam -1$ do:
\begin{enumerate}
\item Sample $u^{(\sam+1)} \sim q(u^{(\sam)},\cdot).$
\end{enumerate}
\STATE{\bf Output}: \index{target distribution}Target approximation $\post \approx \pMCMC: = \frac{1}{\Sam} \sum_{\sam=1}^\Sam \delta(u - u^{(\sam)}).$
\end{algorithmic}
\end{algorithm}
\FloatBarrier

The estimator for $\post(f)$ resulting from Algorithm \ref{alg_1} is then
$$\pMCMC(f) = \frac{1}{\Sam} \sum_{\sam=1}^\Sam f(u^{(\sam)}).$$
Recall the notation $\text{Var}$, $\text{Cov}$ and
$\mathbb{E}^{u^{(0)} \sim \pi}$ from the previous section.
We then have the following result concerning the error in this estimator.

\begin{theorem}[\index{MCMC}MCMC Error]
\label{t:MCMC}
Let $f:\Ru\longrightarrow\mathbb{R}$ satisfy $ {\emph{Var}}_{\post}[f] = 1.$  We have
\begin{align*}
   \Expect^{u^{(0)}\sim \post} \left[\pMCMC (f)-\post(f)\right]  &= 0, \\
   \Expect^{u^{(0)}\sim \post} \left[\left(\pMCMC(f)-\post(f)\right)^2\right] &= \frac{\tau_N^2(f)}{N},
\end{align*}
where 
$$\tau_N^2(f) =  1 + 2  \sum_{m = 1}^{\Sam-1} \frac{\Sam - m}{\Sam}  \emph{Cov}\bigl(f(u^{(0)}), f(u^{(m)}) \bigr) .$$
In particular, 
       \begin{equation*}
        \lim_{N \to \infty} N\Expect^{u^{(0)}\sim \post} \left[\left(\pMCMC(f)-\post(f)\right)^2\right] = \tau^2(f),
        \end{equation*}
where
$$\tau^2(f)=1+2\sum_{m=1}^{\infty}  \emph{Cov}\bigl(f(u^{(0)}), f(u^{(m)}) \bigr),$$
provided that the series converges.
      \end{theorem}
      \begin{proof}
      First note that, under the assumptions that $u^{(0)} \sim \post$ and that the kernel $q$ has \index{invariant distribution}invariant distribution $\post,$ it follows that $u^{(n)} \sim \post$ for all $n \ge 1.$ Therefore, $\pMCMC$ is unbiased for $\post(f)$ by linearity of expectation. 
      Now we characterize the mean squared error of $\pMCMC$, which agrees with its variance:
      \begin{align*}
      \Expect^{u^{(0)}\sim \post} \left[\left(\pMCMC(f)-\post(f)\right)^2\right] &= \text{Var}[\pMCMC (f)]\\
      &=\frac{1}{\Sam^2}\biggl[\sum_{\sam = 1}^\Sam \text{Var}\bigl[f(u^{(n)})\bigr] + 2 \sum_{\sam = 1}^{\Sam-1} \sum_{m >\sam} \text{Cov}\bigl(f(u^{(\sam)}), f(u^{(m)})  \bigr)      \biggr] \\
            &=\frac{1}{\Sam^2} \biggl[  \Sam  + 2  \sum_{\sam = 1}^{\Sam-1} \sum_{m = 1}^{\Sam - \sam} \text{Cov}\bigl(f(u^{(\sam)}), f(u^{(\sam + m)}) \bigr) \biggr]  \\
                         & = \frac{1}{\Sam}\biggl[ 1 + \frac{2}{ \Sam}    \sum_{m = 1}^{\Sam-1} \sum_{\sam = 1}^{\Sam - m} \text{Cov}\bigl(f(u^{(\sam)}), f(u^{(\sam + m)})\bigr)   \biggr] \\
             & = \frac{1}{\Sam}\biggl[ 1 + \frac{2}{ \Sam}    \sum_{m = 1}^{\Sam-1} \sum_{\sam = 1}^{\Sam - m} \text{Cov}\bigl(f(u^{(0)}), f(u^{( m)})\bigr)   \biggr] \\
               & = \frac{1}{\Sam}\biggl[ 1 + 2   \sum_{m = 1}^{\Sam-1} \frac{\Sam - m}{\Sam}  \text{Cov}\bigl(f(u^{(0)}), f(u^{(m)}) \bigr)  \biggr] \\
               &= \frac{\tau_N^2(f)}{\Sam}.
      \end{align*}
The final result follows by the dominated convergence theorem.
      \end{proof}
      
      \begin{remark}
   \index{i.i.d.}  Suppose that $\text{Var}_\post[f] =1.$ If, for $1 \le n \le N,$  $u^{(n)}{\sim}\pi$  are independent, then we have that 
\begin{equation*}
	\text{Var}\left[\frac{1}{N}\sum_{n=1}^N f(u^{(n)})\right] = \frac{1}{N},
\end{equation*}
as we saw in the proof of Theorem \ref{t:MC} for standard \index{Monte Carlo}Monte Carlo.
Thus, if the \index{autocorrelation}\emph{autocorrelations} $\text{Cov}\bigr(f(u^{(0)}) ,f(u^{(m)})\bigr)$ are positive, then the ergodic average will be less accurate than estimated from an\index{i.i.d.} i.i.d. sample. This is because positively correlated random variables have redundant information so are less informative than \index{i.i.d.}i.i.d. random variables. On the other hand, if  the correlations are negative ergodic averages may be more accurate than a direct \index{Monte Carlo}Monte Carlo estimator with\index{i.i.d.} i.i.d. samples.

The theorem is stated in the idealized (and unrealistic) setting in
which the \index{Markov chain}Markov chain starts at the desired \index{target distribution}target distribution.
In general, \index{ergodicity}ergodicity is needed to ensure that chains initialized
far from stationarity will converge to the desired \index{target distribution}target.
Controlling the size of $\tau^2(f)$ and ensuring rapid convergence
to stationarity are the two primary design goals when constructing
\index{Markov chain}Markov chains invariant with respect to $\pi.$
\end{remark}

Addressing the design and analysis of \index{MCMC}MCMC methods in generality and depth is beyond the scope of a single chapter;  entire books are devoted to this subject. We will restrict our discussion to a particular class of \index{MCMC}MCMC methods, known as \index{Metropolis-Hastings}Metropolis-Hastings algorithms. We will prove that the desired \index{target distribution}target distribution is invariant for the \index{Metropolis-Hastings}Metropolis-Hastings kernel, and we will show \emph{geometric \index{ergodicity}ergodicity} of the \index{pCN}pCN \index{Metropolis-Hastings}Metropolis-Hastings algorithm, meaning that the distribution $\post_\sam$ of the $n$-th sample approaches the \index{invariant distribution}invariant distribution exponentially fast in \index{distance!total variation}total variation distance. The idea is illustrated in Figure \ref{fig:MCMC_chain}: after
an initial number of \index{burn-in}{\em burn-in} steps, the samples from 
the chain start to concentrate in regions where the \index{target distribution}target distribution has the greatest mass. We will not discuss sample path \index{ergodicity}ergodicity, noting simply that a general abstract theory exists to deduce it from geometric \index{ergodicity}ergodicity.

  \FloatBarrier
    \begin{figure}[htp]
      \centering
      \includegraphics[width=0.75\columnwidth]{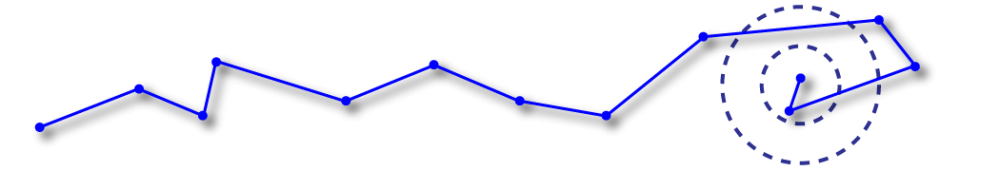}
      \caption{The \index{Markov chain}Markov chain samples points from distribution $\post_\sam$ at step $\sam$, and the sampling distribution converges towards the \index{target distribution}target distribution 
$\post$ whose high density regions are represented by the dashed circles.}
      \label{fig:MCMC_chain}
    \end{figure}
    \FloatBarrier

\section{\index{Metropolis-Hastings}Metropolis-Hastings Sampling}\label{sec:63}

Here we outline the \index{Metropolis-Hastings}Metropolis-Hastings algorithm. The
algorithm has two ingredients: a \index{proposal!kernel}{\em proposal kernel} $q(u,v)$,
which is a Markov transition kernel; and an acceptance probability
$a(u,v)$ that will be used to convert the \index{proposal!kernel}proposal kernel into a kernel $\pmh(u,v)$ for which the given \index{target distribution}target $\post$ is an \index{invariant distribution}invariant distribution. Given the $n$-th sample $u^{(\sam)},$ we generate 
$u^{(\sam+1)}$ by drawing $v^\star$ from the distribution $q(u^{(\sam)}, \cdot).$ The result is 
accepted, which means setting $u^{(\sam+1)}=v^\star$,
with probability $a(u^{(\sam)}, v^{\star})$; it is rejected,
meaning $u^{(\sam+1)} = u^{(\sam)}$, with
the remaining probability $1 - a(u^{(\sam)}, v^{\star})$.
The acceptance probability is given by
\begin{eqnarray}
\label{eq:MHA}
a(u,v) = \text{min} \Bigg( \frac{\post(v) q(v, u)}{\post(u) q(u, v)},\, 1\Bigg).\end{eqnarray}

\FloatBarrier
\begin{algorithm}
\caption{\label{algMH} \index{Metropolis-Hastings}Metropolis-Hastings Algorithm}
\begin{algorithmic}[1]
\vspace{0.1in}
\STATE {\bf Input}: \index{target distribution}Target distribution $\post$, initial distribution $\post_0,$ \index{Markov kernel}Markov kernel $q(u,v)$, number of samples $\Sam.$  \\
\vspace{.04in}
\STATE {\bf Initial Draw}: Draw initial sample $u^{(0)} \sim \post_0.$ 
 \\
\vspace{.04in}
\STATE{{{\bf Subsequent Samples}}}: For $\sam = 0,1,\dots,\Sam -1$ do:
\begin{enumerate}
\item Sample $v^{\star} \sim q(u^{(\sam)},\cdot).$
\item Calculate the acceptance probability $a_\sam:=a(u^{(\sam)},v^{\star}).$
\item Update 
\begin{align*}
u^{(\sam+1)} =
\begin{cases}
 v^{\star}, &  \text{with probability} \,\, a_\sam, \\ 
 u^{(\sam)}, & \text{with probability} \,\,  1 - a_\sam.
\end{cases}
 \end{align*}
\end{enumerate}
\STATE{\bf Output}: \index{target distribution}Target approximation $\post \approx  \pimh^N : = \frac{1}{\Sam} \sum_{\sam=1}^\Sam \delta(u - u^{(\sam)}).$ 
\end{algorithmic}
\end{algorithm}
\FloatBarrier
%
%
%
%
%
%

The estimator resulting from Algorithm \ref{algMH} for $\post(f)$ is then
$$\pimh^N(f)  = \frac{1}{\Sam} \sum_{\sam=1}^\Sam f(u^{(\sam)}).$$
The \index{Metropolis-Hastings}Metropolis-Hastings algorithm implicitly defines a \index{Markov kernel}Markov kernel $\pmh(u,\cdot)$ which specifies the density of the $(\sam+1)$-th sample given that the $n$-th sample is located at $u.$ For $u\neq v,$  the \index{Metropolis-Hastings}Metropolis-Hastings kernel has the following simple expression in terms of the \index{proposal!kernel}proposal kernel and the acceptance probability 
\begin{equation}\label{eq:pmh}
\pmh(u,v) =  a(u,v) q(u,v);
\end{equation}
this expression may be deduced noting that in order to move from $u$ to a new location $v$, the move needs to be proposed and accepted. 

\begin{remark}
 We note the following concerning 
the \index{Metropolis-Hastings}Metropolis-Hastings algorithm. 
\begin{itemize}
\item In order to implement the \index{Metropolis-Hastings}Metropolis-Hastings algorithm one needs to be able to sample from the \index{proposal!kernel}proposal kernel $q(u,\cdot)$ and evaluate the acceptance probability $a(u,v).$ Importantly, the \index{target distribution}target distribution only appears in the acceptance probability $a(u,v),$ and only the ratio $\post(v)/\post(u)$ is involved. Therefore, the \index{Metropolis-Hastings}Metropolis-Hastings algorithm may be implemented for \index{target distribution}target distributions that are only specified up to an unknown normalizing constant. 
\item If $q(u,v) = q(v,u)$ the acceptance probability simplifies to 
$\min \Bigl(1,\post(v)/\post(u)\Bigr).$ This is the setting in which the
original \emph{Metropolis algorithm} was introduced. 
In such a case, moves to regions of higher \index{target distribution}target density are always accepted, while moves to regions of smaller but non-zero target density are accepted with positive probability in order to ensure exploration of the target space. The quantity $\pi(u)q(u,v)$
should be viewed as a joint distribution on the pair $(u,v)$ with $u$ distributed
according to the \index{invariant distribution}invariant distribution and $v|u$ then defined by the \index{Markov kernel}Markov kernel.
In the general \emph{\index{Metropolis-Hastings}Metropolis-Hastings algorithm setting}, when $q$ is not
necessarily symmetric, the method favors moves that are easier 
to be reversed, in the sense that $\pi(v)q(v,u) >\pi(u)q(u,v).$ 
\item The \index{Metropolis-Hastings}Metropolis-Hastings algorithm is extremely flexible due to the freedom in the choice of \index{proposal!kernel}proposal kernel $q(u,v)$. For this algorithm the ergodic behavior, and size 
of $\tau^2(\cdot)$, is heavily dependent on the choice of \index{proposal!kernel}proposal kernel.
\item The accept-reject step may be implemented by drawing, independently
from the \index{proposal!kernel}proposal, a uniformly distributed random variable $\theta_\sam$ in the interval $[0,1].$ Recall $a_\sam$ as defined in
Algorithm \ref{alg_1}. If $\theta_\sam \in [0,a_\sam),$ then the \index{proposal}proposal is accepted ($u^{(\sam+1)}=v^\star$); it is rejected ($u^{(\sam+1)}=u^{(\sam)})$ otherwise. 
\end{itemize}
\end{remark}

\section{Invariance of the \index{target distribution}Target Distribution \texorpdfstring{$\post$}{}} \label{sec:64}
In this section we show that the \index{target distribution}target $\post$ is an \index{invariant distribution}invariant distribution for the \index{Metropolis-Hastings}Metropolis-Hastings kernel. We start by introducing the notion of \index{detailed balance}detailed balance and showing that it implies invariance. We then prove that the \index{Metropolis-Hastings}Metropolis-Hastings kernel satisfies \index{detailed balance}detailed balance with respect to $\post$, and hence $\post$ is invariant.

\subsection{\index{detailed balance}Detailed Balance and its Implication}
A \index{Markov kernel}Markov kernel $p(u,v)$ satisfies  \index{detailed balance}\emph{detailed balance} with respect to $\post$ if, for any $u, v \in \Ru,$
\begin{equation*}
\post(u) p(u,v) = \post(v) p(v,u).
\end{equation*}
\index{detailed balance}Detailed balance of $p(u,v)$ with respect to $\post$ implies that $\post$ is an \index{invariant distribution}invariant distribution for $p(u,v)$. To see this, note that if $p(u,v)$ satisfies \index{detailed balance}detailed balance with respect to $\post,$ then
$$\int_{\Ru} \post(u) p(u,v)\, du = \post(v) \int_{\Ru} p(v,u)\, du = \post(v).$$
Invariance guarantees that, if the chain is distributed according to $\post$ at a given step, then it will also be distributed according to $\post$ in the following step. \index{detailed balance}Detailed balance  
guarantees that the in/out probability
flux between any two \index{state}states is preserved; this is a stronger 
condition, which implies invariance. 

\subsection{\index{detailed balance}Detailed Balance and the \index{Metropolis-Hastings}Metropolis-Hastings Algorithm}
The following theorem establishes the \index{detailed balance}detailed balance of the \index{Metropolis-Hastings}Metropolis-Hastings kernel with respect to the \index{target distribution}target $\post$; it implies, as a consequence, that the \index{target distribution}target is an \index{invariant distribution}invariant distribution for the \index{Metropolis-Hastings}Metropolis-Hastings kernel. 


      \begin{theorem}[\index{Metropolis-Hastings}Metropolis-Hastings and \index{detailed balance}Detailed Balance]
       The \index{Metropolis-Hastings}Metropolis-Hastings kernel satisfies \index{detailed balance}detailed balance with 
respect to the distribution $\post$.
      \end{theorem}

\begin{proof}
We need to show that, for any $u,v \in \R^d,$ 
\begin{equation}
\label{eq:MHP0}
\post(u) \pmh(u,v) = \post(v) \pmh(v,u).
\end{equation}
We let $v^\star$ denote the point proposed from kernel $q(u,\cdot)$,
calculate the joint probability distribution of $(u,v^\star,v)$ and
then integrate out $v^\star$ in order to identify $\post(u) \pmh(u,v).$
We first note that the random variable $v|(u,v^\star)$ has density
\begin{equation}
\label{eq:MHP1}
\delta_{v^\star}(v) a(u,v^\star)+\delta_u(v)\bigl(1-a(u,v^\star)\bigr).
\end{equation}
The density of $(u,v^\star)$ is found from the product of the density of 
$v^\star|u$ and the density of $u$ and is hence given by
\begin{equation}
\label{eq:MHP2}
q(u,v^\star)\pi(u).
\end{equation}
Multiplying \eqref{eq:MHP1} and \eqref{eq:MHP2} gives the density
of $(u,v^\star,v)$ and integrating out $v^\star$ gives the density of
$(u,v)$, namely
\begin{align*}
\pi(u) \pmh(u,v) &=\pi(u)q(u,v)a(u,v)+\pi(u)\delta_u(v)\beta(u),\\
\beta(u) & = \int_{\R^d} \bigl(1-a(u,v^\star)\bigr) q(u,v^\star) dv^\star.
\end{align*}

Now note that
\begin{align*}
       q(u, v)a(u,v) &= \text{min }\biggl( \frac{\post(v) q(v,u)}{\post(u) q(u,v)}  ,1   \biggr) ~ q(u, v)\\
       &= \frac{1}{\post(u)} \times \text{min }\Bigl( \post(u) q(u,v), \post(v) q(v,u)\Bigr).
\end{align*}
Thus, invoking symmetry,
       $$ \post(u) q(u, v)a(u,v) = \text{min }\bigl( \post(u) q(u,v), \post(v) q(v,u)\bigr) = \post(v) q(v, u)a(v,u).$$
It is then apparent that $\pi(u) \pmh(u,v)$ is symmetric with respect to
the pair $(u,v)$, establishing \eqref{eq:MHP0} and completing the proof.
\end{proof}


Invariance of the \index{Metropolis-Hastings}Metropolis-Hastings kernel $\pmh$ 
with respect to $\post$ implies that if the initial sample is drawn from the \index{target distribution}target ($\post_0 = \post)$, then all subsequent samples are also distributed according to the \index{target distribution}target $(\post_n = \post).$ 


\section{Convergence to the \index{target distribution}Target Distribution}\label{sec:65}
In the previous section we showed that if we initialize the \index{Metropolis-Hastings}Metropolis-Hastings algorithm with distribution $\post$, all samples produced by the algorithm will be distributed according to $\post.$
 But the motivation for the 
\index{Metropolis-Hastings}Metropolis-Hastings algorithm
is that we are not able to directly sample from $\post$. 
Our aim in this section is to show that, for certain \index{Metropolis-Hastings}Metropolis-Hastings methods, the law $\pi_n$ of the $n$-th sample converges to $\post$ regardless of the initial distribution $\post_0.$
This is a strong form of  {\em ergodic} behavior which does not hold in general, as illustrated by 
the chain depicted in Figure~\ref{fig-convergence}.

\begin{figure}
  \centering
  \includegraphics[scale=0.70]{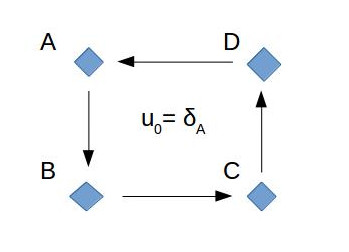}
  \caption{The arrows represent transitions with probability one in a four
\index{state}state \index{Markov chain}Markov chain. The \index{invariant distribution}invariant distribution is the uniform distribution but for $\post_{0} = \delta_{A}$, we have $\post_1 = \delta_{B}, \post_2 = \delta_{C}, \post_3 = \delta_{D}, \post_4 = \delta_A,$ etc. Here $\pi_n$ does not converge to a limit distribution.}\label{fig-convergence}
\end{figure}

In order to understand the mechanisms behind \index{ergodicity}ergodicity we will
first consider \index{Markov chain}Markov chains with finite \index{state-space!finite}state-space. We then study a specific \index{Metropolis-Hastings}Metropolis-Hastings 
algorithm, known as the \index{pCN}pCN (for {\em preconditioned Crank-Nicolson}) method,
which applies to \index{target distribution}targets $\post$ defined by their density with respect to
a \index{Gaussian}Gaussian distribution.

\subsection{\index{state-space!finite}Finite State-Space}

We consider a \index{Markov chain}Markov chain on the \index{state-space!finite}finite state-space 
$S = \left\{ 1 , \cdots , \Du \right\}.$
The \index{Markov kernel}Markov kernel described earlier becomes a $d \times d$ 
\emph{transition matrix}
$P$ with non-negative 
entries $p(i,j)$ satisfying
$$\sum_{j\in S}\ p(i,j)=1.$$ 
The \index{invariant distribution}invariant distribution becomes $d \times 1$ column vector $\pi$, with
non-negative entries which sum to one, satisfying
\begin{equation}
\label{eq:evone}
\pi^\top=\pi^\top P.
\end{equation}
Such an \index{invariant distribution}invariant distribution always exists but is not, in general, unique.
The distribution at each step of the \index{Markov chain}Markov chain
is the $d \times 1$ column vector $\pi_n$ satisfying
\begin{equation}
\label{eq:evtwo}
\pi^\top_{n+1}=\pi_n^\top P,
\end{equation}
where $\pi_0$ is the initial distribution of the chain.
\index{ergodicity}Ergodicity may be defined as convergence of the sequence $\pi_n$ to limit $\pi$
as $n \to \infty;$ this is related to the eigenvalue $1$ of $P$ having 
algebraic and geometric multiplicity one.
We now illustrate a {\em \index{coupling}coupling} approach to proving \index{ergodicity}ergodicity and then,
in the next subsection, generalize the methodology
to study the \index{pCN}pCN method on the \index{state-space!continuous}continuous state-space $\mathbb{R}^d$.
 
\begin{theorem}[\index{ergodicity}Ergodicity in \index{state-space!finite}Finite State-Space]
Let $\{u^{(n)}\}_{n \in \mathbb{Z}^+}$ be a \index{Markov chain}Markov chain with \index{state-space}state-space $S = \left\{ 1 , \cdots , \Du \right\}$, transition matrix $P$ and initial distribution $\post_{0}$. Assume that there is $\eps>0$ such that 
  \begin{equation}
    \label{eqn-unif}
    \min_{(i,j) \in S \times S} p\left( i,j \right)\ge \frac{\eps}{\Du}. 
  \end{equation}
Then there is a unique solution $\post$ to \eqref{eq:evone} within the
class of probability vectors on $S$. Furthermore, the following convergence result holds
for iteration \eqref{eq:evtwo}:
  \begin{equation}\label{eq:tvbound}
    \dtv( \post_{\sam}, \post) \leq \left( 1 - \eps \right)^{\sam}.
  \end{equation}
\end{theorem}
\begin{proof}
First note that the Markov matrix $P$ is a continuous map from
the space of probability distributions on $S$
into itself; it thus continuously maps a compact, convex set into 
itself. By Brouwer's fixed point theorem it follows that $P$ has a fixed point 
in this space, ensuring that an \index{invariant distribution}invariant distribution $\pi$ solving \eqref{eq:evone}
exists. We will now show that for any \index{invariant distribution}invariant distribution $\post$ equation \eqref{eq:tvbound} holds, which also implies the uniqueness of the \index{invariant distribution}invariant distribution within
the class of probability vectors.

Let $\post$ be an \index{invariant distribution}invariant distribution, a probability vector on $S$. Proving 
convergence to equilibrium amounts to ``forgetting the past'', to show that the long time behavior of the \index{Markov chain}Markov chain does not depend on the initial distribution $\post_0$ and in fact converges to $\post$. In general, $u^{(\sam+1)}$ will be strongly dependent on $u^{(\sam)}$, but 
the condition given in \eqref{eqn-unif} implies that there is always some residual chance that the chain jumps to any new \index{state}state, at each step, independently
of where it is currently located, $u^{(\sam)}$. This residual 
probability of the chain to make a ``totally random'' move will 
be shown to diminish the stochastic dependence on $u^{(0)}$ as $\sam$ increases. 

To formalize this idea, let $b_\sam$ be \index{i.i.d.}i.i.d. \index{Bernoulli}Bernoulli random variables with $\Prob\left( b_\sam = 1 \right) = \eps$ and $\Prob\left( b_\sam = 0 \right) = 1-\eps;$ furthermore assume that the sequence $\{b_\sam\}$ is independent of the randomness defining draws from $\{p(u^{(n)},\cdot)\}$. 
Define $r$ to be the uniform transition kernel with equal probability of
transitioning to each \index{state}state in $S: r(i,j)=\Du^{-1}$ for 
all $(i,j) \in S \times S.$ 

Using the lower bound on $p$ we may define 
a new \index{Markov chain}Markov chain $\{w^{(\sam)}\}_{n \in \mathbb{Z}^+}$ as follows:
  \begin{equation} \label{eq:kernel}
    w^{(\sam + 1)} \sim
    \begin{cases}
      &s\left(w^{(\sam)}, \cdot \right), \hfill \text{ for } b_\sam = 0,\\
      &r\left(w^{(\sam)}, \cdot \right), \hfill \text{ for } b_\sam = 1,
    \end{cases}
  \end{equation}
  where 
  \begin{equation*}
    s\left( i, j \right) := \frac{p\left(i, j \right) - \eps r\left(i, j \right)}{1- \eps}.
  \end{equation*}
We make two observations about this construction. First, the lower bound of 
$\eps/\Du$ on $p(i,j)$ means that the probability transition matrix 
$s$ is well-defined; second, the fact that $r(i,j)$ is independent of $i$ is key,
as it means that sampling explicitly forgets the current \index{state}state whenever $b_\sam=1$. 
We may now compute  
  \begin{align*}
    \Prob\left( w^{(\sam + 1)} = j | w^{(\sam)} = i \right) &= \eps \Prob\left( w^{(\sam + 1)} = j | w^{(\sam)} = i, b_\sam = 1 \right) \\
     &\hspace{3.8cm} + \left( 1 - \eps \right) \Prob\left( w^{(\sam + 1)} = j | w^{(\sam)} = i, b_\sam = 0 \right)\\
    &= \eps r\left( i,j \right) + p\left( i,j \right) - \eps r\left( i,j \right) \\
    &= p\left( i,j \right).
  \end{align*}
  Thus the kernel defined by \eqref{eq:kernel} is equivalent in law to that defined
by matrix $P$.
  However, by introducing the ancillary random variables $b_\sam$, we have 
made explicit the concept of ``forgetting the past entirely, with a small probability'' at every step. We may now use this to complete the proof.
Let $f:S \mapsto \R$ be an arbitrary test function with $|f|_\infty \leq 1$ and $\tau \coloneqq \min\left( \sam \in \N: b_\sam = 1 \right).$
Then, regardless of how $w^{(0)}$ is initialized, 
  \begin{align*}
    \Expect\left[ f\left( w^{(\sam)} \right) \right] 
    &= \Expect\left[ f\left( w^{(\sam)} \right) | \tau \geq \sam \right]\Prob\left( \tau \geq \sam \right)  + \sum \limits_{l = 0 }^{\sam-1} \Expect\left[f \left(w^{(\sam)}\right) | \tau = l  \right] \Prob\left( \tau =l \right)  \\
    &=\underbrace{\Expect\left[ f\left( w^{(\sam)} \right) | \tau \geq \sam \right]\Prob\left( \tau \geq \sam \right)}_{|\cdot| \leq \left( 1-\eps \right)^{\sam}}  + \underbrace{\sum \limits_{l = 0 }^{\sam-1} \Expect^{w^{(0)} \sim \mathfrak{u}\left( \cdot \right)}\left[f\left(w^{(\sam-l)}\right) \right] \Prob\left( \tau =l \right)}_{\text{independent of original initial distribution}}, 
  \end{align*}
where $\mathfrak{u}$ denotes the uniform distribution on $S.$

Now consider two \index{Markov chain}Markov chains $\{w^{(\sam)}\}$ and $\{w^{(\sam)'}\}$ with kernel \eqref{eq:kernel}, the first initialized from $\pi_0$ and the second from an \index{invariant distribution}invariant distribution $\post;$ denote their distributions at time $\sam$ by $\pi_n$ and $\pi_n',$ respectively. The law of $w^{(n)}$ agrees with the law $\post_\sam$ of the original chain $u^{(n)}$ when initialized at $\pi_0$; on the other hand, for the second chain it follows from invariance that $\post_n' = \post.$ We will use the \index{variational!characterization of total variation}variational characterization of the \index{distance!total variation}total variation distance established in Lemma \ref{lemmatv}. Employing the preceding identity and noting that the contribution which is independent of the initial distribution will cancel in the two different \index{Markov chain}Markov chains, we obtain
  \begin{align*}
    \dtv( \post_{\sam}, \post_\sam') = \frac{1}{2} \sup \limits_{|f|_\infty \leq 1} \left| \Expect^{\post_{\sam}}\left[ f\left( u \right) \right] - \Expect^{\post_\sam'}\left[ f\left( u \right) \right]\right| \leq \left( 1 - \eps \right)^{\sam}. 
  \end{align*}
Since  $\post_\sam' = \post$ the desired result follows.
\end{proof}

Before extending the above argument to a setting with \index{state-space!continuous}continuous state-space, we make two remarks: 
\begin{remark}
The \index{coupling}coupling proof we have just exhibited may be generalized in a number of
ways; in particular:
  \begin{itemize}
    \item The distribution $r$ does not need to be uniform; it was only chosen so 
for convenience. What is important is that  $r(i,j)$ is lower bounded,
independently of $i$, for all $j$. Adapting $r$ to the matrix $P$ at hand, might 
in some cases greatly improve the above bound ---a larger $\eps$ might
be identified.
    \item Convergence to equilibrium can also be shown if condition~\eqref{eqn-unif} holds with $P$ replaced by the $\sam$-step transition Markov matrix $P^{\sam}$. 
Again, for some chains this may yield faster bounds on the convergence to equilibrium.
  \end{itemize}
  \end{remark}

\subsection{The \index{pCN}pCN Method}

The \index{coupling}coupling argument used in the previous subsection for \index{Markov chain}Markov chains with \index{state-space!finite}finite state-space may also be employed to study \index{ergodicity}ergodicity of \index{Markov chain}Markov chains on a \index{state-space!continuous}continuous state-space.
To illustrate this, we consider a particular \index{Metropolis-Hastings}Metropolis-Hastings algorithm,
the \index{pCN}pCN method, applied to a specific \index{Bayesian!inverse problem}Bayesian \index{inverse problem}inverse problem setting.
Before we get into the details of this setting, we describe the
idea behind the \index{pCN}pCN method at a high level. 
The idea is this. If the desired \index{target distribution}target distribution has the form
\begin{equation}\label{eq:formoftarget}
\post(u)= \frac{1}{Z}  \tilde{\like}(u)  \tilde{\pr}(u)  
\end{equation}
and if the \index{Metropolis-Hastings}Metropolis-Hastings \index{proposal!kernel}proposal kernel $q$ satisfies \index{detailed balance}detailed balance with respect to $\tilde{\pr}$, then \eqref{eq:MHA} simplifies to give 
\begin{eqnarray}
\label{eq:MHA2}
a(u,v) = \text{min} \Bigg( \frac{\tilde{\pr}(v)\tilde{\like}(v) q(v, u)}{\tilde{\pr}(u)\ \tilde{\like}(u) q(u, v)}, 1\Bigg)= \text{min} \Bigg( \frac{  \tilde{\like}(v)}{\tilde{\like}(u)}, 1\Bigg).\end{eqnarray}
We will apply and study this idea in the case where $\tilde{\pr}$ is a \index{Gaussian}Gaussian distribution, in which case it is straightforward to construct a \index{proposal!kernel}proposal kernel that satisfies \index{detailed balance}detailed balance with respect to $\tilde{\pr}$. This scenario arises naturally in \index{Bayesian!inverse problem}Bayesian \index{inverse problem}inverse problems where the \index{prior}prior is either a \index{Gaussian}Gaussian, or it is naturally expressed
via density which is the product of a \index{Gaussian}Gaussian with another function.
We now formalize the \index{inverse problem}inverse problem setting that we consider by imposing certain assumptions on the \index{likelihood}likelihood and the \index{prior}prior, and then relate both to the functions $\tilde{\like}$ and $\tilde{\pr}$ in equation \eqref{eq:formoftarget}.

\begin{assumption}
\label{a:erg}
We make the following assumptions on the \index{Bayesian!inverse problem}Bayesian \index{inverse problem}inverse problem:
\begin{itemize}
    \item Bounded \index{likelihood}likelihood: there are $\like^-, \like^+ >0$ such that, for all $u \in \Ru,$ 
    $$0 <\like^- <\like(u) < \like^+.$$ 
    \item Truncated \index{Gaussian!truncated}Gaussian \index{prior}prior: there is a compact set $B \subset \Ru$ and of positive \index{Lebesgue}Lebesgue measure such that   $\pr(u) \propto\mathbb{1}_{B}(u)   z(u),$ where $z$ is the pdf of \index{Gaussian}Gaussian $\Nc  ( \nc 0,\Cpr  )\nc.$
\end{itemize}
\end{assumption}

Under Assumption \ref{a:erg} we obtain for the \index{posterior}posterior density 
\begin{equation*}
 \post\left( u \right)  \propto \like(u) \mathbb{1}_{B}\left( u \right) z(u),
\end{equation*}
which is of the form in equation \eqref{eq:formoftarget} with $\tilde{\like}(u) = \like(u) \mathbb{1}_B(u)$ and $\tilde{\pr}(u) = z(u).$

The \index{pCN}pCN method is a \index{Metropolis-Hastings}Metropolis-Hastings algorithm with \index{proposal!kernel}proposal kernel 
\begin{equation}\label{eq:pcnp}
q(u,\cdot) \sim \Nc\Bigl( (1-\beta^2)^{1/2} u, \beta^2 \Cpr  \Bigr),
\end{equation}
where $\beta \in (0,1]$ is a user-specified parameter that 
should be tuned to obtain an acceptance probability that, on average,
stays away from $0$ or $1$ ---for example one that is approximately $1/2.$
Thus, given the sample $u^{(\sam)},$ the \index{pCN}pCN proposes a new sample 
\begin{equation*}
  v^\star \sim \left( 1 - \beta^{2} \right)^{1/2}u^{(\sam)} + \beta \xi^{(\sam)}, 
\quad \xi^{(\sam)}\sim \Nc  ( \nc 0, \Cpr  ) \nc,
\end{equation*}
which only requires to sample a \index{Gaussian}Gaussian. Note that
$$\Expect \Bigl[ v^\star (v^\star)^\top \Bigr]= (1 - \beta^{2}) \Expect \Bigl[ u^{(\sam)}
 (u^{(\sam)})^\top \Bigr] +\beta^2 \Cpr ,$$
demonstrating that if $u^{(\sam)} \sim \Nc(0,\Cpr)$ then the \index{proposal}proposal satisfies
$v^\star \sim \Nc(0,\Cpr)$ as well. The following lemma shows the stronger
result that the \index{proposal!kernel}proposal kernel satisfies \index{detailed balance}detailed balance with respect to $z.$

\begin{lemma}\label{lem:detbal}
The \index{proposal!kernel}proposal kernel \eqref{eq:pcnp} satisfies \index{detailed balance}detailed balance with respect to the pdf $z$ of \index{Gaussian}Gaussian $\Nc ( 0,\Cpr )$.
\end{lemma}
\begin{proof}
Recall the notation for the covariance weighted inner-product and
resulting norm described in the introduction to these notes.
We need to show that $z(v) q(v,u)$ is symmetric in $u$ and $v.$ By direct calculation,
\begin{align*}
-\log\bigl( z(v) q(v,u) \bigr)  &=\frac{1}{2}|v|^{2}_{\Cpr} + \frac{1}{2 \beta^{2}} \left| u - \left( 1 - \beta^{2} \right)^{1/2} v\right|_{\Cpr}^{2}\\ 
  &=\left( \frac{1}{2} + \frac{\left( 1 - \beta^{2} \right)}{2 \beta^{2}}  \right) |v|^{2}_{\Cpr} + \frac{1}{2\beta^{2}} |u|^{2}_{\Cpr} - \frac{\left( 1 - \beta^{2} \right)^{1/2}}{\beta^{2}} \langle u , v \rangle_{\Cpr}\\
  &= \frac{1}{2 \beta^{2}}\left( |v|^{2}_{\Cpr} + |u|^{2}_{\Cpr} \right) - \frac{\left( 1 - \beta^{2} \right)^{1/2}}{\beta^{2}} \langle u , v \rangle_{\Cpr}.
\end{align*}
\end{proof}

We now display the \index{pCN}pCN algorithm applied in the setting of Assumption \ref{a:erg} and describe how it leads to an estimator of $\post(f).$ The expression for the acceptance probability in Algorithm \ref{alg_pCN} follows from equation \eqref{eq:MHA2} using Lemma \ref{lem:detbal} and noting that $\post_0$ being supported on $B$ implies that  $u^{(\sam)} \in B$ for all $\sam,$  as any proposed move out of $B$ will be rejected. Thus, $\mathbb{1}_{B}\left( u \right)$ may be dropped from the formula for the acceptance probability in equation \eqref{eq:MHA2}.

\FloatBarrier
\begin{algorithm}
\caption{\label{alg_pCN} \index{pCN}pCN Algorithm}
\begin{algorithmic}[1]
\vspace{0.1in}
\STATE {\bf Input}:  Tuning parameter $\beta \in (0,1),$ covariance $\Cpr$, bounded set $B$, \index{likelihood}likelihood function $g$, initial distribution $\post_0$ supported on $B$, number of samples $\Sam.$  \\
\vspace{.04in}
\STATE {\bf Initial Draw}: Draw initial sample $u^{(0)} \sim \post_0.$
 \\
\vspace{.04in}
\STATE{{{\bf Subsequent Samples}}}: For $\sam = 0,1,\dots,\Sam -1$ do:
\begin{enumerate}
\item Sample $v^{\star} \sim \Nc\Bigl( (1-\beta^2)^{1/2} u^{(\sam)} , \beta^2 \Cpr  \Bigr) .$
\item Calculate the acceptance probability $a_\sam:=a(u^{(\sam)},v^{\star})$ where
\begin{align*}
    a\left( u,v \right)
         =\min\left( \frac{\like(v)}{\like(u)}  \mathbb{1}_{B} \left( v \right) ,1  \right).
  \end{align*}
\item Update
\begin{align*}
u^{(\sam+1)} =
\begin{cases}
 v^{\star}, &  \text{with probability} \,\, a_\sam, \\ 
 u^{(\sam)}, & \text{with probability} \,\,  1 - a_\sam.
\end{cases}
 \end{align*}
\end{enumerate}
\STATE{\bf Output}: \index{target distribution}Target approximation $\post \approx  \ppCN:  = \frac{1}{\Sam} \sum_{\sam=1}^\Sam \delta(u - u^{(\sam)}).$ 
\end{algorithmic}
\end{algorithm}
\FloatBarrier

The estimator for $\post(f)$  resulting from Algorithm \ref{alg_pCN}  is then
$$ \ppCN(f)  = \frac{1}{\Sam} \sum_{\sam=1}^\Sam f(u^{(\sam)}).$$
We can now prove \index{ergodicity}ergodicity using similar techniques to those employed
in the previous subsection in the \index{state-space!finite}finite state-space setting. 
The main idea is that, restricted to the bounded set $B$, the probability 
density of the transition kernel is bounded away from zero by some $\eps$.
Splitting off a ``forgetful part'' that is triggered with probability $\eps$ 
then yields the result.

\begin{theorem}[\index{ergodicity}Ergodicity for \index{pCN}pCN Method]
\label{t:pCN}
Assume that we apply the \index{pCN}pCN method to sample from a
\index{posterior}posterior density $\post$ arising from Assumptions \ref{a:erg} with initial
condition drawn from any density supported on $B$. Then
there exists a constant $\eps \in \left( 0,1 \right)$ such that 
  \begin{equation*}
    \dtv ( \post_{\sam} , \post) \leq \left( 1 - \eps \right)^{\sam},
  \end{equation*}
where $\post_{\sam}$ is the law of the $n$-th sample from the \index{pCN}pCN \index{Metropolis-Hastings}Metropolis-Hastings algorithm.
\end{theorem}

\begin{proof}[Proof of Theorem \ref{t:pCN}]
Note again that since $u^{(0)} \in B$ we have $u^{(n)} \in B$ for all $n \ge 1.$ 
Note further that since $B$ is compact and $q$ is continuous in both of its arguments, there is $q^- >0$ such that, for any $u, v \in B,$
$$q(u,v) \ge q^-.$$
Let $p$ be the \index{Markov kernel}Markov kernel defined by the \index{pCN}pCN \index{Metropolis-Hastings}Metropolis-Hastings algorithm. It follows that, for $u, v \in B,$ 
\begin{align*}
p(u,v) & \ge q(u,v) a(u,v) \\
& \ge q^- \frac{g^-}{g^+} =: \eps \text{Leb} (B),
\end{align*}
where the last equation defines $\eps$ and $\text{Leb}(B)$ denotes the \index{Lebesgue}Lebesgue measure of $B$ (which is assumed to be positive).
Analogously to the discrete proof, we now define $b_\sam$ to be \index{i.i.d.}i.i.d. 
\index{Bernoulli}Bernoulli random variables with $\Prob\left( b_\sam = 1 \right) = \eps$, independently of all other randomness, and consider the transition rule
  \begin{equation*}
    u^{(\sam +1)} \sim
    \begin{cases}
      &s\left( u^{(\sam)} , \cdot \right), \hfill \text{ for } b_\sam = 0,\\
      &r\left( \cdot \right), \hfill \text{ for } b_\sam = 1,
    \end{cases}
  \end{equation*}
  where $r$ denotes the uniform distribution on $B$ and, for $A\subset B$ and $u \in B,$
  \begin{equation*}
    s\left( u, A \right) \coloneqq \frac{p\left(u, A \right) - \eps r\left( A \right)}{1- \eps}.
  \end{equation*}
Just as in the discrete case, one can check that the resulting \index{Markov kernel}Markov kernel is equal to the \index{pCN}pCN \index{Metropolis-Hastings}Metropolis-Hastings kernel $p(\cdot,\cdot)$. Furthermore, exponential convergence is then deduced in exactly the 
same way as in the discrete case.
\end{proof}

\section{Discussion and Bibliography}\label{sec:66}
The idea of \index{sampling}sampling a \index{target distribution}target distribution $\post$ by means
of a $\post$-invariant \index{Markov chain}Markov chain was introduced in the
statistical physics community in \cite{metropolis1953equation}, where a symmetric \index{proposal!kernel}proposal kernel was used. Hastings introduced a powerful generalization of the
method in \cite{hastings1970monte} which allowed for asymmetric \index{proposal!kernel}proposal kernels. 
The \index{Bayesian}Bayesian methodology \cite{gelman2013bayesian}, and in particular \index{MCMC}MCMC-based
exploration of the \index{posterior}posterior, became practical as a result 
of advances in computer power and became widely adopted for many \index{sampling}sampling 
problems arising in science and engineering. 

The book \cite{gamerman2006markov} is a useful
basic introduction to  \index{MCMC}MCMC and the book
\cite{brooks2011handbook} presents state of the art as of 2010.
The paper \cite{cotter2013mcmc} overviews the \index{pCN}pCN method
and related \index{MCMC}MCMC algorithms specifically designed for inverse
problems and other \index{sampling}sampling problems in high-dimensional
\index{state-space}state-spaces.
The book \cite{lindvall2002chapters} describes the \index{coupling}coupling method
in a general setting.
The book \cite{meyn2012markov} contains a wide-ranging presentation
of \index{Markov chain}Markov chains, and their long-time behavior, including \index{ergodicity}ergodicity
and \index{coupling}coupling. Furthermore, the book describes the general 
framework 
for going from convergence of expectations in (possibly weighted) 
\index{distance!total variation}total variation distances to sample path \index{ergodicity}ergodicity and almost sure convergence
of time averages, a topic we did not cover in this chapter.
The paper \cite{mattingly2002ergodicity} describes the 
\index{coupling}coupling technique 
in the context of stochastic differential
equations and their approximations. 

The tuning of parameters in
\index{MCMC}MCMC, such as the parameter $\beta$ appearing in the \index{pCN}pCN method,
is key to their success. If the goal of the \index{MCMC}MCMC sampling method is to
approximate the  expectation of a given test function $f$, then 
the aim of parameter tuning is to minimize
the integrated auto-correlation $\tau^2(f)$ defined in Theorem
\ref{t:MCMC}. In general different $f$ will lead to different optimal
\index{proposal!kernel}proposal parameter choices; however, for a wide class of high-dimensional \index{target distribution}target
distributions and specific \index{proposal!kernel}proposal kernels there are generic rules of
thumb, independent of $f$, for tuning parameters in the proposal
\cite{roberts2001optimal}. This universality often arises
from using suboptimal algorithms and, for specific problems, can be
circumvented by using tailored \index{proposal!kernel}proposals. For example, 
for target measures
that have a density with respect to a \index{Gaussian}Gaussian, the \index{pCN}pCN proposal is preferable
to the \index{random walk Metropolis proposal}random walk Metropolis proposal, as demonstrated 
in \cite{cotter2013mcmc,hairer2014spectral,trillos2017consistency}. Stochastic Newton MCMC methods to sample posterior distributions in function space Bayesian inverse problems are studied in \cite{petra2014computational}.   

 \chapter*{\Huge{\sffamily{Exercises}}}
\label{ch:IPExercises}
\addcontentsline{toc}{chapter}{Exercises}

{\bf Exercise 1} ({\sffamily{Hellinger Distance Between Gaussians}}) Recall the Hellinger distance $\dhell$ between two probability densities introduced in Definition \ref{def:hellinger}.
Consider two Gaussian densities on $\R$: $p_1 = \cN(\mu_1,\sigma_1^2)$
and $p_2 = \cN(\mu_2,\sigma_2^2)$.
Show that the squared Hellinger distance between them is given by
$$\dhell(p_1,p_2)^2 = 1-\sqrt{\exp\biggl(-\frac{(\mu_1-\mu_2)^2}{2(\sigma_1^2
+\sigma_2^2)}\biggr)\frac{2\sigma_1\sigma_2}{(\sigma_1^2+\sigma_2^2)}} \,.$$

\bigskip

\noindent {\bf Exercise 2} ({\sffamily{Kullback-Leibler Divergence Between Gaussians}}) 
Recall the Kullback-Leibler divergence $\dkl$ introduced in Definition \ref{def:KLdivergence}.
Does $\dkl$ define a metric on probability measures? Justify your answer.
Consider two Gaussian densities on $\R$: $p_1 = \cN(\mu_1,\sigma_1^2)$
and $p_2 = \cN(\mu_2,\sigma_2^2)$.
Show that the \index{divergence!Kullback-Leibler}Kullback-Leibler divergence between them is given by
$$\dkl(p_1\|p_2)=\log \Bigl(\frac{\sigma_2}{\sigma_1}\Bigr)+
\frac12\Bigl(\frac{\sigma_1^2}{\sigma_2^2}-1\Bigr)+\frac{(\mu_2-\mu_1)^2}{2\sigma_2^2}.$$
Generalize this result to Gaussians $p_1 = \Nc(\mu_1, \Sigma_1)$ and $p_2 = \Nc(\mu_2, \Sigma_2)$ in $d$ dimensions, with positive definite $\Sigma_1$ and $\Sigma_2,$ to obtain the formula in Example \ref{ex:klg}.

\bigskip

\noindent {\bf Exercise 3} ({\sffamily{Bound Between Hellinger and Kullback-Leibler}}) 
Verify the inequality
$$\dhell(p_1,p_2)^2 \le \frac12 \dkl(p_1\|p_2)$$
for the one-dimensional Gaussian examples in the two previous exercises.


\bigskip

\noindent {\bf Exercise 4} ({\sffamily{Well-posedness of Inverse Problem Under Data Perturbation}}) 
Consider the inverse problem $y = G(u) + \eta$  and noise $\eta \sim \mathcal{N}(0,\gamma^2 I_k).$ Suppose that there is $G^+$ such that
$|G(u)| \le G^+$ for any $u \in \R^d.$ 
Let $y, y'$ with $|y|, |y'| \le r$ be two instances of the data and let $\pi^y$ and $\pi^{y'}$ be the two corresponding posterior distributions with the same prior $\rho = \cN(0,\lambda^{-1}I_d)$. Show that there is $c>0$ such that 
$\dhell(\pi^y, \pi^{y'}) \le c |y-y'|.$


\bigskip

\noindent {\bf Exercise 5} ({\sffamily{Randomized Maximum Likelihood}}) 
Let $u_0^{(n)} \sim \mathcal{N}  (\widehat{m}, \widehat{C})$ i.i.d. with $\widehat{C}$ positive definite. Let  $u^{(n)}$ be the minimizer of \begin{equation}\label{eq:randomobj}
\mathsf{J}^{(n)}(u):= \frac12 |u - u_0^{(n)} |_{\widehat{C}}^2 + \frac12 |y +  \Gamma^{1/2}\xi^{(n)} - Au  |_\Gamma^2, \quad \quad \xi^{(n)} \sim \mathcal{N} (0,I) \,\,{\rm i.i.d.}
\end{equation}
Assume also that $\{u_0^{(n)}\}$ and $\{\xi^{(n)}\}$ are mutually independent
i.i.d. sequences. Show that 
$u^{(n)} \sim \mathcal{N} ( m, C)$ i.i.d., where $m$ and $C$ are defined by
\begin{align*}
m & := \widehat{m} + K(y - A\widehat{m}), \\ 
C &:= (I - K A)\widehat{C},
\end{align*}
and
\begin{equation*}
K := \widehat{C} A^\top (A \widehat{C} A^\top + \Gamma)^{-1}.
\end{equation*}

\bigskip

\noindent {\bf Exercise 6} ({\sffamily{Convergence of Gradient Descent}}) 
\begin{itemize}
\item (i) Suppose that $\J : \R^d \to \R$ has $r$-Lipschitz gradient. Show that, for any $u, v \in \R^d,$ it holds that
\begin{equation}\label{eq:conv}
    \J(v) \le \J(u) + \bigl\langle D\J(u), v - u \bigr\rangle + \frac{r}{2} | v - u |^2. 
\end{equation}
\item (ii)  Let $r>0$ be a real number. We say that $\J$ satisfies an $r$-Polyak-Lojasiewicz ($r$-PL) condition if, for all $u \in \R^d,$ it holds that $r( \J(u) - \J^\star) \le \frac12 | D \J(u)|^2.$ Suppose that $\J$ has an $r$-Lipschitz gradient, satisfies a $c$-PL condition with $0<c<r$, and achieves its infimum $\J^\star.$
Show that the gradient descent algorithm with step-size $r^{-1}$ given by 
$$u_{\ell + 1} = u_\ell - r^{-1} D\J(u_\ell), \quad \ell = 0,1,\ldots$$
has a linear convergence rate. More precisely, show that
$$\J(u_{\ell}) -  \J^\star \le \biggl( 1 - \frac{c}{r} \biggr)^\ell \bigl( \J(u_0) - \J^\star \bigr).$$
\item (iii) We say that $\J$ is $r$-strongly convex if, for all $u, v \in \R^d,$ it holds that 
$$ \J(v) \ge \J(u) + \langle \nabla \J(u) , v-u \rangle + \frac{r}{2}|v- u|^2. $$
Show that $r$-strong convexity implies an $r$-PL condition.
\end{itemize}

\bigskip

\noindent {\bf Exercise 7} ({\sffamily{Best Gaussian Approximation}}) 
Consider the inverse problem of recovering $u \in \mathbb{R}$ from data $y \in \mathbb{R}$ related by
$$y = u + 0.1 u^3 + \eta, \quad \quad \eta \sim \mathcal{N} (0, 0.4).$$
Assume a Gaussian prior $\rho(u)  = \mathcal{N}(0.5, 1),$ and that the observed data is $y = 1.1.$
Write down the posterior pdf $\pi$ of $u$ given $y$ (up to a normalizing constant) and plot it. Propose an algorithm to find the best Gaussian approximation $\dkl(p \| \pi)$ and an algorithm to find the best Gaussian approximation $\dkl( \pi \| p).$ Implement your proposed algorithms and report your results by writing the means and variances that your algorithms output, and plotting the corresponding Gaussian pdfs along with the posterior $\pi$.

      \bigskip      
\noindent {\bf Exercise 8} ({\sffamily{Inferring Correlation From Data}}) 
 In this problem you will invent an MCMC algorithm for a simple inference problem with Gaussians. Specifically, we will infer the correlation between two Gaussian random variables. Consider the model $(y,z) \sim \mathcal{N}(\mu,\Sigma)$, with:
\begin{eqnarray*}
\mu = \begin{bmatrix}
0 \\ 0 \end{bmatrix};\hspace{0.2in}  \Sigma = \Sigma(u) = \begin{bmatrix}
 1 & u \\ u & 1 \end{bmatrix}.
\end{eqnarray*}
Draw $N = 1000$ i.i.d. samples from the distribution $(y,z) \sim \mathcal{N} \bigl( \mu,\Sigma(0) \bigr) = \cN \bigl( 0, I_2 \bigr) $;  henceforth we refer to this
as the data. You will develop a Metropolis-Hastings MCMC algorithm  to find the posterior distribution of $u$, given the data;
you already know that the data was generated using $u = 0 $ which provides intuition
as you develop the algorithm.  There are multiple aspects to developing this algorithm: finding the likelihood, constructing a prior, specifying a proposal distribution, and determining the acceptance function. In the next few parts, you will be stepped through developing each ingredient. 
\begin{itemize}
\item (i) Show that the likelihood $\Prob \bigl(\{y^{(i)},z^{(i)}\}_{i = 1}^N|u \bigr)$ is given by: $$\prod_{i = 1}^N\Prob(y^{(i)},z^{(i)} | u) = \prod_{i =1}^N \frac{1}{2\pi\sqrt{1- u^2}} \exp \biggl(-\frac{1}{2(1-u^2)} \Bigl[{(y^{(i)})}^2-2u{y^{(i)}z^{(i)}} + (z^{(i)})^2 \Bigr]\biggr).$$
\item (ii) Consider Jeffreys prior $\Prob(u) = \frac{1/\pi}{|\text{det}(\Sigma(u))|^{1/2}}$. Show that this defines a probability distribution and,  for our specific choice of $\Sigma(u)$, 
has a closed form expression equal to $1/{\{\pi(1-u^2)^{1/2}\}}.$
\item Using Bayes theorem, find (up to normalization) a formula for
the posterior distribution  $\Prob \bigl(u|\{y^{(i)},z^{(i)}\}_{i = 1}^N \bigr).$ 
\item (iii) Consider the proposal distribution
$$v^\star \sim \text{Uniform} \Bigl(u^{(n)} - 0.1,u^{(n)}+0.1 \Bigr).$$
This proposal distribution is symmetric with respect to $u^{(n)}$, meaning that there is equal probability of moving in either direction of $u^{(n)}$. The Metropolis Hastings algorithms with these types of proposal distributions are often  referred to
as \emph{random walk Metropolis algorithms}.  Using this proposal distribution, find the acceptance probability function. Starting from $u^{(0)} = 0.1$ and after a burn-in\index{burn-in} time of $10^4$ samples, execute the Markov chain to generate $\bar{N} = 10^3$ samples. Keep a running sample mean and variance in the burn-in period. Plot the sample mean and variance as a function of $n$. 
 Discuss your findings.  \\
{\sffamily{Observation 1}}: The running sample mean and variance are often used as a diagnosis for the convergence of the Markov chain. \qed\\
{\sffamily{Observation 2}}: Note that an online method   
to compute the running sample mean and running sample variance is given by:
\begin{eqnarray*}
m^{(n+1)} = \frac{n~m^{(n)} + u^{(n+1)}}{n+1} ~~~ c^{(n+1)} = \frac{(n-1)~{c}^{(n)} + (u^{(n+1)}-m^{(n+1)})^2}{n}.\quad\qed 
\end{eqnarray*}
\item (iv) Repeat the previous experiment but with the step size of the proposal distribution changed from $0.1$ to $0.4$. That is, consider the proposal $v^\star \sim \text{Uniform}(u^{(n)} - 0.4,u^{(n)}+0.4).$ What do you observe about the convergence rate of the MCMC algorithm?
\end{itemize}

\bigskip

\noindent {\bf Exercise 9} ({\sffamily{Gibbs Sampling}}) 
 In this problem we consider Gibbs sampling, an MCMC algorithm for generating approximate samples from a multivariate distribution. Gibbs sampling is used when the conditional distribution of a variable conditioned on the rest is tractable (you will see an example in the next problem). In particular, for a discrete random vector $ y = (y_1,y_2, \dots, y_p)\in \R^p$,  the Gibbs sampling algorithm is given by:
 \begin{algorithm}[H]
   \caption*{{\bf{Algorithm}} Gibbs Sampling}
    \begin{algorithmic}[1]
      \STATE Initialize $y^{(0)} \sim \post_0$ (any choice of distribution $\pi_0$ is reasonable).
      For $n = 1$ to $N$ do:
      \STATE $y_{1}^{(n)} \sim \Prob \Bigl(y_1  | y_2 = y_2^{(n-1)}, y_3 = y_3^{(n-1)}, \dots, y_p = y_p^{(n-1)}\Bigr),$
      \STATE $y_{2}^{(n)} \sim \Prob \Bigl(y_2 | y_1 = y_1^{(n)}, y_3 = y_3^{(n-1)}, \dots, y_p = y_p^{(n-1)}\Bigr),$
      \STATE \vdots
      \STATE $y_p^{(n)} \sim \Prob \Bigl(y_p |y_1 = y_1^{(n)}, y_2 = y_2^{(n)}, \dots, y_{p-1} = y_{p-1}^{(n)} \Bigr).$
\end{algorithmic}
\end{algorithm}

 There is an intimate connection between this algorithm and the Metropolis Hastings algorithm as we now show, through two steps. 

\begin{itemize}
\item (i)  Firstly, consider the state $\Bigl(y_1^{(n-1)},y_2^{(n-1)},\dots,y_p^{(n-1)} \Bigr)$. Let $$v^\star  \sim \Prob \Bigl(y_1|y_{2}^{(n-1)},\dots,y_p^{(n-1)} \Bigr).$$ According to the Gibbs sampling algorithm, with probability 1, we transition to the the state $ \Bigl(v^\star,y_2^{(n-1)},\dots,y_p^{(n-1)} \Bigr).$   Show that this choice 
of  proposal kernel for $v^\star$  satisfies the detailed balance equation with respect to the joint distribution.
\item (ii) Prove that with this choice of Markov kernel, the acceptance function in the Metropolis Hastings Algorithm reduces to $1$. Hence, Gibbs sampling is indeed a special case of the Metropolis Hastings algorithm. 
 \end{itemize}
 
\bigskip 
 
\noindent  {\bf Exercise 10} ({\sffamily{The Ising Model}}) 
 Graphical models are a family of multivariate distributions which are Markov in accordance to a particular undirected graph. Each node in the graph $i \in V$ is associated to a random variable. The set of edges $E \subset {V \choose 2}$ encodes the conditional dependency relationships: a variable conditioned on its neighbours is independent of the remaining variables. \\
In this problem we focus on the setting where the collection of random variables $\{y_i\}_{i = 1}^p$ take on discrete values $\pm 1$. This is known as the Ising model and is described with the following joint distribution over the variables $y$:
$$\Prob(y) = \frac{1}{Z}~\text{exp}\biggl( \, \sum_{\{s,t\} \in E} u_{s,t}y_s{y}_t\Biggr).$$
Here $u \in \mathbb{R}^{p \times p}$ encodes the graph structure. (We set the diagonal elements of $u$ to be zero.) In particular, $u_{s,t}$ is non-zero if variables $s$ and $t$ are connected via an edge. 
\begin{itemize}
\item (i) Suppose that you were tasked with sampling from this joint distribution. One possible approach would be to use importance sampling, a technique that is based on sampling from another distribution, and reweighting the samples based on the likelihood of the original joint distribution. While this is a natural approach, it becomes intractable in the setting where the number of variables $p$ is large (say $p = 20$). Why? 
\item (ii) Show that the conditional distribution of a variable $y_r$ given the rest ($y_{V\textbackslash{r}}$) is given by:
$$ \Prob(y_r|y_{V\textbackslash{r}}) = \frac{\text{exp}(2y_r\sum_{t \in V\textbackslash{r}} u_{rt}y_t)}{\text{exp}(2y_r\sum_{t \in V\textbackslash{r}} u_{rt}y_t) + 1}.$$
{\sffamily{Observation:}} Notice that sampling from the conditional distribution is tractable. Why? This suggests that Gibbs sampling could be used to draw samples from the joint distribution.\qed 
\item (iii) We consider a specific example to showcase the use of Gibbs sampling for this problem. Consider a collection of $y \in \mathbb{R}^{5}$ discrete variables specified by the following $u^\star \in \mathbb{R}^{5 \times{5}}$:
$$ u^\star = \begin{bmatrix} 0 & 0.5 & 0 & 0.5 & 0 \\ 0.5& 0& 0.5& 0& 0.5\\ 0 &0.5& 0 & 0.5& 0 \\ 0.5 & 0 & 0.5 & 0&0 \\ 0 & 0.5 & 0 & 0& 0\end{bmatrix}.$$
 Using a Gibbs sampler with initialization $y^{(0)} = \begin{pmatrix} 1 & -1  &-1 & 1 & 1 \end{pmatrix}^\top$ and burn-in period of $5000$ samples, draw $N = 1000$ samples from the joint distribution. Report the sample mean and sample variance for each of the variables.  From your samples, compute a correlation matrix of all the 5 variables and plot an image of the correlation values. Do you see a pattern? Does this confirm the validity of the sampling technique? 
 
 \item (iv)  You will now reverse engineer $u$ from the samples you generated! You will use the Metropolis Hastings algorithm to get the posterior distribution $\Prob \bigl(u|\{y^{(i)}\}_{i = 1}^N \bigr)$. Notice that $u$ is symmetric and has zeros on the diagonal, meaning that there are $p(p-1)/2$ free parameters.  Hence we  work with a vector $\tilde{u}\in \mathbb{R}^{p(p-1)/2}$ containing all the degrees of freedom of $u$.  Recall that for Metropolis Hastings, we need to construct a prior on $\tilde{u}$ and a proposal distribution.
Since we expect the graph structure to be sparse (i.e. $\tilde{u}$ sparse), a natural prior on each element of $\tilde{u}_{i}$ is the Laplace distribution $\tilde{u}_{i} \sim \text{Laplace}(0,\lambda)$ i.i.d. Further, we use a random-walk proposal distribution:
$$v^\star \sim \tilde{u}^{(n)} + \mathcal{N}(0,\sigma^2 I ),$$
with $\sigma^2 = 0.1$. With $\lambda = 0.2$ and a burn-in time of $10000$ samples, use Metropolis Hastings to generate $\bar{N} = 1000$ samples from the posterior $\Prob \bigl(u|\{y^{(i)}\}_{i = 1}^N \bigr)$. Compute and report the sample mean and variance of these samples. Do your findings match the underlying $u^\star$?\\
{\sffamily{Observation}}: The acceptance probability function in the Metropolis Hasting algorithm often removes the normalization constant in the target probability distribution. In this scenario, this does not happen. Why? What does this say about this method for large $p$? \qed
\item (v) Suppose that the likelihood $\Prob(y_1,y_2,\dots, {y}_p | u)$ is well approximated by:
$$ \Prob(y_1,y_2,\dots, {y}_p | u) \approx\prod_{r = 1}^{p} \Prob(y_r | y_{V\textbackslash{r}}, u). $$
The expression in the right is sometimes referred to as the pseudo log-likelihood. Show that with this approximation, the MAP estimator of $u$ is given by: 
$$\arg\min_{\substack{u,~u = u^\top\\\text{diag}(u) = 0}} \sum_{i = 1}^N \sum_{r = 1}^p -\log\Big(\Prob(y_r^{(i)}|y^{(i)}_{V\textbackslash{r}}, u)\Big) + \frac{1}{\lambda}\| \tilde{u}\|_{\ell_1}.$$ 
{\sffamily{Observation}}: Under some basic regularity conditions, it can be shown that the objective used to define the MAP estimator is a convex function of $u$, and thus the optimization can be solved efficiently.\qed 
\end{itemize}

\part{Data Assimilation}

 \chapter{\Large{\sffamily{Filtering and Smoothing Problems and Well-Posedness}}}\label{lecture7}

In this chapter we introduce \index{data assimilation}data assimilation problems in which
the model of interest, and the data associated with it, have
a time-ordered nature. We distinguish between the \index{filtering}filtering problem
(on-line) in which the data is incorporated sequentially as it
comes in, and the \index{smoothing}smoothing problem (off-line) which is a specific instance 
of the \index{inverse problem}inverse problems that have been the subject of the preceding
chapters. We formulate the \index{filtering}filtering and \index{smoothing}smoothing problems in Section \ref{sec:filtering}. After that, we focus
on the \index{smoothing}smoothing problem in Section \ref{sec:smooth} and describe its interpretation as an \index{inverse problem}inverse problem. 
This interpretation will allow us to seamlessly apply to the \index{smoothing}smoothing problem the  \index{well-posed}well-posedness theory for \index{inverse problem}inverse problems developed in Chapter \ref{ch1}. Section \ref{sec73} is concerned with the on-line \index{filtering}filtering problem. We will establish \index{well-posed}well-posedness of the \index{filtering}filtering problem in \index{distance!total variation}total variation distance as a corollary of the \index{well-posed}well-posedness of the \index{smoothing}smoothing problem. We will also provide a roadmap for the \index{filtering}filtering methods that will be introduced in subsequent chapters, highlighting the settings in which they will be presented. Section \ref{sec:74} closes with extensions and bibliographical remarks.

\section{Formulation of \index{filtering}Filtering and \index{smoothing}Smoothing Problems}
\label{sec:filtering}
Consider the \index{stochastic dynamics model}{\em stochastic dynamics model} given by 
\begin{align*}
v_{j+1} &= \Psi(v_j) + \xi_j, \: j \in \Z^+, \notag \\
v_0 &\sim \Nc(m_0, C_0), \: \xi_j \sim \Nc(0, \Sigma) \: \index{i.i.d.}\text{i.i.d.} \notag, 
\end{align*}
where we assume that $v_0$ is independent of the sequence $\{ \xi_j \}$;
this is often written as $v_0 \perp \{ \xi_j \}.$ 
Now we add the \index{data model}{\em data model} given by
\begin{align*}
y_{j+1} &= h(v_{j+1}) + \eta_{j+1} , \: j \in \Z^+, \notag \\
\eta_j &\sim \Nc(0, \Gamma) \: \index{i.i.d.}\text{i.i.d.} \notag,
\end{align*}
where we assume that $\{\eta_j\} \perp v_0$ and that 
$\eta_k \perp \xi_j $ for all $j,k$. 
The following will be assumed in the remainder of these notes. 
\begin{assumption}
\label{a:fas} 
The matrices $C_0,$ $\Sigma$ and $\Gamma$ are \index{positive definite}positive definite.
Further, we have $\Psi \in C(\Ru, \Ru)$ and $h \in C(\Ru, \Ry)$. 
\end{assumption}

\noindent
We define, for a given and fixed integer $J,$
\begin{equation}
\V := \{v_0, \ldots, v_J\}, \: \Y := \{y_1, \ldots, y_J\}, 
\: \text{and} \: Y_j := \{y_1, \ldots, y_j\}.  \notag
\end{equation}
The sequence $\V$ is often termed the \index{signal}{\em signal} and the sequence
$\Y$ the {\em data}.

\begin{definition} 
The \index{smoothing}{\em smoothing problem} is to find the probability density 
$\Pi(\V):=\Prob(\V|\Y) = \Prob(\{v_0, \ldots, v_J\}|\{y_1, \ldots, y_J\})$ on $\R^{\Du(J+1)}$ 
for some fixed integer $J.$ We refer to $\Pi$ as the \index{smoothing!distribution}{\em  smoothing distribution}.
\end{definition}
\begin{definition} 
The \index{filtering}{\em filtering problem} is to find, and update sequentially in $j$, the probability densities 
$\post_j(v_j):=\Prob(v_j|Y_j) = \Prob( v_j|\{y_1, \ldots, y_j\}) $ on $\Ru$ for $j=1, \dots, J.$
We refer to $\post_j$ as the {\em \index{filtering!distribution}filtering distribution at time $j.$}
\end{definition}

The key conceptual issue to appreciate concerning the \index{filtering}filtering problem, in
comparison with the \index{smoothing}smoothing problem, is that interest is focused
on characterizing, or approximating, a sequence of probability distributions,
defined  in an iterative fashion as the data is acquired sequentially. 

\begin{remark}\label{rem:marginal}
We note the following identity:
\begin{align*}
\int \Pi(v_0, \ldots, v_J) \, dv_0dv_1 \ldots dv_{J-1} = \post_J(v_J). 
\end{align*}
This expresses the fact that the marginal of the \index{smoothing!distribution}smoothing distribution at time $J$ corresponds to the 
\index{filtering!distribution}filtering distribution at time $J$. Note also that, for $j<J,$ in general 
\begin{equation}
\int \Pi(v_0, \ldots , v_J) \, dv_0\ldots dv_{j-1}  dv_{j+1} \ldots dv_{J} \neq \post_j(v_j), \notag
\end{equation}
since the expression on the left-hand side of the equation depends on data 
$Y_J$, whereas that on the right-hand side depends only on $Y_j$, and $j < J$.
\end{remark}

\section{The \index{smoothing}Smoothing Problem}
\label{sec:smooth}

\subsection{Formulation as an \index{inverse problem}Inverse Problem}

If we define
$$\eta := \{\eta_1, \ldots, \eta_J\}$$
and
$$G(\V):=\bigl\{h(v_1), \ldots, h(v_J) \bigr\},$$
then the \index{data model}data model can be written in the form of the inverse
problem \eqref{eq:jc0}:
$$\Y=G(\V)+\eta.$$
The \index{stochastic dynamics model} stochastic dynamics model provides a \index{prior}prior probabilistic description of $\V,$ 
which may then be used to formulate a \index{Bayesian!inverse problem}Bayesian version of the
\index{inverse problem}inverse problem of finding $\V$ from $\Y.$

\subsection{Formula for pdf of the \index{smoothing}Smoothing Problem}
The \index{smoothing!distribution}smoothing distribution can be found by combining a \index{prior}prior on $v$ and a \index{likelihood}likelihood function using \index{Bayes theorem}Bayes theorem. 
The \index{prior}prior is the probability distribution on $v$ implied by the distribution of $v_0$ and the \index{stochastic dynamics model}stochastic dynamics model; the \index{likelihood}likelihood function is defined by the \index{data model}data model. We now derive the \index{prior}prior and the \index{likelihood}likelihood separately.

The \index{prior}prior distribution can be derived as follows: 
\begin{subequations}
\begin{align}
\Prob(\V) &= \Prob(v_J, v_{J-1}, \ldots , v_0) \notag \\
&= \Prob(v_J|v_{J-1}, \ldots , v_0)\Prob(v_{J-1}, \ldots , v_0) \notag \\
&= \Prob(v_J|v_{J-1})\Prob(v_{J-1}, \ldots , v_0). \notag
\end{align}
\end{subequations}
The third equality comes from the Markov, or memoryless, property which
follows from the independence of the elements of the sequence $\{\xi_j\}.$ 
By induction, we have
\begin{subequations}
\begin{align}
\Prob(\V) &=\Prob(v_0)  \prod_{j=0}^{J-1}\Prob(v_{j+1}|v_j) \notag \\
&= \frac{1}{Z_\rho}\exp\bigl(-\reg(\V)\bigr) \notag \\
&=: \pr(\V), \notag
\end{align}
\end{subequations}
where $Z_\rho>0$ is a normalizing constant and
\begin{equation}
\reg(\V) := \frac{1}{2}|v_0-m_0|^2_{C_0} + \frac{1}{2}  \sum_{j=0}^{J-1} |v_{j+1} - \Psi(v_j)|^2_{\Sigma}. \notag
\end{equation}

The \index{likelihood}likelihood function, which incorporates the measurements gathered from observing the system, depends only on the measurement model and may be derived as follows: \index{loss}
\begin{align}
\Prob(\Y|\V) &= \prod_{j=0}^{J-1} \Prob(y_{j+1}|v_0,\ldots , v_{J}) \notag \\
&= \prod_{j=0}^{J-1}\Prob(y_{j+1}|v_{j+1}) \notag \\
& \propto \exp\bigl(-\loss(\V;\Y)\bigr), \notag
\end{align}
 where the loss\index{loss} function is given by
\begin{equation}
\loss(\V;\Y) :=  \frac{1}{2} \displaystyle{\sum_{j=0}^{J-1}}|y_{j+1} - h(v_{j+1})|_{\Gamma}^2. \notag
\end{equation}
The factorization of $\Prob(\Y|\V)$ in terms of the product of the
$\Prob(y_{j+1}|v_{j+1})$ follows from the independence of the
elements of $\{\eta_j\}$ and the fact that the \index{observation}observation at time $j+1$
depends only on the \index{state}state at time $j+1.$

Using \index{Bayes theorem}Bayes Theorem \ref{t:bayes} we find the \index{smoothing!distribution}smoothing distribution by combining the \index{likelihood}likelihood and the \index{prior}prior \index{loss}
\begin{align*}
\Pi(\V) & \propto \Prob(\Y|\V)\Prob(\V)\\
&  =\frac{1}{Z}\exp\bigl(-\reg(\V)-\loss(\V;\Y)\bigr).  \notag
\end{align*}
Note that $\V \in \R^{\Du(J+1)}$ and $\Y \in \R^{\Dy J}$.

\subsection{\index{well-posed}Well-Posedness of the \index{smoothing}Smoothing Problem}
Now we study the \index{well-posed}well-posedness of the \index{smoothing}smoothing problem with respect
to perturbations in the data. To this end, we consider two \index{smoothing!distribution}smoothing
distributions corresponding to different observed data sequences $\Y,\Y'$:  \index{loss}
\begin{align*}
\Pi(\V)&:=\Prob(\V|\Y) = \frac{1}{Z}\exp\bigl(-\reg(\V)-\loss(\V;\Y)\bigr), \notag \\
\Pi'(\V)&:=\Prob(\V|\Y')= \frac{1}{Z'}\exp\bigl(-\reg(\V)-\loss(\V;\Y')\bigr). \notag
\end{align*}
We make the following assumptions:
\begin{assumption}
There is a finite non-negative constant $R$ such that the
data $\Y, \Y'$ and the \index{observation!function}observation function $h$ satisfy:
\label{a:swp}
\begin{itemize}
\item $|\Y|, |\Y'| \le R;$
\item Letting $\phi(\V) :=\Bigl( \sum_{j=1}^J (|h(v_j)|^2) \Bigr)^{1/2}$, it holds that $\Expect^{\pr}[\phi^2(\V)]<\infty$.
\end{itemize}
\end{assumption}

The following theorem shows \index{well-posed}well-posedness of the \index{smoothing}smoothing problem.

\begin{theorem}[\index{well-posed}Well-Posedness of \index{smoothing}Smoothing]
\label{t:wps}
Under Assumption \ref{a:swp}, there is $\kappa \in [0, \infty)$ independent of $\Y$ and $\Y'$  
such that $$\dhell(\Pi, \Pi') \leq \kappa|\Y-\Y'|.$$ 
\end{theorem}

\begin{proof}
We show that the proof of Theorem \ref{thm1}, which established \index{well-posed}well-posedness for \index{Bayesian!inverse problem}Bayesian \index{inverse problem}inverse problems under Assumption \ref{a:dh1}, applies in the \index{smoothing}smoothing context as well. To do so, we rewrite the problem in the same notation used in Chapter 1, and show that Assumption \ref{a:swp} above implies Assumption \ref{a:dh1}.
Write \index{loss}
\begin{align*}
\Pi(\V)&= \frac{1}{Z}\exp\bigl(-\loss(\V;\Y)\bigr) \pr(\V)= \frac{1}{Z} \like(\V;\Y) \pr(\V),\\
\Pi'(\V)& =  \frac{1}{Z'}\exp\bigl(-\loss(\V;\Y')\bigr)  \pr(\V) = \frac{1}{Z'} \like(\V;\Y') \pr(\V), \notag
\end{align*}
where $Z,Z'>0$ are normalization constants. 
Here $ |\Y-\Y'|$ plays the role of $\delta$ in Theorem \ref{thm1}. 
Since the likelihood\index{likelihood}  
$\like(\V;\Y) := \exp\bigl(-\loss(\V;\Y)\bigr)$ 
and $\loss(\V;\Y)$ is positive, we have that 
$$\sup_{v} \Bigl|\sqrt{\like(\V;\Y)}\Bigr| + \Bigl|\sqrt{\like(\V;\Y')}\Bigr|  \le 2,  $$
and so Assumption \ref{a:dh1} (ii) is satisfied. To see that Assumption \ref{a:dh1} (i) is also satisfied, note that $e^{-x}$ is \index{Lipschitz}Lipschitz-1. Therefore, using the \index{Cauchy-Schwarz inequality}Cauchy-Schwarz inequality and some algebraic manipulations, there is  $\kappa$ independent of $\Y$ and $\Y'$ such that
\begin{align*}
\Bigl|\sqrt{\like(\V,\Y)} - \sqrt{\like(\V;\Y')} \Bigr|&\le  \frac{1}{2} \Bigl| \loss(\V;\Y) - \loss(\V;\Y')  \Bigr|  \\
&=\frac{1}{2} \biggl| \sum_{j=0}^{J-1}\frac12 \Bigl( |y_{j+1} - h(v_{j+1})|_\Gamma^2 - |y'_{j+1} - h(v_{j+1})|_\Gamma^2   \Bigr)  \biggr|  \\
&\le \kappa \sum_{j=0}^{J-1} |y_{j+1} - y'_{j+1}|_\Gamma  |y_{j+1} + y'_{j+1} -2 h(v_{j+1})|_\Gamma   
   \\
   &\le \kappa |\Y-\Y'| \phi(\V), 
\end{align*}
where $\phi$ is defined in Assumption \ref{a:swp}.
This shows that under Assumption \ref{a:swp} the \index{likelihood}likelihood function of the \index{smoothing}smoothing problem satisfies Assumption \ref{a:dh1} with $\delta =|\Y-\Y'|$; Theorem \ref{t:wps} follows from Theorem \ref{thm1}. 
\end{proof}

\section{The \index{filtering}Filtering Problem}\label{sec73}

\subsection{Formula for pdf of the \index{filtering}Filtering Problem}
\index{filtering}Filtering concerns the iterative updating of distributions, as new data arrives. We recall that we denote the \index{filtering!distribution}filtering distribution at time $j$ by ${\post}_{j}=\Prob(v_{j}|Y_j),$ and we now introduce $\hat{\post}_{j+1}=\Prob(v_{j+1}|Y_j).$ Then, we decompose in two steps the updating of the \index{filtering!distribution}filtering distribution from time $j$ to time $j+1$:
\begin{align}\label{eq:pna}
\begin{split}
&{ \text{ \bf  Prediction Step:}} ~~~~\;  \hat{\post}_{j+1} = \Pred \post_j.  \index{prediction} \\
& {\text{  \bf Analysis Step:}} ~~~~~\,\,\,\,\,  \post_{j+1} = \An_j  \hat{\post}_{j+1}. \index{analysis} 
\end{split}
\end{align}

The combination of the \index{prediction}prediction and \index{analysis}analysis steps is shown 
schematically in Figure \ref{fig:interaction} and leads to the update 
\begin{equation}
\post_{j+1} = \An_j \Pred\post_j, \notag
\end{equation}
where $\Pred$ is a 
Markov map and $\An_j$ is a 
\index{likelihood}likelihood map (\index{Bayes theorem}Bayes theorem) that we define in what follows. 

\begin{figure}[h]
  \begin{center}
    \includegraphics[width=0.8\textwidth]{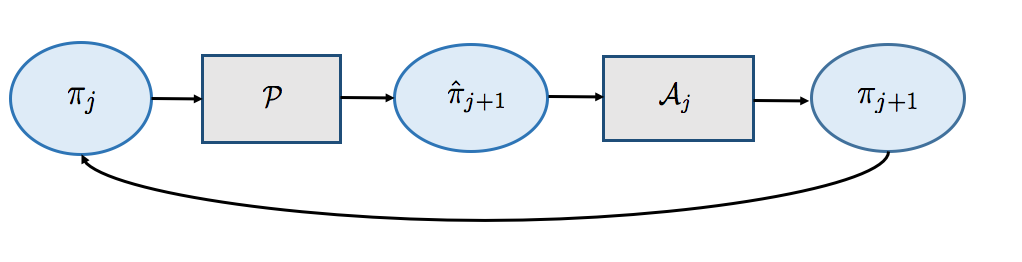}
  \end{center}
  \caption{Prediction and analysis steps combined.} \label{fig:interaction}
\end{figure}

We first derive the map $\Pred,$ which is sometimes termed {\em \index{prediction}prediction}.
By the Markov property of the
\index{stochastic dynamics model}stochastic dynamics model, we have 
\begin{align}
\hat{\post}_{j+1}(v_{j+1}) &= \Prob(v_{j+1}|Y_j)\\
&= \int_{\Ru} \Prob(v_{j+1}|Y_j, v_j)\Prob(v_j|Y_j) \, dv_j \notag \\
&= \int_{\Ru} \Prob(v_{j+1}|v_j)\Prob(v_j|Y_j) \, dv_j \notag \\
&= \int_{\Ru} \Prob(v_{j+1}|v_j)\post_j(v_j) \, dv_j \notag \\
&= \frac{1}{(2\pi)^{\Du/2} ({\rm det}\Sigma)^{1/2}}\int_{\Ru} \exp\Bigl(-\frac{1}{2}|v_{j+1}-\Psi(v_j)|^2_\Sigma \Bigr)\post_j(v_j) \,  dv_j. \notag
\end{align}
This defines the operator $\Pred$; the \index{prediction}prediction step is shown schematically in Figure \ref{fig:model}. Note that $\Pred$ is independent of step $j$ because the \index{Markov chain}Markov chain defined by the \index{stochastic dynamics model}stochastic dynamics model is time-homogeneous. In the absence of data, the
distribution of $v_j$ simply evolves through repeated application of $\Pred$.

\FloatBarrier
\begin{figure}[htp]
  \begin{center}
    \includegraphics[width=0.5\textwidth]{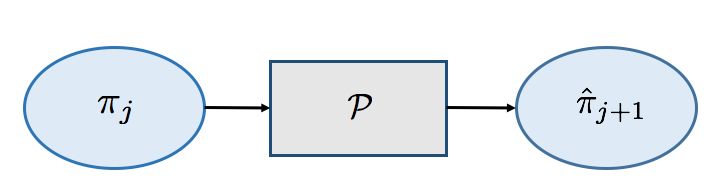}
  \end{center}
  \caption{\index{prediction}Prediction step.} \label{fig:model}
\end{figure}
\FloatBarrier

Now we derive the  \index{likelihood}likelihood map $\An_j$, which is sometimes called  {\em \index{analysis}analysis}. Note that the \index{prediction}prediction step does not make use of the new \index{observation}observation $y_{j+1}$, which is assimilated in the \index{analysis}analysis step through application of \index{Bayes theorem}Bayes theorem, as follows:
\begin{align}
\post_{j+1}(v_{j+1}) &= \Prob(v_{j+1}|Y_{j+1})\\
&= \Prob(v_{j+1}|Y_j, y_{j+1}) \notag \\ 
&= \frac{\Prob(y_{j+1}|v_{j+1}, Y_j)\Prob(v_{j+1}|Y_j)}{\Prob(y_{j+1}|Y_j)} \notag \\
&= \frac{\Prob(y_{j+1}|v_{j+1})\Prob(v_{j+1}|Y_j)}{\Prob(y_{j+1}|Y_j)} \notag \\
&= \frac{\exp(-\frac{1}{2}|y_{j+1}-h(v_{j+1})|^2_{\Gamma})\hat{\post}_{j+1}(v_{j+1})}{\int_{\Ru}  \exp(-\frac{1}{2}|y_{j+1}-h(v_{j+1})|^2_{\Gamma})\hat{\post}_{j+1}(v_{j+1}) \, dv_{j+1}}.\notag
\end{align}
This defines the  map $\An_{j}$ through multiplication
by the \index{likelihood}likelihood, and then normalization to a probability measure. The \index{analysis}analysis update is shown schematically in Figure \ref{fig:data}. 
It depends on $j$ because the data $y_{j+1}$ appears in the equation, and this will change with each set of measurements.

\begin{figure}[h]
  \begin{center}
    \includegraphics[width=0.5\textwidth]{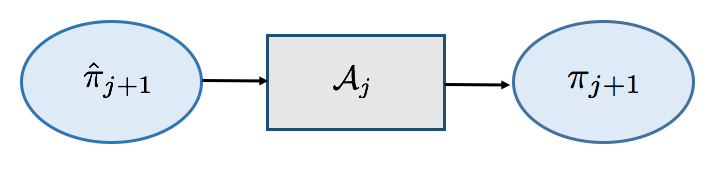}
  \end{center}
  \caption{Update step.} \label{fig:data}
\end{figure}

\subsection{\index{well-posed}Well-Posedness of the \index{filtering}Filtering Problem}
Now we establish the \index{well-posed}well-posedness of the \index{filtering}filtering problem. We let
\begin{equation*}
\post_J = \Prob(v_J|\Y), \quad \quad \post_J' = \Prob(v_J|\Y') 
\end{equation*}
be two \index{filtering!distribution}filtering distributions arising from observed data $\Y = Y_J$ and $\Y' = Y_J'.$ As noted in Remark \ref{rem:marginal}, the \index{filtering!distribution}filtering distribution at time $J$ is the $J$-th marginal of the \index{smoothing!distribution}smoothing distribution; using this observation, the \index{well-posed}well-posedness of the \index{filtering}filtering problem is a direct consequence of the \index{well-posed}well-posedness of the \index{smoothing}smoothing problem in the \index{distance!Hellinger}Hellinger distance. 
However, for the \index{filtering}filtering problem this approach only gives
\index{well-posed}well-posedness in the (weaker) \index{distance!total variation}total variation distance.
\begin{corollary}[\index{well-posed}Well-posedness of \index{filtering}Filtering]
Under Assumption \ref{a:swp},
there exists $\kappa = \kappa(R)$ such that $\dtv(\post_J, \post_J') \leq \kappa |\Y-\Y'|$. 
\end{corollary}

\begin{proof}
Let $\Pi, \Pi'$ be the \index{smoothing!distribution}smoothing distributions $\Pi = \Prob(\V|\Y)$ and $\Pi' = \Prob(\V|\Y')$. We 
note that there exists $\kappa$ such that $\dtv(\Pi,\Pi') \leq \kappa|\Y-\Y'|$
by Theorem \ref{t:wps} and by the fact that the \index{distance!Hellinger}Hellinger distance bounds
the \index{distance!total variation}total variation distance (Lemma \ref{l:dh2}). 
Let $f: \Ru \to \R$ and 
$F: \R^{\Du(J+1)} \to \R.$ Then 
\begin{align}
\dtv (\post_J, \post_{J}') &= \frac{1}{2} \sup_{|f|_\infty \leq 1 }\Bigl|\Expect^{\post_J}[f(v_J)] - \Expect^{\post_{J'}}[f(v_J)] \Bigr| \notag \\
&= \frac{1}{2} \sup_{|f|_{\infty} \leq 1} \Bigl|\Expect^\Pi[f(v_J)] - \Expect^{\Pi'}[f(v_J)] \Bigr| \notag \\
&\leq \frac{1}{2} \sup_{|F|_{\infty} \leq 1} \Bigl| \Expect^\Pi[F(\V)] - \Expect^{\Pi'}[F(\V)]  \Bigr|\notag \\
&= \dtv(\Pi, \Pi') \notag \\
&\leq \kappa|\Y-\Y'|. \notag
\end{align}
Here the first inequality follows from the fact that $\{|f|\leq 1\}$ can be viewed as a subset of $\{|F| \leq 1\}$. 
\end{proof}

\subsection{Roadmap to Discrete Filtering Methods}\label{sec:124}
There are several \index{filtering}filtering methods for performing the \index{prediction}prediction and \index{analysis}analysis steps. Some methods can be applied generally to nonlinear problems. However, others require a linear \index{dynamics}dynamics model $( \Psi(\cdot) = M\cdot)$ and/or linear \index{observation}observations $(h(\cdot) = H\cdot)$. 
Some of the methods provably approximate the \index{filtering!distribution}filtering distributions, while some just estimate the \index{state}state using covariance information to weight the relative importance
of predictions  based on the \index{dynamics}dynamics model and 
on the \index{data model}data model. 

The applicability of the methods that will be studied in the following chapters is summarized in Table \ref{roadmap}, with respect to linearity/nonlinearity of the \index{stochastic dynamics model}dynamics
and the \index{observation}observation model. Furthermore,  ${\bf P}$ is used to denote
 methods which provably approximate the \index{filtering!distribution}filtering distributions 
$\post_j$ in certain large particle limit; ${\bf S}$ denotes methods
which only attempt to estimate the \index{state}state using the data. Some of these constraints on the setting in which they apply
can be relaxed, but the list above describes
the methods as will be presented in these notes.
Furthermore, extended and ensemble  Kalman filters are  observed to 
accurately represent the \index{filtering!distribution}filtering distributions in situations where \index{Gaussian!approximation}approximate Gaussianity holds; this may be induced
by small noise  and/or  by large data.

\begin{table}
	\begin{center}
	\bgroup
\def\arraystretch{1.5}
		\begin{tabular}{ | c | c |c|c| c|}
			\hline
		\index{Kalman filter}Kalman Filter &$ \Psi(\cdot) = M\cdot$ & $h(\cdot) = H\cdot$ &${\bf P}$ & Chapter 8   \\ \hline
		3DVAR& $\text{General}\, \Psi$ & $ h(\cdot) = H\cdot $ &${\bf S}$ & Chapter 9   \\ \hline
		Extended Kalman Filter& $\text{General}\, \Psi $& $h(\cdot) = H\cdot $ &${\bf S}$ & Chapter 10    \\ \hline
			Ensemble Kalman Filter& $\text{General}\,   \Psi $ & $\, h(\cdot) = H\cdot $ &${\bf S}$ & Chapter 10   \\ \hline
			Bootstrap Particle Filter& $ \text{General}\, \Psi$ & $ \text{General}\, h$ & ${\bf P}$ & Chapter 11 \\ \hline
			Optimal Particle Filter& $\text{General}\, \Psi$ & $h(\cdot) = H\cdot$ & ${\bf P}$  & Chapter 12 \\ \hline
		\end{tabular}
\egroup
		\caption{Summary of the filtering methods considered in the following chapters, along with the settings in which they will be presented.}		
		\label{roadmap}
	\end{center}
\end{table}

\section{Discussion and Bibliography}\label{sec:74}
The book \cite{law2015data} gives a mathematical introduction to
\index{data assimilation}data assimilation; 
for further information on the \index{smoothing}smoothing problem as presented here, 
see Section 2.3 in that book;
for further information on the \index{filtering}filtering problem as presented here, 
see Section 2.4. 
 Our notes present new perspectives on data assimilation, different from those
emphasized in \cite{law2015data}, including their formulation as random dynamical
systems, detailed discussion of both the bootstrap and optimal particle filters,
and the use of data assimilation in the solution of inverse problems; 
on the other hand, the book \cite{law2015data} links the
pseudocode to downloadable code, a resource that usefully complements our notes.

The books \cite{abarbanel2013predicting,reich2015probabilistic,asch2016data,sarkka2013bayesian,bain2008fundamentals,crisan2011oxford,evensen2022data} and the review paper \cite{reich2019data}  give
alternative foundational presentations of the subject of data assimilation.  
The books \cite{kalnay2003atmospheric,oliver2008inverse,majda2012filtering,carrassi2018data}
study \index{data assimilation}data assimilation in the context of weather forecasting, oil reservoir
simulation, turbulence modeling, and geophysical sciences, respectively.

In this chapter, we have assumed throughout that $\Sigma$, the model
covariance, is  \index{positive definite}positive definite.  In applications, the \index{stochastic dynamics model}stochastic dynamics model  can be interpreted as arising from discretization of a stochastic differential equation governing the evolution of the \index{state}state. Even if the underlying \index{signal}signal is governed by a deterministic map $\Psi$, the use of a \index{stochastic dynamics model}\emph{stochastic} dynamics model can help account for errors in the modeling of this deterministic map. However, the case where
$\Sigma \equiv 0$ is also of interest as it corresponds to
 \index{dynamics!deterministic}deterministic dynamics without \index{model error}model error. In this case
we again define
$$\eta := \{\eta_1, \ldots, \eta_J\}$$
and then define
$$G_0(v_0):=\Bigl\{h\bigl(\Psi^{(1)}(v_0)\bigr), h\bigl(\Psi^{(2)}(v_0)\bigr)\ldots, h(\Psi^{(J)}v_0)\bigr)\Bigr\},$$
where $\Psi^{(j)}$ denotes $\Psi$ composed with itself $j$ times.
Then the \index{data model}data model can be written in the form of the following inverse 
problem for the determination of the initial condition of the \index{dynamical system}dynamical system:
$$\Y=G_0(v_0)+\eta.$$ 
The \index{Gaussian}Gaussian assumption $v_0 \sim \Nc(m_0,C_0)$ provides a \index{prior}prior model
for a \index{Bayesian}Bayesian formulation of this problem.  We refer to the book chapter \cite{hairer2011signal} for 
the derivation of the \index{posterior}posterior distribution in other related settings, including \index{dynamical system}dynamical systems defined by \index{ordinary differential equation}ordinary and stochastic differential equations with discrete and continuous \index{observation}observations. 

To streamline the presentation, throughout  Part II of these notes we assume to have access to maps $\Psi$ and $h,$ and covariance matrices $C_0$, $\Sigma,$ and $\Gamma$ defining the \index{dynamics}dynamics and \index{data model}data models. In practice,  however, models only reflect imperfectly the evolution of the system and the relationship between  \index{signal}signal and data. 
 For this reason, an important challenge in \index{data assimilation}data assimilation is the identification and correction of \index{model error}model errors, and the estimation of model parameters, along with the \index{state}state, from data. Several recent efforts that leverage machine learning to address \index{model error}model error in \index{dynamical system}dynamical systems are reviewed in \cite{levine2021framework}. Relatedly, several recent frameworks are emerging to blend \index{data assimilation}data assimilation with machine learning to obtain model corrections or surrogate models for the \index{dynamics}dynamics, including \cite{chen2023reduced,chen2021auto,bocquet2020bayesian,brajard2020combining,gottwald2021supervised,krishnan2017structured}.

 \chapter{\Large{\sffamily{The Linear-Gaussian Setting}}}\label{chapter:DAlineargaussian}
Recall the \index{stochastic dynamics model}stochastic dynamics and \index{data model}data models introduced in the previous chapter:
\begin{align}\label{eq:models}
\begin{split} 
        v_{j+1} &= \Psi(v_j) + \xi_j,   \quad \quad \,\, \, \, \;  \xi_j \sim \Nc(0, \Sigma) \index{i.i.d.}\text{ i.i.d.},  \\
        y_{j+1} &= h(v_{j+1}) + \eta_{j+1},  \quad \quad \, \eta_j \sim \Nc(0, \Gamma) \index{i.i.d.}\text{ i.i.d.},
\end{split}
\end{align}
with \(v_0 \sim \Nc(m_0, C_0),\) $C_0, \Sigma$ and $\Gamma$ \index{positive definite}positive definite and $v_0 \perp \{\xi_j\} \perp \{\eta_j\}.$   
Here we study the \index{filtering}filtering and \index{smoothing}smoothing problems under the assumption that both the state-transition function \(\Psi(\cdot)\) and the \index{observation!function}observation function \(h(\cdot)\) are  linear. Throughout, we will assume the following:
\begin{assumption}
The \index{stochastic dynamics model}stochastic dynamics and the \index{data model}data models defined by equation \eqref{eq:models} 
 hold with linear $\Psi(\cdot)$ and $h(\cdot)$: 
\label{a:kf} 
\begin{itemize}
	\item Linear \index{dynamics!linear}dynamics:  $v_{j+1} = Mv_j + \xi_j$ for some $\mtx{M}\in\R^{\Du \times \Du}$.
    \item Linear \index{observation}observation: $y_{j+1} = H v_{j+1} + \eta_{j+1}$ for some $H \in \R^{\Dy\times \Du}.$
\end{itemize}
\end{assumption}
We will be mostly concerned with the case where $\Du>\Dy.$
Under the \index{linear-Gaussian setting}\emph{linear-Gaussian} assumption, the \index{filtering}filtering and \index{smoothing!distribution}smoothing distributions are  \index{Gaussian}Gaussian and therefore are fully characterized by their mean and covariance.
We consider first the \index{Kalman filter}\emph{Kalman filter} in Section \ref{sec:81}, which gives explicit formulae for the iterative update of the mean and covariance of the \index{filtering!distribution}filtering distribution, and then in Section \ref{sec:82} the {\emph{\index{Kalman smoother}Kalman smoother}}, which characterizes the \index{smoothing!distribution}smoothing distribution. Section \ref{sec:83} closes this chapter with bibliographical remarks.
While the \index{Kalman filter}Kalman filter and the \index{Kalman smoother}Kalman smoother only characterize the \index{filtering!distribution}filtering and \index{smoothing!distribution}smoothing distributions in the \index{linear-Gaussian setting}linear-Gaussian setting, their importance extends beyond this setting, as will be demonstrated in the next two chapters.

\section{\index{Kalman filter}Kalman Filter}\label{sec:81}
The \index{filtering}{\em filtering problem} is to estimate the \index{state}state at time \(j\) given the data from the past up to the present time \(j\). That is, we want to determine the pdf \(\post_j = \Prob(v_j|Y_j)\), where $Y_j := \{y_1,\hdots,y_j\}.$ 
We  define $\hat{\post}_{j+1}=\Prob(v_{j+1}|Y_j)$ and recall
the evolution $$\post_{j+1} = \An_j \Pred \post_j, \quad \post_0 = \Nc(m_0, C_0),$$
which can be decomposed in terms of the \index{prediction}prediction and \index{analysis}analysis steps \eqref{eq:pna}.
Note that \(\Pred\) does not depend on $j$ because the same \index{Markov chain}Markov chain defined by the \index{state}state \index{dynamics}dynamics governs each \index{prediction}prediction step, whereas \(\An_j\) depends on $j$ because at each step $j$ the \index{likelihood}likelihood sees different data. The \index{dynamics!linear}linear dynamics assumption implies that applying the operator $\Pred$ to a \index{Gaussian}Gaussian distribution gives again a \index{Gaussian}Gaussian, and the linear \index{observation}observation assumption implies that applying the operator $\An_j$ to a \index{Gaussian}Gaussian gives again a \index{Gaussian}Gaussian. Therefore, we have the following:

\begin{theorem}[Gaussianity of \index{filtering!distribution}Filtering Distributions]
	Under Assumption \ref{a:kf},
$\post_0$, \(\{\post_{j+1}\}_{j \in \Z^+}\) and \(\{\hat{\post}_{j+1}\}_{j \in \Z^+}\) are all \index{Gaussian}Gaussian distributions.
\end{theorem}

As a consequence, the \index{filtering!distribution}filtering distributions can be entirely characterized by 
their mean and covariance. We write
\begin{align*}
	\hat{\post}_{j+1} &= \Prob(v_{j+1}|Y_j) = \Nc(\hat{m}_{j+1}, \hat{C}_{j+1}), && \text{(\index{prediction}prediction)}\\
	\post_{j+1} &= \Prob(v_{j+1}|Y_{j+1}) = \Nc(m_{j+1}, C_{j+1}), && \text{(\index{analysis}analysis)}
\end{align*}
and aim to find update formulae for these means and covariances.
The \index{Kalman filter}Kalman filter achieves this.

\begin{theorem}[Characterization of the \index{Kalman filter}Kalman Filter] 
	\label{kalmanfilter}
Suppose that Assumption \ref{a:kf} holds. Then,
for all $j \in \Z^+$, $C_j$ is \index{positive definite}positive definite and
\begin{subequations}
\label{eq:kfprecision}
\begin{align}
\hat{m}_{j+1} &= Mm_j, \\
\hat{C}_{j+1} &= MC_jM^\top+\Sigma, \\
C_{j+1}^{-1} &= (MC_jM^\top + \Sigma)^{-1} + H^\top \Gamma^{-1} H,\\
C_{j+1}^{-1}m_{j+1} &= (MC_jM^\top + \Sigma)^{-1} M m_j + H^\top \Gamma^{-1} y_{j+1}. 
\end{align}
\end{subequations}
\end{theorem}

\begin{proof}
The proof proceeds by breaking the \index{Kalman filter}Kalman filter step above into the \index{prediction}prediction and the \index{analysis}analysis steps. We first derive the update formulae, 
assuming that $C_j$ and $\hat{C}_{j+1}$ are \index{positive definite}positive definite; we conclude
with an inductive proof that this is indeed the case.

\noindent \textbf{\index{prediction}Prediction}: The mean and variance of the \index{prediction}prediction
step may be calculated as follows. The mean is given by: 
\begin{align*}
	\hat{m}_{j+1} &= \Expect \big[ v_{j+1} | Y_j \big] \\
	&= \Expect \big[ M v_j + \xi_j | Y_j \big]  \\
	&= M \Expect \big[ v_j | Y_j \big] + \Expect \big[\xi_j | Y_j \big] \\
	&= Mm_j, 
\end{align*}
where we used that \(\xi_j\) and \(Y_j\) are independent.
The covariance is given by
\begin{align*}
	\hat{C}_{j+1} =& \Expect \big[ (v_{j+1} - \hat{m}_{j+1}) \otimes (v_{j+1} - \hat{m}_{j+1}) | Y_j \big] \\
	=& \Expect \big[ M(v_j - m_j) \otimes M(v_j - m_j) | Y_j \big] + \Expect \big[ \xi_j \otimes\xi_j | Y_j \big] \\
	&+ \Expect \big[ \xi_j \otimes M(v_j - m_j) | Y_j \big] + \Expect \big[ M(v_j - m_j) \otimes\xi_j | Y_j \big] \\
	=& M \Expect \big[(v_j - m_j) \otimes (v_j - m_j) | Y_j \big] M^\top + \Sigma   \\
	=& M C_j M^\top + \Sigma,
\end{align*}
where we used that \(\xi_j\) and \(v_j\) are independent.
Thus, in the \index{linear-Gaussian setting}linear-Gaussian setting
the \index{prediction}prediction operator \(\Pred\) from \(\post_j = \Nc(m_j, C_j)\) to \(\hat{\post}_{j+1} = \Nc(\hat{m}_{j+1}, \hat{C}_{j+1})\) is given by
\begin{align*}
	\hat{m}_{j+1} &= M m_j, \\
	\hat{C}_{j+1} &= M C_j M^\top + \Sigma.
\end{align*}

\noindent \textbf{\index{analysis}Analysis}: The \index{analysis}analysis step may be derived as follows, 
using \index{Bayes theorem}Bayes Theorem \ref{t:bayes}:
\begin{align*}
	\Prob(v_{j+1} | Y_{j+1}) &= \Prob(v_{j+1} | y_{j+1},Y_{j})\\
& \propto \Prob(y_{j+1} | v_{j+1}, Y_j) \Prob(v_{j+1} | Y_j) \\
	& = \Prob(y_{j+1} | v_{j+1}) \Prob(v_{j+1} | Y_j).
\end{align*}
This gives 
\begin{align}\label{eq:smoot}
\begin{split}
\Prob(v_{j+1} | Y_{j+1}) &\propto \exp\left(-\frac{1}{2} |v_{j+1} - m_{j+1}|_{C_{j+1}}^2 \right)  \\
&\propto \exp\left(-\frac{1}{2} |y_{j+1} - H v_{j+1}|_{\Gamma}^2 \right) \exp\left(-\frac{1}{2} |v_{j+1} - \hat{m}_{j+1}|_{\hat{C}_{j+1}}^2 \right) \\
	&= \exp\left(-\frac{1}{2} |y_{j+1} - H v_{j+1}|_{\Gamma}^2 -\frac{1}{2} |v_{j+1} - \hat{m}_{j+1}|_{\hat{C}_{j+1}}^2 \right).
	\end{split}
\end{align}
Taking logarithms and
matching quadratic and linear terms in $v_{j+1}$ from either side of this
identity gives the update operator \(\An_j\) from \(\hat{\post}_{j+1} = \Nc(\hat{m}_{j+1}, \hat{C}_{j+1})\) to \(\post_{j+1} = \Nc(m_{j+1}, C_{j+1})\):
\begin{align*}
	C_{j+1}^{-1} &= \hat{C}_{j+1}^{-1} + H^\top \Gamma^{-1} H, \\
	C_{j+1}^{-1} m_{j+1} &= \hat{C}_{j+1}^{-1} \hat{m}_{j+1} + H^\top \Gamma^{-1} y_{j+1}.
\end{align*}
Combining the \index{prediction}prediction operator \(\Pred\) and update operator \(\An_j\) yields the desired update formulae. 

\noindent \textbf{Positive-definiteness}: It remains to show that \index{positive definite} \(C_j > 0\) for all \(j \in \Z^+\). We will use induction. By assumption the
result holds true for \(j = 0\). Assume that it is true for \(C_j\). For
the \index{prediction}prediction operator \(\Pred\) we have, for $u \ne 0$,
\begin{align*}
	\langle u, \hat{C}_{j+1} u \rangle &= \langle u, M C_j M^\top u \rangle + \langle u, \Sigma u \rangle \\
	&= \langle M^\top u, C_j  M^\top u \rangle + \langle u, \Sigma u \rangle \\
        & \ge \langle u, \Sigma u \rangle \\
	&> 0,
\end{align*}
where we used that $C_j>0$ and $\Sigma>0.$
Therefore \(\hat{C}_{j+1}, \hat{C}_{j+1}^{-1} > 0\). Then for the update operator \(\An_j\):
\begin{align*}
	\langle u, C_{j+1}^{-1} u \rangle &= \langle u, \hat{C}_{j+1}^{-1} u \rangle + \langle u, H^\top \Gamma^{-1} H u \rangle \\
	&= \langle u, \hat{C}_{j+1}^{-1} u \rangle + \langle H u, \Gamma^{-1} H u \rangle \\
        &\ge \langle u, \hat{C}_{j+1}^{-1} u \rangle \\
	&> 0,
\end{align*}
where we used that $\Gamma>0.$ Therefore, \(C_{j+1}, C_{j+1}^{-1}> 0\), which concludes the proof.
\end{proof}

\begin{remark}
The previous proof reveals two interesting facts about the structure of the \index{Kalman filter}Kalman filter updates. The first is that the covariance update does not involve the observed data; this can be thought of as a consequence of the fact that the \index{posterior}posterior covariance in the \index{linear-Gaussian setting}linear-Gaussian setting for \index{inverse problem}inverse problems does not depend on the observed data, as noted in Chapter 2.  The second is that the update formulae for the covariance are affine in the \index{prediction}prediction step, but nonlinear in the \index{analysis}analysis step; specifically, the \index{analysis}analysis step is affine in the precisions (inverse covariances).  
\end{remark}

\subsection{\index{Kalman filter}Kalman Filter:  Algorithmic Implementation}
We now rewrite the \index{Kalman filter}Kalman filter in an alternative form, which can be advantageous for algorithmic implementation. This formulation is summarized in Algorithm \ref{algKF} below, and it 
is written in terms of covariance matrices instead of precision matrices. 

\FloatBarrier
\begin{algorithm}
\caption{\label{algKF} \index{Kalman filter}Kalman Filter Algorithm}
\begin{algorithmic}[1]
\vspace{0.1in}
\STATE {\bf Input}: Initial distribution $\post_0 = \Nc(m_0,C_0)$ with $m_0 \in \R^{\Du},$ $C_0  \in \R^{\Du \times \Du}.$ \\
\vspace{.04in}
\STATE For $j = 0, 1, \ldots, J-1$ do the following \index{prediction}prediction and \index{analysis}analysis steps:
\STATE {\bf \index{prediction}Prediction}: 
\begin{align}
\hat{m}_{j+1} &= Mm_j, \\
\hat{C}_{j+1} &= MC_jM^\top+\Sigma, 
\end{align}
\vspace{.04in}
\STATE{{{\bf \index{analysis}Analysis}}}: 
	\begin{align}\label{eq:kf update}
	\begin{split}
		m_{j+1} &= \hat{m}_{j+1} + K_{j+1} d_{j+1}, \\
		C_{j+1} &= (I - K_{j+1} H) \hat{C}_{j+1},
		\end{split}
	\end{align}
where
	\begin{align}\label{eq:kf analysis}
	\begin{split}
		d_{j+1} &= y_{j+1} - H \hat{m}_{j+1}, \\
		S_{j+1} &= H \hat{C}_{j+1} H^\top + \Gamma, \\
		K_{j+1} &= \hat{C}_{j+1} H^\top S_{j+1}^{-1}.
		\end{split}
	\end{align}

\STATE{\bf Output}: Predicted distributions $\hat{\post}_{j+1} = \Nc(\hat{m}_{j+1}, \hat{C}_{j+1})$ and \index{filtering!distribution}filtering distributions $\post_{j+1} = \Nc(m_{j+1}, C_{j+1}),$ $j = 0,1, \ldots, J-1.$ 
\end{algorithmic}
\end{algorithm}
\FloatBarrier

Importantly, the formulation in Theorem \ref{kalmanfilter} involves a matrix inversion in the \index{state-space}state-space $\R^{\Du}$ while the one given in Algorithm \ref{algKF} requires 
only inversion in the data space $\R^{\Dy}$ to compute \(S_{j+1}^{-1}\). 
In many applications the \index{observation}observation space dimension is much smaller 
than the \index{state-space}state-space dimension (\(\Dy \ll \Du\)), and then the
formulation given in  Algorithm \ref{algKF} leads to 
much cheaper computations
than the one given in Theorem \ref{kalmanfilter}.

\as I notice that the Algorithm numbering is outside the
sequential numbering system for Theorems/Lemma etc. I propose we change
this and make it part of the sequential system for all environments. \nc


The vector \(d_{j+1}\) is known as the {\em \index{innovation}innovation} 
and the matrix \(K_{j+1}\) as the {\em \index{Kalman gain}Kalman gain}. Note that
\(d_{j+1}\) measures the mismatch of the predicted \index{state}state from the given data.

Combining the form of $d_{j+1}$ and $\hat{m}_{j+1}$ shows that the update
formula for the Kalman mean can be written as
\begin{equation}
\label{eq:kmu}
m_{j+1} = (I-K_{j+1}H)\hat{m}_{j+1} + K_{j+1} y_{j+1}, \quad \hat{m}_{j+1}=Mm_j.
\end{equation}
This update formula has the very natural interpretation that the mean update
is formed as a linear combination of the evolution of the noise-free \index{dynamics}dynamics
and of the data. Equations \eqref{eq:kf update} and \eqref{eq:kmu} show that the \index{Kalman gain}Kalman gain $K_{j+1}$ determines the weight given to the new \index{observation}observation $y_{j+1}$ in the \index{state}state estimation.
The update formula \eqref{eq:kmu} may also be derived from an \index{optimization}optimization perspective,
the topic of the next subsection.

The fact that the \index{analysis}analysis update given by Algorithm \ref{algKF} agrees with the one derived in Theorem \ref{kalmanfilter}  can be established using the following lemma:

\begin{lemma}[\index{Woodbury matrix identity}Woodbury Matrix Identity] \label{eq:woodbury} Let \(A \in \R^{p \times p}\), \(U \in \R^{p \times q}\), \(B \in \R^{q \times q}\), \(V \in \R^{q \times p}\). If \(A, B > 0\), then \(A + UBV\) is invertible and 
	\[(A + UBV)^{-1} = A^{-1} - A^{-1}U (B^{-1} + VA^{-1}U)^{-1} V A^{-1}.\] 
\end{lemma}

Now, to see the agreement between the characterization in terms of precision matrices in Theorem \ref{kalmanfilter} and the covariance characterization in \eqref{eq:kf update} and \eqref{eq:kf analysis}, note that Lemma \ref{eq:woodbury}  applied to
(\ref{eq:kfprecision}c)  gives
\begin{align*}
C_{j+1} &=  \hat{C}_{j+1}  - \hat{C}_{j+1}  H^\top (\Gamma + H \hat{C}_{j+1}  H^\top )^{-1} H \hat{C}_{j+1} \\
& = \Bigl( I - \hat{C}_{j+1} H^\top (\Gamma + H \hat{C}_{j+1}H^\top)^{-1} H \Bigr) \hat{C}_{j+1} \\
& = (I - \hat{C}_{j+1} H^\top S_{j+1}^{-1} H ) \hat{C}_{j+1} \\
& = (I - K_{j+1} H) \hat{C}_{j+1},
\end{align*}
as desired.

\subsection{Optimization Perspective: Mean of  \index{Kalman filter}Kalman Filter}
\label{sec:okf}
Since \(\post_{j+1}\) is \index{Gaussian}Gaussian, its mean agrees with its mode. Thus, formulae \eqref{eq:smoot} implies that
\begin{align*}
	m_{j+1} &=  {\rm argmax}_\varg \post_{j+1}(\varg)  \\
	& = {\rm argmin}_\varg \J(\varg), 
\end{align*}
where 
$$\J(\varg):= \frac{1}{2} |y_{j+1} - H \varg|_{\Gamma}^2 + \frac{1}{2} |\varg - \hat{m}_{j+1}|_{\hat{C}_{j+1}}^2.$$
In other words, \(m_{j+1}\) is chosen to fit both the observed data 
\(y_{j+1}\) and the predictions \(\hat{m}_{j+1}\) as well as possible. 
The covariances \(\Gamma\) and \(\hat{C}_{j+1}\)  determine
the relative weighting between the two quadratic terms.
The solution of the minimization problem is given by \eqref{eq:kmu}, as may be verified by direct differentiation of $\J.$

An alternative derivation which is helpful in more sophisticated contexts is to cast the problem in terms of constrained minimization. Write $v'=\varg-\hat{m}_{j+1}$, $y'=y_{j+1} - H \hat{m}_{j+1}$ and
$C'=\hat{C}_{j+1}.$ Then
 minimization of $\J$ may be reformulated as
\begin{align*}
        m_{j+1} = \hat{m}_{j+1}+{\rm argmin}_{v'} \left( \frac{1}{2} |y' - H v'|_{\Gamma}^2 + \frac{1}{2} \langle v', b \rangle \right), 
\end{align*}
where the minimization is now subject to the constraint $C'b=v'.$ Using
Lagrange multipliers we write
\begin{equation}
\label{eq:LM}
\I(v')= \frac{1}{2} |y' - H v'|_{\Gamma}^2 + \frac{1}{2} \langle v', b \rangle
+\langle \lambda, C'b-v' \rangle;
\end{equation}
computing the derivative and setting to zero gives
\begin{subequations}
\label{eq:LM_EL}
\begin{align*}
-H^\top\Gamma^{-1}(y'-Hv')+\frac12 b-\lambda & =0,\\
\frac12 v'+C'\lambda &=0,\\
v'-C'b &=0.
\end{align*}
\end{subequations}
The last two equations imply that $C'(2\lambda+b)=0.$ Thus we set 
$\lambda=-\frac{1}{2}b$ and drop the second equation, replacing the first by
$$-H^\top\Gamma^{-1}(y'-HC'b)+b=0.$$ 
Solving for $b$ gives
\begin{align*}
\varg&=\hat{m}_{j+1}+v'\\
&=\hat{m}_{j+1}+C'b\\
&=\hat{m}_{j+1}+C'(H^\top\Gamma^{-1}HC'+I)^{-1}H^\top\Gamma^{-1}y'\\
&=\hat{m}_{j+1}+C'(H^\top\Gamma^{-1}HC'+I)^{-1}H^\top\Gamma^{-1}(y_{j+1} - H \hat{m}_{j+1})\\
&=(I-K_{j+1}H)\hat{m}_{j+1}+K_{j+1}y_{j+1},
\end{align*}
where we have defined
$$K_{j+1}=C'(H^\top\Gamma^{-1}HC'+I)^{-1}H^\top\Gamma^{-1}.$$
It remains to show that $K_{j+1}$ agrees with the definition given
in \eqref{eq:kf analysis}. To see this we note that if we choose $S$ to
be any matrix satisfying 
$K_{j+1}=C'H^\top S^{-1},$ then
$$H^\top S^{-1}=(H^\top\Gamma^{-1}HC'+I)^{-1}H^\top\Gamma^{-1}$$
so that
$$(H^\top\Gamma^{-1}HC'+I)H^\top=H^\top\Gamma^{-1} S.$$
Thus 
$$H^\top\Gamma^{-1}HC'H^\top+H^\top=H^\top\Gamma^{-1}S$$
which may be achieved by choosing any $S$ so that
$$\Gamma^{-1}(HC'H^\top+\Gamma)=\Gamma^{-1}S$$
and multiplication by $\Gamma$ gives the desired formula for $S$.

\subsection{Optimality of \index{Kalman filter}Kalman Filter}

The following theorem states that the \index{Kalman filter}Kalman filter gives the best 
estimator of the mean in an online setting. 
In the following, $\Expect$ denotes expectation with respect to
all randomness present in the problem statement, through the initial
condition, the noisy dynamical evolution, and the noisy data. Furthermore,
$\Expect [\cdot|Y_j]$ denotes conditional expectation, given
the data $Y_j$ up to time $j$.

\begin{theorem}[Optimality of \index{Kalman filter}Kalman Filter]
	Let \(\{m_j\}\) be the sequence computed using the \index{Kalman filter}Kalman filter, and \(\{z_j\}\) be any sequence in \(\Ru\) such that $z_j$ is $Y_j$
measurable.\footnote{For practical purposes, this means $z_j$ is a fixed
non-random function of given observed $Y_j.$} Then, for all $j \in \N,$
	\[ \Expect \Big[ |v_j - m_j |^2 \mid Y_j \Big] \leq \Expect \Big[ |v_j - z_j |^2 \mid Y_j \Big].\]
\end{theorem}

\begin{proof}
Note that $m_j$ and $z_j$ are fixed and non-random, given $Y_j$. Thus, we have:
\begin{align*}
	\Expect \Big[ |v_j - z_j |^2 \mid Y_j \Big] =& \Expect \Big[ |v_j - m_j + m_j - z_j |^2 \mid Y_j \Big] \\
	=& \Expect \Big[ |v_j - m_j|^2 + 2 \Big\langle v_j - m_j, m_j - z_j 
\Big \rangle + |m_j - z_j|^2 \mid Y_j \Big] \\
 =& \Expect \Big[ |v_j - m_j|^2 \mid Y_j \Big] + 2 \Big \langle \Expect \Big[ v_j - m_j \mid Y_j \Big] , m_j - z_j \Big \rangle + |m_j - z_j|^2 \\
	=& \Expect \Big[ |v_j - m_j|^2 \mid Y_j \Big] + 2 \Big \langle \Expect \Big[ v_j \mid Y_j \Big]-m_j , m_j - z_j \Big \rangle + |m_j - z_j|^2 \\
	=& \Expect \Big[ |v_j - m_j|^2 \mid Y_j \Big] + 0 + |m_j - z_j|^2 \\
	\geq & \Expect \Big[ |v_j - m_j|^2 \mid Y_j \Big].
\end{align*}
The fifth step follows since \(m_j = \Expect \big[ v_j \mid Y_j \big].\) 
\end{proof}

\section{\index{Kalman smoother}Kalman Smoother}\label{sec:82}
We next discuss the \index{Kalman smoother}Kalman smoother, which refers to the \index{smoothing}smoothing problem in the \index{linear-Gaussian setting}linear-Gaussian setting of Assumption \ref{a:kf}. As with the \index{Kalman filter}Kalman filter, it is possible
to solve the problem explicitly because the \index{smoothing!distribution}smoothing distribution is itself a \index{Gaussian}Gaussian. The explicit formulae computed help to build
intuition about the \index{smoothing!distribution}smoothing distribution more generally. We recall Remark \ref{rem:marginal}, which implies that the \index{filtering!distribution}filtering distribution at time $j=J$ determines the marginal of the \index{Kalman smoother}Kalman smoother on its last coordinate. However, the \index{filtering!distribution}filtering distributions do not determine the
\index{Kalman smoother}Kalman smoother in its entirety.

\subsection{Defining Linear System}
Let $\V= \{v_0,\ldots,v_J\}$ and $\Y =\{y_1,\ldots, y_J\}$. Using Bayes Theorem \ref{t:bayes}
and the fact that $\{\xi_j\}$, $\{\eta_j\}$ are mutually independent \index{i.i.d.}
i.i.d. sequences, independent of $v_0$, we have
\begin{equation*}
\Prob (\V|\Y)\propto \Prob(\Y|\V) \Prob(\V)= \prod\limits_{j=1}^J\Prob(y_j|v_j)\times \prod\limits_{j=0}^{J-1} \Prob (v_{j+1}|v_j)\times \Prob (v_0).
\end{equation*}
Noting that
$$v_{j+1}|v_j \sim \Nc(Mv_j,\Sigma), \quad y_j|v_j \sim \Nc(Hv_j,\Gamma)$$
the \index{smoothing!distribution}smoothing distribution can be expressed as
\begin{equation}
\Prob (\V|\Y) \propto \exp\bigl(-\J(\V)\bigr),
\end{equation}
where 
\begin{equation}\label{eq:jformula}
\J(\V):=\frac{1}{2}|v_0-m_0|^{2}_{C_0}+\frac{1}{2} \sum\limits_{j=0}^{J-1}|v_{j+1}-Mv_j|^{2}_{\Sigma}+\frac{1}{2}\sum\limits_{j=0}^{J-1}|y_{j+1}-Hv_{j+1}|^{2}_{\Gamma}.
\end{equation}

\begin{theorem}[Characterization of the \index{Kalman smoother}Kalman Smoother]
Suppose that Assumption \ref{a:kf} holds. Then $\Prob(\V | \Y)$ is \index{Gaussian}Gaussian with a block tridiagonal precision matrix $\precision>0$ and  mean $m$ solving $\precision m=r$, where 
\begin{align}
\precision  =& \begin{bmatrix}
 \precision_{0,0} & \precision_{0,1} &  &  &  &  \\ 
 \precision _{1,0} & \precision_{1,1} & ... &  & 0  &  \\ 
 0 & ... & ... &  &  &  \\ 
  & 0 & ... & ... & ... &  \\ 
  &  &  & ... & \precision_{J-1,J-1} & \precision_{J-1,J} \\ 
  &  &  &  & \precision_{J,J-1} & \precision_{J,J}
 \end{bmatrix}
\end{align} 
with 
\begin{align*}
\precision_{0,0}&=C_0^{-1} +M^\top\Sigma^{-1}M,\\
\precision_{j,j}&=\Sigma^{-1}+M^\top\Sigma^{-1}M+H^\top\Gamma^{-1}H, \quad  1\leq j \leq J-1,\\
\precision_{J,J}&=\Sigma^{-1}+H^\top\Gamma^{-1}H,\\
 \precision_{j,j+1}&=-\Sigma^{-1}M, \quad   0\leq j\leq J-1, \\
r_0&=C_0^{-1}m_0,\\
r_j&=H^\top\Gamma^{-1}y_j, \quad  1\leq j\leq J.
\end{align*}
\end{theorem}

\begin{proof}
We may write $\J(\V)=\frac{1}{2}\vert \precision^{1/2}(\V-m)\vert^2 +q$ with $q$ independent of $\V$, by definition. Note that  $\precision$ is then the Hessian of $\J(\V)$, and differentiating in equation \eqref{eq:jformula} we obtain that 
\begin{align*}
\precision_{0,0}&= \partial_{v_0}^2\J(\V)=C_0^{-1} + M^\top\Sigma^{-1}M,\\
\precision_{j,j}&= \partial_{v_j}^2\J(\V)=\Sigma^{-1} + M^\top\Sigma^{-1}M+H^\top\Gamma^{-1}H,\\
\precision_{J,J}&= \partial_{v_J}^2\J(\V)=\Sigma^{-1} +H^\top\Gamma^{-1}H,\\
 \precision_{j-1,j} &=\partial_{v_{j-1},v_{j}}^2\J(\V)=-\Sigma^{-1}M. 
\end{align*}
Otherwise, for all other values of indices $\{k,l\}$, $\precision_{k,l}=0.$ This proves that the matrix $\precision$ has a block tridiagonal structure.

Now we focus on finding $m$. We have that $\nabla_\V\J(\V)=\precision(\V-m)$, so that $-\nabla_\V\J(\V)\vert_{\V=0}=\precision m$. Thus, we find $r$ as
\begin{align*}
r_0=&-\nabla_{v_0}\J(\V)\vert_{\V=0}=-(-C_0^{-1}m_0)=C_0^{-1}m_0,\\
r_j=&-\nabla_{v_j}\J(\V)\vert_{\V=0}=-(-H^\top\Gamma^{-1}y_j)=H^\top\Gamma^{-1}y_j.
\end{align*}
We have shown that $\precision$ is symmetric and that  $\precision\ge 0$; to prove that $\precision$ is a 
precision matrix, we need to show that $\precision>0.$ Take, for the sake of argument, $\Y = 0$ and $m_0=0$ in equation \eqref{eq:jformula}, so that every term in the expansion of $\J(\V)$ involves $\V$. It is evident that in such case $\J(\V) = \V^\top\precision \V$. Suppose that $\V^\top\precision \V = 0$ for some nonzero $\V$. Then by \index{positive definite} positive-definiteness of $C_0, \Sigma$, and $\Gamma$, it must be that $v_0 = 0$ and $v_{j+1} = M v_j$ for $j = 0,1,\dots,J$. Thus, we must have $\V = 0$. This proves that $\precision$ is \index{positive definite}positive definite.
\end{proof}
\begin{remark}
Since the \index{smoothing!distribution}smoothing distribution in the \index{linear-Gaussian setting}linear-Gaussian setting is itself \index{Gaussian}Gaussian, its mean agrees with its mode. Therefore, the \index{posterior!mean estimator}posterior mean found above is the unique minimizer of $\J(\V),$ that is, the \index{MAP estimator}MAP estimator.  
\end{remark} 
\subsection{\index{Kalman smoother}Kalman Smoother: Solution of the Linear System}
The mean of the \index{Kalman smoother}Kalman smoother may be obtained by \index{Gaussian!elimination}Gaussian elimination, as summarized in the following algorithm.

\FloatBarrier
\begin{algorithm}
\caption{\label{algKS} \index{Kalman smoother}Kalman Smoother by \index{Gaussian!elimination}Gaussian Elimination}
\begin{algorithmic}[1]
\vspace{0.1in}
\STATE {\bf Input}: Initial distribution $\post_0 = \Nc(m_0,C_0)$ with $m_0 \in \R^{\Du},$ $C_0  \in \R^{\Du \times \Du}.$ \\
\vspace{.04in}
\STATE {\bf Row reduction}: 
Define a matrix sequence $\{\precision_j\}$:
\begin{align}
\precision_0 &=\precision_{0,0}, \nonumber\\
\precision_{j+1} &=\precision_{j+1,j+1}-  M^\top \Sigma^{-1} \precision_j^{-1} \Sigma^{-1} M,  \quad j=0,\dots, J-1;
\label{eqn:Ljm}
\end{align}
and vector sequence $\{z_j\}$: 
\begin{align*}
z_0&= C_0^{-1}m_0, \\
z_{j+1}&= H^\top\Gamma^{-1}y_{j+1}-  M^\top \Sigma^{-1}  \precision_j^{-1}z_j.
\end{align*}

\vspace{.04in}
\STATE{{{\bf Back-substitution}}}: 
Read off $m_J$ by solving the equation $\precision_Jm_J=z_J.$ 
Perform back-substitution to obtain  
$$\precision_jm_j=z_j-\precision_{j,j+1}m_{j+1}, \quad j=J-1,\ldots,1.$$
\STATE{\bf Output}: Mean $m = \{ m_j \}_{j = 0}^J$ of the \index{Kalman smoother}Kalman smoother.
\end{algorithmic}
\end{algorithm}
\FloatBarrier

Note that $m_J$ found this way coincides with the mean of the Kalman
filter at $j=J.$ The rest of this chapter is devoted to proving the following proposition: 
\begin{proposition}
The matrices $\{\precision_j\}$ in \eqref{eqn:Ljm} are \index{positive definite}positive definite. 
\end{proposition}
\begin{proof}
The proof of this theorem relies on the following two lemmas:
\begin{lemma}\label{lemm1}
If $$X:=\begin{bmatrix}
X_1 &  \times& \times & \times \\ 
\times & X_2 & \times & \times \\ 
\times & \times & ... &\times  \\ 
\times & \times & \times & X_\Du
\end{bmatrix}$$ is \index{positive definite}positive definite, then $X_\du$ is
\index{positive definite}positive definite for all $i \in \{1,\dots, \Du\}.$
\end{lemma}
\begin{lemma} \label{lemm2}
Let $B$ be a block  lower (or upper) triangular matrix with identity on the 
diagonal. Then $B$ is an invertible matrix.
\end{lemma}
Using Lemma \ref{lemm1}, we deduce that $\precision_0=\precision_{0,0}$ is \index{positive definite}positive definite. Consider the matrix $B\in\R^{\Du(J+1)\times \Du(J+1)}$ defined as
$$B= \begin{bmatrix}
 I & 0 & & 0 \\ 
 -\precision_{1,0}\precision_0^{-1} & I & ... & ... \\ 
 0 &  & ... & 0 \\ 
  &  & 0 & I
 \end{bmatrix}.$$
We compute 
$$B\precision B^\top= \begin{bmatrix}
\precision_0 & 0 & & &0 \\
0 & \precision_1 & \precision_{1,2} & ... & 0\\
& \precision_{2,1} & \precision_{2,2} & ... &\\
& & & & \\
0& & & \precision_{J-1,J-1} &\precision_{J-1,J} \\
0& & & \precision_{J,J-1}& \precision_{J,J}\\ \end{bmatrix}.$$
By Lemma \ref{lemm1}, the matrix
 $$\tilde{\precision}=\begin{bmatrix}
\precision_1 & \precision_{1,2} & ... & 0\\
 \precision_{2,1} & \precision_{2,2} & ... &\\
 & & & \\
 & & \precision_{J-1,J-1} &\precision_{J-1,J} \\
 & & \precision_{J,J-1}& \precision_{J,J}\\ \end{bmatrix}$$ is \index{positive definite}positive definite, 
and so is $\precision_0.$

Lemma \ref{lemm1} and the \index{positive definite}positive-definiteness of $\tilde{\precision}$ imply that $\precision_1$ is \index{positive definite}positive definite. Therefore, by Lemma \ref{lemm2} the matrix 
 $$B_2=\begin{bmatrix}
I& 0 & & 0\\
-\precision_{2,1}\precision_1^{-1}& I & ... & ...\\
& & ... &0\\
0& & & I\\
\end{bmatrix}$$
is invertible.
Thus, we have $$B_2 \tilde{\precision}B_2^\top = \begin{bmatrix}
\precision_1 & 0 & & & 0\\
0 & \precision_2 & \precision_{2,3} & ... & \\
& \precision_{3,2} & \precision_{3,3} & & 0\\
& & & & \\
& & & \precision_{J-1,J-1} &  \precision_{J-1,J} \\
0& & & \precision_{J,J-1}& \precision_{J,J} \\
\end{bmatrix}, 
$$
giving the \index{positive definite}positive-definiteness of $\precision_2$. Iterating the argument shows that all the $\precision_j$ are  \index{positive definite}positive definite.
\end{proof}

\section{Discussion and Bibliography}\label{sec:83}
The original paper of Kalman \cite{kalman1960new}, in which the
Kalman filter is derived, is arguably the 
first systematic presentation of a methodology to combine 
predictive models with data; it is noteworthy that Kalman did not
employ the Bayesian perspective to derive the filter which bears his
name, but rather invoked a minimum variance hypothesis. 
The continuous time analogue of the Kalman filter, which goes by the
name Kalman-Bucy filter and applies to stochastic differential
equations, may be found in \cite{kalman1961new}. 
We refer to \cite{rauch1965maximum,gelb1974applied,anderson1979optimal,law2015data,reich2015probabilistic,asch2016data,sarkka2013bayesian} for further background on the \index{linear-Gaussian setting}linear-Gaussian setting and for alternative derivations and expressions of the Kalman update formulae.

\index{Kalman filter}Kalman filters and smoothers are the cornerstones of numerous \index{data assimilation}data assimilation algorithms for \index{filtering}filtering and \index{smoothing}smoothing, some of which will be studied in the next two chapters. 
The book \cite{harvey1990forecasting} overviews the subject in the context of time-series analysis and economics.
 The optimality of the \index{Kalman filter}Kalman filter is described in \cite{anderson1979optimal}. 
The paper \cite{sanz2015long} contains an application of the
optimality property of the \index{Kalman filter}Kalman filter (which applies beyond
the \index{linear-Gaussian setting}linear-Gaussian setting to the mean of the \index{filtering!distribution}filtering distribution
in quite general settings).  A link between the standard implementation of the \index{Kalman smoother}Kalman smoother
and \index{Gauss-Newton}Gauss-Newton methods for \index{MAP estimator}MAP estimation is made in
\cite{bell1994iterated}. For further details on the
\index{Kalman smoother}Kalman smoother, in both discrete and continuous time,
see \cite{law2015data} and \cite{hairer2005analysis}. We refer to \cite{krishnan2017structured} for  a machine learning approach to learn linear (and nonlinear) \index{dynamics}dynamics and \index{data model}data models using deep learning.

 \chapter{\Large{\sffamily{Optimization for
Filtering and Smoothing: 3DVAR and 4DVAR}}}\label{ch:optfilteringsmoothing}

This chapter demonstrates the use of \index{optimization}optimization,
namely the \index{3DVAR}3DVAR and \index{4DVAR}4DVAR methodologies, to obtain information 
from the \index{filtering!distribution}filtering and \index{smoothing!distribution}smoothing distributions. 
We emphasize that the methods we present in this chapter do not provide
approximations of the \index{filtering!distribution}filtering and \index{smoothing!distribution}smoothing distributions;
they simply provide estimates of the \index{signal}signal, given data, in the
\index{filtering}filtering (on-line) and \index{smoothing}smoothing (off-line) data scenarios. 
Their relationship to the \index{filtering!distribution}filtering and \index{smoothing!distribution}smoothing distributions is analogous to the relationship of \index{MAP estimator}MAP estimation to 
the full \index{Bayesian}Bayesian \index{posterior}posterior distribution. 
In the previous chapter we showed how the mean of the Kalman
filter could be derived through an \index{optimization}optimization principle,
once the predictive covariance is known; this idea is generalized
to nonlinear \index{forward!model}forward models to obtain \index{3DVAR}3DVAR.
On the other hand, \index{4DVAR}4DVAR is defined directly as a \index{MAP estimator}MAP estimator.

Here ``VAR'' refers to \index{variational}variational, and encodes the concept of
\index{optimization}optimization. The 3D and 4D, respectively, refer to three Euclidean
spatial dimensions and to three Euclidean
spatial dimensions plus a time dimension; this nomenclature reflects
the historical derivation of these problems in the geophysical sciences,
but the specific structure of fields over three-dimensional Euclidean
space plays no role in the generalized form of the methods described
here. The key distinction is that \index{3DVAR}3DVAR solves a sequence of \index{optimization}optimization
problems at each point in time (hence is an on-line \index{filtering}filtering
method); in contrast, \index{4DVAR}4DVAR solves an \index{optimization}optimization
problem which involves data distributed over time (and is an off-line \index{smoothing}smoothing method).

This chapter is organized as follows. We introduce the problem setting in Section \ref{sec:91}. \index{3DVAR}3DVAR and \index{4DVAR}4DVAR are considered, in turn, in Sections \ref{sec:92} and \ref{sec:93}. Section \ref{sec:94} closes with extensions and bibliographical remarks. 

\section{The Setting}\label{sec:91}
\index{3DVAR}3DVAR borrows from the \index{Kalman filter}Kalman filter \index{optimization}optimization principle outlined
in Subsection \ref{sec:okf},
but substitutes a fixed given covariance for the predictive covariance.
Throughout we consider the setting, commonly occurring in applications, in which the \index{dynamics}dynamics model
is nonlinear, but the \index{observation!function}observation function is linear. We thus have
a discrete-time \index{dynamical system}dynamical system with noisy \index{state}state 
transitions and noisy \index{observation}observations given by
\begin{align*}
        \text{\index{stochastic dynamics model}Stochastic Dynamics Model:}  \quad v_{j+1} &= \Psi(v_j) + \xi_j, \quad j \in \Z^+. \\
        \text{\index{data model}Data Model:}  \quad y_{j+1} &= Hv_{j+1} + \eta_{j+1}, \quad j \in \Z^+, \text{ for some} \,\, H \in \R^{\Dy \times \Du}.  \\
        \text{Probabilistic Structure:}  \quad v_0 &\sim \Nc(m_0, C_0), \quad \xi_j \sim \Nc(0, \Sigma), \quad \eta_j \sim \Nc(0, \Gamma).\\
 \text{Probabilistic Structure:} \quad v_0 &\perp \{\xi_j\} \perp \{\eta_j\} \text{ independent.}
\end{align*}

\section{3DVAR}\label{sec:92}
We introduce \index{3DVAR}3DVAR by analogy with the update formula \eqref{eq:kmu}
for the \index{Kalman filter}Kalman filter, and its derivation through \index{optimization}optimization from
Subsection  \ref{sec:okf}. The primary differences between \index{3DVAR}3DVAR and the Kalman
filter mean update are that $\Psi(\cdot)$ can be nonlinear for \index{3DVAR}3DVAR, and
that for \index{3DVAR}3DVAR we have no closed update formula for the covariances.
To deal with this second issue, \index{3DVAR}3DVAR uses a fixed predicted covariance,
independent of time $j$, and pre-specified.
The resulting minimization problem, and
its solution, is described in Table \ref{3DVAR:t1}, making the
analogy with the \index{Kalman filter}Kalman filter.

\begin{table}
\bgroup
\def\arraystretch{1.5}%
	\begin{tabular}{ | c | c |}
		\hline
		\index{Kalman filter}Kalman Filter & 3DVAR  \\ \hline
		$m_{j+1} = \arg \min_\varg \J(\varg)$ & $m_{j+1} = \arg \min_\varg \J(\varg)$  \\ \hline
		$\J(\varg) = \frac{1}{2} |y_{j+1} - H \varg|_{\Gamma}^2 + \frac{1}{2} |\varg - \hat{m}_{j+1}|_{\hat{C}_{j+1}}^2$  &  $\J(\varg) = \frac{1}{2} |y_{j+1} - H \varg|_{\Gamma}^2 + \frac{1}{2} |\varg - \hat{m}_{j+1}|_{\hat{C}}^2$ \\
		\hline
		$\hat{m}_{j+1} = M m_j$ & $\hat{m}_{j+1} = \Psi(m_j)$ \\ \hline
		$m_{j+1} = (I-K_{j+1}H)\hat{m}_{j+1} + K_{j+1} y_{j+1}$ & $m_{j+1} = (I-KH) \hat{m}_{j+1} + K y_{j+1}$\\ \hline
	\end{tabular}
	\egroup
	\caption{Comparison of \index{Kalman filter}Kalman filter and \index{3DVAR}3DVAR update formulae.}		
		\label{3DVAR:t1}
\end{table}

Note that the minimization
itself is of a quadratic functional, and so may be solved by means of
linear algebra. The constraint formulation used for the \index{Kalman filter}Kalman filter,
in Subsection  \ref{sec:okf}, may also be applied and used to derive the mean
update formula.

\subsection{3DVAR: Algorithmic Implementation}
The  \index{3DVAR}3DVAR \index{filtering}filtering method is fully described in the following algorithm.
\FloatBarrier
\begin{algorithm}
\caption{\label{alg3DVAR} 3DVAR}
\begin{algorithmic}[1]
\vspace{0.1in}
\STATE {\bf Input}:  Initial mean $m_0 \in \R^d$ and fixed predictive covariance $\hat{C} \in \R^{d\times d}.$ 
\STATE For $j = 0, 1, \ldots, J-1$ do the following \index{prediction}prediction and \index{analysis}analysis steps:
\STATE {\bf \index{prediction}Prediction}: 
\begin{align}\label{eq:pred3DVAR}
\hat{m}_{j+1} &= \Psi(m_j). 
\end{align}
\STATE{{{\bf \index{analysis}Analysis}}}: 
	\begin{align}\label{eq:an3DVAR}
	\begin{split}
		m_{j+1} &= (I-KH)\hat{m}_{j+1} + K y_{j+1}.	
		\end{split}
	\end{align}
\STATE{\bf Output}: Estimates $\{m_j\}_{j=1}^J$ of the \index{signal}signal.
\end{algorithmic}
\end{algorithm}
\FloatBarrier

The \index{Kalman gain}Kalman gain $K$ for \index{3DVAR}3DVAR is fixed, because the predicted covariance
$\hat{C}$ is fixed. Precisely we have, by analogy with the \index{Kalman filter}Kalman filter, the following formulae for the \index{3DVAR}3DVAR gain matrix $K:$
\begin{align*}
                S &= H \hat{C} H^\top + \Gamma, \\
                K &= \hat{C} H^\top S^{-1}.
        \end{align*}
The method also delivers an implied \index{analysis}analysis covariance $C= (I - K H) \hat{C}.$ 
Note that the resulting algorithm which maps $m_j$ to $m_{j+1}$
may be specified directly in terms of the gain $K$, without
need to introduce $\hat{C}, C,$ and $S$. In the remainder of this
section we simply view $K$ as fixed and given.
In this setting we show that the \index{3DVAR}3DVAR algorithm produces accurate \index{state}state estimation
under vanishing noise assumptions in the \index{dynamics}dynamics/\index{data model}data model.

\subsection{3DVAR: Long-Time Accuracy}
We will make the following assumptions on the \index{dynamics}dynamics/\index{data model}data model:

\begin{assumption}
\label{a:3dvar}
Consider the \index{stochastic dynamics model}dynamics/\index{data model}data model under the assumptions that
$\xi_j \equiv 0, \Gamma = \gamma^2 \Gamma_0, |\Gamma_0| = 1$
and assume that the data $y_{j+1}$ used in the \index{3DVAR}3DVAR
algorithm is found from observing
a true \index{signal}signal $v_j^{\dagger}$ given by
\begin{align*}
       \index{dynamics} \emph{Dynamics Model:}  \quad v_{j+1}^{\dagger} &= \Psi(v_j^{\dagger}), \quad j \in \Z^+. \\
        \index{data model}\emph{Data Model:}  \quad y_{j+1} &= Hv_{j+1}^{\dagger} + \gamma  \eta^{\dagger}_{j+1,0}, \quad j \in \Z^+.
\end{align*}
\end{assumption}

With this assumption of noise-free \index{dynamics!deterministic}dynamics ($\xi_j \equiv 0$)
we deduce that the \index{3DVAR}3DVAR filter produces output which, 
asymptotically, has an error of the same size as the \index{observation!noise}observational noise
error $\gamma.$ The key additional assumption in the theorem that
allows this deduction is
a relationship between the \index{Kalman gain}Kalman gain $K$ and
the derivative $D\Psi(\cdot)$ of the \index{dynamics}dynamics model. 
Encoded in the assumption are two ingredients: that the \index{observation!function}observation
function $H$ is rich enough in principle to learn enough components
of the system to synchronize the whole system; and that $K$ is designed
cleverly enough to effect this synchronization. The proof of the
theorem is simply using these two ingredients and then controlling 
the small stochastic perturbations, arising from noisy \index{observation}observations in Assumption \ref{a:3dvar}.

\begin{theorem}[Accuracy of 3DVAR]
Let Assumption \ref{a:3dvar} hold with $\eta^{\dagger}_{j,0} \sim \Nc(0,\Gamma_0)$ an \index{i.i.d.}i.i.d. sequence. Assume 
that, for the gain matrix $K$ appearing in the \index{3DVAR}3DVAR method,
there exists a norm $\lVert \cdot \rVert$ on $\Ru$ and constant $\lambda \in (0,1)$ such that, for all $v \in \Ru,$ 
$$\lVert (I-KH)D \Psi (v) \rVert \leq \lambda.$$
Then, there is a constant $c>0$ such that the \index{3DVAR}3DVAR algorithm satisfies the following large-time asymptotic
error bound: 
$$\limsup_{j\to\infty}  \Expect[\lVert m_j - v^{\dagger}_j \rVert] \leqslant \frac{c\gamma}{1-\lambda},$$
where the expectation is taken with respect to the sequence $\{\eta^{\dagger}_{j,0}\}.$ 
\end{theorem}

\begin{proof}
We have
	\begin{align*}
		v^{\dagger}_{j+1} &= \Psi(v^{\dagger}_{j}), \\
m_{j+1} &= (I-KH)\Psi(m_j) + Ky_{j+1},
\end{align*}
and hence that
\begin{align*}
	v^{\dagger}_{j+1} &= (I-KH)\Psi(v^{\dagger}_{j}) + KH\Psi(v^{\dagger}_{j}),\\	
	m_{j+1}	&= (I-KH)\Psi(m_j) + KH\Psi(v^{\dagger}_{j}) + \gamma K\eta^{\dagger}_{j+1,0}.
	\end{align*}
Define $e_j=m_j-v_j^{\dagger}.$ By subtracting the evolution
equation for $v_j^{\dagger}$ from that for $m_j$ we obtain, using
the mean value theorem, 
\begin{align*} \label{3DVAR:eq3}
	e_{j+1} &= m_{j+1} - v^{\dagger}_{j+1} & \\
	&= (I-KH) \bigl(\Psi(m_j)-\Psi(v^{\dagger}_j)\bigr) + \gamma K \eta^{\dagger}_{j+1,0} & \\
	&= \left( (I-KH) \int_0^1  {D} \Psi \left(s m_j + (1-s)v_j^\dagger \right) ds \right) e_j + \gamma K \eta_{j+1,0}^\dagger. && 
\end{align*}

\noindent As a result, by the triangle inequality,
\begin{align*} 
	\lVert e_{j+1} \rVert  & \leq  \left\lVert \left(  \int_0^1 (I-KH)D \Psi \left(s m_j + (1-s)v_j^\dagger \right) ds \right) e_j \right\rVert + \lVert \gamma K \eta_{j+1,0}^\dagger \rVert \\ 
	& \leq   \left(  \int_0^1 \left\lVert (I-KH)D \Psi \left(s m_j + (1-s)v_j^\dagger \right)  \right\rVert ds \right)  \lVert e_j \rVert  + \lVert \gamma K \eta_{j+1,0}^\dagger \rVert &  \\
	& \leq \lambda \lVert e_j \rVert + \gamma \lVert K \eta_{j+1,0}^{\dagger} \rVert. &
\end{align*}
Taking expectations on both sides, we obtain, for $c:= \Expect[\lVert K \eta_{j+1,0}^{\dagger} \rVert] > 0$, 
\begin{equation} \label{3DVAR:eq5}
	\begin{split}
		\Expect [\lVert e_{j+1} \rVert] & \leq \lambda \Expect [\lVert e_j \rVert] + \gamma \Expect [\lVert K \eta_{j+1,0}^{\dagger} \rVert] \\
		& \leq \lambda \Expect [\lVert e_j \rVert] + \gamma c.
	\end{split}
\end{equation}
Using the  \index{discrete Gronwall inequality}discrete Gronwall inequality of Theorem \ref{t1.6} we have that:
\begin{equation} \label{3DVAR:eq6}
	\begin{split}
		\Expect [\lVert e_j \rVert] & \leq \lambda^j \Expect [\lVert e_0 \rVert] + \sum_{i=0}^{j-1} c \lambda^i \gamma \\
		& \leq \lambda^j \Expect [\lVert e_0 \rVert] +  c \gamma \frac{1-\lambda^j}{1-\lambda},
	\end{split}
\end{equation}
where $e_0 =  m_0 - v_0$. Since $\lambda<1$, the desired statement follows.
\end{proof}

\section{4DVAR}\label{sec:93}

Recall that \index{3DVAR}3DVAR differs from \index{4DVAR}4DVAR because, whilst also based on an
\index{optimization}optimization principle, \index{4DVAR}4DVAR is applied in a distributed fashion over
all data in the time interval $j=1,\ldots, J$; in contrast \index{3DVAR}3DVAR is 
applied sequentially from time $j-1$ to time $j,$ for $j=1,\ldots, J$.
We consider two forms of the methodology: {\em \index{4DVAR!weak constraint}weak constraint 4DVAR (w4DVAR)}, in which the fact that the \index{dynamics}dynamics model contains randomness is accounted 
for in the \index{optimization}optimization; and \index{4DVAR}{\em 4DVAR} (sometimes known as {\em \index{4DVAR!strong constraint}strong constraint 4DVAR}), which can be derived from \index{4DVAR!weak constraint}w4DVAR in the limit of $\Sigma \to 0$ (no randomness in the \index{dynamics!deterministic}dynamics). 

The \index{objective}objective function minimized in \index{4DVAR!weak constraint}w4DVAR is
\begin{equation}\label{eq:objectivew4dvar}
	\J(\V) = \frac{1}{2}|v_0 - m_0|^2_{C_0} + \frac{1}{2}  \sum^{J-1}_{j=0}  |v_{j+1} - \Psi(v_j)|^2_{\Sigma} +  \frac{1}{2} \sum^{J-1}_{j=0} |y_{j+1} - H{v_{j+1}}|^2_{\Gamma},
\end{equation}
where $\V=\{v_{j}\}^{J}_{j=0} \in \R^{\Du(J+1)}$, 
$\Y=\{y_{j}\}^{J}_{j=1} \in \R^{\Dy J}$, $v_j \in \Ru$, $y_j \in \Ry$, $H$
is the \index{observation!function}observation function, $\Sigma$ is the random \index{dynamical system}dynamical system covariance, $\Gamma$ is the data noise covariance, and $m_0$ and $C_0$ are the mean and
covariance of the initial \index{state}state. 
The three terms in the \index{objective}objective function enforce, in turn, information about the initial
condition $v_0$, the \index{dynamics}dynamics model, and the \index{data model}data model. Note that, because
$\Psi$ is nonlinear, the \index{objective}objective is not quadratic and cannot be optimized
in closed form.  Implementation of the \index{4DVAR}4DVAR \index{smoothing}smoothing algorithm involves therefore using a suitable numerical 
\index{optimization}optimization algorithm; a brief discussion  of some guiding principles for the construction of gradient-based \index{optimization}optimization methods can be found in Chapter \ref{chap:optimization}, but whole books are devoted to this subject. 
 In contrast, each step of \index{3DVAR}3DVAR requires solution of a
quadratic \index{optimization}optimization problem, tractable in closed form. 

\begin{theorem}[Minimizer Exists for \index{4DVAR!weak constraint}w4DVAR]
Assume that $\Psi$ is bounded and continuous. Then $\J$ has a minimizer, which is a \index{MAP estimator}MAP estimator for the \index{smoothing}smoothing problem.
\end{theorem}

\begin{proof}
Recall Theorem \ref{thm:achievable-obj}, which shows that the \index{MAP estimator}MAP estimator based on the \index{smoothing!distribution}smoothing distribution $\Prob(\V|\Y) \propto \exp\bigl( -\J(\V) \bigr)$ is attained provided that $\J$ is guaranteed to be non-negative, continuous, and satisfy $\J(\V) \to \infty$ as $|\V| \to \infty.$ Now, the \index{objective}objective $\J$ defined by equation \eqref{eq:objectivew4dvar} is clearly non-negative, and it is continuous since $\Psi$ is assumed to be continuous. It remains to show that $\J(\V) \to \infty$ as $|\V| \to \infty.$
Let $R$ be a bound for $\Psi,$ so that $|\Psi(v_j)|_\Sigma \le R$ for all $v_j \in \Ru.$ 
Then, since 
$$\J(\V) \ge \frac{1}{2}|v_0|_{C_0}^2 - |v_0|_{C_0} |m_0|_{C_0}  + \frac{1}{2} \sum_{j=0}^{J-1} \Bigl( |v_{j+1} |^2_\Sigma 
- 2R|v_{j+1}|_{\Sigma}  \Bigr),$$
it follows that $\J(\V) \to \infty$ as $|\V| \to \infty$ and the proof is complete.  
\end{proof}

We now consider the vanishing dynamical noise limit of \index{4DVAR!weak constraint} w4DVAR.
This is to minimize 
\begin{equation*}
	\J_0(\V) = \frac{1}{2}|v_0 - m_0|^2_{C_0} +  \frac{1}{2}  \sum^{J-1}_{j=0} |y_{j+1} - H{v_{j+1}}|^2_{\Gamma}
\end{equation*}
subject to the hard constraint that
$$v_{j+1}=\Psi(v_j), \quad j=0,\ldots, J-1.$$
This is \index{4DVAR}4DVAR.
Note that by using the constraint, \index{4DVAR}4DVAR can be written as a minimization over
$v_0$, rather than over the entire sequence $\{v_j\}_{j=0}^{J}$ as is required
in \index{4DVAR!weak constraint}w4DVAR.

We let $\J_{\sigma}$ denote the \index{objective}objective function $\J$ from \index{4DVAR!weak constraint}w4DVAR
in the case where $\Sigma$ is replaced by  $\sigma^2\Sigma_0$.
Roughly speaking, the following result shows that minimizers of $\J_{\sigma}$ 
converge as $\sigma \to 0^+$ to points in ${\R}^{\Dy(J+1)}$ which
satisfy the hard constraint associated with \index{4DVAR}4DVAR. 

\begin{theorem}[Small \index{signal}Signal Noise Limit of \index{4DVAR!weak constraint}w4DVAR]
Suppose that $\Psi$ is bounded and continuous and let $\V^{\sigma}$ be a minimizer of $\J_{\sigma}.$ 
Then as $\sigma \to 0^+$ there is a convergent 
subsequence of $\V^{\sigma}$ with 
limit $\V^*$ satisfying $v^*_{j+1} = \Psi(v^*_j)$. 
\end{theorem}

\begin{proof}
Throughout this proof $c$ is a constant which may change from
instance to instance, but is independent of $\sigma.$
Consider $\V \in \R^{\Du(J+1)}$ defined by  $v_0=m_0$
and  $v_{j+1}=\Psi(v_{j})$. Then $\V$ is bounded, as $\Psi(\cdot)$ is bounded, 
and the bound is independent of $\sigma$. Furthermore, 
\begin{equation*}
	\J_{\sigma}(\V) =  \frac{1}{2} \sum^{J-1}_{j=0} |y_{j+1}-Hv_{j+1}|^2_\Gamma \leq c,
\end{equation*}
where $c$ is independent of $\sigma$.
It follows that
\begin{equation*}
	\J_{\sigma}(\V^{\sigma}) \leq \J_{\sigma}(\V) \leq c.
\end{equation*}
Thus,
\begin{align*}
	\frac{1}{2}|v^{\sigma}_{j+1} - \Psi(v^{\sigma}_{j})|^2_{\Sigma_0} &= \frac{\sigma^2}{2}|v^{\sigma}_{j+1} - \Psi(v^{\sigma}_{j})|^2_{\Sigma}  \le  \sigma^2 \J_\sigma(\V^\sigma) \leq \sigma^2 c, \\
	 \frac{1}{2}|v^{\sigma}_{0} - m_0|^2_{C_0} &\leq  \J_\sigma(\V^\sigma) \le  c .
\end{align*}
Since $\Psi$ is bounded, these bounds imply that $|\V^{\sigma}|$ is bounded
above independently of $\sigma.$
Therefore,
there is a limit $\V^{*} : \V^{\sigma} \to \V^*$ along
a subsequence. By continuity 
\begin{equation*}
	0 \leq \frac{1}{2}|v^{*}_{j+1} - \Psi(v^{*}_{j})|^2_{\Sigma_0} \leftarrow \frac{1}{2}|v^{\sigma}_{j+1} - \Psi(v^{\sigma}_{j})|^2_{\Sigma_0} \leq \sigma^{2}c. \\
\end{equation*}
Letting $\sigma \to 0^+$ we obtain that $v^*_{j+1} = \Psi(v^*_j).$
\end{proof}

\section{Discussion and Bibliography}\label{sec:94}

The \index{3DVAR}3DVAR and \index{4DVAR}4DVAR methodologies, in the context of weather forecasting,
are discussed in \cite{lorenc1986analysis} and \cite{fisher2009data},
respectively. The implementation of these methodologies by the 
 UK Meteorological Office 
is overviewed in \cite{lorenc2000met,rawlins2007met}.
The accuracy analysis presented here is similar to
that which first appeared in the papers
\cite{brett2013accuracy,moodey2013nonlinear}
and was developed further in \cite{law2012analysis,sanz2015long,law2016filter}. 
It arises from considering stochastic perturbations of the
seminal work of Titi and collaborators, exemplified by the paper
\cite{hayden2011discrete}; this in turn is linked to earlier work on
synchronization in dynamical systems \cite{pecora1990synchronization}. 
In all of these works, particular emphasis is placed in estimating the \index{state}state of deterministic chaotic \index{dynamical system}dynamical systems from partial and noisy \index{observation}observations \cite{lalley1999beneath,paulin2019optimization,paulin2018concentration,branicki2018accuracy,oljaca2018almost}.
For an overview of \index{variational!data assimilation}variational \index{data assimilation}data assimilation methods,
and their links to problems in physics and mechanics, see the book
\cite{abarbanel2013predicting} and the references therein; see
also the paper \cite{broecker2017existence}.

 \chapter{\Large{\sffamily{The Extended and Ensemble Kalman Filters }}}\label{lecture10}

In this chapter we describe the \index{Kalman filter!extended}Extended Kalman Filter  (\index{ExKF}ExKF)\footnote{The \index{Kalman filter!extended}extended
Kalman filter is often termed the EKF in the literature, a terminology 
introduced before the existence of the \index{EnKF}EnKF;
we find it useful to write \index{ExKF}ExKF to unequivocally distinguish it from the \index{EnKF}EnKF.}
and the \index{Kalman filter!ensemble}Ensemble Kalman Filter (\index{EnKF}EnKF). The \index{ExKF}ExKF approximates the predictive covariance by linearization, while the \index{EnKF}EnKF approximates it by the empirical covariance of a collection of particles. 
The \index{ExKF}ExKF is a provably accurate approximation of the 
\index{filtering!distribution}filtering distribution if the \index{dynamics}dynamics are approximately linear and \index{small noise limit}small noise is present in both \index{signal}signal and data, in which case the \index{filtering!distribution}filtering distribution is well approximated by a 
\index{Gaussian!approximation}Gaussian. In such settings, the \index{EnKF}EnKF can also provide a good approximation of the
\index{filtering!distribution}filtering distribution if a sufficiently large number of particles
is used.  For problems where the filtering distributions are not well approximated by Gaussians, ExKF and EnKF can still be successful online optimizers for state estimation; they may be thought of as generalizations of 3DVAR in which
the model covariance, which weights the model contribution to the optimization
problem solved at every step, is updated on the basis of linearized (ExKF)
or ensemble (EnKF) information.

This chapter is organized as follows. We introduce the problem setting in Section \ref{sec:101}. The \index{ExKF}ExKF and \index{EnKF}EnKF are described, in turn, in Sections \ref{sec:102} and \ref{sec:103}. We close in Section \ref{sec:104} with extensions and bibliographical remarks.

\section{The Setting}\label{sec:101}
Throughout this chapter we consider the setting in which 3DVAR was introduced
and may be applied: the \index{dynamics}dynamics model
is nonlinear, but the \index{observation!function}observation function is linear. 
For purposes of exposition, we summarize it again here:
\begin{align*}
        v_{j+1} &= \Psi(v_j) + \xi_j, && \xi_j \sim \Nc(0, \Sigma) \index{i.i.d.} \text{ i.i.d.}, \\
        y_{j+1} &= Hv_{j+1} + \eta_{j+1}, && \eta_j \sim \Nc(0, \Gamma) \index{i.i.d.}\text{ i.i.d.},
\end{align*}
with,  as in previous chapters,  \(v_0 \sim \Nc(m_0, C_0)\) independent of the independent \index{i.i.d.}i.i.d.
sequences \(\{\xi_j\}\) and  \(\{\eta_j\}\).
Throughout this chapter we assume that $v_j \in \R^\Du, \,
y_j \in \R^\Dy.$ 

\FloatBarrier
\begin{table}
\bgroup
\def\arraystretch{1.5}
        \begin{tabular}{ | c | c |}
                \hline
                Kalman Filter & \index{ExKF}ExKF \\ \hline
                $m_{j+1} = \arg \min_\varg \J(\varg)$ & $m_{j+1} = \arg \min_\varg \J(\varg)$  \\ \hline
                $\J(\varg) = \frac{1}{2} |y_{j+1} - H \varg|_{\Gamma}^2 + \frac{1}{2} |\varg - \hat{m}_{j+1}|_{\hat{C}_{j+1}}^2$  &  $\J(\varg) = \frac{1}{2} |y_{j+1} - H \varg|_{\Gamma}^2 + \frac{1}{2} |\varg - \hat{m}_{j+1}|_{\hat{C}_{j+1}}^2$ \\
                \hline
                $\hat{m}_{j+1} = M m_j$ & $\hat{m}_{j+1} = \Psi(m_j)$ \\ \hline
$\hat{C}_{j+1}$ update exact & $\hat{C}_{j+1}$ update by linearization \\ \hline 
                $m_{j+1} = (I-K_{j+1}H)\hat{m}_{j+1} + K_{j+1} y_{j+1}$ & $m_{j+1} = (I-K_{j+1}H)\hat{m}_{j+1} + K_{j+1} y_{j+1}$\\ \hline
        \end{tabular}
\egroup        
        \caption{Comparison of \index{Kalman filter}Kalman filter and \index{ExKF}ExKF update formulae.}
                \label{3DVAR:t2}
\end{table}
\FloatBarrier

\section{The \index{Kalman filter!extended}Extended Kalman Filter}\label{sec:102}
This method is derived by applying the Kalman methodology, using
linearization to propagate the covariance $C_j$ to the predictive
covariance $\hat{C}_{j+1}.$ Table \ref{3DVAR:t2} summarizes the idea,
and in what follows we calculate the formulae required in full detail.

We first recall the \index{Kalman filter}Kalman filter update formulae and their derivation. We
have
\begin{align}
\widehat{v}_{j+1} = Mv_{j}+\xi_{j}, \quad v_{j}  \sim \Nc(m_{j},C_{j}), \quad \xi_{j} \sim \Nc(0,\Sigma).
\end{align}
From this we deduce, by taking expectations, that
\begin{equation}
\widehat{m}_{j+1} = \Expect[\widehat{v}_{j+1} \mid   Y_{j}]=\Expect[Mv_{j}+\xi_{j} \mid  Y_{j}]= \Expect[Mv_{j}\mid  Y_{j}] + \Expect[\xi_{j}\mid  Y_{j}] =Mm_{j} .
\end{equation}
The covariance update is derived as follows:
\begin{align}
\widehat{C}_{j+1} &= \Expect\Bigl[(\widehat{v}_{j+1}-\widehat{m}_{j+1})\otimes(\widehat{v}_{j+1}-\widehat{m}_{j+1})\mid  Y_{j} \Bigr] \nonumber \\
				  &= \Expect\Bigl[(M(v_{j}-m_{j})+\xi_{j})\otimes(M(v_{j}-m_{j})+\xi_{j})\mid  Y_{j}\Bigr] \nonumber \\
                  &= \Expect \Bigl[(M(v_{j}-m_{j}))\otimes(M(v_{j}-m_{j}))\mid  Y_{j} \Bigr]+\Expect\Bigl[\xi_{j}\otimes\xi_{j}\mid  Y_{j}\Bigr]  \\
&\quad  +\Expect\Bigl[(M(v_{j}-m_{j}))\otimes \xi_{j} \mid Y_{j}\Bigr]
+\Expect\Bigl[\xi_{j}\otimes (M(v_{j}-m_{j})) \mid Y_{j}\Bigr] \nonumber \\
                  &=M\Expect\Bigl[(v_{j}-m_{j})\otimes(v_{j}-m_{j}) \mid Y_{j}\Bigr]M^{\top} + \Sigma \nonumber \\
                 &=MC_{j}M^{\top}+\Sigma. \nonumber
\end{align}

For the \index{ExKF}ExKF, the \index{prediction}prediction map $\Psi$ is no longer linear.
But since 
$\xi_{j}$ is independent of $Y_{j}$ and $v_{j}$,
we obtain
\begin{equation*}
\widehat{m}_{j+1} = \Expect\Bigl[\Psi(v_{j})+\xi_{j} \mid Y_{j}\Bigr] = \Expect\Bigl[\Psi(v_{j})\mid Y_{j}\Bigr]+\Expect\Bigl[\xi_{j} \mid Y_{j}\Bigr] = \Expect \Bigl[\Psi(v_{j})\mid Y_{j}\Bigr]. 
\end{equation*}
If we assume that the fluctuations of $v_j$ around its mean $m_j$ (conditional
on data) are small, then a reasonable approximation is to take
$\Psi(v_{j}) \approx \Psi(m_{j})$ so that
\begin{equation}
\label{eq:1a}
\widehat{m}_{j+1} = \Psi(m_{j}).    
\end{equation}
For the predictive covariance we use linearization; we have
\begin{align*}
\widehat{C}_{j+1} \nonumber &= \Expect\Bigl[(\widehat{v}_{j+1}-\widehat{m}_{j+1})
\otimes(\widehat{v}_{j+1}-\widehat{m}_{j+1})\mid  Y_{j}\Bigr]\nonumber \\ 
                 &= \Expect\Bigl[(\Psi(v_{j})-\Psi(m_{j})+\xi_{j})\otimes(\Psi(v_{j})-\Psi(m_{j})+\xi_{j})\mid  Y_{j}\Bigr]\\ 
&=\Expect \Bigl[(\Psi(v_{j})-\Psi(m_{j}))\otimes(\Psi(v_{j})-\Psi(m_{j}))\mid  Y_{j}\Bigr]+\Sigma\\
&\approx D\Psi(m_{j})\Expect\Bigl[(v_{j}-m_{j})\otimes(v_{j}-m_{j})\mid  Y_{j}\Bigr]D\Psi(m_{j})^\top+\Sigma,
\end{align*}
and so, again assuming that fluctuations of $v_j$ around its mean $m_j$ 
(conditional on data) are small, we invoke the approximation
\begin{equation}
\label{eq:1b}
\widehat{C}_{j+1}=D\Psi(m_{j})C_{j}D\Psi(m_{j})^{\top}+\Sigma.
\end{equation}
To be self-consistent, $\Sigma$ itself should be small.  We next summarize the steps of the \index{ExKF}ExKF.

\FloatBarrier
\begin{algorithm}
\caption{\label{algExKF} \index{Kalman filter!extended}Extended Kalman Filter}
\begin{algorithmic}[1]
\STATE {\bf Input}:  Initial mean $m_0 \in \R^d$ and covariance $C_0 \in \R^{d\times d}.$ 
\STATE For $j = 0, 1, \ldots, J-1$ do the following \index{prediction}prediction and \index{analysis}analysis steps:
\STATE {\bf \index{prediction}Prediction}: 
\begin{align}\label{eq:predExKF}
\hat{m}_{j+1} &= \Psi(m_j), \\
\hat{C}_{j+1} &= D\Psi(m_j) C_j D\Psi(m_j)^\top+\Sigma.
\end{align}
\vspace{-0.5cm}
\STATE{{{\bf \index{analysis}Analysis}}}: 
	\begin{align}\label{eq:anExKF}
	\begin{split}
		m_{j+1} &= (I-K_{j+1}H)\hat{m}_{j+1} + K_{j+1} y_{j+1},\\
		C_{j+1}&=(I-K_{j+1}H)\widehat{C}_{j+1}.
		\end{split}
	\end{align}
\STATE{\bf Output}: Predictive means $\{\hat{m}_j\}_{j=1}^J$ and covariances $\{\hat{C}_{j}\}_{j=1}^J,$ and \index{analysis}analysis means $\{m_j\}_{j=1}^J$ and covariances $\{C_j\}_{j=1}^J.$
\end{algorithmic}
\end{algorithm}
\FloatBarrier


Note that the \index{Kalman gain}Kalman gain $K_{j+1}$ in equation \eqref{eq:anExKF} is defined in the same way as for the \index{Kalman filter}Kalman filter, namely 
	\begin{align*}
		K_{j+1} &= \hat{C}_{j+1} H^\top S_{j+1}^{-1}, \quad \quad S_{j+1} = H \hat{C}_{j+1} H^\top + \Gamma.
	\end{align*}
Thus, the \index{analysis}analysis step is the same as for the \index{Kalman filter}Kalman filter. However, for the \index{ExKF}ExKF
the maps $C_{j} \mapsto\widehat{C}_{j+1}\mapsto C_{j+1}$ depend
on the observed data through the dependence of the predictive
covariance on the filter mean. 
To be self-consistent with the ``small fluctuations around the mean''
assumptions made in the derivation of the \index{ExKF}ExKF, $\Sigma$ and $\Gamma$
should both be small.

The \index{analysis}analysis step can also be defined by 
\begin{align*}
C_{j+1}^{-1} &= \widehat{C}_{j+1}^{-1}+H^{\top}\Gamma^{-1}H, \\
m_{j+1}&=\arg \min_\varg \J(\varg),
\end{align*}
where
\begin{equation}\label{eq:Ifilter}
\J(\varg) = \frac{1}{2} |y_{j+1}-H\varg |_{\Gamma}^{2} + \frac{1}{2}|\varg-\widehat{m}_{j+1} |^{2}_{\widehat{C}_{j+1}}
\end{equation}
and $\widehat{m}_{j+1}, \widehat{C}_{j+1}$ are calculated as above
in the \index{prediction}prediction step \eqref{eq:predExKF}. 
The constraint formulation of the minimization
problem, derived for the \index{Kalman filter}Kalman filter in Section \ref{sec:okf}, may
also be used to derive the update formulae above.


\section{\index{Kalman filter!ensemble}Ensemble Kalman Filter}\label{sec:103}

When the \index{dynamical system}dynamical system is in high dimension, evaluation and storage of 
the predictive covariance, and in particular the Jacobian required for 
the update formula \eqref{eq:1b},
becomes computationally inefficient and expensive for the 
\index{ExKF}ExKF. The \index{EnKF}EnKF was developed to overcome
this issue. The basic idea is to maintain an ensemble of particles,
and to use their empirical covariance
within a Kalman-type update. The method is summarized in  
Table \ref{3DVAR:t3}.  It may be thought of as an ensemble 3DVAR technique
in which a collection of particles are generated similarly to 3DVAR, but
interact through an ensemble estimate of their covariance.

In the basic form which we present here, the 
\index{EnKF}EnKF is applied  when $\Psi$ is nonlinear, while the \index{observation!function}observation
function $H$ is linear. 
The $\Sam$ particles used at step $j$ are denoted $\{v_{j}^{(\sam)}\}^{\Sam}_{\sam=1}.$ 
They are all given equal weight, so it is possible, in principle, to 
make an approximation to the \index{filtering!distribution}filtering distribution of the form
$$\post_j^\Sam (v_j) \approx \frac{1}{\Sam}\sum_{\sam=1}^\Sam 
\delta\bigl(v_j-v_{j}^{(\sam)}\bigr).$$
This approximation can in principle be accurate if $\Sam$ is sufficiently large and the filtering distributions are approximately Gaussian. In problems where approximate Gaussianity of the filtering distribution fails ---for instance due to strong nonlinearity of $\Psi$ and large observation noise--- EnKF is better understood as a sequential \index{optimization}optimization method, similar in spirit to 3DVAR,  as described in the introduction to the chapter.

The \index{state}state of all the particles at time $j+1$ are predicted 
to give $\{\widehat{v}_{j+1}^{(\sam)}\}^{\Sam}_{\sam=1}$ using the dynamical
model. The resulting empirical covariance is then used
to define an \index{objective}objective function which is minimized
in order to perform the \index{analysis}analysis step and obtain $\{{v}_{j+1}^{(\sam)}\}^{\Sam}_{\sam=1}.$ 
The updates are denoted schematically by
\[\{v_{j}^{(\sam)}\}^{\Sam}_{\sam=1} \xmapsto{} \{\widehat{v}_{j+1}^{(\sam)}\}^{\Sam}_{\sam=1}\xmapsto{} \{v_{j+1}^{(\sam)}\}^{\Sam}_{\sam=1}.\]

The idea of the \index{EnKF}EnKF is summarized in Table \ref{3DVAR:t3} below, which is followed by a full description of the algorithm.  

\FloatBarrier
\begin{table}
\bgroup
\def\arraystretch{1.5}
        \begin{tabular}{ | c | c |}
                \hline
                Kalman Filter & \index{EnKF}EnKF \\ \hline
                $m_{j+1} = \arg \min_\varg \J(\varg)$ &  $v_{j+1}^{(\sam)} = \arg \min_\varg \J_\sam(\varg)$  \\ \hline
                $\J(\varg) = \frac{1}{2} |y_{j+1} - H \varg|_{\Gamma}^2 + \frac{1}{2} |\varg - \hat{m}_{j+1}|_{\hat{C}_{j+1}}^2$  &  $\J_n(\varg) = \frac{1}{2} |y_{j+1}^{(\sam)} - H  v |_{\Gamma}^2 + \frac{1}{2} |\varg - \hat{v}_{j+1}^{(\sam)}|_{\hat{C}_{j+1}}^2$ \\
                \hline
                $\hat{m}_{j+1} = M m_j$ & $\hat{v}_{j+1}^{(\sam)} = \Psi(v_j^{(\sam)})+\xi_j^{(\sam)}$ \\ \hline
$\hat{C}_{j+1}$ update exact & $\hat{C}_{j+1}$ update by ensemble estimate\\ \hline 
                $m_{j+1} = (I-K_{j+1}H)\hat{m}_{j+1} + K_{j+1} y_{j+1}$ &  $v_{j+1}^{(\sam)} = (I-K_{j+1}H)\hat{v}_{j+1}^{(\sam)} + K_{j+1} y_{j+1}^{(n)}$  \\ \hline
        \end{tabular}
\egroup
        \caption{Comparison of \index{Kalman filter}Kalman filter and \index{EnKF}EnKF update formulae.}
                \label{3DVAR:t3}
\end{table}
\FloatBarrier

\subsection{Algorithmic Implementation of \index{EnKF}EnKF}
We next summarize the steps of the \index{EnKF}EnKF:
\FloatBarrier
\begin{algorithm}
\caption{\label{algEnKF} \index{Kalman filter!ensemble} Ensemble Kalman Filter}
\begin{algorithmic}[1]
\STATE {\bf Input}: Ensemble size $N.$ Initial ensemble $\{v_0^{(n)}\}_{n=1}^N.$ Parameter $s \in \{0,1\}.$
\STATE For $j = 0, 1, \ldots, J-1$ do the following \index{prediction}prediction and \index{analysis}analysis steps:
\STATE {\bf \index{prediction}Prediction}: 
\begin{align}\label{eq:predEnKF}
\begin{split}
\xi_{j}^{(\sam)} &\sim \Nc(0,\Sigma), \quad  \index{i.i.d.}\text{i.i.d.}, \quad \sam=1,\ldots,\Sam,\\
\widehat{v}_{j+1}^{(\sam)} &= \Psi(v_{j}^{(\sam)})+\xi^{(\sam)}_{j}, \quad \sam=1,\ldots,\Sam, \\
\widehat{m}_{j+1} &= \frac{1}{\Sam}\sum^{\Sam}_{\sam=1} \widehat{v}_{j+1}^{(\sam)}, \\
\widehat{C}_{j+1} &= \frac{1}{\Sam}\sum^{\Sam}_{\sam=1}\bigl(\widehat{v}^{(\sam)}_{j+1}-\widehat{m}_{j+1}\bigr)\otimes \bigl(\widehat{v}^{(\sam)}_{j+1}-\widehat{m}_{j+1}\bigr).
\end{split}
\end{align}
\vspace{-0.5cm}
\STATE{{{\bf \index{analysis}Analysis}}}: 
	\begin{align}\label{eq:analEnKF}
	\begin{split}
	\eta_{j+1}^{(\sam)} &\sim \Nc(0,\Gamma),   \quad \sam=1,\ldots,\Sam,\\
y_{j+1}^{(\sam)}&=y_{j+1}+s\eta_{j+1}^{(\sam)},\quad \sam=1,\ldots,\Sam,\\
v_{j+1}^{(\sam)} &= (I-K_{j+1}H)\widehat{v}_{j+1}^{(\sam)}+K_{j+1}y_{j+1}^{(\sam)},\quad \sam=1,\ldots,\Sam. 
\end{split}
\end{align}
\STATE{\bf Output}: Ensembles  $\{v_j^{(n)}\}_{n=1}^N,\quad  j = 0,1, \ldots, J.$
\end{algorithmic}
\end{algorithm}
\FloatBarrier

%
%
%
Once again  the \index{Kalman gain}Kalman gain $K_{j+1}$ in equation \eqref{eq:analEnKF} is defined in the same way as for the \index{Kalman filter}Kalman filter, namely 
	\begin{align*}
		K_{j+1} &= \hat{C}_{j+1} H^\top S_{j+1}^{-1}, \quad \quad S_{j+1} = H \hat{C}_{j+1} H^\top + \Gamma.
	\end{align*}
 However $\hat{C}_{j+1}$ is estimated in a novel fashion, using 
an ensemble of particles;
this is the key innovation behind the EnKF. 
The parameter $s$ may be chosen to be $0$ or $1$.
The choice $s=1$ is natural when aiming at approximating the \index{Kalman filter}Kalman filter in \index{linear-Gaussian setting}linear-Gaussian settings; in such case the $y_{j+1}^{(\sam)}$ are referred to as \index{observation!perturbed}{\em perturbed observations}. The choice $s=0$ is natural if viewing the algorithm as a sequential
optimizer  in problems where the filtering distributions are not well approximated by Gaussians.

The \index{analysis}analysis step may be written as 
\begin{equation}
v_{j+1}^{(\sam)}=\arg \min_\varg \J_{\sam}(v),    
\end{equation}
where
\begin{equation}
\label{eq:deenz2}
\J_{n}(\varg) := \frac{1}{2} |y_{j+1}^{(\sam)} - H\varg|^{2}_{\Gamma} + \frac{1}{2} |\varg-\widehat{v}_{j+1}^{(\sam)} |^{2}_{\widehat{C}_{j+1}}
\end{equation}
and the predictive mean and covariance are given by  \eqref{eq:predEnKF}. 
Note that $\widehat{C}_{j+1}$ is typically not invertible as it is a rank
$\Sam$ matrix and $\Sam$ is usually less than the dimension $\Du$ of the space
on which $\widehat{C}_{j+1}$ acts; this is since the typical use
of ensemble methods is for high-dimensional \index{state-space}state-space estimation, 
with a small ensemble size. 
The minimizing solution can be found
by regularizing $\widehat{C}_{j+1}$ by  adding $\epsilon I$ for
$\epsilon>0$, deriving the update equations as above,
and then letting $\epsilon \to 0^+.$ Alternatively, the constraint 
formulation of the minimization problem, derived for the \index{Kalman filter}Kalman filter 
in Subsection \ref{sec:okf}, may also be used to derive the update formulae above.

The following theorem explains why perturbing the observations with $ s = 1$ may be favored when aiming at approximating the \index{Kalman filter}Kalman filter in 
 (close to) \index{linear-Gaussian setting}linear-Gaussian settings. Setting $s=1$  ensures that if each \index{prediction}prediction particle  $\widehat{v}_{j+1}^{(\sam)}$ is distributed according to a non-degenerate \index{Gaussian}Gaussian predictive distribution $\mathcal{N}  ( \widehat{m}_{j+1}, \widehat{C}_{j+1}),$ 
then,  in the linear Gaussian setting, each analysis particle $\widehat{v}_{j+1}^{(\sam)}$ will be \index{Gaussian}Gaussian distributed with mean and covariance given by the \index{filtering!distribution}filtering distribution found by the \index{Kalman filter}Kalman filter formulae. This is achieved by updating each particle minimizing an objective defined using a randomization of the 
likelihood\index{likelihood} function. 
\begin{theorem}[Perturbed Observation EnKF –- Randomized Likelihood Viewpoint]\label{th:enkfpo}
Suppose that $\widehat{v}_{j+1}^{(\sam)} \sim \mathcal{N}  ( \widehat{m}_{j+1}, \widehat{C}_{j+1})$  with $\widehat{C}_{j+1}$ positive definite. Let  $v_{j+1}^{(n)}$ be the minimizer of 
\begin{equation}\label{eq:randomobj}
\mathsf{J}_n(v):=  \frac12 |y_{j+1} + \eta_{j+1}^{(n)} - Hv  |_\Gamma^2 + \frac12 |v - \widehat{v}_{j+1}^{(\sam)} |_{\widehat{C}_{j+1}}^2 , \quad \quad \eta_{j+1}^{(n)}\sim \mathcal{N} (0,\Gamma),
\end{equation}
where $\widehat{v}_{j+1}^{(\sam)}$ and $\eta_{j+1}^{(n)}$ are independent.
Then $v_{j+1}^{(n)} \sim \mathcal{N} ( m_{j+1}, C_{j+1}),$ where $m_{j+1}$ and $C_{j+1}$ are defined by
\begin{align}
m_{j+1} & = \widehat{m}_{j+1} + K_{j+1} (y_{j+1} - H\widehat{m}_{j+1}), \label{eq:statement1}\\ 
C_{j+1} &= (I - K_{j+1} H)\widehat{C}_{j+1}, \label{eq:statement2}
\end{align}
and
\begin{equation*}
K_{j+1} := \widehat{C}_{j+1} H^\top (H\widehat{C}_{j+1} H^\top + \Gamma)^{-1}.
\end{equation*}
\end{theorem}
\begin{proof}
The minimizer of \eqref{eq:randomobj} is given by 
\begin{align}
v_{j+1}^{(\sam)} &= \widehat{v}_{j+1}^{(\sam)} + K_{j+1}( y_{j+1} + \eta_{j+1}^{(\sam)} - H \widehat{v}_{j+1}^{(\sam)} ) \label{eq:aux1} \\
&= C_{j+1} \Bigl\{ \widehat{C}_{j+1}^{-1} \widehat{v}_{j+1}^{(\sam)} + H^\top \Gamma^{-1}(y_{j+1} + \eta_{j+1}^{(\sam)} ) \Bigr \}, \label{eq:aux2}
\end{align}
where $C_{j+1}$ is defined in \eqref{eq:statement2} and the equivalence between \eqref{eq:aux1} and \eqref{eq:aux2} follows from the equivalence of precision and covariance characterizations of the Kalman filter in Theorem \ref{kalmanfilter}  and equations \eqref{eq:kf update} and \eqref{eq:kf analysis}. Notice that \eqref{eq:aux1} and \eqref{eq:aux2} show that $v_{j+1}^{(\sam)} $ can be written as a linear combination of Gaussian random variables, so $v_{j+1}^{(\sam)} $ is Gaussian. We next show that its mean and covariance are given by \eqref{eq:statement1} and \eqref{eq:statement2}.

First, from \eqref{eq:aux1} we deduce that
\begin{align*}
\Expect[  v_{j+1}^{(\sam)} ] &= \Expect\Bigl[\widehat{v}_{j+1}^{(\sam)} + K_{j+1}( y_{j+1} + \eta_{j+1}^{(\sam)} - H \widehat{v}_{j+1}^{(\sam)} )  \Bigr]  \\
& =  \widehat{m}_{j+1} + K_{j+1} (y_{j+1} - H\widehat{m}_{j+1}),
\end{align*}
where we used that by assumption $\Expect[ \widehat{v}_{j+1}^{(\sam)}] = \widehat{m}_{j+1}$ and that $\Expect[\eta_{j+1}^{(\sam)} ] = 0.$ 

Second, from \eqref{eq:aux2} we deduce that
\begin{align*}
\Expect[ ( v_{j+1}^{(\sam)}    - m_{j+1}) \otimes ( v_{j+1}^{(\sam)}    - m_{j+1}) ]   
&= C_{j+1} \widehat{C}_{j+1}^{-1} C_{j+1} + C_{j+1} H^\top \Gamma^{-1} H C_{j+1} \\ 
&= C_{j+1} \bigl(  \widehat{C}_{j+1}^{-1}  + H^\top \Gamma^{-1} H     \bigr)  C_{j+1} \\
& = C_{j+1},
\end{align*}
where we used that $C_{j+1}^{-1} = \widehat{C}_{j+1}^{-1}  + H^\top \Gamma^{-1} H $ by  the equivalent characterization of the Kalman filter covariance in Theorem \ref{kalmanfilter} and equations \eqref{eq:kf update} and \eqref{eq:kf analysis}.
\end{proof}

\subsection{Subspace Property of EnKF}

We now give another way to think of, and exploit in algorithms, the 
low rank property of $\widehat{C}_{j+1}.$
Note that $\J_{\sam}(\varg)$ is undefined unless
$$\varg-\widehat{v}_{j+1}^{(\sam)}=\widehat{C}_{j+1}a$$
for some $a \in \Ru.$ 
From the structure of $\widehat{C}_{j+1}$ it follows that 
\begin{equation}
\label{eq:deenz3}
 \varg=\widehat{v}_{j+1}^{(\sam)}+\frac{1}{\Sam}\sum_{m=1}^\Sam b_m\bigl(\widehat{v}^{(m)}_{j+1}-\widehat{m}_{j+1}\bigr) 
\end{equation}
for some unknown  vector $b = \{b_m \}_{m=1}^\Sam \in \R^\Sam$ to be determined. Note that both $a$ and $b$  depend on the ensemble member $n$, but we suppress that dependence from the notation. This form for $\varg$ can be substituted into \eqref{eq:deenz2} to obtain
a functional $\I_{\sam}(b)$ to be minimized over $b \in \R^\Sam.$ 
We re-emphasize that $\Sam$ will typically be much smaller than $\Du$, the
\index{state-space}state-space dimension. 
Once $b$ is determined, it may be substituted back into
\eqref{eq:deenz3} to obtain the solution to the minimization problem. 
 
To dig a little deeper into this calculation, we define
$$e^{(m)}=\widehat{v}^{(m)}_{j+1}-\widehat{m}_{j+1}$$
and note that then
$$\widehat{C}_{j+1}=\frac{1}{\Sam}\sum_{m=1}^\Sam e^{(m)} \otimes e^{(m)}.$$
Since
$$\widehat{C}_{j+1} a=\frac{1}{\Sam}\sum_{m=1}^\Sam b_m e^{(m)}$$
we deduce that 
$$b_m= \langle e^{(m)}, a \rangle.$$
Now note that
$$\frac{1}{2} | v-\widehat{v}_{j+1}^{(\sam)} |^{2}_{\widehat{C}_{j+1}}
= \frac12 \langle a, \widehat{C}_{j+1} a \rangle=\frac{1}{2\Sam}\sum_{m=1}^\Sam b_m^2.$$ 
Therefore, defining
\begin{equation}
\label{eq:z2}
 \mathsf{F}_{\sam}(b)  := \frac{1}{2} \Bigl|y_{j+1}^{(\sam)} - H\widehat{v}_{j+1}^{(\sam)} -\frac{1}{\Sam}\sum_{m=1}^\Sam b_m H(\widehat{v}_{j+1}^{(\sam)}-\widehat{m}_{j+1}) \Bigr|^{2}_{\Gamma}+\frac{1}{2\Sam}\sum_{m=1}^\Sam b_m^2
\end{equation}
we have proved the following:

\begin{theorem}[Implementation of \index{EnKF}EnKF in $\Sam$--Dimensional Subspace] 
\label{eq:extra1}
Given the \index{prediction}prediction defined by \eqref{eq:predEnKF}, the Kalman 
update formulae \eqref{eq:analEnKF} may be found
by minimizing $ \mathsf{F}_{\sam}(b)  $ with respect to $b$ and substituting
into \eqref{eq:deenz3}.
\end{theorem}

\section{Discussion and Bibliography}\label{sec:104}
In this chapter we have considered a derivative-based \index{filtering}filtering algorithm (ExKF) and an
ensemble-based \index{filtering}filtering algorithm (EnKF). Extended and ensemble Kalman algorithms for the \index{smoothing}smoothing problem are also available, see e.g. \cite{evensen2000ensemble,bell1994iterated}. 

The development and theory of the \index{ExKF}ExKF is 
documented in the text \cite{jazwinski2007stochastic}. 
A methodology for analyzing evolving probability distributions
with small variance, and establishing the validity of the
\index{Gaussian!approximation}Gaussian approximation, is described in \cite{sanz2016gaussian}.
The use of the \index{ExKF}ExKF
for weather forecasting was proposed in \cite{ghilande}.
However, the dimension of the \index{state-space}state-space in most geophysical
applications renders the \index{ExKF}ExKF impractical.

The \index{EnKF}EnKF provided an innovation with far
reaching consequences in geophysical applications, because
it allowed for the use of  partial, low-rank,  empirical correlation
information, without the computation of the full covariance.
An overview of ensemble Kalman methods may be found in the
book \cite{evensen2009data}, including a historical perspective
on the subject, originating from papers of Evensen and
Van Leeuwen in the mid 
1990s \cite{evensen1994sequential,evensen1996assimilation};
a similar idea was also developed by Houtekamer
within the Canadian Meteorological Service,  
around the same time; \cite{houtekamer1995methods,houtekamer1998data}.
The presentation of the \index{EnKF}EnKF 
as a smart  sequential \index{optimization}optimization tool, adopted
here, is developed in \cite{law2015data}. The derivation of
the update equations in a space whose dimension is that of
the ensemble is well-known to practitioners
in the field \cite{asch2016data}
and is also described in \cite{albers2019ensemble}. 
The form of \index{EnKF}EnKF with perturbed observations ($s = 1$) presented in these notes is closely related to \index{likelihood!randomized maximum}randomized maximum likelihood \cite{chen2012ensemble}, but other implementations of the algorithm are available, see e.g. \cite{tippett2003ensemble,anderson2001ensemble,bishop2001adaptive,majda2012filtering}.
See also \cite{kelly2014well} for  a proof of Theorem \ref{th:enkfpo} and for further details on the connection between randomized maximum likelihood\index{likelihood} and perturbed observation EnKF. 

The analysis of ensemble methods
is difficult and theory is only just starting to emerge.
In the linear case the method converges in the large ensemble
limit to the \index{Kalman filter}Kalman filter \cite{le2009large,mandel2011convergence,kwiatkowski2015convergence}, 
but in the nonlinear case the
limit does not reproduce the \index{filtering!distribution}filtering distribution
\cite{ernst2015analysis}. 
An overview of ensemble Kalman methods, adopting a unifying
mean-field framework in which $N \to \infty$, may be found in
\cite{calvello22}. That framework provides the basis for an analysis of
the accuracy of the EnKF \cite{carrillo2022ensemble},
in terms of its ability to approximate
the true filtering distribution. 
However it is arguable that a major advantage of ensemble methods is that
they can provide good \index{state}state estimation when the number of
particles is {\em not} large; this subject
is discussed in \cite{gottwald2013mechanism,kelly2014well,tong2015nonlinear,tong2016nonlinear,ghattassanzalonso22}. 
In particular the paper \cite{ghattassanzalonso22} develops a
unified non-asymptotic analysis of ensemble Kalman
methods from the perspective of high-dimensional statistics,
which explains why a small sample size $N$ suffices in applications
where the covariance models have moderate effective dimension.

 \chapter{\Large{\sffamily{Particle Filter}}}
\label{ch11}

This chapter is devoted to the \index{particle filter}particle filter, a method that
approximates the \index{filtering!distribution}filtering distribution by a sum of \index{Dirac}Dirac masses.
\index{particle filter}Particle filters provably converge to the \index{filtering!distribution}filtering distribution as the number of particles, and hence the number of \index{Dirac}Dirac masses,
approaches infinity. We focus on the \index{particle filter!bootstrap}bootstrap particle filter, also known as \index{sequential importance resampling}sequential importance resampling;
it is linked to the material on \index{Monte Carlo}Monte Carlo and \index{importance sampling}importance sampling described in Chapter \ref{lecture5}.
We note that the \index{Kalman filter}Kalman filter completely characterizes the \index{filtering!distribution}filtering
distribution in the \index{linear-Gaussian setting}linear-Gaussian setting. The Kalman-based methods 
introduced in the two previous chapters apply outside the \index{linear-Gaussian setting}linear-Gaussian 
setting and are built by approximating the predictive distribution using
a \index{Gaussian!ansatz}Gaussian ansatz, and then applying the Kalman formulae for the \index{analysis}analysis 
step. The \index{particle filter!bootstrap}bootstrap particle filter approximates the predictive distribution by a sum of \index{Dirac}Dirac masses and, using this structure, exactly solves the 
\index{analysis}analysis step. Thus, both Kalman-based methods (with linear \index{observation}observations) and 
the \index{particle filter!bootstrap}bootstrap particle filter use exact application of \index{Bayes formula}Bayes formula, but
with approximate \index{prior}priors found by approximating the outcome of the
\index{prediction}prediction step. However, whilst  Kalman-based methods use an approximation
that is only valid for problems which are close to \index{Gaussian}Gaussian, 
\index{particle filter}particle filters have the potential of recovering an accurate approximation to the \index{filtering!distribution}filtering distribution in nonlinear, \index{non-Gaussian}non-Gaussian settings provided that the number of particles is large enough. 
 However, an  important disadvantage of particle filters is that they tend to struggle in high-dimensional problems for practically implementable particle numbers. 
In contrast, Kalman-based methods are robust, but harder to interpret in
a rigorous fashion except for \index{linear-Gaussian setting}linear-Gaussian problems.

This chapter is organized as follows. We describe the problem setting in Section \ref{sec:111}. We then introduce the \index{particle filter!bootstrap}bootstrap particle filter in Section \ref{sec:112} and analyze its convergence in Section \ref{sec:113}. Section \ref{ssec:BPFdynamics} describes how the \index{particle filter!bootstrap}bootstrap particle filter can be interpreted as a \index{dynamical system}random dynamical system. We close in Section \ref{sec:115} with extensions and bibliographical remarks. 

\section{The Setting}\label{sec:111}
Let us return to the setting in which we introduced \index{filtering}filtering
and \index{smoothing}smoothing in Chapter \ref{lecture7}, with nonlinear stochastic \index{dynamics!stochastic}dynamics and nonlinear
\index{observation}observation function, namely the model
\begin{align*}
	v_{j+1} &= \Psi(v_j) + \xi_j, && \xi_j \sim \Nc(0, \Sigma) \index{i.i.d.}\text{ i.i.d.}, \\
	y_{j+1} &= h(v_{j+1}) + \eta_{j+1}, && \eta_j \sim \Nc(0, \Gamma) \index{i.i.d.}\text{ i.i.d.},
\end{align*}
with \(v_0 \sim \post_0 := \Nc(m_0, C_0)\) independent of the \index{i.i.d.}i.i.d.
sequences \(\{\xi_j\}\) and  \(\{\eta_j\}\). Here \(\Psi(\cdot)\) drives the \index{dynamics}dynamics and \(h(\cdot)\) is the \index{observation!function}observation function. 
Recall that we denote by \(Y_j = \{y_1,\hdots,y_j\}\) all the data up to time 
\(j\) and by \(\post_j\) the pdf of \(v_j | Y_j\), that is, \(\post_j = \Prob(v_j|Y_j)\). The \index{filtering}filtering problem is to determine \(\post_{j+1}\) from \(\post_j\). We may do so in two steps: first, we run forward the \index{Markov chain}Markov chain generated by the stochastic \index{dynamical system}dynamical system
(\index{prediction}prediction), and second, we incorporate the data by an application of \index{Bayes theorem}Bayes 
theorem (\index{analysis}analysis). 

For the \index{prediction}prediction step, we define the operator \(\Pred\) acting on a pdf \(\post\) as an application of a \index{Markov kernel}Markov kernel defined by
\begin{equation}
\label{eq:markov}
	(\Pred \post)(v) = \int_{\R^\Du} p(u,v) \post(u) \, du,
\end{equation}
where \(p(u,v)\) is the associated pdf of the stochastic \index{dynamics!stochastic}dynamics, so that 
\begin{equation*}
	p(u,v) = \frac{1}{\sqrt{(2\pi)^\Du {\rm det}\Sigma}} \text{exp}\left( -\frac{1}{2}|v - \Psi(u)|_{\Sigma}^2 \right).
\end{equation*}
Thus, we obtain
\begin{equation*}
	\Prob(v_{j+1}|Y_j)=\hat{\post}_{j+1} = \Pred \post_j.
\end{equation*}
We then define the \index{analysis}analysis operator \(\An_j\) acting on a pdf \(\post\) to correspond to an application of \index{Bayes theorem}Bayes theorem, namely  
$$(\An_j \post)(v) = \frac{\like_j(v) \post(v)}{\int_{\Ru} \like_j(v) \post(v) \, dv}, \quad \quad  \like_j(v) = \text{exp} \Bigl(-\frac{1}{2}|y_{j+1} - h(v)|_{\Gamma}^2 \Bigr).$$
Finally, combining the \index{prediction}prediction and \index{analysis}analysis steps, we obtain 
\begin{equation*}
	\post_{j+1} = \An_j \hat{\post}_{j+1} = \An_j \Pred \post_j. 
\end{equation*}
We now describe a way to numerically approximate, and update,
the pdfs $\post_j.$

\section{\index{particle filter!bootstrap}The Bootstrap Particle Filter}\label{sec:112}
The \index{particle filter!bootstrap}Bootstrap Particle Filter (\index{BPF}BPF) can be thought of as 
performing \index{sequential importance resampling}sequential importance resampling. Let \(S^\Sam\) be an operator acting on a pdf \(\post\) by producing an $\Sam$-samples \index{Dirac}Dirac approximation of \(\post\), that is
\begin{equation*}
	(S^\Sam \post)(u) = \sum_{\sam=1}^\Sam w_\sam \delta(u - u^{(\sam)}),
\end{equation*}
where \(u^{(1)},\hdots,u^{(\Sam)}\) are \index{i.i.d.}i.i.d samples from \(\post\) that are weighted uniformly i.e. \(w_\sam = \frac{1}{\Sam}\). Note that $S^\Sam \post=\pMC$, as introduced in Chapter \ref{lecture5}, equation \eqref{eq:indeed}. We will use the operator $S^\Sam$ to approximate 
the measure produced by the \index{Markov kernel}Markov kernel step $\Pred$ within the overall
\index{filtering!map} filtering map $\An_j \Pred$. 
Note that $S^\Sam$ is a {\em random} map taking pdfs into pdfs 
if we interpret weighted sums of \index{Dirac}Dirac masses as a pdf.

Let \(\post_0^\Sam = \post_0 = \Nc(m_0, C_0)\) and let $\post^\Sam_j$ denote
a particle approximation of the pdf $\post_j$ 
that we will determine in what
follows. We define
\begin{equation*}
	\hat{\post}^\Sam_{j+1} = S^\Sam \Pred \post^\Sam_j;
\end{equation*}
this is an approximation of \(\hat{\post}_{j+1}\) from the previous section. We then apply the operator \(\An_j\) to act on \(\hat{\post}^\Sam_{j+1}\) by appropriately reconfiguring the weights \(w_j\) according to the data.

To understand this reconfiguration of the weights
we use the fact that, if 
$$\post(v)=\frac{1}{\Sam} \sum_{\sam=1}^\Sam \delta(v-v^{(\sam)}),$$
then
$$(\An_j \post)(v)= \sum_{\sam=1}^\Sam w^{(\sam)} \delta(v-v^{(\sam)}),$$
where 
$$\bar{w}^{(\sam)}=\like_j(v^{(\sam)})$$
and the ${w}^{(\sam)}$ are found from the $\bar{w}^{(\sam)}$ by renormalizing
them to sum to one.
We use this calculation concerning the application of \index{Bayes formula}Bayes formula
to sums of \index{Dirac}Dirac masses within the following desired approximation of the
\index{filtering!update}filtering update formula:
\begin{equation*}
	\post_{j+1} \approx \post_{j+1}^\Sam = \An_j \hat{\post}_{j+1}^\Sam = \An_j S^\Sam \Pred \post_j^\Sam.
\end{equation*}
The steps for the method are summarized in Algorithm \ref{alg:bpf}.

\FloatBarrier
\begin{algorithm}
\caption{\index{particle filter!bootstrap}Bootstrap Particle Filter}
\begin{algorithmic}[1]
\vspace{0.1in}
\STATE {\bf Input}: Initial distribution $\post_0^\Sam = \post_0$, number of particles $\Sam.$  \\
\vspace{.04in}
\STATE {\bf Particle Generation}: For $j = 0,1,\dots,J-1$, perform
\begin{enumerate}
\item Draw $v_j^{(\sam)} \sim \post_j^\Sam$ for $\sam=1,\ldots,\Sam$ \index{i.i.d.}i.i.d.
\item Set $\hat{v}^{(\sam)}_{j+1} = \Psi(v_j^{(\sam)}) + \xi_j^{(\sam)}$ with $\xi_{j}^{(\sam)}$\index{i.i.d.} i.i.d.  $\Nc(0,\Sigma).$
\item Set $\bar{w}_{j+1}^{(\sam)} = \text{exp} \Bigl(-\frac{1}{2}|y_{j+1} - h(\hat{v}^{(\sam)}_{j+1})|_{\Gamma}^2\Bigr).$
\item Set $w_{j+1}^{(\sam)} = \bar{w}_{j+1}^{(\sam)} / \sum_{\sam=1}^\Sam \bar{w}_{j+1}^{(\sam)}.$
\item Set $\post_{j+1}^\Sam (u) = \sum_{\sam=1}^\Sam w_{j+1}^{(\sam)} \delta(u - \hat{v}^{(\sam)}_{j+1}).$
\end{enumerate}
\vspace{.04in}
\STATE{\bf Output:} Particle approximations $\post_{j}^\Sam \approx \post_j, \,\, j = 1, \ldots, J.$
\end{algorithmic}
\label{alg:bpf}
\end{algorithm}
\FloatBarrier


\section{\index{particle filter!bootstrap}Bootstrap Particle Filter Convergence}\label{sec:113}

We will now show that, under certain conditions, the \index{BPF}BPF 
converges to the true \index{filtering!distribution}filtering distribution in
the limit \(\Sam \rightarrow \infty\). The proof is similar to that of the
Lax-Equivalence Theorem from the numerical approximation of
evolution equations, part of which is the statement that
consistency and stability together imply convergence.
For the \index{BPF}BPF, consistency refers to a Monte Carlo error estimate, similar to that
derived in the chapter on \index{importance sampling}importance sampling, and stability
manifests in bounds on the
\index{Lipschitz}Lipschitz constants for the operators $\Pred$ and $\An_j.$

Our first step is to define what we mean by convergence, that is, we need 
a metric on probability measures. Notice that the operators \(\Pred\) and \(\An_j\) are deterministic, but the operator \(S^\Sam\) is random since it requires \index{sampling}sampling.
As a consequence, the approximate pdfs \(\post_j^N\) are also random. 
Thus, in fact, we need a \index{distance!between random probability measures}distance between random probability measures. To this
end, for random pdfs $\post$ and $\post'$, we define
\begin{equation*}
	d(\post, \post') = \sup_{|f|_\infty \leq 1} \Bigl( \Expect\Bigl[\bigl(\post(f) - \post'(f)\bigr)^2\Bigr] \Bigr)^{1/2},
\end{equation*}
where the expectation is taken over the random variable, in our case, 
the randomness from \index{sampling}sampling with \(S^\Sam\). 
 This \index{distance!between random probability measures}distance between random probability measures was introduced in Chapter \ref{lecture5} to study Monte Carlo integration: see equation \eqref{eq:rdist}. 
%


We now prove three lemmas, which together will enable us to prove convergence
of the \index{BPF}BPF.  The first shows consistency; the second
and third show stability estimates for $\Pred$ and $\An_j,$ respectively.

\begin{lemma}\label{lemmaSampling}
 For any pdf $\post,$ it holds that 
	\[d(\post, S^\Sam \post) \leq \frac{1}{\sqrt{\Sam}}.\]
\end{lemma}
\begin{proof} 
This is a consequence of Theorem \ref{t:MC}, since 
$S^N\post$ agrees with $\pMC$ as defined in Chapter \ref{lecture5}. 
\end{proof}

Now we prove a stability bound for the operator $\Pred$ defined in equation \eqref{eq:markov}.

\begin{lemma}\label{lemmaPrediction}
 For any pdfs $\post, \post',$ it holds that 
	\[d(\Pred \post, \Pred\post') \leq d(\post,\post').\]
\end{lemma}
	
\begin{proof}  For $|f|_\infty \le 1$  define a function $q$ on $\R^\Du$ by
\[q(v') = \int_{\R^\Du}p(v',v)f(v) \, dv,  \]
where, recall, $p$ denotes the transition pdf associated to the \index{stochastic dynamics model}stochastic dynamics model. 
Note that
\[|q(v')| \leq \int_{\R^\Du}p(v',v) \, dv = 1 ,\]
and so $|q|_\infty \le 1.$ Moreover, it holds that
\[\post(q) = (\Pred\post)(f).\]
To see this, note that by exchanging the order of integration, we have
\begin{align*}
	\post(q) &= \int_{\R^\Du}q(v')\post(v') \, dv' = \int_{\R^\Du} \Bigg[ \int_{\R^\Du}p(v',v)f(v) \, dv \Bigg] \post(v') \, dv' \\
	&= \int_{\R^\Du} \Bigg[ \int_{\R^\Du}p(v',v)\post(v') \, dv' \Bigg] f(v) \, dv \\
	&= \int_{\R^\Du} (\Pred\post)(v)f(v) \, dv.
\end{align*}
Finally, using that $|q|_\infty \le 1$ and $\post(q) = (\Pred\post)(f)$ we deduce that
\begin{align*}
	d(\Pred\post,\Pred\post') &= \sup_{|f|_\infty \leq 1} \Bigl(\Expect \Bigl[\bigl((\Pred\post)(f) - (\Pred\post')(f)\bigr)^2\Bigr]\Bigr)^{1/2} \\
	&\leq \sup_{|q|_\infty \leq 1} \Bigr( \Expect  \Bigl[\bigl(\post(q) - \post'(q)\bigr)^2\Bigr]\Bigr)^{1/2} \\
	&= d(\post,\post').
\end{align*}
\end{proof}

To prove the next lemma and the main convergence
theorem of the \index{BPF}BPF below,
we will make the following assumption, which encodes the idea of a bound
on the \index{observation!function}observation function.

\begin{assumption}
\label{a:LE}
There exists \(\kappa \in (0,1)\) such that, for all  $v\in\R^\Du$ and  $j \in \{0,\hdots,J-1\},$
\begin{equation*}
	\kappa \leq \like_j(v) \leq \kappa^{-1}.
\end{equation*}
\end{assumption}

It may initially appear strange to use the same constant $\kappa$ in the
upper and lower bounds, but recall that $\like_j$ is undefined up to a multiplicative
constant. Consequently, given any upper and lower bounds, $\like_j$ can be
scaled to achieve the bound as stated. Relatedly, it is $\kappa^{-2}$
which appears in the stability constant in the next lemma; 
if $\like_j$ is not scaled to produce the same constant $\kappa$ in the
upper and lower bounds in  
Assumption \ref{a:LE}, then it is the ratio of the upper and lower
bounds which would appear in the stability bound.

\begin{lemma} \label{lemma:analysis}
Let Assumption \ref{a:LE} hold. Then, for all pdfs $\post, \post'$ and $j \in \{0,\hdots,J-1\},$ it holds that
	\[d(\An_j \post,\An_j \post') \leq \frac{2}{\kappa^2} d(\post,\post').\]
\end{lemma}

\begin{proof}
To ease the notation, we drop the $j$ subscripts in $\like_j$ and $\An_j.$ Let $|f|_\infty \le 1.$ Following the proof of Theorem \ref{IS} in Chapter \ref{lecture5}, we use the following identity:
\begin{align*}
	(\An\post)(f) - (\An\post')(f) &= \frac{\post(f\like)}{\post(\like)} - \frac{\post'(f\like)}{\post'(\like)} \\
	&= \frac{\post(f\like)}{\post(\like)} - \frac{\post'(f\like)}{\post(\like)} + \frac{\post'(f\like)}{\post(\like)} - \frac{\post'(f\like)}{\post'(\like)} \\
	&= \frac{1}{\kappa} \bigg( \frac{\post(\kappa f\like) - \post'(\kappa f\like)}{\post(\like)} + \frac{\post'(f\like)}{\post'(\like)} \frac{\post'(\kappa \like) - \post(\kappa \like)}{\post(\like)} \bigg).
\end{align*}
Applying \index{Bayes theorem}Bayes Theorem \ref{t:bayes} and using that $|f|_\infty \le 1$ gives
\[\Bigl|\frac{\post'(f\like)}{\post'(\like)}\Bigr| = |(\An\post')(f)| \leq 1. \]
Therefore,
\[\bigl|(\An\post)(f) - (\An\post')(f)\bigr| \leq \frac{1}{\kappa^2}\Bigl( \bigl|\post(\kappa f\like) - \post'(\kappa f\like)\bigr| + \bigl|\post'(\kappa \like) - \post(\kappa \like)\bigr|   \Bigr).  \]
It follows that
\[\Expect \Bigl[\bigl( (\An\post)(f) - (\An\post')(f)\bigr)^2\Bigr] \leq  \frac{2}{\kappa^4}  \Bigl(\Expect \Bigl[(\post(\kappa f\like) - \post'(\kappa f\like))^2\Bigr] + \Expect\Bigl[\bigl(\post'(\kappa \like) - \post(\kappa \like)\bigr)^2\Bigr] \Bigr).  \]
Since $|\kappa \like| \le 1,$ we find that
\[\sup_{|f|_\infty \leq 1} \Expect\Bigl[\bigl((\An\post)(f) - (\An\post')(f)\bigr)^2\Bigr] \leq \frac{4}{\kappa^4} \sup_{|f|_\infty \leq 1} \Expect\Bigl[\bigl(\post(f) - \post'(f)\bigr)^2\Bigr],  \]
and hence
\[d(\An \post,\An\post') \leq \frac{2}{\kappa^2} d(\post,\post'). \]
\end{proof}


\begin{theorem}[Convergence of the \index{BPF}BPF] \label{th:convergenceBPF}
Let Assumption \ref{a:LE} hold.
Then there exists a $c = c(J,\kappa)$ independent of $N$ such that, for all $j = 1, \ldots, J,$
	\[d(\post_j,\post_j^\Sam) \leq \frac{c}{\sqrt{\Sam}}. \]
\end{theorem}

\begin{proof}
Let $e_j = d(\post_j,\post_j^\Sam)$. Using the triangle inequality,
\begin{align*}
	e_{j+1} &= d(\post_{j+1},\post_{j+1}^\Sam) = d(\An_j\Pred\post_j,\An_jS^\Sam \Pred\post_j^\Sam) \\
	&\leq d(\An_j\Pred\post_j  , \An_j\Pred\post_j^\Sam) + d(\An_j\Pred\post_j^\Sam,\An_jS^\Sam \Pred\post_j^\Sam) .
\end{align*}
Applying the stability bound for $\An_j$, 
we have
\[e_{j+1} \leq \frac{2}{\kappa^2}\Bigl[d(\Pred\post_j   , \Pred\post_j^\Sam) + d(\hat{\post}_j^\Sam,S^\Sam\hat{\post}_j^\Sam)\Bigr],\]
where $\hat{\post}_j^\Sam = \Pred\post_j^\Sam$. By the stability bound for $\Pred$,
\[d(\Pred\post_j  , \Pred\post_j^\Sam) \leq d(\post_j, \post_j^\Sam) \]
and by the consistency bound for $S^\Sam,$ 
\[d(\hat{\post}_j^\Sam,S^\Sam \hat{\post}_j^\Sam) \leq \frac{1}{\sqrt{\Sam}}.\]
Therefore,
\begin{align*}
	e_{j+1} &\leq \frac{2}{\kappa^2}\Bigl(d(\post_j,\post_j^\Sam) + \frac{1}{\sqrt{\Sam}}\Bigr) \\
	&\leq \frac{2}{\kappa^2}\Bigl(e_j + \frac{1}{\sqrt{\Sam}}\Bigr).
\end{align*}
We let $\lambda = 2/\kappa^2$ and note that \(\lambda \geq 2\) since \(\kappa \in (0,1]\). Then the \index{discrete Gronwall inequality}discrete Gronwall inequality of Theorem \ref{t1.6} gives
\begin{equation*}
	e_{j} \leq \lambda^je_0 + \frac{\lambda}{\sqrt{\Sam}}\frac{1-\lambda^j}{1-\lambda}.
\end{equation*}
Recall that \(\post_0^\Sam = \post_0\) hence \(e_0 = 0\). Thus, letting
\[c = \frac{\lambda(1 - \lambda^J)}{1 - \lambda}\]
completes the proof since \(\lambda(1 - \lambda^j)/(1 - \lambda)\) is increasing in \(j\).
\end{proof}

\section{\index{particle filter!bootstrap}The Bootstrap Particle Filter as a Random \index{dynamical system}Dynamical System}\label{ssec:BPFdynamics}

A nice interpretation of the \index{BPF}BPF is to view it as a random \index{dynamical system}dynamical system for a set of interacting
particles $\{v_j^{(\sam)}\}_{n=1}^N.$ 
To this end, a measure 
\begin{equation*}
	\bar{\post}_j^{\Sam}(u) = \frac{1}{\Sam}\sum_{\sam=1}^\Sam \delta(u-v_j^{(\sam)}) \approx \post_j^\Sam(u) \approx \post_j(u)
\end{equation*}
with equally weighted particles may be naturally defined after the \index{resampling}resampling step from $\post_j^\Sam.$
It can then be seen that the \index{BPF}BPF updates the particle positions 
\begin{equation*}
\{v_j^{(\sam)} \}_{\sam = 1}^\Sam  \mapsto \{v_{j+1}^{(\sam)} \}_{\sam = 1}^\Sam 
\end{equation*}
via the random map 
\begin{align*}
	\hat{v}_{j+1}^{(\sam)} & = \Psi(v_j^{(\sam)}) + \xi_j^{(\sam)}, && \xi_j^{(\sam)} \sim \Nc(0,\Sigma) \index{i.i.d.}\index{i.i.d.}\text{ i.i.d.,}\\
	v_{j+1}^{(\sam)} & = \sum_{m=1}^\Sam \mathbbm{1}_{I_{j+1}^{(m)}}\left(r_{j+1}^{(\sam)}\right) \hat{v}_{j+1}^{(m)}, && r_{j+1}^{(\sam)} \sim \text{Uniform}(0,1)\index{i.i.d.}\text{ i.i.d.}
\end{align*}
Here the supports $I_{j}^{(m)}$ of the indicator functions have widths
given by the weights appearing in  $ {\post}_j^{\Sam}(u).$ Specifically, we have  
\begin{align*}
	I_{j+1}^{(m)} & = \left[\alpha_{j+1}^{(m-1)},\alpha_{j+1}^{(m)}\right), &  \alpha_{j+1}^{(m+1)} = \alpha_{j+1}^{(m)}+w_{j+1}^{(m)}, \quad \alpha_{j+1}^{(0)} = 0.
\end{align*}
Note that, by construction, $\alpha_j^{(\Sam)} = 1.$

Thus, the underlying \index{dynamical system}dynamical system on particles comprises $\Sam$ particles
governed by two steps: (i) the underlying \index{stochastic dynamics model}stochastic dynamics model, in
which the particles do not interact; (ii) a \index{resampling}resampling of the resulting collection of particles, to reflect the different weights associated with them, in
which the particles do then interact. The interaction is driven by the
weights, which see all the particle positions and measure their goodness
of fit to the data. Note that the same particle may be replicated more than once
through the \index{resampling}resampling in (ii) and, relatedly, a particle may disappear
through the \index{resampling}resampling.


%

\section{Discussion and Bibliography}\label{sec:115}

\index{particle filter}Particle filters are overviewed from an algorithmic viewpoint
in \cite{doucet2001introduction,doucet2000sequential}, 
and from a more mathematical perspective in \cite{del2004feynman,chopin2020introduction}. A variety of ways to resample the weights are reviewed and compared in \cite{chopin2020introduction}.
The convergence of \index{particle filter}particle filters is addressed
in \cite{crisan1998discrete}; the clean proof presented here
originates in \cite{rebeschini2015can} and may also
be found in \cite{law2015data}. We refer to \cite{CD02} for a review paper on convergence results for \index{particle filter}particle filters. 
For problems in which the \index{dynamics}dynamics evolve in relatively low-dimensional 
spaces they have been enormously successful. However, \index{particle filter}particle filters often perform poorly in high-dimensional systems due to the fact that the particle weight typically concentrates on one,  or
a small number, of particles --- the phenomenon of weight collapse; see 
\cite{bickel2008sharp, snyder2008obstacles, Snyder2011}.
Generalizing them so that they work for the high-dimensional problems that arise,
for example, in geophysical applications, provides a major challenge
\cite{van2015nonlinear}.
This fact also motivates the widespread adoption of the EnKF in
the geophysical sciences -- despite the relative paucity of theoretical
justification, in comparison with the particle filter,
the EnKF automatically avoids weight collapse since all particles are
equally weighted.

\chapter{\Large{\sffamily{Optimal Particle Filter }}}
\label{ch12}

This chapter is devoted to the \index{particle filter!optimal}Optimal Particle Filter (\index{OPF}OPF). 
Like the
Bootstrap Particle Filter\index{particle filter!bootstrap} (\index{BPF}BPF)  from the previous chapter, the \index{OPF}OPF
approximates the \index{filtering!distribution}filtering distribution by a sum of \index{Dirac}Dirac masses. But while the \index{BPF}BPF is conceptually derived
by factorizing the update of the \index{filtering!distribution}filtering distribution into a \index{prediction}prediction and an \index{analysis}analysis step, the \index{OPF}OPF uses a different
factorization which can result in improved performance. 



We introduce the decomposition of the \index{filtering!update}filtering update used by the \index{OPF}OPF in Section \ref{sec:121}.
The setting will initially be the same as for the \index{BPF}BPF 
(nonlinear \index{dynamics!stochastic}stochastic dynamics and nonlinear \index{observation}observations), and in this
general setting we will prove a convergence result, similar to that
for the \index{BPF}BPF from the previous chapter. However, we will see that 
the \index{OPF}OPF cannot be implemented in the fully nonlinear case without additional
approximate \index{sampling}sampling. 
For this reason, we will specify in Section \ref{sec:122} to the case of linear \index{observation}observations, 
where the \index{OPF}OPF can be  implemented 
in a straightforward fashion, without additional approximate
\index{sampling}sampling; indeed we will see that in this setting the method
may be characterized as a set of interacting
\index{3DVAR}3DVAR filters. Section \ref{sec:123}  discusses the sense in which the 
\index{OPF}OPF has desirable properties in comparison with the \index{BPF}BPF. 
We close in Section \ref{sec:125} with bibliographical remarks.

\section{The \index{particle filter!bootstrap}Bootstrap and \index{particle filter!optimal}Optimal Particle Filters Compared}\label{sec:121}
We initially work in the setting in which we introduced \index{filtering}filtering
and \index{smoothing}smoothing in Chapter \ref{lecture7}, with nonlinear \index{dynamics!stochastic}stochastic dynamics and nonlinear
\index{observation!function}observation function, namely the model
\begin{align*}
        v_{j+1} &= \Psi(v_j) + \xi_j, && \xi_j \sim \Nc(0, \Sigma) \index{i.i.d.}\text{ i.i.d.}, \\
        y_{j+1} &= h(v_{j+1}) + \eta_{j+1}, && \eta_j \sim \Nc(0, \Gamma) \index{i.i.d.}\text{ i.i.d.},
\end{align*}
with \(v_0 \sim \post_0 := \Nc(m_0, C_0)\) independent of the \index{i.i.d.}i.i.d.
sequences \(\{\xi_j\}\) and  \(\{\eta_j\}\). Here \(\Psi(\cdot)\) drives the \index{dynamics}dynamics and \(h(\cdot)\) is the \index{observation!function}observation function.
Recall that we denote by \(Y_j = \{y_1,\hdots,y_j\}\) all the data up to time
\(j\) and by \(\post_j\) the pdf of \(v_j | Y_j\), that is, \(\post_j = \Prob(v_j|Y_j)\). The \index{filtering}filtering problem is to determine \(\post_{j+1}\) from \(\post_j\).

The fundamental \index{filtering}filtering problem that we are interested in is thus 
determination of $\Prob(v_{j+1}|Y_{j+1})$ from $\Prob(v_{j}|Y_{j}).$ 
The \index{BPF}BPF is based on applying \index{sampling}sampling to the outcome of
the following manipulation:
\begin{align*}
	\Prob(v_{j+1}|Y_{j+1}) & = \Prob(v_{j+1}|y_{j+1}, Y_j)\\
	& = \frac{\Prob(y_{j+1}|v_{j+1}, Y_j)\Prob(v_{j+1}|Y_j)}{\Prob(y_{j+1}|Y_j)}\\
	& =\frac{\Prob(y_{j+1}|v_{j+1}, Y_j)}{\Prob(y_{j+1}|Y_j)}     \int_{\Ru}\Prob(v_{j+1}|v_j,Y_j)\Prob(v_j|Y_j)\, dv_j\\
	&=\frac{\Prob(y_{j+1}|v_{j+1})}{\Prob(y_{j+1}|Y_j)}       \int_{\Ru}\Prob(v_{j+1}|v_j)\Prob(v_j|Y_j) \, dv_j\\
	&= \An_j \Pred \Prob(v_j|Y_j).
\end{align*}
The \index{Markov kernel}Markov kernel $\Pred$ acts on arbitrary density $\post$ by
\begin{equation*}
	\Pred \post(v_{j+1}) = \int_{\Ru}\Prob(v_{j+1}|v_j)\post(v_j) \, dv_j,
\end{equation*}
and $\An_j$ acts on an arbitrary density $\post$ by application of \index{Bayes theorem}Bayes theorem, taking into account the \index{likelihood}likelihood of the data
\begin{equation*}
	\An_j\post(v_{j+1}) = \frac{1}{Z}\Prob(y_{j+1}|v_{j+1})\post(v_{j+1}),
\end{equation*}
with $Z$ normalization to a probability density.
The above manipulations are summarized by the relationship
\begin{equation}
\label{eq:pf}
\post_{j+1}=\An_j \Pred\post_j.
\end{equation}
Note that in this factorization we apply a \index{Markov kernel}Markov kernel and then 
\index{Bayes theorem}Bayes theorem.  In contrast, to derive the \index{OPF}OPF we perform the following manipulation: 
\begin{align*}
	\Prob(v_{j+1}|Y_{j+1}) & = \int_{\Ru}\Prob(v_{j+1},v_j|Y_{j+1})  \, dv_j \\
	& = \int_{\Ru}\Prob(v_{j+1}|v_j,Y_{j+1})\Prob(v_j|Y_{j+1}) \, dv_j  \, \\
& = \int_{\Ru}\Prob(v_{j+1}|v_j,y_{j+1},Y_j)\Prob(v_j|y_{j+1},Y_j) \, dv_j\\
	& = \int_{\Ru} \Prob(v_{j+1}|v_j,y_{j+1})\Prob(v_j|y_{j+1},Y_{j}) \, dv_j\\
	& = \int_{\Ru}  \Prob(v_{j+1}|v_j,y_{j+1})  \frac{\Prob(y_{j+1}|v_j,Y_{j}) }{\Prob(y_{j+1}|Y_{j})}  \Prob(v_j|Y_{j}) \,dv_j\\
& = \int_{\Ru} \Prob(v_{j+1}|v_j,y_{j+1})  \frac{  \Prob(y_{j+1}|v_j) }{\Prob(y_{j+1}|Y_{j})}  \Prob(v_j|Y_{j}) \, dv_j \\ 
	& = \Pred^{OPF}_j \An_j^{OPF}\Prob(v_j|Y_{j}),
\end{align*}
with \index{Markov kernel}Markov kernel for particle update 
\begin{equation*}
	\Pred_j^{OPF}\post(v_{j+1}) = \int_{\Ru}\Prob(v_{j+1}|v_j,y_{j+1})\post(v_j) \, dv_j
\end{equation*}
and application of \index{Bayes theorem}Bayes theorem to include the \index{likelihood}likelihood
\begin{equation*}
	\An_j^{OPF}\post(v_j) = \frac{1}{Z}\Prob(y_{j+1}|v_j)\post(v_j).
\end{equation*}
Thus, we have
\begin{equation}
\label{eq:opf}
\post_{j+1}=\Pred_j^{OPF} \An_j^{OPF}\post_j.
\end{equation}
Note that in the factorization given by \index{OPF}OPF we apply \index{Bayes theorem}Bayes theorem and then a \index{Markov kernel}Markov kernel, the opposite order to the \index{BPF}BPF. Moreover, the propagation mechanism is different ---it sees the data through
the \index{Markov kernel}Markov kernel $\Pred_j^{OPF}$--- and hence the weighting 
of the particles is also different:  the \index{BPF}BPF weights are proportional to the \index{likelihood}likelihood $\Prob(y_{j+1}|v_{j+1})$ and the \index{OPF}OPF weights are proportional to  $\Prob(y_{j+1}|v_{j})$ which may be, in general, not available in closed form.  In the \index{BPF}BPF, the evolution of the particles and the observation of the data are kept separate from each other ---the \index{Markov kernel}Markov kernel $\Pred$ depends only on the \index{dynamics}dynamics and not the observed data and is thus independent of $j$.
Furthermore, \index{sampling}sampling from the \index{Markov kernel}Markov kernel $\Pred_j^{OPF}$ may
not be possible and may require further approximation.
In the next subsection we will see that these two issues may be overcome
when the \index{observation!function}observation function is linear, and 
particle updates use a \index{3DVAR}3DVAR procedure. 
However, in the remainder of this subsection we study particle approximations
of \eqref{eq:opf},  simply assuming that the OPF weights are computed
exactly and that $\Pred_j^{OPF}$ can be sampled from without approximation. 

The natural particle approximation of \eqref{eq:opf}, generalizing the
\index{BPF}BPF from the preceding chapter, is to consider the iteration
$$\post_{j+1}^\Sam =\Pred_j^{OPF}S^\Sam \An_j^{OPF} \post_j^\Sam,  \quad \post_0^N=S^\Sam \post_0. $$
We refer to this as the \index{OPF}OPF.
It is possible to show that, under suitable assumptions, the \index{OPF}OPF satisfies a convergence result analogous to Theorem \ref{th:convergenceBPF} for the \index{BPF}BPF. Here we will analyze a slight modification of the \index{OPF}OPF, called the \index{particle filter!Gaussianized optimal}Gaussianized Optimal Particle Filter (\index{GOPF}GOPF), which reorders the \index{resampling}resampling and propagation steps.  We first write the resulting algorithm and then establish a convergence result.

The \index{GOPF}GOPF satisfies the recursion
$$\post_{j+1}^\Sam =S^\Sam \Pred_j^{OPF}\An_j^{OPF} \post_j^\Sam,  \quad \post_0^N=S^\Sam \post_0. $$
This recursion is similar in spirit to the one we derived for the \index{BPF}BPF, but note that the order of the analysis, \index{sampling}sampling and \index{prediction}prediction steps is different. Our goal now is to show a convergence result for the \index{GOPF}GOPF. We will make the following assumption, which is analogous to Assumption  \ref{a:LE} in Chapter 11 for the Bootstrap filter. 

\begin{assumption}\label{assumption:like}
There exists $\kappa \in (0,1)$ such that, for  all $v_j \in \R^\Du,$ and for all $j \in \{ 0, \ldots, J-1 \}$, 
\begin{equation*}
\kappa \le \Prob( y_{j+1} |v_j)\le \kappa^{-1}. 
\end{equation*}
\end{assumption}

We are ready to establish a convergence result for the \index{GOPF}GOPF.  The proof employs the same  \index{distance!between random probability measures}distance \eqref{eq:rdist}  between random probability measures used in Theorem \ref{th:convergenceBPF} to establish convergence for the \index{BPF}BPF and in Chapter \ref{lecture5} to study \index{Monte Carlo}Monte Carlo and \index{importance sampling}importance sampling. 

\begin{theorem}[Convergence of \index{GOPF}GOPF]
Let Assumption \ref{assumption:like} hold. Then there is a $c = c(J,\kappa)$ independent of $N$ such that, for all $j = 1\ldots, J,$
\begin{equation*}
d(\pi_j, \pi_j^N) \le\frac{c}{\sqrt{N}}
\end{equation*}
\end{theorem} 

\begin{proof}
Let $e_{j}  = d(\pi_j, \pi_j^N).$ Then,
\begin{align*}
e_{j+1} &= d(\pi_{j+1}, \pi_{j+1}^N) \\
&= d( \Pred_j^{OPF}   \An_j^{OPF}  \pi_j,     S^N \Pred_j^{OPF}   \An_j^{OPF}  \pi_j^N) \\
& \le  d( \Pred_j^{OPF}   \An_j^{OPF}  \pi_j,     \Pred_j^{OPF}   \An_j^{OPF}  \pi_j^N) +  d( \Pred_j^{OPF}   \An_j^{OPF}  \pi_j^N,   S^N   \Pred_j^{OPF}   \An_j^{OPF}  \pi_j^N) \\
& \le \frac{2}{\kappa^2}  e_j + \frac{1}{\sqrt{N}},
\end{align*}
where we have used Lemmas \ref{lemmaSampling}, \ref{lemmaPrediction} and Lemma \ref{lemma:analysis}, replacing Assumption \ref{a:LE}  by Assumption \ref{assumption:like} in order to guarantee the stability of $\An^{OPF}.$
The rest of the proof is identical to that of Theorem \ref{th:convergenceBPF}.
\end{proof}


\section{Implementation: Linear \index{observation}Observation Setting }\label{sec:122}
In general it is not possible to implement the \index{OPF}OPF in the fully nonlinear setting because of two computational bottlenecks:
\begin{itemize}
\item There may not be a closed formula for evaluating the \index{likelihood}likelihood $\Prob(y_{j+1}|v_{j}),$ making unfeasible the computation of the particle weights. 
\item It may not be possible to sample from the \index{Markov kernel}Markov kernel $\Prob(v_{j+1}|v_j,y_{j+1}),$ making unfeasible the propagation of particles. 
\end{itemize}
However, when the \index{observation!function}observation function $h(\cdot)$ is linear, i.e. $h(\cdot) = H\cdot$ for some $H \in R^{\Dy \times \Du},$ both bottlenecks are overcome.
We thus consider the following setting, which arises in many applications:
\begin{align*}
	v_{j+1} & = \Psi(v_j) + \xi_j, && \xi_j \sim \Nc(0,\Sigma) \index{i.i.d.}\text{ i.i.d.}, \\
	y_{j+1} & = Hv_{j+1} +\eta_{j+1}, && \eta_{j} \sim \Nc(0,\Gamma)\index{i.i.d.}\text{ i.i.d.}, 
\end{align*}
with $v_0 \sim \Nc(m_0,C_0)$ and $v_0, \{\xi_j\}, \{\eta_j\}$ independent. 
First, note that combining the \index{dynamics!stochastic} stochastic dynamics and \index{data model}data models we may write 
\begin{equation*}
	y_{j+1} = H\Psi(v_j) + H\xi_j + \eta_{j+1},
\end{equation*}
which shows that the conditional distribution for $y_{j+1}$ given $v_j$ is
\begin{equation*}
	\Prob(y_{j+1} | v_j) = \Nc(H\Psi(v_j),S),
\end{equation*}
where $S= H\Sigma H^\top + \Gamma.$
We will use this formula to compute the weights, thus overcoming the first computational bottleneck.  

We now show that, under the linear \index{observation}observation assumption, $\Pred^{OPF}_j$ is a \index{Gaussian}Gaussian kernel, and hence the second computational bottleneck is overcome too. 
We have
\begin{align*}
	\Prob(v_{j+1}|v_j,y_{j+1}) &\propto \Prob(y_{j+1}|v_{j+1}, v_j)\Prob(v_{j+1}|v_j)\\ 
	&= \Prob(y_{j+1}|v_{j+1})\Prob(v_{j+1}|v_j)\\
	&\propto \exp \left(-\frac{1}{2}|y_{j+1}-Hv_{j+1}|^2_\Gamma -\frac{1}{2}|v_{j+1}-\Psi(v_{j})|^2_\Sigma \right)\\
	&=\exp\bigl(-\J_{\textsc{opt}}(v_{j+1})\bigr).
\end{align*}
 This is a \index{Gaussian}Gaussian distribution for $v_{j+1}$ as 
$$\J_{\textsc{opt}}(v_{j+1}):=\frac{1}{2}|y_{j+1}-Hv_{j+1}|^2_\Gamma +\frac{1}{2}|v_{j+1}-\Psi(v_{j})|^2_\Sigma$$ is quadratic 
with respect to $v_{j+1}$.\footnote{$\J_{\textsc{opt}}$ is identical to 
$\J$ on the right-hand side of Table \ref{3DVAR:t1}, with $\hat{C}$ replaced by $\Sigma.$} Consequently, we can compute the mean $m_{j+1}$ and covariance $C$ 
(which, note, is independent of $j$) of this \index{Gaussian}Gaussian by matching the mean
and quadratic terms in the relevant quadratic forms: 
\begin{align*}
	C^{-1} &= H^\top\Gamma^{-1}H + \Sigma^{-1},\\
	C^{-1}m_{j+1} &= \Sigma^{-1}\Psi(v_j) + H^\top\Gamma^{-1}y_{j+1}.
\end{align*}
Then $\Prob(v_{j+1}|y_{j+1}, v_j)=\Nc(m_{j+1},C).$ 
This is hence a special case of \index{3DVAR}3DVAR in which the
\index{analysis}analysis covariance is fixed at $C$; note that when we derived \index{3DVAR}3DVAR we
fixed the predictive covariance $\hat{C}$ which, here, is fixed at $\Sigma.$ 
As with the \index{Kalman filter}Kalman filter, and with \index{3DVAR}3DVAR, it is possible to implement 
the \index{prediction}prediction step through the following mean and covariance formulae
which avoid inversion in \index{state-space}state-space, and require inversion only in data
space: 
\begin{align*}
	m_{j+1} & = (I-KH)\Psi(v_j) + Ky_{j+1},\\
	C & = (I-KH)\Sigma,\\
	K & = \Sigma H^\top S^{-1},\\
	S & = H \Sigma H^\top + \Gamma.
\end{align*}
Furthermore, as for \index{3DVAR}3DVAR, the inversion of $S$ need only be performed once
in a pre-processing step before the algorithm is run.
Since the expression for $\Prob(v_{j+1}|v_j,y_{j+1})$ is \index{Gaussian}Gaussian we now
have the ability to sample directly from $\Pred_j^{OPF}$.
The \index{OPF}OPF is thus given by the following update algorithm for approximations
$\post_j^\Sam \approx \post_j$ in which we generalize the notational conventions 
used in the previous chapter to formulate \index{particle filter}particle filters as random
\index{dynamical system}dynamical systems:\footnote{The notation used in step 4 for the \index{resampling}resampling step was introduced in Subsection \ref{ssec:BPFdynamics}.}

\FloatBarrier
\begin{algorithm}
\caption{\index{particle filter!optimal}Optimal Particle Filter }
\begin{algorithmic}[1]
\vspace{0.1in}
\STATE {\bf Input}: Initial distribution $\Prob(v_0) = \post_0$, number of particles $\Sam.$
\STATE {\bf Initial Sampling}: Draw $\Sam$ particles $v_0^{(\sam)} \sim \post_0$ so that $\post_0^\Sam=S^\Sam\post_0.$
\vspace{0.04in}
\STATE{\bf Subsequent Sampling} For $j = 0,1,\dots,J-1$, perform:
\begin{enumerate}
\item Set $\hat{v}_{j+1}^{(\sam)} = (I-KH)\Psi(v_{j}^{(\sam)}) + Ky_{j+1} + 		\zeta_{j+1}^{(\sam)}$ with $\zeta_{j+1}^{(\sam)}$ \index{i.i.d.}i.i.d.  $\Nc(0,C).$
\item Set $\bar{w}_{j+1}^{(\sam)} = \exp\left(-\frac{1}{2}|y_{j+1}-H\Psi(v_j^{(\sam)})|_S^2\right).$
\item Set $w_{j+1}^{(\sam)} = \bar{w}_{j+1}^{(\sam)} / \sum_{\sam=1}^\Sam\bar{w}_{j+1}^{(\sam)}.$
\item Set $v_{j+1}^{(\sam)}=\sum_{m=1}^{\Sam} {\mathbbm{1}}_{I_{j+1}^{(m)}}(r_{j+1}^{(\sam)})\hat{v}_{j+1}^{(m)}.$
\item Set $\post_{j+1}^{\Sam}(v_{j+1}) = \frac{1}{\Sam}\sum_{\sam=1}^\Sam \delta\left(v_{j+1}-{v}_{j+1}^{(\sam)}\right).$
\end{enumerate}
\vspace{0.04in}
\STATE {\bf Output}: Particle approximations $\post_{j}^\Sam \approx \post_j, \,\, j = 1, \ldots, J.$
\end{algorithmic}
\label{alg:opf}
\end{algorithm}
\FloatBarrier

The \index{GOPF}GOPF has a similar form, after a reordering of the \index{sampling}sampling and propagation
steps:
\footnote{Here again, the \index{resampling}resampling step 3 follows the notation introduced in Subsection \ref{ssec:BPFdynamics}.}

\FloatBarrier
\begin{algorithm}
\caption{\index{particle filter!Gaussianized optimal} Gaussianized Optimal Particle Filter}
\begin{algorithmic}[1]
\vspace{0.1in}
\STATE {\bf Input}: Initial distribution $\Prob(v_0) = \post_0$, number of particles $\Sam.$
\STATE {\bf Initial Sampling}: Draw $\Sam$ particles $v_0^{(\sam)} \sim \post_0$ so that $\post_0^\Sam=S^\Sam\post_0.$
\vspace{0.04in}
\STATE{\bf Subsequent Sampling} For $j = 0,1,\dots,J-1$, perform:
\begin{enumerate}
\item Set $\bar{w}_{j+1}^{(\sam)} = \exp\left(-\frac{1}{2}|y_{j+1}-H\Psi(v_j^{(\sam)})|_S^2\right).$
\item Set $w_{j+1}^{(\sam)} = \bar{w}_{j+1}^{(\sam)} / \sum_{\sam=1}^\Sam\bar{w}_{j+1}^{(\sam)}.$
\item Set $\hat{v}_{j}^{(\sam)}=\sum_{m=1}^{\Sam} {\mathbbm{1}}_{I_{j+1}^{(m)}}(r_{j+1}^{(\sam)}){v}_{j}^{(m)}.$
\item Set ${v}_{j+1}^{(\sam)} = (I-KH)\Psi(\hat{v}_{j}^{(\sam)}) + Ky_{j+1} +          \zeta_{j+1}^{(\sam)}$ with $\zeta_{j+1}^{(\sam)}$\index{i.i.d.} i.i.d.  $\Nc(0,C).$
\item Set $\post_{j+1}^{\Sam}(v_{j+1}) = \frac{1}{\Sam}\sum_{\sam=1}^\Sam \delta\left(v_{j+1}-{v}_{j+1}^{(\sam)}\right).$
\end{enumerate}
\vspace{0.04in}
\STATE {\bf Output}: Particle approximations $\post_{j}^\Sam \approx \post_j, \,\, j = 1, \ldots, J.$
\end{algorithmic}
\label{alg:opf}
\end{algorithm}
\FloatBarrier

\section{``Optimality'' of the \index{particle filter!optimal}Optimal Particle Filter}\label{sec:123}
\index{particle filter}Particle filter methods rely on approximating the target distribution by a swarm of \index{Dirac}Dirac masses; it is clear that the distribution will not be well approximated by only a small number of particles in most cases. Consequently, a performance requirement for \index{particle filter}particle filter methods is that they do not lead to degeneracy of the particles. Resampling leads to degeneracy if
a few particles have all the weights. Conversely, non-degeneracy 
may be promoted by ensuring that the weights $w_j^{(\sam)}$ are similar in magnitude, so that a small number of particles are not overly favored during the \index{resampling}resampling step. This condition can be formulated as a requirement that the variance of the weights be minimized; doing this results in the \index{OPF}OPF.

To understand this perspective,
we consider an arbitrary particle update kernel of the form 
$\pi(v_{j+1}|v_j^{(\sam)},Y_{j+1})$
and we study the resulting \index{particle filter}particle filter without \index{resampling}resampling.
It is then the case that the  unnormalized particle weights are updated according
to the formula
			\begin{equation}\label{eq:weights}
				\bar{w}^{(\sam)}_{j+1} = \bar{w}^{(\sam)}_j \frac{\Prob(y_{j+1}|v_{j+1})\Prob(v_{j+1}|v^{(\sam)}_{j})}{\pi(v_{j+1}|v_j^{(\sam)},Y_{j+1})}.
			\end{equation}

\begin{theorem}[Meaning of Optimality]
	The choice of $\Prob(v_{j+1}|v_j^{(\sam)}, y_{j+1})$ as the particle update kernel $\pi(v_{j+1}|v_j^{(\sam)},Y_{j+1})$ results in the minimal
variance of the weight $w^{(\sam)}_{j+1}$ with respect to all possible
choices of the particle
update kernel $\pi(v_{j+1}|v_j^{(\sam)},Y_{j+1}).$
\end{theorem}

\begin{proof}
We calculate the variance of the unnormalized weights (treated as random variables) $\bar{w}^{(\sam)}_{j+1}$ with respect to the transition density $\pi(v_{j+1}|v_j^{(\sam)},Y_{j+1})$ and obtain
\begin{align*}
	{\rm Var}_{\pi(v_{j+1}|v_j^{(\sam)},Y_{j+1})} &[\bar{w}^{(\sam)}_{j+1}]  = \int_{\Ru} \left(\bar{w}^{(\sam)}_{j+1}\right)^2\pi(v_{j+1}|v_j^{(\sam)},Y_{j+1}) \, dv_{j+1}\\
	 &- \left[\int_{\Ru} \bar{w}^{(\sam)}_{j+1}\pi(v_{j+1}|v_j^{(\sam)},Y_{j+1}) \, dv_{j+1} \right]^2\\
	& = \left(\bar{w}^{(\sam)}_{j}\right)^2\int_{\Ru} \frac{\left(\Prob(y_{j+1}|v_{j+1})\Prob(v_{j+1}|v^{(\sam)}_{j})\right)^2}{\pi(v_{j+1}|v_j^{(\sam)},Y_{j+1})} \, dv_{j+1}\\
	& \quad\quad\quad\quad\quad\quad\quad - \left(\bar{w}^{(\sam)}_{j}\right)^2\left[\int_{\Ru} \Prob(y_{j+1}|v_{j+1})\Prob(v_{j+1}|v^{(\sam)}_{j}) \, dv_{j+1}\right]^2\\
	& = \left(\bar{w}^{(\sam)}_{j}\right)^2\left[\int_{\Ru}\frac{\left(\Prob(y_{j+1}|v_{j+1})\Prob(v_{j+1}|v^{(\sam)}_{j})\right)^2}{\pi(v_{j+1}|v_j^{(\sam)},Y_{j+1})} \, dv_{j+1} - \Prob(y_{j+1}|v^{(\sam)}_{j})^2\right].
\end{align*}
Choosing $\pi(v_{j+1}|v_j^{(\sam)},Y_{j+1}) = \Prob(v_{j+1}|v_j^{(\sam)}, y_{j+1})$, as in the \index{OPF}OPF, we obtain
\begin{align*}
	{\rm Var}_{\Prob(v_{j+1}|v_j^{(\sam)}, Y_{j+1})}[\bar{w}^{(\sam)}_{j+1}]
	& = \left(\bar{w}^{(\sam)}_{j}\right)^2\left[\int_{\Ru}\frac{\left(\Prob(y_{j+1}|v_{j+1})\Prob(v_{j+1}|v^{(\sam)}_{j})\right)^2}{\Prob(v_{j+1}|v_j^{(\sam)}, y_{j+1})} \, dv_{j+1} - \Prob(y_{j+1}|v^{(\sam)}_{j})^2\right]\\
	& = \left(\bar{w}^{(\sam)}_{j}\right)^2\left[\Prob(y_{j+1}|v^{(\sam)}_{j})^2 - \Prob(y_{j+1}|v^{(\sam)}_{j})^2\right]\\
	& = 0.
\end{align*}
\end{proof}

\begin{remark}
Note that directly from  \eqref{eq:weights}  we can see that choosing $\pi(v_{j+1}|v_j^{(\sam)},Y_{j+1}) = \Prob(v_{j+1}|v_j^{(\sam)}, y_{j+1})$ gives the weight update
			\begin{equation*}
				\bar{w}^{(\sam)}_{j+1} = \bar{w}^{(\sam)}_j \Prob(y_{j+1} | v_j^{(n)}),
			\end{equation*}
which does not depend on the draw $v_{j+1} \sim  \Prob(v_{j+1}|v_j^{(\sam)}, y_{j+1}),$ and is deterministic given $y_{j+1}$ and $v_j^{(n)}.$
\end{remark}

\begin{remark}
The \index{OPF}OPF is optimal in the very precise sense of the
theorem. Note that no optimality criterion is asserted by
this theorem with respect to iterating the particle updates, and in particular
when \index{resampling}resampling is included. The nomenclature ``optimal'' should thus be
treated with caution.
\end{remark}

\begin{example}[Linear-Gaussian One-Step Filter]
Recall Example \ref{ExampleIS} from Chapter \ref{lecture5}. We considered a \index{linear-Gaussian setting} linear-Gaussian one-dimensional \index{inverse problem}inverse problem with \index{prior}prior $\pr(u) = \Nc(0, {\widehat{c}}\, ^2)$ and \index{likelihood}likelihood $\Prob(y| u) = \Nc( au, \gamma^2)$, and we showed that the \index{divergence!$\chi^2$}$\chi^2$ divergence between the \index{posterior}posterior $\post$ and the \index{prior}prior $\pr$ is given by
 \begin{align*}
    \zeta &= \dchi( \post \| \pr)  + 1  \\
     & = \frac{\delta^2 + 1}{ \sqrt{2\delta^2 + 1} }\exp \Bigl( \frac{\delta^2}{2\delta^2 + 1} z^2  \Bigr), \quad \quad z \sim \Nc(0,1), 
     \end{align*}  
     where $\delta^2 := a^2 {\widehat{c}}\, ^2/\gamma^2.$ It is easy to see that $\zeta$ is monotonically increasing as a function of $\delta.$ We saw in Chapter \ref{lecture5} that a large \index{divergence!$\chi^2$}$\chi^2$ divergence between \index{target distribution}target (\index{posterior}posterior) and \index{proposal!distribution}proposal (\index{prior}prior) leads to a poor approximation of the \index{target distribution}target by reweighing \index{prior}prior samples. 
  
Now we consider a scalar, \index{linear-Gaussian setting}linear-Gaussian \index{filtering}filtering step
    \begin{align*}
   v_{1} &= \alpha v_0+\xi, \quad \quad
    v_0\sim\Nc(0,c_0^2), \quad
    \xi\sim\Nc(0,\sigma^2),\\
    y_1 &=h v_{1}+\eta,\quad \quad
    \eta \sim\Nc(0,\gamma^2).
    \end{align*}
In the \index{analysis}analysis step, the \index{BPF}BPF updates the \index{prior}prior $\Prob(v_1) = \Nc (0, \alpha^2 c_0^2 + \sigma^2) $ with \index{likelihood}likelihood $\Prob(y_1|v_1) = \Nc(hv_1, \gamma^2),$ while the \index{OPF}OPF updates the \index{prior}prior $\Prob(v_0) = \Nc(0, c_0^2)$ with \index{likelihood}likelihood $\Prob(y_1|v_0) = \Nc(h\alpha v_0, h^2 \sigma^2 + \gamma^2).$ Both bootstrap and optimal \index{analysis}analysis steps reweigh samples from their respective \index{prior}priors using their given \index{likelihood}likelihoods; since in both cases the \index{prior}prior is \index{Gaussian}Gaussian and the \index{observation}observation model is linear, we are in the setting of Example \ref{ExampleIS}. Here, the \index{divergence!$\chi^2$}$\chi^2$ divergence between the \index{target distribution}target and \index{proposal}proposal for the bootstrap and optimal filters are determined by 
\begin{align*}
\delta_{\text{BPF}} &=  \frac{ h^2 \alpha^2 c_0^2 + h^2 \sigma^2 }{\gamma^2}     ,      \\
\delta_{\text{OPF}} &= \frac{h^2 \alpha^2 c_0^2}{ h^2 \sigma^2 + \gamma^2}.
\end{align*}
  Clearly, $\delta_{\text{OPF}} \le \delta_{\text{BPF}},$ which indicates that the \index{divergence!$\chi^2$}$\chi^2$ divergence between \index{target distribution}target and \index{proposal!distribution}proposal is smaller for the optimal than for the bootstrap filter. In particular, note that in the \index{small noise limit}small \index{observation!noise}observation noise limit $\gamma \to 0^+$, the \index{divergence!$\chi^2$}$\chi^2$ divergence for the bootstrap filter diverges, while for the optimal filter it remains bounded provided that $h^2 \sigma^2 >0.$ In such a small \index{observation!noise}observation noise regime, the \index{OPF}OPF is clearly advantageous over the \index{BPF}BPF. 
  Finally, it is illustrative to see that the bootstrap and optimal filter agree and $\delta_{\text{BPF}}  = \delta_{\text{OPF}}$  if there is no noise in the \index{stochastic dynamics model}stochastic dynamics model, i.e. if $\sigma^2 =0.$
    \end{example}

\section{Discussion and Bibliography}\label{sec:125}
The \index{OPF}OPF is discussed, and further references given,
in the paper \cite{doucet2000sequential}; see section IID. 
 Throughout much of this chapter we assume
\index{Gaussian}Gaussian additive noise
and linear \index{observation}observation function, in which case the \index{prediction}prediction step is tractable; the order in which the \index{prediction}prediction and \index{resampling}resampling is performed can be commuted,
leading to the distinction between
what we term the \index{GOPF}GOPF and the \index{OPF}OPF. 
The paper \cite{doucet2000sequential} discusses the general setting,
beyond that in which \index{Gaussian}Gaussian additive noise
and linear \index{observation}observation function are assumed;
the idea that the order of \index{prediction}prediction and 
\index{resampling}resampling can be commuted was observed in
the general setting in \cite{pitt1999filtering}.
The convergence of the \index{OPF}OPF is studied in \cite{johansen2008note}.
The formulation of the bootstrap and optimal particle filters as random
\index{dynamical system}dynamical systems may be found in \cite{kelly2016ergodicity}.

The performance of the \index{BPF}BPF is poor when the filtering distributions 
 are far from the predictive distributions, a situation that arises in high-dimensional or small \index{observation!noise}observation noise \index{filtering}filtering settings. In such cases, the update of the weights in the \index{analysis}analysis step of the \index{BPF}BPF results in a degenerate distribution of weights, with the largest weight being close to $1$ \cite{bickel2008sharp, snyder2008obstacles, Snyder2011}. 
 This is the issue that the \index{OPF}OPF tries to ameliorate; the papers \cite{snyder2015performance,agapiou2017importance,sanz2020bayesian} show calculations which demonstrate
the extent to which this amelioration is manifest in theory. In practice, further exploiting decay of correlations through \emph{localization} is often needed. A review of local \index{particle filter}particle filters can be found in \cite{farchi2018comparison} and the paper  \cite{rebeschini2015can} investigates, from a theoretical viewpoint, whether localization can help to beat the curse of dimension. Attempts to bridge \index{particle filter}particle filters with \index{Kalman filter!ensemble}ensemble Kalman filters to alleviate the curse of dimension include \cite{FK13,stordal2011bridging}, and the relation between the collapse of ensemble and particle methods is investigated in the paper \cite{morzfeld2017collapse}, which also emphasizes the importance of localization.

 \chapter*{\Huge{\sffamily{Exercises}}}
\label{ch:DAExercises}
\addcontentsline{toc}{chapter}{Exercises}

{\bf Exercise 1} ({\sffamily{Scalar Linear-Gaussian Dynamics}}) 
Consider the scalar stochastic dynamics model 
\begin{subequations}
\begin{align*}
v_{j+1}&=a v_j+ \xi_j, \quad \quad \xi_j\sim \Nc(0,\sigma^2)\,\,\, \text{i.i.d.},\\
v_0&\sim \Nc(m_0,c_0^2), \quad \quad v_0 \perp \{\xi_j\}.
\end{align*}
\end{subequations}
\begin{itemize}
\item (i) Show that
\begin{equation*}
v_j= a^j v_0 + \sum_{i=0}^{j-1} a^{j-i-1} \xi_i.
\end{equation*}
Why does it follow that $v_j$ is Gaussian?

\item (ii) Show that the mean and variance of $v_j$
are given by
\begin{align*}
m_j& = a^j m_0, \\
c_j^2 
&= a^{2j}c_0^2 +\sigma^2 \sum_{i=0}^{j-1} a^{2i}.
\end{align*}

\item (iii) Find explicit formulae for the maps $m_j \mapsto m_{j+1}$ and
$c_j^2 \mapsto c_{j+1}^2.$

\item (iv)
If $|a|<1$ find the limit of $m_j$ and $c_j^2$
as $j \to \infty.$ What happens if $a=-1, 1$ or if $|a|>1$? 

\end{itemize}

\bigskip

\noindent {\bf Exercise 2} ({\sffamily{Filtering and Smoothing: Scalar Linear-Gaussian Setting}}) 
Consider the scalar stochastic dynamics and observation models given by 
\begin{align*}
v_{j+1}&=a v_j+ \xi_j, \quad \quad \xi_j\sim \Nc(0,\sigma^2),  \quad \quad v_0 \sim \Nc(m_0,c_0^2), \\
y_{j+1} &= v_{j+1} + \eta_{j+1}, \quad \quad \eta_{j+1} \sim \Nc(0,\gamma^2),
\end{align*}
where \(\{\xi_j\}\) and  \(\{\eta_j\}\) are i.i.d. sequences and $v_0 \perp \{\xi_j\} \perp \{\eta_j\}.$ 

\begin{itemize}
    \item (i) Set $a = 1.25,$ $\sigma^2 = 0.5,$ $m_0 = 1,$ $c_0^2 = 1$ and $\gamma^2 =0.1.$ Generate synthetic data $\{y_j\}_{j=1}^J$ with $J = 10$ from this model following these steps: 
   \begin{enumerate}[label=(\alph*),leftmargin=2.5\parindent]
    \item  Sample $v_0^\dagger \sim \Nc(1, 1).$ 
    \item For $j=0, \ldots, 9,$ sample $v_{j+1}^\dagger \sim \Nc(1.25 v_j^\dagger, 0.5)$. 
    \item For $j=0, \ldots, 9$, sample $y_{j+1} \sim \Nc(v_{j+1}^\dagger, 0.1).$ 
    \end{enumerate}
    We interpret $\{v_j^\dagger\}_{j=0}^{10}$ as the true signal underlying the synthetic data $\{y_j\}_{j=1}^{10}$. 
    \item (ii) For the synthetic data generated above  find, 
using the Kalman filter,  the filtering distributions $\pi_j(v_j) = \mathbb{P}(v_j|y_1, \ldots, y_j)$ for $j=0, 1, \ldots, 10$. Using the Kalman smoother, find the smoothing distribution $\mathbb{P} (\{v_j\}_{j=0}^{10} | \{y_j\}_{j=1}^{10})$. 
    \item (iii) Plot, for discrete time $1\le j \le 10,$  the mean of the filtering and smoothing distributions, together with the true signal and the observations.
\end{itemize}

\noindent {\bf Exercise 3} ({\sffamily{The Pendulum Problem}}) 
The dynamics of a pendulum are characterized by the following linear system:
$$\ddot{u} + \delta{\dot{u}} + \sin(u) = 0,$$
where $u(t)$ denotes the location of the pendulum, $\dot{u}(t)$ denotes the velocity and $\delta$ is a scalar parameter.
\begin{itemize}
\item (i)  Show  that this dynamical system implies the identity: 
$$\frac{\mathrm{d}}{\mathrm{d}t} \Big[\frac{1}{2}\dot{u}^2 - \cos(u)\Big] = -\delta\dot{u}^2.$$
 Use this identity to prove  that for $\delta \geq 0$ 
the solution will not blow up in finite time. 
What happens when $\delta = 0$?
\item  (ii)  Show that the dynamical system can be equivalently expressed 
 using the following first order differential equation: 
\begin{eqnarray}
\begin{cases}
\dot{u} = w,\\
\dot{w} = -\delta{w} - \sin(u);
\end{cases}
\label{eqn:diff}
\end{eqnarray}
 let $v =( u , w)^\top$ and consider this as an equation for $v$. 
Let $\Psi(\mathsf{v})$ denote the solution of \eqref{eqn:diff} at time $t = 0.2$ with initial condition $v(0)=\mathsf{v}.$ Set  $v_0 = (u_0,w_0)^\top = (\pi/4,0)^\top$ and consider the deterministic dynamics model and observation model given by
\begin{align}
\begin{split}\label{eqn:diff2}
v_{j+1} &= \Psi(v_{j}), \\
y_{j+1} &= u_{j+1} + \eta_{j+1}, \quad \eta_j \sim \mathcal{N}(0, \gamma^2) \index{i.i.d.} \text{ i.i.d.}
\end{split}
\end{align}
Set $\gamma^2 = 0.01$ and $\delta = 0.1$.  Solve numerically the differential equation \eqref{eqn:diff} to generate solutions of \eqref{eqn:diff2} and thereby  obtain $20$ observations $\{y_j\}_{j=1}^{20}$. Plot these values. 
\item (iii) This question continues from the setting of the previous questions.
Recall that 3DVAR estimates the
state of a partially observed dynamical system from the following sequential updates: 
\begin{eqnarray*}
m_{j+1} = (I - KH)\Psi(m_j) + Ky_{j+1}.
\end{eqnarray*}
Here, the $\{y_j\}$ denote the observations and the $\{m_j\}$ the state estimates;
the map encapsulates a
tradeoff between fitting to data and respecting the dynamics. 
Note that here $H = \begin{pmatrix} 1 & 0 \end{pmatrix}$ and recall that for 3DVAR the matrix $K$ satisfies the relations:
\begin{eqnarray*}
S = H \widehat{C} H^\top + \gamma^2, ~~~ K = \widehat{C} H^\top S^{-1},
\end{eqnarray*}
with $\widehat{C} \in \mathbb{R}^{2\times{2}}$  to be specified. Consider initializing 3DVAR from $m_0 = (0,0)^\top$ and suppose that $\widehat{C}$ is chosen to be a diagonal matrix. In this problem, show that there is only one degree of freedom. Play around with this parameter to find one that gives you a small $\sum_{j = 1}^{20} |v_{j} - m_{j}|$. (This an open-ended problem; you are encouraged to  experiment.) 
\item (iv) Repeat items (ii) and (iii)  with the observations being on the state $w$ instead of $u$, e.g. $y_{j+1} = w_{j+1} + \eta_{j+1}$. What do you notice?
\item (v) Starting with $N \in\{5,20,50,100\}$ random particles, sampled from a Gaussian $\Nc(0, I_2),$  apply particle filtering and report the effective sample size after each iteration of the dynamics. What do you observe?
\item (vi) What is the advantage of particle filtering for this problem over 3DVAR? Are there any disadvantages?\end{itemize}
\bigskip

\noindent {\bf Exercise 4} ({\sffamily{Estimation of Model Parameters: the EM Algorithm}}) 
Consider \index{stochastic dynamics model}stochastic dynamics and \index{data model}data models given by 
\begin{align}
        v_{j+1} &= \Psi_\theta(v_j) + \xi_j, && \xi_j \sim \Nc(0, \Sigma) \index{i.i.d.} \text{ i.i.d.}, \label{eq:SSM1}\\
        y_{j+1} &= Hv_{j+1} + \eta_{j+1}, && \eta_j \sim \Nc(0, \Gamma) \index{i.i.d.}\text{ i.i.d.}  \label{eq:SSM2}
\end{align}
with \(v_0 \sim \Nc(m_0, C_0)   \perp \{\xi_j\} \perp \{\eta_j\} \). Here, the vector $\theta \in \R^p$ parameterizes the dynamics. 
For a given and fixed integer $J,$ set
$\V := \{v_0, \ldots, v_J\}$ and  $\Y := \{y_1, \ldots, y_J\}.$
We seek to find $\theta$ that maximizes the \index{likelihood}likelihood function of $\theta$ given the observed data $Y$:
\begin{equation}
\Prob(Y | \theta) = \int \Prob (V, Y |\theta) \, dV.
\end{equation}
Here and below $\Prob (Y | \theta)$ denotes the pdf of $Y$ given that the dynamics map $\Psi_\theta$ in \eqref{eq:SSM1} is parameterized by $\theta;$ $\Prob(V,Y |\theta)$ is defined similarly. 
 
 \begin{itemize}
\item (i) Show that the joint distribution of $V$ and $Y$ admits the characterization
\begin{align}\label{eq:joint}
\begin{split}
\hspace{-0.5cm} \log & \Prob(V,Y | \theta) 
 =  -  \frac{1}{2} \displaystyle{\sum_{j=0}^{J-1}}|y_{j+1} - Hv_{j+1}|_{\Gamma}^2 
 -  \frac{1}{2}|v_0-m_0|^2_{C_0} - \frac{1}{2}  \sum_{j=0}^{J-1} |v_{j+1} - \Psi_\theta(v_j)|^2_{\Sigma}  + c,
\end{split}
\end{align}
where $c$ is a constant independent of $V,$ $Y$ and $\theta.$

\item (ii) Show that, for any pdf $q$ with compatible support, it holds that 
\begin{equation*}
\log \Prob (Y | \theta) = \mathcal{L}(q, \theta) + \dkl \bigl( q \| \Prob(V|Y,\theta)\bigr),
\end{equation*}
where 
$$\mathcal{L}(q, \theta) := \int\log\biggl( \frac{\Prob (V, Y| \theta)}{q(V)} \biggr) \, q(V)  \, dV$$
is a lower bound for the log-likelihood $\log \Prob (Y | \theta)$ since the Kullback-Leibler divergence is non-negative.
\item (iii)
You will now derive an iterative algorithm to maximize the likelihood 
given the current iterate $\theta_\ell$. In particular, we define the new iterate $\theta_{\ell+1}$ in two steps, maximizing in turn the two components of the lower bound $\mathcal{L}(q, \theta)$:
\begin{enumerate}[label=(\alph*),leftmargin=2.5\parindent]
\item  First, show that $q_\ell(V) = \Prob(V|Y,\theta_\ell)$ maximizes the lower bound $\mathcal{L}(q, \theta_\ell)$ over pdf $q.$  
\item Second, you will obtain $\theta_{\ell + 1}$ by maximizing the lower bound  $\mathcal{L}(q_\ell, \theta)$ over $\theta.$ Show that the quantity to maximize 
is the expected value of the joint log-density $\log  \Prob (V, Y | \theta)$ with respect to $q_\ell(V) = \Prob(V|Y, \theta_\ell).$
\end{enumerate}

Combining these two steps, you have derived the Expectation Maximization \index{EM}(EM) algorithm summarized below. 

\begin{algorithm}
\caption*{\label{alg:EM}{\bf Algorithm}  Expectation Maximization}
\begin{algorithmic}[1]
\STATE {\bf Input}:  Initialization $\theta_0.$
\STATE For $\ell = 0, 1, \ldots, L-1$ do the following expectation and maximization steps:
\STATE {\bf E-Step}: Compute 
\begin{align*}
\Expect^{V \sim  \Prob(V|Y, \theta_\ell)} \Bigl[ \log \Prob(V, Y | \theta) \Bigr]= \int \log \Prob(V, Y | \theta) \Prob(V|Y, \theta_\ell) \, dV. 
\end{align*}
\vspace{-0.5cm}
\STATE{{{\bf M-Step}}}: Compute 
\begin{equation*}
\theta_{\ell+1} = \arg \max_\theta \Expect^{V \sim  \Prob(V|Y, \theta_\ell)} \Bigl[ \log \Prob(V, Y | \theta) \Bigr].
\end{equation*}
\STATE{\bf Output}: Parameter $\theta^L.$
\end{algorithmic}
\end{algorithm}
\FloatBarrier

\item (iv)
Let $\{\theta_\ell \}_{\ell = 0}^{L-1}$ be the iterates of the EM algorithm. Show that, for $0 \leq \ell \leq L-1,$ it holds that 
\begin{equation}\label{eq:decreaselikelihood}
\log \Prob(Y | \theta_\ell) \le \log \Prob (Y | \theta_{\ell + 1}).
\end{equation}

\begin{observation}
\label{rem:em}
As a consequence of \eqref{eq:decreaselikelihood} it is possible to deduce, under mild assumptions, that the iterates $\theta_\ell$ of the \index{EM}EM algorithm converge, as $\ell \to \infty,$ to a local maximizer of the \index{likelihood}likelihood function. It is important to note, however, that the expectation in the E-step and the \index{optimization}optimization in the M-step are often intractable. \index{Monte Carlo}Monte Carlo, \index{filtering}filtering, or \index{smoothing}smoothing algorithms may be employed to approximate the E-step, and \index{optimization}optimization algorithms to approximate the M-step. Such approximations can cause loss of monotonicity and convergence guarantees. 
\end{observation}
\end{itemize}


\bigskip

\noindent {\bf Exercise 5} ({\sffamily{EM Algorithm with Ensemble Kalman Filter}}) 
Consider the scalar stochastic dynamics and observation models given by 
\begin{align*}
v_{j+1}&= \theta v_j+ \xi_j, \quad \quad \xi_j\sim \Nc(0,\sigma^2),  \quad \quad v_0 \sim \Nc(m_0,c_0^2), \\
y_{j+1} &= v_{j+1} + \eta_{j+1}, \quad \quad \eta_{j+1} \sim \Nc(0,\gamma^2),
\end{align*}
where $\{\xi_j\}$ and  $\{\eta_j\}$ are \index{i.i.d.}i.i.d. sequences and
\(v_0  \perp \{\xi_j\} \perp \{\eta_j\}. \)
Generate synthetic data $\{y_j\}_{j=1}^{10}$ from this model as in Exercise 2 using parameter $\theta^\star = 1.25.$ You will derive an EM algorithm to find $\theta^\star.$ Notice that the methodology you will derive is also applicable in nonlinear settings. 
\begin{itemize}
    \item (i) Implement an ensemble Kalman filtering algorithm with $N = 100$ particles, so that given parameter $\theta_\ell$ it outputs an ensemble $\{v_j^{(n)}\}_{n=1}^{100}$ for discrete time $j = 0, 1, \ldots, 10.$  
    \item (ii) Using this ensemble Kalman filtering algorithm with parameter $\theta_\ell$, derive an (approximate) EM algorithm by setting 
    $$\theta_{\ell +1} = \arg \max_\theta \frac{1}{N} \sum_{n=1}^N \log \Prob(V^{(n)}, Y | \theta), $$
    where we define, as in \eqref{eq:joint},
      \begin{align*}
\begin{split}
\hspace{-0.5cm} \log  \Prob(V^{(n)},Y | \theta)  
 =  -  \frac{1}{2 \gamma^2} \displaystyle{\sum_{j=0}^{9}}|y_{j+1} - v_{j+1}^{(n)}|^2
 -  \frac{1}{2 c_0^2}|v_0^{(n)} - m_0|^2 - \frac{1}{2 \sigma^2}  \sum_{j=0}^{9} |v_{j+1}^{(n)} - \theta v_j^{(n)}|^2.
\end{split}
\end{align*}

    Implement this EM algorithm with an initialization $\theta_0 = 1.5$ to recover $\theta^\star.$
\end{itemize}

\bigskip

\noindent {\bf Exercise 6} ({\sffamily{Likelihood: Linear-Gaussian Setting}}) 
 Suppose that, for each $0 \le j \le J-1,$ the predictive distribution $\Prob (v_{j+1}| Y_j, \theta)$ of the stochastic dynamics and data models \eqref{eq:SSM1} and \eqref{eq:SSM2} is \index{Gaussian}Gaussian with mean $\widehat{m}_{j+1}(\theta)$ and covariance $\widehat{C}_{j+1}(\theta).$ Show that then the \index{likelihood}log-likelihood function admits the following characterization 
\begin{equation*}
\log \Prob (Y | \theta) = - \frac12 \sum_{j = 0 }^{J-1} |y_{j+1} - H\widehat{m}_{j+1}(\theta)|_{S_{j+1}(\theta)}^2 - \frac12 \sum_{j=0}^{J-1} \log \text{{det}} \bigl(S_{j+1}(\theta)\bigr) + c,
\end{equation*}
where $S_{j+1}(\theta) = H \widehat{C}_{j+1}(\theta) H^\top + \Gamma$ and $c$ is a constant independent of $\theta.$

\bigskip

\noindent {\bf Exercise 7} ({\sffamily{$\chi^2$ Divergence in the Exponential Family}}) 
 The pdf $\post_\theta(u)$ is in the exponential family $\mathcal{E}_F$ if it can be written in the form  $\post_\theta(u)=e^{\langle t(u),\theta\rangle-F(\theta)+\kappa(u)}$. We denote this distribution as $\mathcal{E}_F(\theta)$.
    Suppose the natural parameter space $\Theta=\{\theta|\int_{\mathbb{R}^d}\post_\theta(u)du=1\}$ is affine, meaning that $\sum_iw_i\theta_i\in\Theta$ if $\theta_i\in\Theta$ and $\sum_iw_i=1$. Show that the $\chi^2$ divergence within the same exponential family $\mathcal{E}_F$ is characterized by
    \[\dchi(\mathcal{E}_F \bigl(\theta_1)\|\mathcal{E}_F(\theta_2) \bigr)=e^{F(2\theta_1-\theta_2)-2F(\theta_1)+F(\theta_2)}-1.\]

\bigskip

\noindent {\bf Exercise 8} ({\sffamily{$\chi^2$ Divergence Between Gaussians}}) Recall the exponential family introduced in Exercise 7.
   \begin{itemize}
 \item (i) Let $p=\mathcal{N}(\mu,\Sigma)$ be a Gaussian on $\R^d$ with positive definite covariance matrix $\Sigma.$  Show that it belongs to the exponential family, with parameter $\theta=[\Sigma^{-1}\mu;-\frac{1}{2}\Sigma^{-1}]$ where the natural parameter space $\mathbb{R}^d\bigoplus\mathbb{R}^{d\times d}$ inherits the inner products from the Euclidean space and the matrix space\footnote{The canonical inner product in the space of square matrices is defined to be the trace of the matrix product. We can also view this as an extension of the Euclidean inner product where the scalar in each coordinate is replaced by vectors.}, $t(u)=[u;uu^\top]$, $F(\theta)=\frac{1}{2}\mu^\top \Sigma^{-1}\mu+\frac{1}{2}\log \text{det} \Sigma,$ and $\kappa(u)=-\frac{d}{2}\log(2\pi)$. 
 
\item (ii) Let $p_1=\mathcal{N}(\mu_1,\Sigma_1)$ and $p_2=\mathcal{N}(\mu_2,\Sigma_2)$ be Gaussians on $\R^d$ with positive definite covariance matrices. Show that
    \[\dchi(p_1\| p_2)=\frac{ \text{det} W}{\sqrt{\text{det}(2W-I)}}e^{w^\top \Sigma_1^{-1}(2W-I)^{-1}w} -1,\]
    where $W=\Sigma_2 \Sigma_1^{-1}$ and $w=\mu_1-\mu_2$.
    \end{itemize}
    
    \bigskip

\part{Kalman Inversion}
\chapter{\Large{\sffamily{\index{inverse problem}  Blending Inverse Problems and Data Assimilation \index{data assimilation}}}}

This chapter brings together the material in the first two  parts of these notes, demonstrating how the principles and ideas underpinning the derivation of \index{Kalman filter!extended}extended and \index{Kalman filter!ensemble}ensemble Kalman filters for \index{data assimilation}data assimilation can be used to design  ensemble Kalman methods for \index{inverse problem}inverse problems. We adopt an \index{optimization}optimization perspective to the \index{inverse problem}inverse problem and study gradient-based and  ensemble  algorithms for the minimization of two \index{objective}objective functions: a \index{data-misfit}\emph{data-misfit} \index{objective}objective defined by a \index{loss}loss function; and a \index{Tikhonov-Phillips}\emph{Tikhonov-Phillips} \index{objective}objective defined by appending the \index{loss}loss term with a \index{regularization}regularization term. These \index{objective}objective functions will be introduced in Section \ref{sec:setting}, where we also show that they are particular instances of a general family of \index{nonlinear least-squares}nonlinear least-squares objectives. Section \ref{sec:NLS} contains a short overview of Gauss-Newton and Levenberg-Marquardt  \index{optimization}optimization algorithms for \index{nonlinear least-squares}nonlinear least-squares. In Section \ref{ssec:derivative-based} we consider gradient-based extended Kalman methods for both \index{objective}objectives, highlighting their interpretation as standard \index{Gauss-Newton}Gauss-Newton and \index{Levenberg-Marquardt}Levenberg-Marquardt \index{optimization}optimization algorithms. Finally, in Section \ref{ssec:Ensemble-based} we consider ensemble Kalman methods that avoid the calculation of gradients by invoking a statistical linearization\index{statistical linearization}
defined with an ensemble of particles. The chapter closes in Section \ref{sec:discussion} with extensions and bibliographical remarks. 

\section{Problem Setting and \index{objective}Objective Functions}\label{sec:setting}
Recall the \index{inverse problem}inverse problem of finding an unknown $u \in \Ru$ from data $y \in \Ry,$ where
\begin{align}
\label{IP}
y=G(u)+\eta, \quad  \eta\sim \Nc(0,\Gamma),
\end{align}
and $G$ represents a given \index{forward!model}forward model.
We consider an \index{optimization}optimization approach to the \index{inverse problem}inverse problem, seeking to recover the unknown $u$ by minimizing a \index{data-misfit}data-misfit or a \index{Tikhonov-Phillips}Tikhonov-Phillips \index{objective}objective function defined, respectively, by 
\begin{equation}\label{eq:TPobjective}
\Jdm(u) :=  \frac12|y-G(u)|^2_\Gamma, \quad \quad \Jtp(u) := \frac12|y-G(u)|^2_\Gamma   +  \frac12|u-  \hat{m} |_{\hat{C}}^2.
\end{equation}
As discussed in Section \ref{sec:31} and the examples therein,   
the \index{data-misfit}data-misfit \index{objective}objective function $\Jdm$ 
can be interpreted as a \index{loss}loss function and minimizing it promotes fitting the given data $y$; and the
\index{Tikhonov-Phillips}Tikhonov-Phillips \index{objective}objective $\Jtp$ comprises a \index{loss}loss function appended with a \index{regularization}regularization term that helps prevent overfitting the data. While in this chapter we focus on the \index{optimization}optimization perspective, we recall that in the \index{Bayesian}Bayesian perspective the \index{regularization}regularization term can be interpreted as the negative log-density of a \index{Gaussian}Gaussian \index{prior}prior $\pr(u) = \Nc( \hat{m}, \hat{C}),$ in which case minimizing the \index{Tikhonov-Phillips}Tikhonov-Phillips \index{objective}objective is equivalent  to finding the \index{MAP estimator}MAP estimator.

The \index{data-misfit}data-misfit and \index{Tikhonov-Phillips}Tikhonov-Phillips \index{objective}objectives are examples of \index{nonlinear least-squares}nonlinear least-squares \index{objective}objectives of the general form
\begin{equation}\label{eq:obj}
	\J(u) = \frac12 |r(u)|^2.
\end{equation}
To see this, note first that 
the \index{data-misfit}data-misfit \index{objective}objective 
can be written in the form
\begin{align}\label{eq:objectiveLM}
	 \Jdm(u) = \frac12|\rdm(u)|^2, \quad \quad \rdm(u) :=\Gamma^{-1/2}\bigl(y - G(u) \bigr). 
\end{align}
 Secondly, note that  the 
\index{Tikhonov-Phillips}Tikhonov-Phillips \index{objective}objective
may be written in the form 
\begin{align*}
	\Jtp(u) &= \frac12|y-G(u)|^2_\Gamma + \frac12|u-\hat{m}|^2_{\hat{C}} \\
	&= \frac12|z-h(u)|^2_Q,
\end{align*}
where
$$z := \begin{bmatrix}
	y\\ \hat{m}
\end{bmatrix},  \quad \quad \quad 
h(u) := \begin{bmatrix}
	G(u)\\ u
\end{bmatrix},
\quad \quad \quad
Q := 
\begin{bmatrix}
	\Gamma & 0 \\ 0 & \hat{C}
\end{bmatrix}.
$$
Therefore, 
\begin{align}\label{eq:objectiveGN}
	\Jtp(u) = \frac12|\rtp(u)|^2, \quad \quad \rtp(u) :=Q^{-1/2}\bigl( z - h(u) \bigr).
\end{align}
Equations \eqref{eq:objectiveLM} and \eqref{eq:objectiveGN} show that 
 both the \index{data-misfit}data-misfit and the 
\index{Tikhonov-Phillips}Tikhonov-Phillips 
\index{objective}objectives can be written in the general form \eqref{eq:obj}.

\section{Nonlinear-Least Squares Optimization}\label{sec:NLS}
Gradient-based \index{optimization}optimization algorithms  for the
\index{nonlinear least-squares}nonlinear least-squares problem of
minimizing \eqref{eq:obj} can  be broadly classified into line-search and trust region methods, exemplified by the classical \index{Gauss-Newton}Gauss-Newton and \index{Levenberg-Marquardt}Levenberg-Marquardt algorithms, respectively. 
 We overview each of these in turn. 

\subsection{Gauss-Newton Method}\label{ssec:GN}

The \index{Gauss-Newton}Gauss-Newton method applied to the general least-squares \index{objective}objective \eqref{eq:obj} is a line-search method which, starting from an initialization $u_0,$ sets
\begin{align*}
	u_{\ell+1} = u_\ell + \alpha_\ell v_\ell, \quad \quad \ell  = 0, 1,\ldots, L-1,
\end{align*}
where $v_\ell$ is a search direction defined by

\begin{align}\label{eq:definitionJl}
	v_\ell = \arg\min_v \Jl(v), \quad \quad \Jl(v):= \frac12|Dr(u_\ell)v + r(u_\ell) |^2,
\end{align}
and $\alpha_\ell>0$ is a user-chosen step-size parameter. Here and throughout this chapter, $Dr$ will denote the Jacobian of $r,$ which here is assumed 
to exist.  However, a significant outcome of the presentation in this
chapter is the derivation of ensemble Kalman formulae for the search 
direction update, avoiding the need for the calculation of the Jacobian; these ensemble
methods can be used when the Jacobian does not exist, or is too expensive
to compute. 

Our presentation in Sections \ref{ssec:derivative-based} and 
\ref{ssec:Ensemble-based} will focus on the derivation of 
extended and ensemble Kalman formulae, respectively, for the search direction 
update. Although the choice of step-size is crucial to the
efficiency of all Gauss-Newton methods, it is not the focus of these notes.
Consequently,  we introduce algorithms viewing the number $L$ of iterations, 
and the mechanism for determining the step-size schedule 
$\{\alpha_\ell\}_{\ell = 0}^{L-1}$, as given inputs.

\begin{remark}
In practice each step-size $\alpha_\ell$ is chosen adaptively based
on the current state $u_\ell$ and search direction $v_\ell.$
A unifying idea shared by many sophisticated line search strategies is to find an interval of desirable step-sizes and then try out a sequence of candidates within that interval, stopping when certain conditions are satisfied. For instance, a simple condition is to require reduction of $\J$ in which case $\alpha_\ell$ is required to satisfy 
$$\J(u_\ell + \alpha_\ell v_\ell) < \J(u_\ell).$$ 
However, this condition is not sufficient to guarantee convergence and
motivates the stronger Armijo condition: for some constant $c_1 \in (0,1)$
$$\J(u_\ell + \alpha_\ell v_\ell)  \le \J(u_\ell) + c_1 \alpha_\ell \langle D\J(u_\ell), v_\ell \rangle.$$ 
The choice of stopping criteria and adaptive  step-sizes will be further discussed in Section \ref{sec:discussion}.
\end{remark}

\subsection{Levenberg-Marquardt Method}\label{ssec:LM}

The \index{Levenberg-Marquardt}Levenberg-Marquardt method applied to the general least-squares \index{objective}objective \eqref{eq:obj} is a trust region method which, starting from an initialization $u_0,$ sets
\begin{align*}
	u_{\ell+1} = u_\ell +  v_\ell, \quad \quad \ell  = 0, 1,\ldots, L-1,
\end{align*}
where
\begin{equation*}
	v_\ell =  \arg\min_{v} \Jl(v), \quad \text{s.t.} \,\, |v|_{\hat{C}}^2 \le \delta_\ell,  \quad \quad \Jl(v):= \frac12|Dr(u_\ell)v + r(u_\ell) |^2.
\end{equation*}
Similar to \index{Gauss-Newton}Gauss-Newton methods, the increment $v_\ell$ is defined as the minimizer of a linearized \index{objective}objective, but now the minimization is constrained to a ball $\{ |v|_{\hat{C}}^2 \le \delta_\ell\}$ in which we \emph{trust} that the \index{objective}objective can be replaced by its linearization. 
For any $\delta_\ell$ there is an $\alpha_\ell \in (0, \infty]$ such that 
\begin{equation*}
	v_\ell = \arg\min_v  \Juc \hspace{-0.2cm}(v),
\end{equation*}
where
\begin{equation}\label{eq:LMobjectivewithlengthstep}
	\Juc \hspace{-0.2cm}(v)  =  \Jl(v)   + \frac{1}{2\alpha_\ell} | v |_{\hat{C}}^2 .
\end{equation}
The parameter $\alpha_\ell\in (0,\infty]$ acts as a Lagrange multiplier and 
plays an analogous role to the  step-size  in 
\index{Gauss-Newton}Gauss-Newton methods. We study Levenberg-Marquardt 
methods from the perspective of the unconstrained minimization
problem for $\Juc$ given by \eqref{eq:LMobjectivewithlengthstep}. 
Our presentation in Sections \ref{ssec:derivative-based} and \ref{ssec:Ensemble-based} will focus on the derivation of Kalman formulae for the increments $v_\ell.$ As for Gauss-Newton methods, we view the number $L$ of iterations and 
the mechanism for determining the step-size schedule $\{\alpha_\ell\}_{\ell = 0}^{L-1}$ as inputs to the algorithms we state here.

\begin{remark} In practice the parameter $\delta_\ell$ is chosen adaptively, for instance by monitoring the ratio $$ s_\ell = \frac{\J(u_\ell) - \J(u_\ell + v_\ell)}{\mathsf{Q}(0) - \mathsf{Q}(v_\ell)},$$
where $\mathsf{Q}(v)$ is a quadratic approximation to $\J(u_\ell + v).$ If $s_\ell$ is close to $1,$ this indicates that the objective \eqref{eq:obj} can be well approximated by a quadratic in a neighborhood of size $\delta_\ell$ around $u_\ell$, and thus that the next trust region can be enlarged. On the other hand, if $s_\ell$ is small, we may shrink the trust region in the next iteration. 
The choice of stopping criteria and adaptive  step-sizes  will be  further  
discussed in Section \ref{sec:discussion}.  
\end{remark}

Note that the \index{Levenberg-Marquardt}Levenberg-Marquardt increment is the  unconstrained  minimizer of a \emph{regularized} \index{objective}objective. It is for this reason that we say that \index{Levenberg-Marquardt}Levenberg-Marquardt provides an implicit \index{regularization}regularization. This regularization helps avoid overfitting when applied to the \index{data-misfit}data-misfit \index{objective}objective which, unlike the Tikhonov-Phillips objective,
is not regularized. On the other hand, \index{Gauss-Newton}Gauss-Newton methods do not provide implicit \index{regularization}regularization and 
therefore  should not be applied to the \index{data-misfit}data-misfit \index{objective}objective when solving \index{ill-posed}ill-posed \index{inverse problem}inverse problems. We will therefore focus on gradient and ensemble methods that arise from the following three choices of \index{objective}objective function and \index{optimization}optimization algorithm (see Table \ref{Updateswithgradients}):

\begin{table}
	\begin{center}
	\bgroup
\def\arraystretch{1.5}
		\begin{tabular}{ | c | c |c|c|}
			\hline
			\index{objective}Objective & \index{optimization}Optimization & Gradient Method & Ensemble Method    \\ \hline
		 $\Jtp$ & \index{Gauss-Newton}Gauss-Newton &  \index{IExKF}IExKF & \index{IEnKF-SL}IEnKF-SL  \\ \hline
			 $\Jdm$ & \index{Levenberg-Marquardt}Levenberg-Marquardt  & \index{ExKI}ExKI  & \index{EnKI-SL}EnKI-SL    \\ \hline
			 $\Jtp$ & \index{Levenberg-Marquardt}Levenberg-Marquardt & \index{TExKI}TExKI   & \index{TEnKI-SL}TEnKI-SL  \\ \hline
		\end{tabular}
\egroup
		\caption{Summary of the algorithms considered in this chapter.}		
		\label{Updateswithgradients}
	\end{center}
\end{table}

\begin{enumerate}
\item \index{Tikhonov-Phillips}Tikhonov-Phillips and \index{Gauss-Newton}Gauss-Newton, leading to \index{Kalman filter!iterative extended}Iterative Extended and \index{Kalman filter!iterative ensemble}Iterative Ensemble Kalman Filters (\index{IExKF}IExKF and \index{IEnKF-SL}IEnKF-SL);
\item \index{data-misfit}Data-misfit and \index{Levenberg-Marquardt}Levenberg-Marquardt, leading to \index{Kalman inversion!extended}Extended and \index{Kalman inversion!ensemble}Ensemble Kalman Inversion (\index{ExKI}ExKI and \index{EnKI-SL}EnKI-SL); and
\item \index{Tikhonov-Phillips}Tikhonov-Phillips and \index{Levenberg-Marquardt}Levenberg-Marquardt, leading to \index{Kalman inversion!Tikhonov extended}Tikhonov Extended and \index{Kalman inversion!Tikhonov ensemble}Tikhonov Ensemble Kalman Inversion (\index{TExKI}TExKI and \index{TEnKI-SL}TEnKI-SL).
\end{enumerate}
 Gradient methods will be introduced in Section \ref{ssec:derivative-based} while their ensemble counterparts will be introduced in Section \ref{ssec:Ensemble-based}. 

\section{Extended Kalman Methods}\label{ssec:derivative-based}
In this section we derive closed formulae for the \index{Gauss-Newton}Gauss-Newton method applied to the \index{Tikhonov-Phillips}Tikhonov-Phillips \index{objective}objective $\Jtp$, as well as for the \index{Levenberg-Marquardt}Levenberg-Marquardt method  applied to the \index{data-misfit}data-misfit \index{objective}objective $\Jdm$ and the \index{Tikhonov-Phillips}Tikhonov-Phillips \index{objective}objective $\Jtp.$ These formulae are the basis for the ensemble methods considered in the next section. 

Since the search directions of \index{Gauss-Newton}Gauss-Newton and \index{Levenberg-Marquardt}Levenberg-Marquardt methods are found by minimizing a linearization of the \index{objective}objective, it is instructive to consider first linear least-squares \index{optimization}optimization before delving into the nonlinear setting.  The following  result characterizes the minimizer $m$ of the \index{Tikhonov-Phillips}Tikhonov-Phillips \index{objective}objective $\Jtp$ in the case of linear $G(u) =  A  u$. Note that it is a consequence of completing
the square and is derived in the \index{linear-Gaussian setting}linear-Gaussian setting for \index{inverse problem}inverse problems studied in Chapter \ref{sec: posteriordistribution};  we record it here, as it will be used extensively in this chapter. 
\begin{lemma}\label{lemma:linear}
	It holds that 
	\begin{equation}\label{eq:linearTikhonov}
		\frac12| y -  A  u|_\Gamma^2 + \frac12| u - \hat{m}|^2_{\hat{C} }
		= \frac12 | u - m|_C^2 + \beta,
	\end{equation}
	where $\beta$ does not depend on $u,$ and 
	\begin{align}
		C^{-1} &=  A^\top \Gamma^{-1}  A  + \hat{C}^{-1},  \label{eq:precision} \\ 
		C^{-1} m & =  A^\top  \Gamma^{-1} y + \hat{C}^{-1} \hat{m}. \label{eq:precisionmean}
	\end{align}
	Equivalently, 
	\begin{align}
		m &= \hat{m} + K(y -  A  \hat{m}), \label{eq:linearmean}\\ 
		C &= (I - K A)\hat{C},\label{eq:linearcovariance}
	\end{align}
	where $K$ is the \index{Kalman gain}Kalman gain matrix given by
	\begin{equation}\label{eq:kalmangain}
		K = \hat{C}  A^\top (A  \hat{C} A^\top + \Gamma)^{-1} =C  A^\top \Gamma^{-1}.
	\end{equation} 
\end{lemma}
\begin{proof}
	The  formulae  \eqref{eq:precision} and \eqref{eq:precisionmean}  
 follow  by matching linear and quadratic coefficients in $u$ between  
	\begin{equation}
		\frac{1}{2} |u-m|^2_C  \quad \quad \text{and}  \quad \quad \frac12|u-\hat{m}|^2_{\hat{C}} + \frac12|y- A  u|^2_\Gamma. 
	\end{equation}
	The formulae  \eqref{eq:linearmean} and \eqref{eq:linearcovariance} as well as the equivalent expressions for the \index{Kalman gain}Kalman gain $K$ in Equation \eqref{eq:kalmangain} can be obtained using 
 the  \index{Woodbury matrix identity}Woodbury matrix identity, Lemma \ref{eq:woodbury}.
\end{proof}

%

\subsection{\index{Kalman filter!iterative extended}Iterative Extended Kalman Filter (IExKF)}\label{ssec:GNTP}
In this subsection we introduce two ways of writing the  \index{Gauss-Newton}Gauss-Newton update applied to the \index{Tikhonov-Phillips}Tikhonov-Phillips \index{objective}objective $\Jtp$. In order to apply the \index{Gauss-Newton}Gauss-Newton method to the \index{Tikhonov-Phillips}Tikhonov-Phillips \index{objective}objective, we use \eqref{eq:objectiveGN}. The following result is a direct consequence of Lemma \ref{lemma:linear}. 
\begin{lemma}
	The \index{Gauss-Newton}Gauss-Newton method applied to the \index{Tikhonov-Phillips}Tikhonov-Phillips \index{objective}objective $\Jtp$ admits the characterizations: 
	\begin{equation}
		u_{\ell+1}  = u_\ell +  \alpha_\ell C_\ell\Bigl( G_\ell^\top \Gamma^{-1} \bigl( y - G(u_\ell) \bigr)  + \hat{C}^{-1}(\hat{m} - u_\ell)  \Bigr)\label{GNTP1},
	\end{equation}
	and 
	\begin{equation}
		u_{\ell+1}  =  u_\ell  + \alpha_\ell  \Bigl(K_\ell \bigl( y - G(u_\ell) \bigr) +  (I- K_\ell G_\ell) (\hat{m} - u_\ell)   \Bigr), \label{GNTP2} 
	\end{equation}
 where $G_\ell = DG(u_\ell)$ and 
	\begin{align*}
		K_\ell &= \hat{C} G_\ell^\top (G_\ell \hat{C}G_\ell^\top + \Gamma)^{-1},  \\
		C_\ell & = (I - K_\ell G_\ell) \hat{C}.
	\end{align*}
\end{lemma}
\begin{proof}
	The search direction $v_\ell$  of \index{Gauss-Newton}Gauss-Newton for the \index{objective}objective $\Jtp$ is given by 
	\begin{align}
		v_\ell &= \arg \min_v  \frac12 \bigl|D\rtp(u_\ell)v + \rtp(u_\ell) \bigr|^2 \\
		& = \arg \min_v   \frac12 \bigl| z - h(u_\ell) - Dh(u_\ell) v \bigr|_Q^2      \\
		& = \arg \min_v  \biggl\{  \frac12 \bigl| y - G(u_\ell) - DG(u_\ell) v \bigr|^2_\Gamma + \frac12 \bigl|v - (\hat{m} - u_\ell) \bigr|_{\hat{C}}^2 \biggr\} . \label{eq:GN_obj}
	\end{align}
	Applying Lemma \ref{lemma:linear},  using formulae \eqref{eq:precisionmean} and \eqref{eq:linearcovariance},  we deduce that
	\begin{align*}
		v_\ell =  C_\ell\Bigl( G_\ell^\top \Gamma^{-1} \bigl( y - G(u_\ell) \bigr)  + \hat{C}^{-1}(\hat{m} - u_\ell)  \Bigr),
	\end{align*}
	which establishes the characterization \eqref{GNTP1}.
	The equivalence between \eqref{GNTP1} and \eqref{GNTP2} follows from the identity \eqref{eq:kalmangain}, which implies that $C_\ell G_\ell^\top \Gamma^{-1} = K_\ell$ and $C_\ell \hat{C}^{-1} = I-K_\ell G_\ell.$ 
\end{proof}

We refer to the \index{Gauss-Newton}Gauss-Newton method applied to $\Jtp$ as the \index{Kalman filter!iterative extended}Iterative Extended Kalman Filter (\index{IExKF}IExKF) algorithm.   Discussion of how to choose the  step-sizes  adaptively can be found in Section \ref{sec:discussion}. 
\FloatBarrier
\begin{algorithm}
	\caption{\index{Kalman filter!iterative extended}Iterative Extended Kalman Filter (\index{IExKF}IExKF) \label{itexKF}}
	\begin{algorithmic}[1]
	\STATE {\bf Input}: Initialization $u_0 = \hat{m}$,  rule for choosing the step-sizes $\{\alpha_\ell\}_{\ell = 0}^{L-1}.$  \\
	\STATE For $\ell = 0, 1, \ldots, L-1$ do:
	\begin{enumerate}
		\item Set $K_\ell = \hat{C} G_\ell^\top (G_\ell \hat{C} G_\ell^\top + \Gamma)^{-1}, \quad \quad G_\ell = DG(u_\ell). $
		\item  Set 
		\begin{equation}\label{eq:updateiexKF}
			u_{\ell+1} = u_\ell + \alpha_\ell \Bigl( K_\ell \big( y - G(u_\ell) \big) + (I - K_\ell G_\ell)(\hat{m}-u_\ell)   \Bigr).
		\end{equation}
	\end{enumerate}
	\STATE{\bf Output}: $u_1, u_2, \ldots, u_L.$
	\end{algorithmic}
\end{algorithm}
\FloatBarrier
The next proposition shows that in the linear case,  if $\alpha_\ell = 1$ for all $\ell \ge 0,$ \index{IExKF}IExKF finds the minimizer of the \index{objective}objective \eqref{eq:TPobjective} in one iteration, and further iterations still 
 stay at  the minimizer.
\begin{proposition}\label{proposition:IExKF}
	Suppose that $G(u) =  A  u$ is linear and that $\alpha_\ell = 1$ for all $\ell \ge 0.$ Then the output of Algorithm \ref{itexKF} satisfies
	$$ u_\ell = m, \quad \quad   \ell  = 1, 2, \ldots $$
	where $m$ is the minimizer of the \index{Tikhonov-Phillips}Tikhonov-Phillips \index{objective}objective \eqref{eq:TPobjective}.
\end{proposition}
\begin{proof}
	In the linear case we have 
	$$G_\ell =  A , \quad \quad K_\ell = K = \hat{C} A^\top ( A  \hat{C}  
A^\top + \Gamma)^{-1}, \quad \quad \ell  = 0, 1,\ldots$$
	Therefore,  update \eqref{eq:updateiexKF} simplifies as 
	$$u_{\ell +1} = \hat{m} + K(y -  A \hat{m}), \quad \ell  = 0, 1,\ldots$$
	This implies that, for all $\ell \ge 1,$   it holds that $u_\ell=m$ with $m$ defined in \eqref{eq:linearmean}. 
\end{proof}

\subsection{\index{Kalman inversion!extended}Extended Kalman Inversion (ExKI)}\label{ssec:LMDM}
In this subsection we study the application of the \index{Levenberg-Marquardt}Levenberg-Marquardt algorithm to the \index{data-misfit}data-misfit \index{objective}objective $\Jdm$.  In order to apply the \index{Levenberg-Marquardt}Levenberg-Marquardt method to the \index{data-misfit}data-misfit \index{objective}objective $\Jdm,$  recall that this objective can be written in standard \index{nonlinear least-squares}nonlinear least-squares form:
\begin{align}
	 \Jdm(u) = \frac12|\rdm(u)|^2, \quad \quad \rdm(u) :=\Gamma^{-1/2}\bigl(y - G(u) \bigr). 
\end{align}

\begin{lemma}
	The \index{Levenberg-Marquardt}Levenberg-Marquardt method applied to the \index{data-misfit}data-misfit \index{objective}objective $\Jdm$ admits the following characterization:
	\begin{equation}
		u_{\ell+1} = u_\ell + K_\ell \Bigl(y - G(u_\ell) \Bigr),
	\end{equation}
	where $$K_\ell = \alpha_\ell \hat{C} G_\ell^\top (\alpha_\ell G_\ell \hat{C} G_\ell^\top + \Gamma)^{-1}, \quad \quad G_\ell = DG(u_\ell).$$
\end{lemma} 
\begin{proof}
	Note that the increment $v_\ell$ is defined as the unconstrained minimizer of
	\begin{align}
		\begin{split}
			\Jucdm(v) &=   \frac12|D\rdm(u_\ell)v + \rdm(u_\ell) |^2  + \frac{1}{2\alpha_\ell} |v|^2_{\hat{C}}  \\
			& = \frac12| y - G(u_\ell) - DG(u_\ell)v |^2_\Gamma + \frac{1}{2\alpha_\ell} |v|^2_{\hat{C}}.\label{eq:LMDM}
		\end{split}
	\end{align}
	The result follows from Lemma \ref{lemma:linear}. 
\end{proof}
The previous lemma motivates the following \index{Kalman inversion!extended}Extended Kalman Inversion (\index{ExKI}ExKI) algorithm.  Discussion of how to choose the  step-sizes adaptively can be found in Section \ref{sec:discussion}. 
\FloatBarrier
\begin{algorithm}
	\caption{\index{Kalman inversion!extended}Extended Kalman Inversion (\index{ExKI}ExKI) \label{ILMDM}}
	\begin{algorithmic}[1]
	\STATE {\bf Input}: Initialization $u_0 = \hat{m}$, rule for choosing the step-sizes  $\{\alpha_\ell\}_{\ell = 0}^{L-1}.$  \\
	\STATE For $\ell = 0, 1, \ldots, L-1$ do:
	\begin{enumerate}
		\item Set  $K_\ell = \alpha_\ell \hat{C} G_\ell^\top (\alpha_\ell G_\ell \hat{C} G_\ell^\top + \Gamma)^{-1}, \quad \quad G_\ell = DG(u_\ell).$
		\item Set 
		\begin{equation}\label{eq:ILMDM}
			u_{\ell +1} = u_\ell + K_\ell \Bigl( y - G(u_\ell)  \Bigr).
		\end{equation}
	\end{enumerate}
	\STATE{\bf Output}: $u_1, u_2, \ldots, u_L .$
	\end{algorithmic}
\end{algorithm}
\FloatBarrier
When $\alpha_0 = 1$, the following linear-case result shows that \index{ExKI}ExKI reaches the minimizer of $\Jtp$ in one iteration. However, in contrast to \index{IExKF}IExKF, further iterations of \index{ExKI}ExKI  will typically no longer agree with the \index{optimization}minimizer of $\Jtp$.

\begin{proposition}\label{proposition:ILMDM}
	Suppose that $G(u) =  A  u$ is linear  and $\alpha_0 = 1$. Then the output of Algorithm \ref{ILMDM} satisfies
	$$ u_1= \arg\min_u \Jtp(u),$$
	where $\Jtp$ is the \index{Tikhonov-Phillips}Tikhonov-Phillips \index{objective}objective \eqref{eq:TPobjective}.
\end{proposition}
\begin{proof}
	The proof is identical to that of Proposition \ref{proposition:IExKF}, noting that in the linear case
	$u_{\ell +1} = u_\ell + K(y -   A  u_\ell).$
\end{proof}

\subsection{\index{Kalman inversion!Tikhonov extended}
Tikhonov Extended Kalman Inversion (TExKI)}\label{ssec:LMTP}
In this subsection we describe the application of the \index{Levenberg-Marquardt}Levenberg-Marquardt algorithm to the \index{Tikhonov-Phillips}Tikhonov-Phillips \index{objective}objective $\Jtp$. 

\begin{lemma}
	The \index{Levenberg-Marquardt}Levenberg-Marquardt method applied to the \index{Tikhonov-Phillips}Tikhonov-Phillips \index{objective}objective $\Jtp$ admits the following characterization:
	\begin{align*}
		u_{\ell +1} = u_\ell + K_\ell  \Bigl( z - h(u_\ell) \Bigr),
	\end{align*}
	where 
	$$K_\ell = \alpha_\ell \hat{C} H_\ell^\top ( \alpha_\ell H_\ell \hat{C} H_\ell^\top + Q )^{-1}, \quad \quad H_\ell = Dh(u_\ell).$$
\end{lemma} 
\begin{proof}
	Note that the increment $v_\ell$ is defined as the unconstrained minimizer of
	\begin{align}
		\Juctp(v) &=  \Jtpl(v)  + \frac{1}{2\alpha_\ell}|v|^2_{\hat{C}} \\
		& = \frac12 |z - h(u_\ell) -Dh(u_\ell)v |_Q^2 + \frac{1}{2\alpha_\ell}|v|^2_{\hat{C}}. \label{eq:LMTP_obj}
	\end{align}
 This has the form of Equation \eqref{eq:LMDM}, 
replacing $y$ with $z,$  $G$ with $h,$ and $\Gamma$ with $Q.$  
\end{proof}

The previous lemma motivates the following \index{Kalman inversion!Tikhonov extended}Tikhonov Extended Kalman Inversion (\index{TExKI}TExKI) algorithm. Discussion of how to choose the step-sizes adaptively can be found in Section \ref{sec:discussion}. 
\FloatBarrier
\begin{algorithm}
	\caption{\index{Kalman inversion!Tikhonov extended}Tikhonov Extended Kalman Inversion (\index{TExKI}TExKI) \label{ILMTP}}
	\begin{algorithmic}
	\STATE {\bf Input}: Initialization $u_0 = \hat{m}$,  rule for choosing the step-sizes  $\{\alpha_\ell\}_{\ell = 0}^{L-1}.$  \\
	\STATE For $\ell = 0, 1, \ldots, L-1$ do:
	\begin{enumerate}
		\item Set $K_\ell = \alpha_\ell \hat{C} H_\ell^\top (\alpha_\ell H_\ell \hat{C} H_\ell^\top + Q)^{-1}, \quad \quad H_\ell = Dh(u_\ell).$
		\item Set 
		\begin{equation}\label{eq:ILMTP}
			u_{\ell +1} = u_\ell + K_\ell \bigl( z - h(u_\ell)  \bigr).
		\end{equation}
	\end{enumerate}
	\STATE{\bf Output}: $u_1, u_2, \ldots, u_L .$
	\end{algorithmic}
\end{algorithm}
\FloatBarrier

When $\alpha_0 = 1$, the following linear-case result shows that \index{TExKI}TExKI reaches in one iteration the minimizer of a  $\Jtp$ objective appended with an additional regularization term.

\begin{proposition}\label{proposition:TExKI}
	Suppose that $G(u) =  A u$ is linear  and $\alpha_0 = 1$. Then the output of Algorithm \ref{ILMDM} satisfies
	$$ u_1= \arg\min_u \Bigl(\Jtp(u) + \frac12 |u  - \hat{m}|_{\hat C}^2\Bigr),$$
	where $\Jtp$ is the \index{Tikhonov-Phillips}Tikhonov-Phillips \index{objective}objective \eqref{eq:TPobjective}.
\end{proposition}
\begin{proof}
	Notice that setting $H = \begin{bmatrix}
	A\\ I
\end{bmatrix}$ we have 
	$$u_1 = \widehat{m} + \widehat{C} H^\top( H \widehat{C} H^\top + Q)^{-1} (z - H \widehat{m}).$$
 Lemma \ref{lemma:linear} then implies that $u_1$ minimizes
$$ \frac12 |z - Hu |^2_Q  + \frac12 | u - \widehat{m} |^2_{\hat C},$$
which implies the result.
\end{proof}
\begin{remark}
It is illustrative to compare Propositions \ref{proposition:IExKF}, \ref{proposition:ILMDM}, and \ref{proposition:TExKI}. These results show that in a linear setting: (i) \index{IExKF}IExKF reaches in one iteration the minimizer of $\Jtp,$  and that further iterates remain at the minimizer; (ii) \index{ExKI}ExKI reaches in one iteration the minimizer of $\Jtp$; and (iii) \index{TExKI}TExKI reaches in one iteration the minimizer of a $\Jtp$ \index{objective}objective appended with an additional regularization term. 
\end{remark}

\section{Ensemble Kalman Methods}
\label{ssec:Ensemble-based}
In this section we review three subfamilies of iterative methods that update an ensemble $\{ u_\ell^{(n)} \}_{n=1}^N$ employing Kalman-based formulae, where $\ell = 0, 1, \ldots$ denotes the iteration index and $N$ is a fixed ensemble size. Each ensemble member $u_\ell^{(n)}$ is updated by optimizing an \index{objective}objective defined using the current ensemble $\{ u_\ell^{(n)} \}_{n=1}^N$. The \index{optimization}optimization is performed without evaluating derivatives by invoking a \emph{statistical linearization}\index{statistical linearization} of a \index{Gauss-Newton}Gauss-Newton or \index{Levenberg-Marquardt}Levenberg-Marquardt algorithm. In analogy with the previous section, the three subfamilies of ensemble methods we consider differ in the choice of the \index{objective}objective and in the choice of the \index{optimization}optimization algorithm.

Given an ensemble $\{u_\ell^{(n)} \}_{n=1}^N$ we use the following notation for ensemble empirical means 
\begin{align*}
m_\ell &:= \frac{1}{N} \sum_{n=1}^N u_\ell^{(n)}, \quad \quad \bar{G}_\ell := \frac{1}{N} \sum_{n=1}^N G(u_\ell^{(n)}),
\end{align*}
and empirical covariances and cross-covariances
\begin{align*}
\hat{C}_\ell^{uu}  &:= \frac{1}{N} \sum_{n=1}^N (u_\ell^{(n)} - m_\ell)  (u_\ell^{(n)} - m_\ell)^\top,  \quad \quad  \\
\hat{C}_\ell^{uy} &:= \frac{1}{N} \sum_{n=1}^N \bigl( u_\ell^{(n)} - m_\ell  \bigr)   \bigl( G(u_\ell^{(n)}) - \bar{G}_\ell  \bigr)^\top,  \\
\hat{C}_\ell^{yy} &:= \frac{1}{N} \sum_{n=1}^N \bigl(G(u_\ell^{(n)}) - \bar{G}_\ell \bigr)  \bigl(G(u_\ell^{(n)}) - \bar{G}_\ell  \bigr)^\top.
\end{align*}
 Here and in what follows $(\hat{C}_\ell^{uu})^{-1}$ denotes the pseudoinverse of $\hat{C}_\ell^{uu}.$


The overarching theme that underlies the derivation of the ensemble methods studied in this section is the use of statistical linearization to avoid evaluation of derivatives. 
The idea behind statistical linearization\index{statistical linearization} is  
this: if  $G(u) =  A  u$ 
is linear, we have
\begin{equation*}
\hat{C}_\ell^{uy} = \hat{C}_\ell^{uu} A^\top.
\end{equation*}
Thus, if $\hat{C}_\ell^{uu}$ is invertible,
$A= (\hat{C}_\ell^{uy})^\top (\hat{C}_\ell^{uu})^{-1}$. 
Here and in what follows $(\hat{C}_\ell^{uu})^{-1}$ denotes the inverse 
of $\hat{C}_\ell^{uu}$ if this inverse exists, and the pseudoinverse otherwise. Noting that
$A$ is the derivative of $G(\cdot)$ in the linear case, this calculation
motivates the following \emph{approximation} in the general nonlinear case:
\begin{equation}\label{eq:stat_linearization}
DG(\uin) \approx G_\ell:=(\hat{C}_\ell^{uy})^\top (\hat{C}_\ell^{uu})^{-1}, 
\quad \quad n = 1,\dots,N,
\end{equation}
Note that \eqref{eq:stat_linearization} gives the same approximation of the derivative for every particle $n,$ and
indeed that it leads to an approximation that may be used at any point.

Other useful approximations follow from this. For example, note that
the exact gradient of $\Jdm(u)$ from \eqref{eq:TPobjective} is given by
$$D \Jdm(u)=DG(u)^\top \Gamma^{-1}\bigl(y-G(u)\bigr).$$ 
This suggests the approximation, for $G_\ell$ given 
by \eqref{eq:stat_linearization},
\begin{subequations}
\label{eq:SLgrad}
\begin{align}
D \Jdm(\uin) &\approx G_{\ell}^\top \Gamma^{-1}\bigl(y-G(\uin)\bigr),\\
& =(\hat{C}_\ell^{uu})^{-1} \hat{C}_\ell^{uy} \Gamma^{-1}\bigl(y-G(\uin)\bigr).
\end{align}
\end{subequations}

\subsection{\index{Kalman filter!iterative ensemble}Iterative Ensemble Kalman Filter with Statistical Linearization (IEnKF-SL)}\label{ssec:IenKF}
Given an ensemble $\{u_\ell^{(n)} \}_{n=1}^N$, consider the following \index{Gauss-Newton}Gauss-Newton update for each $n$:
\begin{equation}\label{eq:GNupdateensemble}
u_{\ell +1}^{(n)} = u_\ell^{(n)} + \alpha_\ell v_\ell^{(n)},
\end{equation}
where $\alpha_\ell>0$ is the step-size, and $\vin$ is the minimizer of the following (linearized) \index{Tikhonov-Phillips}Tikhonov-Phillips \index{objective}objective  (see \eqref{eq:GN_obj})
\begin{equation}\label{eq:IEKF_new_C}
\Jtpn(v) = \frac12 \bigl| y- G(u_\ell^{(n)} )  - G_\ell v \bigr|^2_\Gamma + \frac12 \bigl| \widehat{m} - \uin - v \bigr|^2_{\widehat{C}}.
\end{equation}
 It is important to appreciate that we adopt the 
statistical linearization\index{statistical linearization}  \eqref{eq:stat_linearization} in the above formulation. This couples
the different objective functions $\Jtpn$ indexed by $\ell.$ 
Applying Lemma \ref{lemma:linear}, the minimizer $\vin$ of
$\Jtpn$  can be calculated as
\begin{equation}\label{eq:IEKF_step1}
v_\ell^{(n)} = C_\ell \Bigl( G_\ell^\top \Gamma^{-1} \big( y - G(\uin) \big) + \widehat{C}^{-1} \big( \widehat{m} - \uin \big) \Bigr),
\end{equation}
or, in an equivalent form, 
\begin{equation}\label{eq:IEKF_step2}
v_\ell^{(n)} = K_\ell \big( y - G(\uin) \big) + (I - K_\ell G_\ell) (\widehat{m} - u_\ell^{(n)}), 
\end{equation}                      
where
\begin{align*}
C_\ell &= \big( G_\ell^\top \Gamma^{-1} G_\ell + \widehat{C}^{-1} \big)^{-1}, \\
K_\ell &= \widehat{C} G_\ell^\top (G_\ell \widehat{C} G_\ell^\top + \Gamma)^{-1}.
\end{align*}
Crucially each $v_\ell^{(n)}$ depends on all the $\{u_\ell^{(m)}\}_{m=1}^N.$

Combining \eqref{eq:GNupdateensemble} and \eqref{eq:IEKF_step2}  leads to the \index{Kalman filter!iterative ensemble}Iterative Ensemble Kalman Filter with Statistical Linearization (\index{IEnKF-SL}IEnKF-SL) algorithm. 
Discussion on how to choose the step-sizes adaptively can be found in Section \ref{sec:discussion}. 
\begin{algorithm}
\caption{\index{Kalman filter!iterative ensemble}Iterative Ensemble Kalman Filter with Statistical Linearization \index{IEnKF-SL} \label{itenKF}} 
\begin{algorithmic}
\STATE {\bf Input}: Initial ensemble $\{ u_0^{(n)}\}_{n=1}^N$ sampled from $\Nc(\widehat{m},\widehat{C})$,  rule for choosing the step-sizes  $\{\alpha_\ell\}_{\ell = 0}^{L-1}.$ \\
\STATE For $\ell = 0, 1, \ldots, L-1$ do:
\begin{enumerate}
\item Set $K_\ell = \widehat{C} G_\ell^\top (G_\ell \widehat{C} G_\ell^\top + \Gamma)^{-1}, \quad \quad G_\ell = (\hat{C}_\ell^{uy})^\top (\hat{C}_\ell^{uu})^{-1}.  $
\item Set
\begin{equation}\label{eq:IEKF_update}
u_{\ell +1}^{(n)} = \uin + \alpha_\ell  \Bigl( K_\ell \big( y - G(u_\ell^{(n)}) \big) + (I - K_\ell G_\ell)  \big( \widehat{m} - \uin \big)   \Bigr), \quad \quad 1 \le n \le N. 
\end{equation}
\end{enumerate}
\STATE{\bf Output}: Ensemble means $m_1, m_2, \ldots, m_L .$
\end{algorithmic}
\end{algorithm}
\FloatBarrier
Notice that \index{IEnKF-SL}IEnKF-SL is a natural ensemble-based version of the derivative-based \index{IExKF}IExKF Algorithm \ref{itexKF}  with update \eqref{eq:updateiexKF}. 
Other statistical linearizations\index{statistical linearization} 
and approximations of the \index{Gauss-Newton}Gauss-Newton scheme are possible.


\subsection{\index{Kalman inversion!ensemble}Ensemble Kalman Inversion with Statistical Linearization (EnKI-SL)}\label{ssec:EKI}
Given an ensemble $\{u_\ell^{(n)} \}_{n=1}^N$, consider the following \index{Levenberg-Marquardt}Levenberg-Marquardt update for each $n$:
\begin{equation}\label{eq:updateensembleLM}
u_{\ell +1}^{(n)} = u_\ell^{(n)} + v_\ell^{(n)},
\end{equation}
where $v_\ell^{(n)}$ is the minimizer of the following regularized (linearized) \index{data-misfit}data-misfit \index{objective}objective  (see \eqref{eq:LMDM}) 
\begin{equation}\label{eq:EKI_new_C}
\Jucdmn \hspace{-0.2cm}(v) = \frac12 \bigl| y - G(u_\ell^{(n)} )  - G_\ell v \bigr|^2_\Gamma  + \frac{1}{2\alpha_\ell} \bigl|v\bigr|^2_{\widehat{C}},
\end{equation}
and $\alpha_\ell>0$ will be regarded as a step-size. Notice that we adopt the 
statistical linearization\index{statistical linearization} \eqref{eq:stat_linearization} in the above formulation. Applying Lemma \ref{lemma:linear}, we can calculate the minimizer $v_\ell^{(n)}$ explicitly:
\begin{equation}\label{eq:EKI_step1}
v_\ell^{(n)} = (G_\ell^\top \Gamma^{-1} G_\ell + \alpha_\ell^{-1} \hat{C}^{-1})^{-1} G_\ell^\top \Gamma^{-1} \big( y - G(\uin) \big),
\end{equation}
or, in an equivalent form, 
\begin{equation}\label{eq:EKI_step2}
v_\ell^{(n)} = \hat{C} G_\ell^\top (G_\ell \hat{C} G_\ell^\top + \alpha_\ell^{-1} \Gamma)^{-1} \big( y - G(\uin) \big).
\end{equation}
 As in the preceding subsection,
each $v_\ell^{(n)}$ depends on all the $\{u_\ell^{(m)}\}_{m=1}^N.$ 
This leads to the \index{Kalman inversion!ensemble}Ensemble Kalman Inversion (\index{EnKI-SL}EnKI-SL) with Statistical Linearization method. Discussion of how to choose the step-sizes adaptively can be found in Section \ref{sec:discussion}. 
\begin{algorithm}
	\caption{\index{Kalman inversion!ensemble}Ensemble Kalman Inversion with Statistical Linearization\index{EnKI-SL}\label{algorithm:EKI_step}}
	\begin{algorithmic}
	\STATE {\bf Input}: Initial ensemble $\{ u_0^{(n)}\}_{n=1}^N$ sampled from $\Nc(\widehat{m},\widehat{C})$, rule for choosing the step-sizes  $\{\alpha_\ell\}_{\ell = 0}^{L-1}.$ \\ 
	\STATE For $\ell = 0, 1, \ldots, L-1$ do:
	\begin{enumerate}
		\item Set $K_\ell = \hat{C} G_\ell^\top (G_\ell \hat{C} G_\ell^\top + \alpha_\ell^{-1} \Gamma)^{-1}, \quad \quad G_\ell = (\hat{C}_\ell^{uy})^\top (\hat{C}_\ell^{uu})^{-1}.  $ 
		\item Set
		\begin{equation}\label{eq:EKI_update}
		u_{\ell +1}^{(n)} = u_\ell^{(n)}+ K_\ell \Bigl( y - G(u_\ell^{(n)})     \Bigr) ,\quad \quad 1 \le n \le N. 
		\end{equation}
	\end{enumerate}
	\STATE{\bf Output}: Ensemble means $m_1, m_2, \ldots, m_L.$
	\end{algorithmic}
\end{algorithm}
\FloatBarrier
Notice that \index{EnKI-SL}EnKI-SL is a natural ensemble-based version of the derivative-based \index{ExKI}ExKI Algorithm \ref{ILMDM}.

\subsection{\index{Kalman inversion!Tikhonov ensemble}Tikhonov Ensemble Kalman 
Inversion with Statistical Linearization (TEnKI-SL)}\label{ssec:TEKI}
Recall that we define
$$z := \begin{bmatrix}
y\\ \hat{m}
\end{bmatrix},  \quad \quad \quad 
h(u) := \begin{bmatrix}
G(u)\\ u
\end{bmatrix},
\quad \quad \quad
Q := 
\begin{bmatrix}
\Gamma & 0 \\ 0 & \hat{C}
\end{bmatrix}.
$$
Then, given an ensemble $\{u_\ell^{(n)} \}_{n=1}^N$, we can define 
$$\bar{h}_\ell := \frac{1}{N} \sum_{n=1}^N h(u_\ell^{(n)}),$$ and empirical covariances
\begin{align*}
\hat{C}_\ell^{zz} &:= \frac{1}{N} \sum_{n=1}^N \bigl( h(u_\ell^{(n)}) - \bar{h}_\ell  \bigr)   \bigl( h(u_\ell^{(n)}) - \bar{h}_\ell  \bigr)^\top,\\
\hat{C}_\ell^{uz} &:= \frac{1}{N} \sum_{n=1}^N \bigl( u_\ell^{(n)} - m_\ell  \bigr)   \bigl( h(u_\ell^{(n)}) - \bar{h}_\ell  \bigr)^\top.
\end{align*}
Furthermore, we define the statistical linearization\index{statistical linearization} $H_\ell$:
\begin{equation}
Dh(\uin)
\approx
(\hat{C}_\ell^{uz})^\top (\hat{C}_\ell^{uu})^{-1}
=: H_\ell.
\end{equation}
Notice that
\begin{equation*}
H_\ell = \begin{bmatrix}
G_\ell \\ I
\end{bmatrix},
\end{equation*}
with $G_\ell$ defined in \eqref{eq:stat_linearization}.

Given an ensemble $\{u_\ell^{(n)} \}_{n=1}^N$, consider the following \index{Levenberg-Marquardt}Levenberg-Marquardt update for each $n$:
\begin{equation*}
u_{\ell +1}^{(n)} = u_\ell^{(n)} + v_\ell^{(n)},
\end{equation*}
where $v_\ell^{(n)}$ is the minimizer of the following regularized (linearized) \index{Tikhonov-Phillips}Tikhonov-Phillips \index{objective}objective  (see \eqref{eq:LMTP_obj}) 
\begin{equation}\label{eq:TEKI_new_C}
\Juctpn \hspace{-0.2cm}(v) = \frac12 \bigl| z - h(u_\ell^{(n)} )  -  H_\ell v \bigr|^2_Q  + \frac{1}{2\alpha_\ell} \bigl|v\bigr|^2_{\hat{C}}, 
\end{equation}
and $\alpha_\ell>0$ will be regarded as a step-size. We can calculate the minimizer $v_\ell^{(n)}$ explicitly, applying Lemma \ref{lemma:linear}:
\begin{equation}\label{eq:TEKI_step1}
v_\ell^{(n)} = (H_\ell^\top Q^{-1} H_\ell + \alpha_\ell^{-1} \hat{C}^{-1})^{-1} H_\ell^\top Q^{-1} \big( z - h(\uin) \big),
\end{equation}
or, in an equivalent form, 
\begin{equation}\label{eq:TEKI_step2}
v_\ell^{(n)} = \hat{C} H_\ell^\top (G_\ell \hat{C} H_\ell^\top + \alpha_\ell^{-1} Q)^{-1} \big(  z - h(\uin) \big).
\end{equation}
Once again each $v_\ell^{(n)}$ depends on all the $\{u_\ell^{(m)}\}_{m=1}^N.$

This leads to \index{Kalman inversion!Tikhonov ensemble}Tikhonov Ensemble Kalman Inversion with Statistical Linearization (\index{TEnKI-SL}TEnKI-SL), described in Algorithm \ref{algorithm:TEKI_step}. Discussion on how to choose the  step-sizes adaptively can be found in Section \ref{sec:discussion}.
\begin{algorithm}
	\caption{\index{Kalman inversion!Tikhonov ensemble}Tikhonov Ensemble Kalman Inversion with Statistical Linearization \index{TEnKI-SL} \label{algorithm:TEKI_step}}
	\begin{algorithmic}
	\STATE {\bf Input}: Initial ensemble $\{ u_0^{(n)}\}_{n=1}^N$ sampled from $\Nc(m,P)$, rule for choosing the step-sizes  $\{\alpha_\ell\}_{\ell = 0}^{L-1}.$ \\ 
	\STATE For $\ell = 0, 1, \ldots, L-1$ do:
	\begin{enumerate}
		\item Set $K_\ell =\hat{C} H_\ell^\top (G_\ell \hat{C} H_\ell^\top + \alpha_\ell^{-1} Q)^{-1}, \quad \quad 
H_\ell = (\hat{C}_\ell^{uz})^\top (\hat{C}_\ell^{uu})^{-1}.$
		\item Set
		\begin{equation}\label{eq:TEKI_update}
		u_{\ell +1}^{(n)} = u_\ell^{(n)}+ K_\ell \Bigl( z  - h(u_\ell^{(n)})     \Bigr) ,\quad \quad 1 \le n \le N. 
		\end{equation}
	\end{enumerate}
	\STATE{\bf Output}: Ensemble means $m_1, m_2, \ldots, m_L.$
	\end{algorithmic}
\end{algorithm}
\FloatBarrier
Notice that \index{TEnKI-SL}TEnKI-SL is a natural ensemble-based version of the derivative-based \index{TExKI}TExKI Algorithm \ref{ILMTP}. 

\section{Discussion and Bibliography}\label{sec:discussion}
The presentation in this chapter follows the conceptual approach
to this subject overviewed and systematized 
in the paper \cite{chada2021iterative}: Kalman methods 
for \index{inverse problem}inverse problems are studied from the
optimization perspective, and classified in terms of the 
\index{objective}objective function they seek to minimize 
and the \index{nonlinear least-squares}nonlinear 
least-squares \index{optimization}optimization algorithm they are based on.  
For background on \index{nonlinear least-squares}nonlinear least-squares 
\index{optimization}optimization we refer to 
\cite{nocedal2006numerical,dennis1996numerical} where, in particular, 
a detailed discussion on the adaptive choice of the 
step-size parameters can be found;  note that the algorithms stated
in this chapter have been agnostic regarding the step-size choice strategy
as we have concentrated on the use of ideas from Kalman filtering
within \index{optimization}optimization.  Furthermore, following the presentation in \cite{chada2021iterative}, we have considered only 
\index{nonlinear least-squares}nonlinear least-squares 
\index{objective}objectives and quadratic \index{regularization}regularizers. 
However, ensemble Kalman methods for 
\index{inverse problem}inverse problems that use other objective
functions (or loss functions\index{loss}) and other regularizers are 
starting to emerge; in particular cross-entropy loss \cite{kovachki2019ensemble},
logistic loss \cite{SRPR21} and regularizers that promote sparsity
\cite{lee2020lp,schneider2020imposing,kim2022hierarchical} have all been considered.

There are a number of other ways in which Kalman filtering methods
may be used to study inverse problems.
The review article \cite{calvello22} emphasizes the Bayesian approach
to inversion and, in particular, shows how ideas from 
the sequential Monte Carlo \index{sequential Monte Carlo}(SMC)
\cite{del2006sequential} approach to Bayesian 
inversion\index{Bayesian!inversion} can be
adapted to the use of ensemble Kalman methods.
This possibility is highlighted in the paper
\cite{reich2011dynamical}, which is focused on 
sequential data assimilation; note, however, that the
analysis step \eqref{eq:pna} in sequential data assimilation requires
solution of a Bayesian inverse problem and thus the ideas in that paper 
are relevant for inverse problems in general, beyond data
assimilation. The reader interested in the use of SMC
for inverse problems is directed to the papers \cite{kantas2014sequential,beskos2015sequential} and the references therein; the
former paper demonstrates use of the methodology for an inverse
problem arising from the Navier-Stokes equation, and the latter
paper contains a simple proof of convergence of the particle
filter in the context of SMC for inverse problems, 
following the analysis in \cite{rebeschini2015can} for particle
filters in sequential data assimilation.

Another class of methods for inverse problems, which may be
applied in both the optimization and Bayesian approaches,
revolves around the idea of preconditioned gradient descent in 
ensemble Kalman methods for inversion;
in particular, use of the pre-conditioned gradient
\begin{equation*}
\hat{C}_\ell^{uu} D \Jdm(\uin) \approx
\hat{C}_\ell^{uy} \Gamma^{-1}\bigl(y-G(\uin)\bigr).
\end{equation*}
which follows from \eqref{eq:SLgrad}.
This leads to iterative optimization methods
\cite{iglesias2013ensemble,schillings2017analysis,schillings2018convergence},
based on gradient descent, and to Bayesian sampling methods
\cite{garbuno2019interacting,garbuno2020affine}. A key feature of the
preconditioned gradient is that it leads to algorithms
which are affine invariant \cite{affine}, and hence to
convergence rates which are uniform across wide classes of problems;
see the review article \cite{calvello22} for further discussion.

Finally, we note that Kalman methods are based on a \index{Gaussian!approximation}Gaussian approximation,
and hence on matching first and second order moments when studying Bayesian
inversion. Therefore, it is natural to study Kalman methods for inverse
problems which are applied to (possibly stochastic) dynamical systems
whose long-term properties exactly solve the \index{optimization}optimization or Bayesian
approach to inversion in the \index{linear-Gaussian setting}linear-Gaussian setting; this idea is
developed in \cite{huang2022iterated,huang2022efficient}.

The gradient-based \index{IExKF}IExKF algorithm was developed in the control theory literature \cite{jazwinski2007stochastic} without reference to the \index{Gauss-Newton}Gauss-Newton \index{optimization}optimization method; the 
 correspondence  between both methods was established in \cite{bell1993iterated}.   Ensemble Kalman methods were also first introduced as \index{filtering}filtering schemes for sequential \index{data assimilation}data assimilation, as described in Chapter \ref{lecture10}. Their use for \index{state}state and parameter estimation and \index{inverse problem}inverse problems was further developed in \cite{anderson2001ensemble,lorentzen2001underbalanced,skjervheim2011ensemble}. The idea of \emph{iterating} these methods was considered in \cite{chen2012ensemble,emerick2013ensemble,reich2011dynamical}. Ensemble Kalman methods are now popular in both \index{inverse problem}inverse problems and \index{data assimilation}data assimilation; they have also shown some potential in machine learning applications \cite{haber2018never,guth2020ensemble,kovachki2019ensemble}. There are two main computational benefits in updating an ensemble of candidate reconstructions rather than a single estimate. First, the ensemble update can be performed without evaluating derivatives of $G,$ effectively approximating them using 
statistical linearization.\index{statistical linearization} This is important in applications where computing derivatives of $G$ is expensive, or where the map $G$ needs to be treated as a black-box. Second, the use of empirical rather than model covariances can significantly reduce the computational cost whenever the ensemble size $N$ is smaller than the dimension $d$ of the unknown $u$. Another advantage of the ensemble approach is that, for problems that are not strongly nonlinear, the spread of the ensemble may contain meaningful information on the uncertainty in the reconstruction. Statistical linearization has also been used within unscented Kalman methods, see \cite{ungarala2012iterated,huang2021unscented1,huang2021unscented2}. 

In this chapter we have considered three families of ensemble algorithms characterized by a choice of \index{objective}objective function and \index{optimization}optimization algorithm: (i)  \index{Tikhonov-Phillips}Tikhonov-Phillips and \index{Gauss-Newton}Gauss-Newton; (ii) \index{data-misfit}data-misfit and \index{Levenberg-Marquardt}Levenberg-Marquardt; and (iii) \index{Tikhonov-Phillips}Tikhonov-Phillips and \index{Levenberg-Marquardt}Levenberg-Marquardt. Each family of ensemble Kalman methods stems from a choice of \index{objective}objective and a derivative-based \index{optimization}optimization scheme that is approximated with the ensemble. There is substantial freedom as to how to use the ensemble to approximate a derivative-based method. We have focused on \index{likelihood!randomized maximum}randomized-maximum likelihood implementations \cite{gu2007iterative,kelly2014well}, but square-root  approaches \cite{anderson2001ensemble,tippett2003ensemble} can also be considered.

Algorithms in the first family were first introduced in petroleum engineering and the geophysical sciences \cite{aanonsen2009ensemble,chen2012ensemble,emerick2013ensemble,gu2007iterative,li2007iterative,reynolds2006iterative} and were inspired by iterative, gradient-based, extended Kalman filters \cite{bell1994iterated,bell1993iterated,jazwinski2007stochastic}. More challenging problems with strongly nonlinear \index{dynamics}dynamics are considered in \cite{sakov2012iterative}. In this chapter we have presented the Iterative \index{IEnKF-SL}IEnKF-SL as a prototypical example of an algorithm that belongs to this family.
\index{IEnKF-SL}IEnKF-SL was introduced in \cite{chada2021iterative} as a slight modification of the \index{Kalman filter!iterative ensemble}iterative ensemble Kalman algorithm proposed in \cite{ungarala2012iterated}. One of the earliest applications of iterative ensemble Kalman methods for inversion in the petroleum engineering literature was proposed in \cite{reynolds2006iterative}, which considered the alternative characterization of the \index{Gauss-Newton}Gauss-Newton update \eqref{eq:IEKF_step1}. Moreover, instead of using a different preconditioner $C_\ell$ for each step, \cite{reynolds2006iterative} used a fixed preconditioner.
  
  Algorithms in the second family were introduced in the applied mathematics literature \cite{iglesias2016regularizing,iglesias2013ensemble} building on ideas from classical \index{inverse problem}inverse problems \cite{hanke1997regularizing}. Recent theoretical work has focused on developing continuous-time and mean-field limits, as well as various convergence results \cite{blomker2019well,blomker2018strongly,chada2019convergence,herty2018kinetic,ding2019ensemble,kovachki2019ensemble,schillings2017analysis}. Methodological extensions based on \index{Bayesian}Bayesian hierarchical techniques were introduced in \cite{chada2018analysis,chada2018parameterizations} and the incorporation of constraints has been investigated in \cite{albers2019ensemble,chada2019incorporation}. In this chapter we use  EnKI-SL as a prototypical example of an algorithm  that belongs to this subfamily. Its connection with the Ensemble Kalman Inversion algorithm from \cite{iglesias2013ensemble} is discussed in \cite{chada2021iterative}.
  
  The third family, which has emerged more recently, combines explicit \index{regularization}regularization through the \index{Tikhonov-Phillips}Tikhonov-Phillips \index{objective}objective and an implicitly regularizing \index{optimization}optimization scheme \cite{chada2019tikhonov,chada2019convergence}. In this chapter we use  TEnKI-SL as a prototypical example of an algorithm  that belongs to this subfamily.

Our presentation has focused on the derivation of Kalman formulae for the search direction update of \index{Gauss-Newton}Gauss-Newton and \index{Levenberg-Marquardt}Levenberg-Marquardt algorithms and their ensemble approximations. All the algorithms studied in this chapter require specifying appropriate  step-size parameters that determine the size of the updates along the search direction. 
 For gradient-based methods, there is abundant literature on the adaptive choice of  step-sizes \cite{nocedal2006numerical}.   \index{Gauss-Newton}Gauss-Newton methods can be shown to converge when the step-sizes are chosen according to Armijo or Wolfe conditions; the line search is often performed with a backtracking strategy \cite{nocedal2006numerical,dennis1996numerical}. When using \index{Levenberg-Marquardt}Levenberg-Marquardt schemes for \index{inverse problem}inverse problems, it is important to ensure that the  step-sizes,  as well as the stopping criteria, provide sufficient implicit \index{regularization}regularization to alleviate the \index{ill-posed}ill-posedness of \index{inverse problem}inverse problems \cite{hanke1997regularizing}. 
  For ensemble Kalman methods, the use and analysis of adaptive  step-sizes  is a topic of current research \cite{chada2019convergence,iglesias2021adaptive}. In practice, ensemble methods are often run with short  step-sizes,  in which case the algorithms may be interpreted as being defined by discretization of (stochastic) differential equations, see e.g. \cite{reich2011dynamical,chada2021iterative,schillings2017analysis}. Finally, we point out that the original description of some of the algorithms studied in this paper, e.g. \index{EnKI}EnKI  and \index{TEnKI}TEnKI in \cite{iglesias2013ensemble,chada2019tikhonov}, do not discuss the inclusion of  step-size  parameters. This would correspond to setting $\alpha_\ell =1$ for all $\ell \ge 0$ in our terminology.

\newpage
\bibliographystyle{abbrvnat} 

\bibliography{references}

\newpage
\printindex
\end{document}